\documentclass[a4paper]{article}
\usepackage{a4wide}
\usepackage[greek,UKenglish]{babel}
\usepackage{graphicx,subfigure,caption}
\graphicspath{{./}{figures/}}
\usepackage[colorlinks=true,linkcolor=blue,citecolor=red]{hyperref}
\usepackage{xcolor}
\usepackage{amssymb,amsthm,empheq,bbold,mathtools}
\usepackage[inline]{enumitem}
\usepackage{calrsfs}
	\DeclareMathAlphabet{\pazocal}{OMS}{zplm}{m}{n}
\usepackage{authblk}
\usepackage[ruled,vlined,linesnumbered]{algorithm2e}
\usepackage{accents}

\DeclarePairedDelimiter\abs{\lvert}{\rvert}
\DeclarePairedDelimiter\ave{\langle}{\rangle}
\newcommand{\bA}{\boldsymbol{A}}
\newcommand{\ddiv}{\operatorname{div}}
\newcommand{\cE}{\pazocal{E}}
\newcommand{\bbf}{\boldsymbol{f}}
\newcommand{\bbfhat}{\boldsymbol{\hat{f}}}
\newcommand{\bF}{\boldsymbol{F}}
\newcommand{\bg}{\boldsymbol{g}}

\newcommand{\bG}{\boldsymbol{G}}
\newcommand{\bh}{\boldsymbol{h}}
\newcommand{\bhhat}{\boldsymbol{\hat{h}}}
\newcommand{\bI}{\boldsymbol{I}}
\newcommand{\cI}{\pazocal{I}}

\newcommand{\cM}{\mathfrak{M}}
\newcommand{\cpM}{\pazocal{M}}
\newcommand{\N}{\mathbb{N}}
\newcommand{\norm}[2]{\Vert#1\Vert_{#2}}
\newcommand{\norminf}[1]{\Vert#1\Vert_\infty}
\newcommand{\pr}[1]{{}^\prime\!#1}
\newcommand{\Prob}[1]{\operatorname{Prob}{(#1)}}
\newcommand{\bP}{\boldsymbol{P}}
\newcommand{\cP}{\mathcal{P}}
\newcommand{\pP}{\pazocal{P}}
\newcommand{\cQ}{\pazocal{Q}}
\newcommand{\R}{\mathbb{R}}
\newcommand{\brho}{\boldsymbol{\rho}}
\newcommand{\sS}{\mathbb{S}}
\newcommand{\supp}[1]{\operatorname{supp}#1}
\newcommand{\trinorm}[2]{\vert\kern-0.25ex\vert\kern-0.25ex\vert#1\vert\kern-0.25ex\vert\kern-0.25ex\vert_{#2}}
\newcommand{\ubar}[1]{\underaccent{\bar}{#1}}
\newcommand{\cV}{\pazocal{V}}
\newcommand{\Var}[1]{\operatorname{Var}(#1)}

\newcommand \commentout[1] {}

\newtheorem{assumption}{Assumption}[section]
\newtheorem{proposition}[assumption]{Proposition}
\newtheorem{theorem}[assumption]{Theorem}
\newtheorem{lemma}[assumption]{Lemma}

\theoremstyle{remark}\newtheorem{remark}[assumption]{Remark}
\theoremstyle{remark}\newtheorem{example}[assumption]{Example}
\allowdisplaybreaks

\title{Essentials of the kinetic theory of multi-agent systems}
\author{Nadia Loy}
\author{Andrea Tosin}
\affil{{\small Department of Mathematical Sciences ``G. L. Lagrange'' \\ Politecnico di Torino, Italy}}
\date{}
		
\begin{document}
\maketitle
	
\begin{abstract}
In this paper, we present a critical collection of essential mathematical tools and techniques for the analysis of Boltzmann-type kinetic equations, which in recent years have established themselves as a flexible and powerful paradigm to model interacting multi-agent systems. We consider, in particular, scalar equations implementing linear symmetric interaction rules, for which we develop the theory of well-posedness, trend to equilibrium, and Fokker--Planck asymptotics by relying extensively on Fourier methods. We also outline the basics of Monte Carlo algorithms for the numerical solution of such equations. Finally, we elaborate the theory further for Boltzmann-type equations on graphs, a recent generalisation of the standard setting motivated by the modelling of networked multi-agent systems.

\medskip


\medskip

\end{abstract}

\tableofcontents

\section{Introduction}
\label{sect:intro}
Towards the end of the 19th century, the Austrian physicist \textbf{Ludwig Boltzmann} (Vienna, 1844 -- Duino, 1906) formulated the celebrated integro-differential equation that nowadays bears his name~\cite{boltzmann1970CHAPTER} as a part of a scientific activity which imparted a significant momentum to the development of \textit{statistical mechanics}. Boltzmann's goal was to explain the complex macroscopic concepts of thermodynamics starting from the elementary physics of the microscopic collisions among gas molecules, thereby elucidating how macroscopic measurable quantities, such as e.g., the bulk velocity, the temperature, and the energy of a gas, emerge from certain microscopic trends fluctuating around an average behaviour. It was the dawn of the \textit{kinetic\footnote{The adjective \textit{kinetic} comes from the ancient Greek noun \textgreek{k\'{i}nhsis} (\textit{kin\={e}sis}) meaning ``movement''. It indicates something which is concerned with movement.} theory} of gases as the historically first explicit implementation of the ideas of statistical mechanics. Indeed, by relying heavily on statistical arguments and probability theory, Boltzmann envisaged a mathematical-physical description in which microscopic gas molecules could be replaced by the \textit{statistical distribution} of their positions and velocities. This resulted in a mathematical model consisting of only one integro-differential equation for that distribution instead of a system of $O(10^{24})$ coupled ordinary differential equations, one for each molecule of the gas (cf. the Avogadro's number).

The Boltzmann equation writes
\begin{equation}
    \partial_tf+v\cdot\nabla_xf
        =\frac{1}{4\pi}\int_{\R^3}\int_{\sS^2}B((v_\ast-v)\cdot n)
            \bigl(f(x,v',t)f(x,v_\ast',t)-f(x,v,t)f(x,v_\ast,t)\bigr)\,dn\,dv_\ast,
    \label{eq:Boltz}
\end{equation}
where $x,\,v\in\R^3$ are the position and the velocity, respectively, of a \textit{generic representative} molecule of the gas and $f=f(x,v,t)$ is their joint statistical distribution function at time $t\in\R_+$.

The left-hand side of~\eqref{eq:Boltz}, where $\nabla_x$ stands for the gradient with respect to the variable $x$, is a linear advection operator describing the free motion with constant velocity of the gas molecules in the absence of mutual collisions. Notice indeed that if we set momentarily the right-hand side to zero the equation reduces to
$$ \partial_tf+v\cdot\nabla_xf=0, $$
whose solution is
$$ f(x,v,t)=f_0(x-vt,v), $$
i.e. a rigid translation in space of the initial distribution function $f_0$. This is a consequence of the fact that a gas molecule travelling freely from an initial point $x_0$ with constant velocity $v$ follows the path $x(t)=x_0+vt$.

The right-hand side of~\eqref{eq:Boltz} is instead a bilinear integro-differential operator, termed the \textit{collision operator}, which describes the average statistical variation of the molecule velocities due to the molecular collisions. There, $v',\,v_\ast'\in\R^3$ are the \textit{post-collisional} velocities of any two molecules colliding with \textit{pre-collisional} velocities $v,\,v_\ast\in\R^3$. Assuming elastic collisions, from elementary physics it is well known that momentum and kinetic energy of the colliding molecules are conserved. If all molecules have the same mass, this leads to the conditions
\begin{subequations}
    \begin{align}
        v'+v_\ast' &= v+v_\ast \label{eq:Boltz.momentum} \\
        \abs{v'}^2+\abs{v_\ast'}^2 &= \abs{v}^2+\abs{v_\ast}^2, \label{eq:Boltz.energy}
    \end{align}
    \label{eq:Boltz.momentum_energy}
\end{subequations}
whence
\begin{equation}
    v'=v+[(v_\ast-v)\cdot n]n, \qquad
        v_\ast'=v_\ast+[(v-v_\ast)\cdot n]n,
    \label{eq:Boltz.coll_rules}
\end{equation}
where $n\in\sS^2\subset\R^3$ is a unit vector pointing in the direction of the collision, i.e. the direction connecting the centres of the colliding molecules, and $\cdot$ denotes the inner product in $\R^3$. The relationships~\eqref{eq:Boltz.coll_rules} allow one to express the post-collisional velocities in~\eqref{eq:Boltz} as functions of the pre-collisional ones. For this reason, they are called \textit{collision rules}. Moreover, the term $B((v_\ast-v)\cdot n)$ in~\eqref{eq:Boltz} is the \textit{collision kernel}, a quantity which accounts for further mechanical features of the molecular collisions which can affect the rate at which molecules collide. A typical choice for the function $B:\R\to\R_+$ is $B(\nu)=\abs{\nu}$, which entails the following expression of the collision kernel:
\begin{equation}
    B((v_\ast-v)\cdot n)=\abs{(v_\ast-v)\cdot n}.
    \label{eq:Boltz.B}
\end{equation}
The physical meaning is that the more the relative pre-collisional velocity $v_\ast-v$ is oriented in the direction $n$ of the collision the more frequent, or in a sense ``probable'', the collision.

In Section~\ref{sect:Boltz.deriv} we shall consider in detail the construction of the collision operator of the Boltzmann equation, starting from the first principles~\eqref{eq:Boltz.momentum_energy} of molecular collisions. For the moment, we observe that~\eqref{eq:Boltz} expresses a clear separation of the effects leading to the variation of velocity and position of the molecules: on one hand, the velocity changes in consequence of the collisions (right-hand side), which do not entail a change in the position; on the other hand, the position changes in consequence of the free transport (left-hand side), which does not entail a change in the velocity.

The distribution function $f$, which in principle can be obtained from~\eqref{eq:Boltz} complemented with an initial condition $f_0$, allows one to compute average quantities, relevant for a macroscopic description of the gas, as \textit{statistical moments} with respect to the velocity. This ideally completes Boltzmann's programme to recover the complex (thermo)dynamical phenomena of gases from the fundamental collisions among the molecules. The main macroscopic quantities usually considered are the \textit{density} $\rho$ of the gas, the \textit{bulk velocity} $u$, the \textit{total energy} $E$, and the \textit{internal energy} $e$ in a point $x\in\R^3$ at time $t\geq 0$:
\begin{align}
    \begin{aligned}[c]
        & \rho(x,t):=\int_{\R^3}f(x,v,t)\,dv, & & u(x,t):=\frac{1}{\rho(x,t)}\int_{\R^3}vf(x,v,t)\,dv, \\
        & E(x,t):=\frac{1}{\rho(x,t)}\int_{\R^3}\abs{v}^2f(x,v,t)\,dv, & & e(x,t):=\frac{1}{\rho(x,t)}\int_{\R^3}\abs{v-u(x,t)}^2f(x,v,t)\,dv.
    \end{aligned}
    \label{eq:Boltz.moments}
\end{align}
Notice that the bulk velocity and the total and internal energies are linked by the relationship $E=\abs{u}^2+e$. Another macroscopic quantity of interest is the \textit{temperature} of the gas:
$$ \theta(x,t):=\frac{1}{3}e(x,t), $$
which is linked to the bulk velocity and the total energy by the relationship $E=\abs{u}^2+3\theta$. As a matter of fact, computing these macroscopic quantities out of $f$ is typically unfeasible, because it would require to solve the Boltzmann equation~\eqref{eq:Boltz}, which is not that friendly as far as explicit solutions are concerned. Therefore, many research efforts have been devoted to obtaining evolution equations directly for the quantities~\eqref{eq:Boltz.moments} by considering proper averages of~\eqref{eq:Boltz}, possibly in suitable limit regimes. This is the problem of the \textit{hydrodynamic limits} of the Boltzmann equation. Classical examples of macroscopic models that can be obtained as hydrodynamic limits of the Boltzmann equation are the Euler equations for an ideal gas and the Navier--Stokes equations for Newtonian fluids.

The mathematical theory of the Boltzmann equation is a lively field of study, as documented by the continuously increasing number of contributions in the pertinent literature. Without even attempting a review, here we confine ourselves to mentioning some classical references, such as~\cite{cercignani1988BOOK,cercignani1994BOOK,perthame2004BAMS,villani2002CHAPTER}, together with a couple of others about the parallel theory of the so-called \textit{discrete} Boltzmann equation~\cite{gatignol1975BOOK,toscani1989CMP}, in which the velocity ranges in a discrete set of selected values rather than continuously in $\R^3$. A reference specifically devoted to the hydrodynamic limits is instead~\cite{saint-raymond2009BOOK}.

Formally, the leading ideas of the Boltzmann's kinetic theory of gases are applicable to model any system, possibly different from a gas, composed by \textit{interacting} elements that can be assimilated to \textit{particles}, i.e. the generalisation of gas molecules. The key point is that these particles be \textit{indistinguishable} and follow universal interaction rules.

One of the very first attempts in this direction was due to \textbf{Ilya Prigogine} (Moscow, 1907 -- Brussels, 2003), mostly known for his work on complex systems and Nobel Prize in Chemistry in 1977, who, starting from the early Sixties, proposed a \textit{Boltzmann-type approach} to car traffic~\cite{prigogine1960OR,prigogine1971BOOK}. Prigogine's idea was to identify cars along a road as particles, whose microscopic state is given by the position $x\in\R$ and speed $v\in\R_+$. Notice that, unlike gases, here the microscopic variables are one-dimensional, because only car movements in the longitudinal direction of the road are taken into account, and that the speed is non-negative, because only a one-directional lane of the road is considered. The physical collision rules~\eqref{eq:Boltz.coll_rules} are replaced by empirical principles of acceleration and deceleration of a car when it interacts with a faster or a slower car ahead. On the whole, if $f=f(x,v,t)$ denotes the joint distribution of the pair $(x,\,v)\in\R\times\R_+$ at time $t\geq 0$ Prigogine kinetic equation in the unknown $f$ is
$$ \partial_tf+v\partial_xf=(1-P)\int_0^{+\infty}(v_\ast-v)f(x,v,t)f(x,v_\ast,t)\,dv_\ast+\frac{f_0-f}{T}. $$
The left-hand side is the one-dimensional counterpart of the advection operator appearing in the Boltzmann equation~\eqref{eq:Boltz}. The first term on the right-hand side is instead a ``collision'' operator accounting for the mean variation of the speed of the cars produced by the afore-mentioned acceleration and deceleration dynamics. There, $(1-P)(v_\ast-v)$ plays the role of the ``collision'' kernel and, in particular, the coefficient $P\in [0,\,1]$ is the probability of overtaking. The second term on the right-hand side expresses a \textit{relaxation}, with characteristic relaxation time $T>0$, of the distribution function $f$ towards a prescribed \textit{desired speed distribution} $f_0=f_0(x,v)$. This term models the natural inclination of the drivers to adapt locally the speed of their cars to a desired one in the absence of the disturbance caused by the interactions with other cars. It does not have a direct equivalent in the Boltzmann equation~\eqref{eq:Boltz} but reminds closely of the so-called \textit{BGK approximation} of the Boltzmann collision operator, whose seminal idea was introduced in~\cite{bhatnagar1954PR}\footnote{The acronym BGK is formed by the initials of the surnames of the authors of~\cite{bhatnagar1954PR}.} and further developed from then on, see~\cite{puppo2019RMUP}. Such an approximation consists in replacing the right-hand side of~\eqref{eq:Boltz} with a term proportional to $M-f$, where $M=M(x,v)$ is a \textit{local equilibrium distribution}, termed the \textit{Maxwellian distribution}, which makes the Boltzmann collision operator vanish. Because of this property, the Maxwellian distribution depicts the local \textit{statistical equilibrium} of a gas, i.e. a situation in which in a certain point $x\in\R^3$ the collisions among the gas molecules do not produce anymore changes in the statistical distribution of the velocity. As asserted by the celebrated Boltzmann's \textit{H-theorem}, cf. e.g.,~\cite{cercignani1988BOOK}, the collision dynamics described by the Boltzmann collision operator lead the distribution function $f$ to relax locally on the Maxwellian $M$. Therefore, the BGK approximation of the Boltzmann collision operator can be seen as a way to reproduce this very same trend by means of a much simpler mathematical term. In the case of the Prigogine equation, the relaxation towards the desired speed distribution $f_0$ is instead postulated as an additional trend besides the one dictated by car interactions.

Prigogine's work on car traffic paved the way to the application of the ideas and methods of the kinetic theory of gases to contexts very distant from the original one. Initially, some other contributions were given still in the realm of car traffic, among which we recall in particular~\cite{klar1997JSP,paveri1975TR}. Lately, in the early 2000s \textbf{Giuseppe Toscani} and coworkers initiated the systematic development of a mathematical theory based on \textit{Boltzmann-type equations} for \textit{interacting multi-agent systems}~\cite{pareschi2013BOOK}, focussing on applications motivated by econophysics and sociophysics such as the redistribution of wealth~\cite{cordier2005JSP} and the formation of opinions~\cite{toscani2006CMS} in human societies.

Toscani's theory concerns mainly one-dimensional \textit{homogeneous} Boltzmann-type models, in which the distribution function $f$ depends on only one scalar variable $v$, representing the microscopic state -- possibly not the speed -- of the agents of the system, and on time: $f=f(v,t)$. In general, $v$ belongs to a set $I\subseteq\R$, which might not coincide with the whole real line. These models are said to be of Boltzmann type because they are formulated by means of integro-differential equations mimicking the structure of the collision operator of the Boltzmann equation~\eqref{eq:Boltz}. Moreover, they are homogeneous because they do not feature a dependence on the space variable and consequently the equations do not contain transport terms in space. Instead, they describe pure binary interaction dynamics responsible for the variation in time of the distribution of the microscopic state $v$.

A prototypical form of such equations is
\begin{equation}
    \partial_tf=\int_I\ave*{\frac{B(\pr{v},\pr{v}_\ast)}{\abs{J}}f(\pr{v},t)f(\pr{v}_\ast,t)
        -B(v,v_\ast)f(v,t)f(v_\ast,t)}\,dv_\ast,
    \label{eq:Boltz.Toscani}
\end{equation}
where the right-hand side is the Boltzmann-type ``collision'' operator. In Section~\ref{sect:formal_derivation} we shall examine closely the derivation of kinetic equations like~\eqref{eq:Boltz.Toscani}. Here, we mention instead analogies and differences of~\eqref{eq:Boltz.Toscani} with respect to the homogeneous version of~\eqref{eq:Boltz}. First, in~\eqref{eq:Boltz.Toscani} $\pr{v}$, $\pr{v}_\ast$ denote the \textit{pre-interaction} states of the interacting agents, which generate the \textit{post-interaction} states $v$, $v_\ast$. Comparing with~\eqref{eq:Boltz}, we notice that there we find instead the post-collisional velocities $v'$, $v_\ast'$ generated by the pre-collisional velocties $v$, $v_\ast$. The reason is that the collision rules~\eqref{eq:Boltz.coll_rules} are \textit{reversible}, meaning that if  pre-collisional and post-collisional velocities are exchanged the collision rules remain the same; or, in other words, that $\pr{v}=v'$ and $\pr{v}_\ast=v_\ast'$. Hence, the collision operator in~\eqref{eq:Boltz} could be rewritten using in fact $\pr{v}$, $\pr{v}_\ast$ in place of $v'$, $v_\ast'$. This shall be formally clearer in Section~\ref{sect:Boltz.deriv} but for the moment we accept it intuitively and observe consequently that the distribution functions appearing in the collision operators in~\eqref{eq:Boltz} and~\eqref{eq:Boltz.Toscani} are not that different as they could seem at first glance. Next, in~\eqref{eq:Boltz.Toscani} the coefficient $J$ is the Jacobian determinant of the transformation from the pre-interaction to the post-interaction states. In~\eqref{eq:Boltz} this term is apparently missing because, owing to the afore-mentioned reversibility, the collision rules~\eqref{eq:Boltz.coll_rules} have unitary Jacobian determinant. Also, in~\eqref{eq:Boltz.Toscani} the ``collision'' kernel $B$ is not factored out like in~\eqref{eq:Boltz}. The reason is that the specific form~\eqref{eq:Boltz.B} together with the particular collision rules~\eqref{eq:Boltz.coll_rules} entails $B((\pr{v}_\ast-\pr{v})\cdot n)=B((v_\ast'-v')\cdot n)=B((v_\ast-v)\cdot n)$ (technical details again deferred to Section~\ref{sect:Boltz.deriv}), which allows one to factor $B$ out in~\eqref{eq:Boltz}, whereas this is not true in general for interaction rules and a collision kernel different from~\eqref{eq:Boltz.coll_rules} and~\eqref{eq:Boltz.B}, respectively. The corresponding generalisation is precisely that indicated in~\eqref{eq:Boltz.Toscani}. Finally, in~\eqref{eq:Boltz.Toscani} the notation $\ave{\cdot}$ stands for the expectation with respect to possibly \textit{random parameters} contained in the ``collision'' rules. Indeed, the latter might not be fully deterministic, especially when they are concerned with the human behaviour. If $\eta\in\pazocal{B}\subseteq\R$ is one such random parameter with law $h=h(\eta):\pazocal{B}\to\R_+$ then
$$ \ave{\cdot}:=\int_\pazocal{B}(\cdot)h(\eta)\,d\eta. $$
Although not immediately apparent, also the Boltzmann collision operator contains something similar. Indeed, the integral $\frac{1}{4\pi}\int_{\sS^2}(\cdot)\,dn$ can be understood as the average of the molecular collisions with respect to all possible directions of collision, which are uniformly distributed on $\sS^2$ if no preferential direction of collision exists. It is worth mentioning that a first stochastic interpretation of the collision process underlying the Boltzmann equation is due to Kac~\cite{kac1956CHAPTER}, who considered both the direction and the time of collision of pairs of molecules as random variables. The ``master equation'' of such a process is the celebrated \textit{Kac model}, which can be regarded as a simplified version of~\eqref{eq:Boltz}.

Equation~\eqref{eq:Boltz.Toscani}, complemented with the specification of the set $I$ and of the rules describing the interactions among the agents, constitutes a flexible and powerful paradigm for a rigorous mathematical formalisation of models of particle-like phenomena, such as those recalled above, which are possibly not (yet) based on consolidated physical theories. Typically, their interaction rules are indeed postulated heuristically, then~\eqref{eq:Boltz.Toscani} provides a sound mathematical framework where to set non-heuristic theoretical investigations. For instance, Toscani and coworkers developed a refined theory on the formation of wealth distribution curves, which formalises qualitatively and explains quantitatively the empirical observations made at the beginning of the $20$th century by the economist Vilfredo Pareto (Paris, 1848 -- C\'{e}ligny, 1923) about the inequalities in the wealth distribution of western societies. See~\cite{duering2009RMUP,matthes2008JSP}.

Just like the Boltzmann equation is not the only model of the statistical mechanics of particle systems so Boltzmann-type equations are not the only option to describe the interaction dynamics of multi-agent systems aggregately. The Boltzmann equation assumes implicitly \textit{short-range} interactions among gas molecules, which need to be in contact to collide. Therefore, it is in general not suited to model e.g., particle systems characterised by \textit{collisionless long-range} interactions such as those taking place in a gas of charged particles, viz. a plasma. In the mid-20th century, the Russian physicist \textbf{Anatoly Vlasov} (Balashov, 1908 -- Moscow, 1975) proposed a kinetic equation, which nowadays bears his name, in which the Boltzmann collision operator is replaced by a term accounting for the self-consistent \textit{collective} force field generated by the charged particles~\cite{vlasov1945URMSU}. If $f=f(x,v,t)$ is, like in~\eqref{eq:Boltz}, the distribution function of the position $x\in\R^3$ and the velocity $v\in\R^3$ of a generic plasma particle, the Vlasov equation reads
\begin{equation}
    \partial_tf+v\cdot\nabla_xf+\frac{1}{m}\ddiv_v{(Ff)}=0,
    \label{eq:Vlasov}
\end{equation}
where $m>0$ is the particle mass (assuming that all plasma particles have the same mass) and $F=F(x,v,t)\in\R^3$ is the Coulomb--Lorentz force due to the electric and magnetic fields, say $E=E(x,t)\in\R^3$ and $B=B(x,t)\in\R^3$ respectively, created \textit{collectively} in the point $x\in\R^3$ at time $t\geq 0$ by all charged plasma particles as predicted by Maxwell's equations:
$$ F:=q(E+v\times B), $$
where $q$ is the charge of the particles and $\times$ denotes the cross product in $\R^3$. The term $-\frac{1}{m}\ddiv_v(Ff)$ in~\eqref{eq:Vlasov}, where $\ddiv_v$ is the divergence with respect to the variable $v$, is the aforesaid replacement for the Boltzmann collision operator (the minus sign is due to the fact that one should write this term on the right-hand side of the equation to compare it directly with the Boltzmann collision operator). Since $E$, $B$ are independent of $v$, it is not difficult to see that the force field $F$ above is $v$-divergence-free. Using this in $\ddiv_v{(Ff)}=f\ddiv_v{F}+F\cdot\nabla_vf$ one obtains that the Vlasov equation~\eqref{eq:Vlasov} can be given the form
$$ \partial_tf+v\cdot\nabla_xf+\frac{q}{m}(E+v\times B)\cdot\nabla_vf=0 $$
in which it is indeed typically found.

To better appreciate the physical rationale for the Vlasov equation it is instructive to sketch the derivation of this equation from a particle point of view. Let $x_i=x_i(t)$ and $v_i=v_i(t)$ be the position and the velocity of the $i$th plasma particle at time $t$. The Newton equations of motion of this particle in the force field $F$ are
$$ \dot{x}_i=v_i, \qquad \dot{v}_i=\frac{F(x_i,v_i,t)}{m}, $$
where $\dot{}$ stands for the time derivative. The force field $F$ has to be understood as the result of the superposition of all force fields generated by every charged plasma particle. As such, it is the underlying means by which plasma particles interact collectively and collisionlessly. Let us introduce now the \textit{empirical distribution} of the particles:
$$ f_N(x,v,t):=\frac{1}{N}\sum_{i=1}^{N}\delta_{(x_i(t),\,v_i(t))}(x,v), $$
where $N\in\N$ is the total number of plasma particles and $\delta_{(x_i(t),\,v_i(t))}$ is the Dirac delta distribution centred in the point $(x_i(t),\,v_i(t))\in\R^6$. If $\varphi=\varphi(x,v):\R^3\times\R^3\to\R$ is a sufficiently smooth and compactly supported test function, we observe that
\begin{align*}
    \frac{d}{dt}\int_{\R^3}\int_{\R^3}\varphi(x,v)f_N(x,v,t)\,dx\,dv &= \frac{1}{N}\sum_{i=1}^{N}\frac{d}{dt}\varphi(x_i(t),v_i(t)) \\
    &= \frac{1}{N}\sum_{i=1}^{N}\bigl(\nabla_x\varphi(x_i(t),v_i(t))\cdot\dot{x}_i(t)+\nabla_v\varphi(x_i(t),v_i(t))\cdot\dot{v}_i(t)\bigr) \\
    &= \frac{1}{N}\sum_{i=1}^{N}\biggl(\nabla_x\varphi(x_i(t),v_i(t))\cdot v_i(t) \\
    &\phantom{=\frac{1}{N}\sum_{i=1}^{N}\biggl(} +\nabla_v\varphi(x_i(t),v_i(t))\cdot\frac{F(x_i(t),v_i(t),t)}{m}\biggr) \\
    &= \int_{\R^3}\int_{\R^3}\biggl(\nabla_x\varphi(x,v)\cdot v \\
    &\phantom{=\int_{\R^3}\int_{\R^3}\biggl(} +\nabla_v\varphi(x,v)\cdot\frac{F(x,v,t)}{m}\biggr)f_N(x,v,t)\,dx\,dv,
\end{align*}
hence, owing to the arbitrariness of $\varphi$, the empirical distribution satisfies~\eqref{eq:Vlasov} weakly for every number of plasma particles. If, in the limit $N\to\infty$, the sequence $\{f_N\}_{N\in\N}$ converges weakly-$\ast$ in the sense of measures to a distribution function $f$ then we deduce formally that $f$ satisfies in turn the Vlasov equation~\eqref{eq:Vlasov}. In other words, if for an increasing number of particles the empirical description of the plasma (i.e. the one provided by $f_N$) approximates an aggregate, viz. particle-less, statistical description (i.e. the one provided by $f$) then Vlasov's statistical model~\eqref{eq:Vlasov} applies. This is typically the case when the microscopic states $(x_i(t),\,v_i(t))$ of the particles belong to a given compact subset of $\R^6$ for all $t>0$ and all $N\in\N$, for then Prokhorov theorem~\cite{ambrosio2008BOOK} implies that $\{f_N\}_{N\in\N}$ converges up to subsequences.

In the abstract, the main difference of the Vlasov paradigm with respect to the Boltzmann one is that particles need not be in contact to interact. More precisely, particles do not collide at all but modify their velocity in consequence of long-range interactions caused by a force field that they contribute collectively to. This idea has been borrowed and generalised in the context of multi-agent systems, cf. e.g.,~\cite{dobrushin1979FAA}, to provide a statistical mechanics description of particle models based on Newton-type differential equations rather than on collision-like algebraic relationships. A prominent example is the celebrated Cucker--Smale model~\cite{cucker2007TAC,cucker2007JJM}, which was proposed to describe a system of autonomous agents, such as e.g., a flock of birds, that can possibly reach a consensus based on mutual interactions without central coordination. The differential version of the Cucker--Smale model is usually written as
$$ \dot{x}_i=v_i, \qquad \dot{v}_i=\frac{1}{N}\sum_{j=1}^{N}\frac{K}{{\bigl(\sigma^2+\abs{x_j-x_i}^2\bigr)}^\beta}(v_j-v_i), $$
where $x_i,\,v_i\in\R^3$ are the position and velocity of the $i$th agent of the flock, $i=1,\,\dots,\,N$, and $K,\,\sigma,\,\beta>0$ are model parameters. Each term of the sum in the acceleration equation can be understood as the force that the $j$th agent applies to the $i$th agent, so that the total force acting on the $i$th agent is the \textit{average} (cf. the coefficient $1/N$ in front of the sum) of these pairwise contributions. The Vlasov-type version of this model can be obtained formally by a procedure analogous to that presented above, see e.g.,~\cite{canizo2011M3AS,carrillo2010MSSET}. The result is~\eqref{eq:Vlasov} with $m=1$ and
$$ F(x,v,t):=\int_{\R^3}\int_{\R^3}\frac{K}{{\bigl(\sigma^2+\abs{x_\ast-x}^2\bigr)}^\beta}(v_\ast-v)f(x_\ast,v_\ast,t)\,dx_\ast\,dv_\ast. $$
Since, consistently with the observation above, this field is clearly the expectation of the generic pairwise force with respect to the statistical distribution of the particles generating it, the Vlasov-type equation with such an $F$ is also called a \textit{mean field equation} and the limit $N\to\infty$ under which it is obtained from the Newtonian dynamics a \textit{mean field limit}. As an aside, we mention that there are interesting relationships between Vlasov-type mean field equations and Boltzmann-type equations passing through the Fokker--Planck equations that we shall present in Section~\ref{sect:Fokker--Planck}. Here, we refrain from discussing this topic and refer to~\cite{carrillo2010SIMA,carrillo2010MSSET} for details.

In this paper, we are interested in Boltzmann-type equations like~\eqref{eq:Boltz.Toscani} for interacting multi-agent systems. The main goal is to provide a critical collection of fundamental mathematical tools and techniques, that can complement the modelling of multi-agent systems, often partly heuristic, with a rigorous and organic analysis of their basic theoretical properties. In more detail, the paper is organised as follows. In Section~\ref{sect:formal_derivation}, we propose a formal derivation of homogeneous Boltzmann-type equations from stochastic particle models of interacting agents, including the classical homogeneous Boltzmann equation as a particular case. Moreover, sticking to the case of linear symmetric interaction rules on the whole real line, which will then be the leitmotif in the whole paper, we provide a first taste of how such equations allow one to link aggregate trends of the system to distinctive features of the individual interactions. In Section~\ref{sect:Fourier}, we introduce the Fourier transform and the Fourier metric as essential tools that we shall use extensively to develop the subsequent mathematical theory of Boltzmann-type equations. In Section~\ref{sect:basic_theory}, we present the basic well-posedness theory of Boltzmann-type equations, namely existence, uniqueness, and continuous dependence of the solution. We discuss also how the theory of interactions on the whole real line can be borrowed to address interactions in a subset of the real line, which is often the case when e.g., the microscopic state of the agents is a non-negative variable due to physical limitations. In Section~\ref{sect:trend_equil}, we study the trend to equilibrium of Boltzmann-type equations, i.e. the possible convergence of the solutions to steady distributions which depict the statistical configurations of the system emerging spontaneously in the long run out of the interactions among the agents. We also provide some general characterisations of the steady distributions in terms of their moments and tails. In Section~\ref{sect:Fokker--Planck}, we push the study of equilibrium distributions forward by introducing the quasi-invariant limit. This is an asymptotic procedure which, in special regimes of the interaction parameters, transforms Boltzmann-type integro-differential equations in Fokker--Planck differential equations potentially more tractable as far as the explicit computation of equilibrium solutions is concerned. We consider, in particular, a few of such regimes meaningful for applications and exhibit in each of them the analytical equilibrium distribution obtained from the corresponding Fokker--Planck equation. In addition to this, considering that the limit Fokker--Planck equations can be possibly regarded as kinetic models \textit{per se} replacing Boltzmann-type equations in appropriate regimes of the parameters, we address the uniqueness and continuous dependence of their time-evolving solutions. In Section~\ref{sect:MonteCarlo}, we sketch the basics of the Monte Carlo method for the numerical solution of Boltzmann-type equations. In mathematical physics and applied mathematics numerical simulations are often an essential complement to the development of analytical theories, as they allow one to visualise the predicted solutions or to catch a glimpse beyond the boundaries of the established theoretical results. In the case of Boltzmann-type equations, the Monte Carlo numerical method is not only a discretisation technique but is intimately correlated to the particle physics underlying the derivation of the equations in a virtuous circle among modelling, analysis, and numerics. The paper is concluded by Section~\ref{sect:graph}, where we show how the tools and methods set out previously can be employed to address Boltzmann-type equations on graphs, a recent extension of the standard kinetic approach 
conceived to model networked multi-agent systems.

\section{Derivation of a homogeneous Boltzmann-type equation}
\label{sect:formal_derivation}
\subsection{Agent-based model}
\label{sect:agent-based_model}
We consider a large system of \textit{indistinguishable} agents that interact in pairs, whereby they update over time their microscopic state. The latter is described, at time $t\geq 0$, by a scalar random variable $V_t\in\R$. Notice that we do not include, in the symbol $V_t$, any label referring to the agent (such as e.g., $V_{t,i}$, $V_t^i$ or similar) because, as said, agents are indistinguishable. This means that any of them is representative of all the agents of the system, a fundamental assumption at the basis of the \textit{statistical} Boltzmann-type approach. We observe that this is different from e.g., the Vlasov-type mean field approach mentioned in Section~\ref{sect:intro}, where the distinction among the agents is initially preserved and is possibly lost only in the limit of an infinite number of agents.

To model binary-interaction-based dynamics we fix a time step $\Delta{t}>0$ and sample \textit{independently} two agents with states, say, $V_t$, $V^\ast_t$. We assume that within the time step $\Delta{t}$ they interact with a certain probability and, if the interaction occurs, they get the new states $V'_t,\,{V^\ast_t}'\in\R$, respectively. In formulas, we write:
\begin{equation}
	V_{t+\Delta{t}}=(1-\Theta)V_t+\Theta V'_t, \qquad V^\ast_{t+\Delta{t}}=(1-\Theta)V^\ast_t+\Theta{V^\ast_t}',
	\label{eq:Vt+Dt}
\end{equation}
where $\Theta\in\{0,\,1\}$ is a Bernoulli random variable, independent of the sampling of $V_t$, $V^\ast_t$, such that:
\begin{enumerate*}[label=(\roman*)]
\item if $\Theta=0$ then the interaction does not occur within the time step $\Delta{t}$ and the post-interaction states $V_{t+\Delta{t}}$, $V^\ast_{t+\Delta{t}}$ coincide therefore with the pre-interaction ones $V_t$, $V^\ast_t$;
\item conversely, if $\Theta=1$ then the interaction occurs within the time step $\Delta{t}$, leading the post-interaction states to become $V'_t$, ${V^\ast_t}'$.
\end{enumerate*}
Specifically, we let
\begin{equation}
	\Theta\sim\operatorname{Bernoulli}{(\Delta{t})},
	\label{eq:Theta}
\end{equation}
therefore
$$ \Prob{\Theta=1}=\Delta{t}, \qquad \Prob{\Theta=0}=1-\Delta{t} $$
under the constraint $\Delta{t}\leq 1$.

\begin{remark}
As we shall see in a moment, the constraint $\Delta{t}\leq 1$ is not a limitation from the analytical point of view but might pose some numerical restrictions, cf. Section~\ref{sect:MonteCarlo}. A formally different definition of $\Theta$, however equivalent to~\eqref{eq:Theta} to all our purposes, is
$$ \Theta\sim\operatorname{Bernoulli}{\left(\frac{\Delta{t}}{1+\Delta{t}}\right)}, $$
which is free from constraints on $\Delta{t}$ because $\frac{\Delta{t}}{1+\Delta{t}}<1$ for all $\Delta{t}>0$.

We also remark that it is possible to include an \textit{interaction rate} $\mu>0$ in the particle description by letting 
\begin{equation}
 \Theta\sim\operatorname{Bernoulli}{\left(\mu\Delta{t}\right)} \label{eq:Bern.mu}
 \end{equation}
under the constraint $\mu\Delta{t}\leq 1$. This way, a high rate $\mu$ implies a small interval of time $\Delta{t}$ needed to observe an interaction. Assumption~\eqref{eq:Theta} corresponds to $\mu=1$, a simplification that we shall invariably make in the subsequent developments.
\end{remark}

In~\eqref{eq:Vt+Dt}, $V'_t$ and ${V^\ast_t}'$ are placeholders for the physical models of the post-interaction states of the agents when an interaction occurs. In the collisional kinetic theory, they are expressed as functions of the pre-interaction states $V_t$, $V^\ast_t$. In this paper, we focus on the case of \textit{linear} and \textit{symmetric} interactions, in which one has:
\begin{equation}
	V'_t=pV_t+qV^\ast_t, \qquad {V^\ast_t}'=pV^\ast_t+qV_t,
	\label{eq:linsymint}
\end{equation}
$p,\,q\in\R_+$ being possibly random parameters independent of $V_t$, $V^\ast_t$, $\Theta$ and with prescribed laws. The interaction rules~\eqref{eq:linsymint} express the new post-interaction states as a mixing (linear combination) of the pre-interaction states. They are said to be symmetric because the two rules correspond to each other up to a switch of the roles of $V_t$ and $V^\ast_t$.

\begin{remark}
In the abstract, we could consider $p,\,q\in\R$ rather than restricting these parameters to $\R_+$. Nevertheless, physical model parameters are most of the times non-negative. Moreover, with $p,\,q\geq 0$ we get rid of some minor technicalities of the theory that we shall develop, still without loss of generality.
\end{remark}

\subsection{Statistical description}
\label{sect:stat_descr}
Building on the original idea of Boltzmann, we aim to provide a statistical description of the system of interacting agents in terms of the evolution of the probability distribution of the microscopic state $V_t$. For this, we introduce the \textit{kinetic distribution function}
$$ f=f(v,t):\R\times [0,\,+\infty)\to\R_+, $$
which expresses the law of $V_t$. This means that
$$ \Prob{V_t\in A}=\int_Af(v,t)\,dv $$
for every measurable set $A\subseteq\R$, together with the normalisation condition $\int_\R f(v,t)\,dv=1$ for all $t\geq 0$.

\begin{remark}
We use for $f$ the classical notation of functions, as if the law of $V_t$ were a measure absolutely continuous with respect to the Lebesgue measure in $\R$ with density $f$. Nevertheless, $f(v,t)\,dv$ has to be understood just as a practical customary writing, which stands more in general for $f(dv,t)$. The theory that we shall develop covers indeed the case in which $f$ is an abstract probability measure in $\R$ with respect to $v$ parametrised by $t$. Clearly, when $f(\cdot,t)$ is a non-negative integrable function we can refer to it as the \textit{probability density function} of the random variable $V_t$.
\end{remark}

To derive a model of the evolution of $f$ motivated by the particle dynamics described above the idea is to average the equations~\eqref{eq:Vt+Dt}. However, averaging only the information delivered by~\eqref{eq:Vt+Dt} would be insufficient to obtain the entire evolution of the distribution function $f$, because clearly the law of a random variable is not univocally identified by its mean value only. For this, we introduce the concept of \textit{observable quantity}, namely an arbitrary function $\varphi=\varphi(v):\R\to\R$ which can be computed out of the knowledge of the values taken by the random variable $V_t$. Evaluating $\varphi$ on both sides of~\eqref{eq:Vt+Dt} yields
$$ \varphi(V_{t+\Delta{t}})=\varphi\bigl((1-\Theta)V_t+\Theta V'_t\bigr), \qquad \varphi(V^\ast_{t+\Delta{t}})=\varphi\bigl((1-\Theta)V^\ast_t+\Theta{V^\ast_t}'\bigr). $$
Averaging now these relationships and using the fact that $\Theta$ is, by construction, independent of the other random variables we obtain
\begin{align}
	\begin{aligned}[c]
		\ave{\varphi(V_{t+\Delta{t}})} &= \ave{\varphi(V_t)}(1-\Delta{t})+\ave{\varphi(V_t')}\Delta{t} \\
		\ave{\varphi(V^\ast_{t+\Delta{t}})} &= \ave{\varphi(V^\ast_t)}(1-\Delta{t})+\ave{\varphi({V^\ast_t}')}\Delta{t},
	\end{aligned}
	\label{eq:expectation_w.r.t._Theta}
\end{align}
where, here and henceforth, $\ave{\cdot}$ denotes expectation. Rearranging the terms and dividing by $\Delta{t}$ we deduce
$$ \frac{\ave{\varphi(V_{t+\Delta{t}})}-\ave{\varphi(V_t)}}{\Delta{t}}=\ave{\varphi(V'_t)}-\ave{\varphi(V_t)}, \qquad
	\frac{\ave{\varphi(V^\ast_{t+\Delta{t}})}-\ave{\varphi(V^\ast_t)}}{\Delta{t}}=\ave{\varphi({V^\ast_t}')}-\ave{\varphi(V^\ast_t)}, $$
whence, passing formally to the continuous-time limit $\Delta{t}\to 0^+$,
\begin{equation}
	\frac{d}{dt}\ave{\varphi(V_t)}=\ave{\varphi(V'_t)}-\ave{\varphi(V_t)}, \qquad
		\frac{d}{dt}\ave{\varphi(V^\ast_t)}=\ave{\varphi({V^\ast_t}')}-\ave{\varphi(V^\ast_t)}.
	\label{eq:cont_time_lim}
\end{equation}
It is in this passage that the constraint $\Delta{t}\leq 1$ imposed by~\eqref{eq:Theta} becomes uninfluential.

The remaining expectations can be computed using $f$, for instance:
$$ \ave{\varphi(V_t)}=\int_\R\varphi(v)f(v,t)\,dv $$
and likewise for $\varphi(V^\ast_t)$. We notice, in particular, that owing to~\eqref{eq:linsymint} the expectations $\ave{\varphi(V'_t)}$ and $\ave{\varphi({V^\ast_t}')}$ would require the joint law of $V_t$, $V^\ast_t$. Nevertheless, since by assumption the interacting agents are sampled independently their microscopic states are independent at the moment of the interaction. Therefore, their joint law is simply $f(v,t)f(v_\ast,t)$ and we have e.g.,
$$ \ave{\varphi(V'_t)}=\int_\R\int_\R\ave{\varphi(v')}f(v,t)f(v_\ast,t)\,dv\,dv_\ast, $$
where further expectation $\ave{\varphi(v')}$ is meant with respect to the laws of the possibly random parameters $p$, $q$. On the whole, summing the two equations in~\eqref{eq:cont_time_lim} and recalling that $V_t$, $V^\ast_t$ share the same distribution $f$ we get
\begin{equation}
	\frac{d}{dt}\int_\R\varphi(v)f(v,t)\,dv=\int_\R\int_\R\left(\frac{\ave{\varphi(v')}+\ave{\varphi(v_\ast')}}{2}-\varphi(v)\right)f(v,t)f(v_\ast,t)\,dv\,dv_\ast,
	\label{eq:Boltztype.weak.general}
\end{equation}
which holds for every observable quantity $\varphi$. Notice in particular that, using $\int_\R f(v_\ast,t)\,dv_\ast=1$ for all $t\geq 0$, we have written
$$ \int_\R\varphi(v)f(v,t)\,dv=\int_\R\int_\R\varphi(v)f(v,t)f(v_\ast,t)\,dv\,dv_\ast $$
to make the right-hand side more compact.

If the interaction rules are symmetric, like in~\eqref{eq:linsymint}, then it is easy to see that
$$ \int_\R\int_\R\ave{\varphi(v')}f(v,t)f(v_\ast,t)\,dv\,dv_\ast=\int_\R\int_\R\ave{\varphi(v'_\ast)}f(v,t)f(v_\ast,t)\,dv\,dv_\ast. $$
Therefore, the equation for $f$ we shall deal with henceforth is finally
\begin{equation}
	\frac{d}{dt}\int_\R\varphi(v)f(v,t)\,dv=\int_\R\int_\R\ave{\varphi(v')-\varphi(v)}f(v,t)f(v_\ast,t)\,dv\,dv_\ast,
	\label{eq:Boltztype.weak}
\end{equation}
which holds for every observable $\varphi$, with
\begin{equation}
	v'=pv+qv_\ast.
	\label{eq:v'}
\end{equation}
In~\eqref{eq:Boltztype.weak} we have used the linearity of the expectation together with $\ave{\varphi(v)}=\varphi(v)$ (because $\varphi(v)$ is constant with respect to $p$, $q$) to further compact the notation at the right-hand side.

\begin{remark}
The form~\eqref{eq:Boltztype.weak} of the equation for $f$ holds, in general, for every symmetric interaction rule possibly different from~\eqref{eq:v'} (in particular, possibly also non-linear). Conversely, if the interaction rule is not symmetric then the general form~\eqref{eq:Boltztype.weak.general} has to be used.
\end{remark}

\subsection{Weak and strong forms of a Boltzmann-type equation}
Since~\eqref{eq:Boltztype.weak} holds for every $\varphi$, the latter can be regarded as a test function, whereby~\eqref{eq:Boltztype.weak} is actually the \textit{weak form} of the evolution equation of $f$. Its meaning can be expressed in words by saying that the time variation of the mean value of an observable quantity (left-hand side) is the mean value of the variation of that quantity in a generic representative interaction (right-hand side).

To pass formally to the strong form of~\eqref{eq:Boltztype.weak}, i.e. the one which does not involve $\varphi$, the strategy is to collect $\varphi$ in each term of the equation. The problem then arises of how to reduce the term containing $\varphi(v')$ to one involving $\varphi(v)$. This can be done by a proper change of variable in the integral, which requires to introduce the \textit{inverse interaction rules}
\begin{equation}
	v=\frac{p}{p^2-q^2}v'-\frac{q}{p^2-q^2}v_\ast', \qquad v_\ast=\frac{p}{p^2-q^2}v_\ast'-\frac{q}{p^2-q^2}v'.
	\label{eq:invint}
\end{equation}
To have them well-defined we are going to assume $q\neq \pm p$. This change of variables implies also $dv\,dv_\ast=\frac{1}{\abs{p^2-q^2}}\,dv'\,dv'_\ast$, being $p^2-q^2$ the Jacobian determinant of the change of variables, whence
\begin{align*}
	\int_\R\int_\R\ave{\varphi(v')}f(v,t)f(v_\ast,t)\,dv\,dv_\ast &= \ave*{\int_\R\int_\R\varphi(v')f(v,t)f(v_\ast,t)\,dv\,dv_\ast} \\
	&= \int_\R\int_\R\varphi(v')\ave*{\frac{1}{\abs{p^2-q^2}}f(v,t)f(v_\ast,t)}\,dv'\,dv'_\ast,
\end{align*}
where on the left-hand side $v'$ is thought of as a function of $v$, $v_\ast$ through the interaction rule~\eqref{eq:v'} while on the last right-hand side $v$, $v_\ast$ are thought of as functions of $v'$, $v_\ast'$ through the inverse interaction rules~\eqref{eq:invint}.

Plugging this into~\eqref{eq:Boltztype.weak} while noticing that formally $\frac{d}{dt}\int_\R\varphi(v)f(v,t)\,dv=\int_\R\varphi(v)\partial_tf(v,t)\,dv$ yields
\begin{align*}
    \int_\R\varphi(v)\partial_tf(v,t)\,dv &= \int_\R\varphi(v')\int_\R\ave*{\frac{1}{\abs{p^2-q^2}}f(v,t)f(v_\ast,t)}\,dv'_\ast\,dv' \\
    &\phantom{=} -\int_\R\varphi(v)\int_\R f(v,t)f(v_\ast,t)\,dv_\ast\,dv.
\end{align*}
Now, to make the notation uniform between the first and second term at the right-hand side it is customary to rename, in the first term, the pre-interaction states as $\pr{v}$, $\pr{v}_\ast$ and the post-interaction states, which are dummy variables, as $v$, $v_\ast$. Clearly, $\pr{v}$, $\pr{v}_\ast$ are then thought of as functions of $v$, $v_\ast$ according to~\eqref{eq:invint} upon renaming the variables there consistently. With this trick we get
$$ \int_\R\varphi(v)\partial_tf(v,t)\,dv=\int_\R\varphi(v)\int_\R\ave*{\frac{1}{\abs{p^2-q^2}}f(\pr{v},t)f(\pr{v}_\ast,t)-f(v,t)f(v_\ast,t)}\,dv_\ast\,dv, $$
whence, invoking the arbitrariness of $\varphi$,
\begin{equation}
	\partial_tf=\int_\R\ave*{\frac{1}{\abs{p^2-q^2}}f(\pr{v},t)f(\pr{v}_\ast,t)-f(v,t)f(v_\ast,t)}\,dv_\ast,
	\label{eq:Boltztype.strong}
\end{equation}
which is the strong form of~\eqref{eq:Boltztype.weak}.

Owing to its clear structural analogy to the classical homogeneous Boltzmann equation, that we shall present in Section~\ref{sect:Boltz.deriv},~\eqref{eq:Boltztype.strong} is called a \textit{Boltzmann-type equation} and, consequently,~\eqref{eq:Boltztype.weak} is called a Boltzmann-type equation in weak form. Notice that~\eqref{eq:Boltztype.strong} is an integro-differential equation because the right-hand side can be read as the action of a bilinear integral operator on the kinetic distribution function $f$:
\begin{equation}
	Q(f,f)(v,t):=\int_\R\ave*{\frac{1}{\abs{p^2-q^2}}f(\pr{v},t)f(\pr{v}_\ast,t)-f(v,t)f(v_\ast,t)}\,dv_\ast.
	\label{eq:Q.strong}
\end{equation}
Borrowing the jargon of the classical kinetic theory, this operator $Q$ is called the \textit{collisional operator}. Comparing~\eqref{eq:Boltztype.strong} with~\eqref{eq:Boltztype.weak} we see that $Q$ has the following property:
\begin{equation}
	\int_\R\varphi(v)Q(f,f)(v,t)\,dv=\int_\R\int_\R\ave{\varphi(v')-\varphi(v)}f(v,t)f(v_\ast,t)\,dv\,dv_\ast
	\label{eq:Q.weak}
\end{equation}
for all $\varphi$. The fact that $\int_\R\varphi(v)Q(f,f)(v,t)\,dv$ looks mathematically more friendly than $Q$ itself is at the basis of our preference, in this work, for the weak form~\eqref{eq:Boltztype.weak} of the Boltzmann type equation over the strong one~\eqref{eq:Boltztype.strong}.

\begin{remark}
For generic interaction rules, and generic interaction frequency $\mu$, the strong form~\eqref{eq:Boltztype.strong} of the Boltzmann type equation generalises simply as
\begin{equation} \partial_tf=\mu Q(f,f), \qquad Q(f,f)(v,t):=\int_\R\ave*{\frac{1}{\abs{J}}f(\pr{v},t)f(\pr{v}_\ast,t)-f(v,t)f(v_\ast,t)}\,dv_\ast,\label{eq:Boltztype.strong.gen} 
\end{equation}
where $J$ is the Jacobian determinant of the inverse interactions. Clearly, in order for this form to be well-defined it is necessary that the mapping $(v,\,v_\ast)\mapsto (v',\,v_\ast')$ be a diffeomorphism. Notice that the weak form~\eqref{eq:Boltztype.weak} does not require instead such a smoothness of the interaction rule.
\end{remark}

\subsection{Evolution of the moments}
\label{sect:moments_evol}
The weak form~\eqref{eq:Boltztype.weak} of the Boltzmann-type equation is particularly useful to study the evolution of the statistical moments of $f$, i.e. the quantities
$$ M_n(t):=\int_\R v^nf(v,t)\,dv, \qquad n\in\N. $$
Notice that $M_0(t)=1$ for all $t\geq 0$. The first moment $M_1$ is the mean state of the system while the second moment $M_2$ is typically understood with the physical meaning of \textit{energy} of the system. By means of the first and second moment one can also define
$$ \Var{f}(t):=M_2(t)-M_1^2(t)=\int_\R v^2f(v,t)\,dv-\left(\int_\R vf(v,t)\,dv\right)^2, $$
namely the variance of the distribution $f$ which conveys the physical meaning of \textit{internal energy} of the system. The Cauchy-Schwarz inequality implies straightforwardly that $\Var{f}(t)\geq 0$ for all $t$, as it is well-known from the probability theory.

In many applications the trend of the statistical moments, and especially that of these first moments, provides useful hints on the behaviour of the system, in particular for large times, also in the absence of an explicit characterisation of the distribution $f$, which is often difficult to obtain. It is therefore interesting to get a direct picture of the moment evolution from the Boltzmann-type equation. This can be done by letting $\varphi(v)=v^n$ in~\eqref{eq:Boltztype.weak}-\eqref{eq:v'}, which invoking the binomial theorem:
\begin{align*}
	\varphi(v')={(pv+qv_\ast)}^n &= \sum_{k=0}^{n}\binom{n}{k}p^kq^{n-k}v^kv_\ast^{n-k} \\
	&= p^nv^n+q^nv_\ast^n+\sum_{k=1}^{n-1}\binom{n}{k}p^kq^{n-k}v^kv_\ast^{n-k}
\end{align*}
and with some further little algebra yields
$$ \frac{dM_n}{dt}=\left(\ave{p^n+q^n}-1\right)M_n+\sum_{k=1}^{n-1}\binom{n}{k}\ave{p^kq^{n-k}}M_kM_{n-k}. $$
In particular, the mean value and the energy satisfy
$$ \frac{dM_1}{dt}=(\ave{p+q}-1)M_1, \qquad \frac{dM_2}{dt}=\left(\ave{p^2+q^2}-1\right)M_2+2\ave{pq}M_1^2, $$
whence we argue for instance that:
\begin{itemize}
\item If $\ave{p+q}=1$ then $M_1$ is constant in time; if simultaneously $\ave{p^2+q^2}<1$ then $M_2$ tends asymptotically in time to a finite non-zero value, therefore so does in general the variance $\Var{f}$. In such a scenario, it is reasonable to expect that the system evolves towards an emerging statistical profile described by a non-trivial kinetic distribution function $f^\infty(v)=\lim_{t\to +\infty}f(v,t)$, i.e. one which neither collapses in a single point (Dirac mass) nor spreads on the whole real axis. Notice indeed that a Dirac mass would have null variance whereas a distribution smearing over the whole $\R$ would tend to have infinite variance;
\item If both $\ave{p+q}<1$ and $\ave{p^2+q^2}<1$ then $M_1$, $M_2$ decay exponentially fast to zero in time, therefore so does $\Var{f}$. In this case, we can expect that the system converges in time to a statistical profile described by $f^\infty(v)=\delta_0(v)$, namely the Dirac delta centred in $v=0$. This corresponds to an emerging behaviour in which all agents acquire, in the long run, the state $v=0$ independently of their initial states. Depending on the application, this situation can be referred to as an \textit{aggregation} or \textit{consensus}.
\end{itemize}

\subsection{The homogeneous Boltzmann equation}
\label{sect:Boltz.deriv}
For conceptual reference, it is instructive to see that the procedure set forth in Section~\ref{sect:stat_descr} can be used to derive formally also the classical \textit{homogeneous} Boltzmann equation, i.e.~\eqref{eq:Boltz} with $f$ independent of $x$, so that the advection term on the left-hand side vanishes. Such an equation is used to investigate closely the dynamics of pure collisions among gas molecules. Physically, it corresponds to assuming that the gas is uniformly distributed in space, in such a way that in every point the statistical distribution of the velocities of the molecules is the same.

First, we show how the collision rules~\eqref{eq:Boltz.coll_rules} are derived from the physical principles~\eqref{eq:Boltz.momentum_energy}. Let $\sS^2$ be the unit sphere in $\R^3$. Given a pair of colliding molecules with pre-collisional velocities $v,\,v_\ast\in\R^3$, let $n\in\sS^2$ be \textit{any} unit vector of $\R^3$ and let us consider the \textit{ansatz}
$$ v'=v-\gamma n, \qquad v_\ast'=v_\ast+\gamma n, $$
where $\gamma\in\R$ is a free parameter. Plugging it into~\eqref{eq:Boltz.momentum_energy} we obtain that~\eqref{eq:Boltz.momentum} is satisfied for every $\gamma\in\R$. Conversely, using $\abs{v'}^2=v'\cdot v'$ (and similarly for $\abs{v_\ast'}^2$) -- where $\cdot$ denotes the inner product in $\R^3$ -- and $\abs{n}=1$, from~\eqref{eq:Boltz.energy} we get
$$ \gamma^2+\gamma(v_\ast-v)\cdot n=0. $$
Ruling out the trivial solution $\gamma=0$, which would imply no collision, this yields
$$ \gamma=(v-v_\ast)\cdot n, $$
whence the collision rules~\eqref{eq:Boltz.coll_rules} follow.

Notice that, as stated above, $n$ can be in principle any unit vector of $\R^3$. Nevertheless, from the physical point of view it makes sense to take $n$ as a unit vector parallel to the relative position of the colliding molecules or, in other words, aligned with the direction of their centres. Indeed, this direction can be understood as the one along which the collision possibly occurs, whereas the components of the velocities of the molecules in the plane orthogonal to it can in no way lead the molecules to collide. Moreover, this choice of $n$ is suitable to devise physically meaningful expressions of the collision kernel, such as~\eqref{eq:Boltz.B}, which builds on the idea that the more the relative velocity of the colliding molecules is oriented in the direction of the collision the more frequent, or ``probable'', the collision is.

This setting can be recast in the statistical description of Section~\ref{sect:stat_descr} by taking advantage of the stochastic particle model~\eqref{eq:Vt+Dt}, where now we understand $V_t,\,V^\ast_t\in\R^3$ as the random variables representing the pre-collisional velocities, whose realisations are $v,\,v_\ast$, and $V_t',\,{V^\ast_t}'\in\R^3$ as the post-collisional velocities when a collision occurs, whose realisations are $v',\,v_\ast'$. From~\eqref{eq:Boltz.coll_rules} we deduce therefore
$$ V_t'=V_t+[(V^\ast_t-V_t)\cdot n]n, \qquad {V^\ast_t}'=V^\ast_t+[(V_t-V^\ast_t)\cdot n]n, $$
where $n$ has to be regarded as a random quantity, because the direction of collision of two random molecules is in turn random. These rules are the equivalent of~\eqref{eq:linsymint}; in particular, $n$ plays the role of a random coefficient like $p,\,q$. It is customary to assume that it is uniformly distributed on the sphere, i.e.
$$ n\sim\pazocal{U}(\sS^2), $$
to mean that there are in principle no preferential directions of collision.

Concerning the random variable $\Theta$ appearing in~\eqref{eq:Vt+Dt}, here we include in its law the collision kernel $B:\R\to\R_+$, namely the counterpart of the interaction rate $\mu$:
$$ \Theta\sim\operatorname{Bernoulli}(B((V^\ast_t-V_t)\cdot n)\Delta{t}). $$
Depending on $B$, the condition $B((V^\ast_t-V_t)\cdot n)\Delta{t}\leq 1$ might or might not be satisfied with a constant time step $\Delta{t}>0$. If it is not, like in~\eqref{eq:Boltz.B} where $B$ is unbounded, then one can conceptually use an adaptive $\Delta{t}$, which changes with every pair of colliding molecules.

Repeating the procedure described in Section~\ref{sect:stat_descr} and taking advantage of the symmetry of the collision rules~\eqref{eq:Boltz.coll_rules}, we arrive at the equivalent of~\eqref{eq:Boltztype.weak}, namely the weak form of the Boltzmann equation:
$$ \frac{d}{dt}\int_{\R^3}\varphi(v)f(v,t)\,dv=\int_{\R^3}\int_{\R^3}
    \ave{B((v_\ast-v)\cdot n)(\varphi(v')-\varphi(v))}f(v,t)f(v_\ast,t)\,dv\,dv_\ast, $$
where $\ave{\cdot}$ is the expectation with respect to the random parameter $n$, i.e.
$$ \ave{\cdot}=\frac{1}{4\pi}\int_{\sS^2}(\cdot)\,dn, $$
being $4\pi$ the Hausdorff measure of $\sS^2$. We can therefore rewrite the equation more explicitly as
$$ \frac{d}{dt}\int_{\R^3}\varphi(v)f(v,t)\,dv=\frac{1}{4\pi}\int_{\R^3}\int_{\R^3}\int_{\sS^2}
    B((v_\ast-v)\cdot n)(\varphi(v')-\varphi(v))f(v,t)f(v_\ast,t)\,dn\,dv\,dv_\ast. $$

\begin{remark}
It is worth pointing out that, unlike~\eqref{eq:Theta}, here the random variable $\Theta$ is not stochastically independent of $V_t$, $V^\ast_t$, and also of $n$, as these quantities enter the definition of its law. Therefore, the passage~\eqref{eq:expectation_w.r.t._Theta} has to be developed by appealing formally to the conditional expectation:
\begin{align*}
    \ave{\varphi(V_{t+\Delta{t}})} &= \ave{\ave{\varphi\bigl((1-\Theta)V_t+\Theta V_t')\bigr)\vert V_t,\,V^\ast_t,\,n}} \\
    &= \ave{\varphi(V_t)\bigl(1-B((V^\ast_t-V_t)\cdot n))\Delta{t}\bigr)}+\ave{\varphi(V_t')B((V^\ast_t-V_t)\cdot n)}\Delta{t}
\end{align*}
and similarly for $\ave{\varphi(V^\ast_{t+\Delta{t}})}$.
\end{remark}

To recover the strong formulation, we take advantage of the invertibility of the collision rules~\eqref{eq:Boltz.coll_rules} with unit Jacobian determinant to get:
\begin{align*}
    \frac{d}{dt}\int_{\R^3}\varphi(v)f(v,t)\,dv &= \frac{1}{4\pi}\int_{\R^3}\varphi(v')\left(
        \int_{\R^3}\int_{\sS^2}B((v_\ast-v)\cdot n)f(v,t)f(v_\ast,t)\,dn\,dv_\ast'\right)dv' \\
    &\phantom{=} -\frac{1}{4\pi}\int_{\R^3}\varphi(v)\left(\int_{\R^3}\int_{\sS^2}B((v_\ast-v)\cdot n)f(v,t)f(v_\ast,t)\,dn\,dv_\ast\right)dv,
\end{align*}
where, in the first term on the right-hand side, we understand the pre-collisional velocities $v,\,v_\ast$ as functions of the post-collisional ones $v',\,v_\ast'$ through the \textit{inverse collision}
$$ v=v'+[(v_\ast'-v')\cdot n]n, \qquad v_\ast=v_\ast'+[(v'-v_\ast')\cdot n]n. $$
Now, in the same integral we switch to the notation $\pr{v},\,\pr{v_\ast}$ for the pre-collisional velocities and simultaneously to the notation $v,\,v_\ast$ for the post-collisional velocities, i.e.:
\begin{equation}
    \pr{v}=v+[(v_\ast-v)\cdot n]n, \qquad \pr{v}_\ast=v_\ast+[(v-v_\ast)\cdot n]n,
    \label{eq:Boltz.invcoll}
\end{equation}
for homogeneity with the second integral. Thus, we reformulate the equation as
\begin{align*}
    \frac{d}{dt}\int_{\R^3}\varphi(v)f(v,t)\,dv &= \frac{1}{4\pi}\int_{\R^3}\varphi(v)\biggl[\int_{\R^3}\int_{\sS^2}
        \bigl(B((\pr{v}_\ast-\pr{v})\cdot n)f(\pr{v},t)f(\pr{v}_\ast,t) \\
    &\phantom{=} -B((v_\ast-v)\cdot n)f(v,t)f(v_\ast,t)\bigr)dn\,dv_\ast\biggr]dv.
\intertext{Then, comparing~\eqref{eq:Boltz.coll_rules} and~\eqref{eq:Boltz.invcoll} we see that formally $\pr{v}=v'$ and $\pr{v}_\ast=v_\ast'$, whence}
    &= \frac{1}{4\pi}\int_{\R^3}\varphi(v)\biggl[\int_{\R^3}\int_{\sS^2}\bigl(B((v_\ast'-v')\cdot n)f(v',t)f(v_\ast',t) \\
    &\phantom{=} -B((v_\ast-v)\cdot n)f(v,t)f(v_\ast,t)\bigr)dn\,dv_\ast\biggr]dv.
\end{align*}
From~\eqref{eq:Boltz.coll_rules} we compute $v'_\ast-v'=v_\ast-v+2[(v-v_\ast)\cdot n]n$ and consequently $(v_\ast'-v')\cdot n=-(v_\ast-v)\cdot n$, which, if $B$ is an even function like~\eqref{eq:Boltz.B}, implies $B((v_\ast'-v')\cdot n)=B((v_\ast-v)\cdot n)$. Therefore, collecting the collision kernel at the right-hand side and invoking the arbitrariness of $\varphi$ we are finally led to the following strong form of the homogeneous Boltzmann equation:
$$ \partial_tf=\frac{1}{4\pi}\int_{\R^3}\int_{\sS^2}B((v_\ast-v)\cdot n)\bigl(f(v',t)f(v_\ast',t)-f(v,t)f(v_\ast,t)\bigr)dn\,dv_\ast, $$
consistently with~\eqref{eq:Boltz}.

The Boltzmann collisional operator reads explicitly:
$$ Q(f,f)(v,t):=\frac{1}{4\pi}\int_{\R^3}\int_{\sS^2}B((v_\ast-v)\cdot n)\bigl(f(v',t)f(v_\ast',t)-f(v,t)f(v_\ast,t)\bigr)dn\,dv_\ast, $$
where $v',\,v_\ast'$ are given in terms of $v,\,v_\ast$ by~\eqref{eq:Boltz.coll_rules}. We observe that the possibility to express the joint velocity distribution of the colliding molecules as the product of the respective marginal distributions requires the hypothesis of stochastic independence of the molecules at the moment of the collision -- the so-called \textit{Boltzmann ansatz}. In Section~\ref{sect:agent-based_model}, dealing with an abstract multi-agent system, we enforced this hypothesis by postulating that the agents participating in an interaction are sampled independently. In the case of gases, instead, an analogous principle should be justified with a closer reference to the underlying physics. This issue has a long story in the mathematical-physical theory of the Boltzmann equation, especially as far as the rigorous derivation of the latter is concerned, along with the companion question of whether it is reasonable to consider only \textit{pairwise}, viz. \textit{binary}, collisions among gas molecules and neglect all collisions among more than two molecules at once. See e.g.,~\cite{cercignani1988BOOK,cercignani1994BOOK,cercignani1993CHAPTER}. A customary assumption, which grants an empirical validity of these two facts, is that the gas is \textit{rarefied}. The rationale behind it is that if the molecules are not densely packed then:
\begin{enumerate*}[label=(\roman*)]
\item on one hand, it is highly improbable that more than two of them are simultaneously so close to collide;
\item on the other hand, after a reciprocal collision any two molecules can wander a long way, meanwhile undergoing collisions with many other molecules, before possibly colliding again between them, so that should this occur they would have lost any mutual dependence.
\end{enumerate*}
    
\section{Fourier transform in kinetic theory}
\label{sect:Fourier}
The analytical study of the Boltzmann-type equation is expected to make formal arguments, such as those outlined in Section~\ref{sect:moments_evol} about the emergence of a stationary profile plus a number of others concerning e.g., the well-posedness -- hence the mathematical soundness -- of the equation itself, rigorous. To this purpose, it turns out that a particularly convenient mathematical tool is the \textit{Fourier transform}, which for a generic probability measure $\mu$ is defined as
$$ \hat{\mu}(\xi):=\int_\R e^{-i\xi v}\,d\mu(v). $$
Here and henceforth $i$ stands for the imaginary unit.

Let us denote by $\cP(\R)$ the set of probability measures defined on a convenient $\sigma$-algebra of $\R$, for instance the Borel $\sigma$-algebra, and by $C^0_b(\R)$ the space of bounded continuous functions on $\R$. Then we record preliminarily that:
\begin{lemma} \label{lemma:Fourier_basic}
If $\mu\in\cP(\R)$ then $\hat{\mu}\in C^0_b(\R)$ with $\norminf{\hat{\mu}}=1$.
\end{lemma}
\begin{proof}
To show that $\hat{\mu}$ is continuous at an arbitrary point $\bar{\xi}\in\R$ we examine
$$ \abs{\hat{\mu}(\xi)-\hat{\mu}(\bar{\xi})}=\abs*{\int_\R\left(e^{-i\xi v}-e^{-i\bar{\xi}v}\right)d\mu(v)}
	\leq\int_\R\abs{e^{-i\xi v}-e^{-i\bar{\xi}v}}d\mu(v). $$
Since $\abs{e^{-i\xi v}-e^{-i\bar{\xi}v}}\leq 2$, by dominated convergence and invoking the continuity of the exponential function we get
$$ \lim_{\xi\to\bar{\xi}}\abs{\hat{\mu}(\xi)-\hat{\mu}(\bar{\xi})}\leq\int_\R\lim_{\xi\to\bar{\xi}}\abs{e^{-i\xi v}-e^{-i\bar{\xi}v}}d\mu(v)=0, $$
whence the continuity of $\hat{\mu}$ follows.

As for the boundedness, it is sufficient to observe that
$$ \abs{\hat{\mu}(\xi)}\leq\int_\R\abs{e^{-i\xi v}}\,d\mu(v)=\mu(\R)=1 $$
for all $\xi\in\R$. Moreover, the claim on the $\infty$-norm follows straightforwardly from this and the fact that $\hat{\mu}(0)=\mu(\R)=1$.
\end{proof}

\subsection{Fourier-transformed Boltzmann-type equation}
The first to realise that the Fourier transform could constitute an effective mathematical tool to attack the Boltzmann equation was Bobylev, who in~\cite{bobylev1975DANSSSR} noticed how the equation was considerably simplified under such a transformation.

To write our Boltzmann-type equation in Fourier transform we take advantage of the weak form~\eqref{eq:Boltztype.weak} in which we let $\varphi(v)=e^{-i\xi v}$, thereby getting
$$ \partial_t\hat{f}(\xi,t)=\widehat{Q}(\hat{f},\hat{f})(\xi,t), $$
where $\widehat{Q}(\hat{f},\hat{f})$ stands for the Fourier-transformed collisional operator. From~\eqref{eq:Q.weak} we see in particular that
\begin{align*}
	\widehat{Q}(\hat{f},\hat{f})(\xi,t) &= \int_\R e^{-i\xi v}Q(f,f)(v,t)\,dv \\ 
	&= \int_\R\int_\R\ave{e^{-i\xi v'}-e^{-i\xi v}}f(v,t)f(v_\ast,t)\,dv\,dv_\ast \\
	&= \int_\R\int_\R\ave{e^{-i\xi(pv+qv_\ast)}}f(v,t)f(v_\ast,t)\,dv\,dv_\ast-\hat{f}(\xi,t) \\
	&= \ave{\hat{f}(p\xi,t)\hat{f}(q\xi,t)}-\hat{f}(\xi,t)
\end{align*}
so that the Fourier-transformed Boltzmann-type equation reads finally
\begin{equation}
	\partial_t\hat{f}(\xi,t)=\ave{\hat{f}(p\xi,t)\hat{f}(q\xi,t)}-\hat{f}(\xi,t).
	\label{eq:FourierBoltz}
\end{equation}
Clearly, the Fourier-transformed collisional operator $\widehat{Q}(\hat{f},\hat{f})$ has a much simpler form than that of the collisional operator $Q(f,f)$. For this reason,~\eqref{eq:FourierBoltz} will be at the basis of most of the theory that we shall develop in the forthcoming sections.

A first interesting result concerns the regularity of the Fourier transform of a solution to the Boltzmann-type equation. Let $f_0=f_0(\cdot,0)\in\cP(\R)$ be the initial datum prescribed to~\eqref{eq:Boltztype.strong}.
\begin{proposition} \label{prop:Lr_regularity}
Assume that either coefficient $p$, $q$ of the interaction rule~\eqref{eq:v'} is non-zero. If $\hat{f}_0\in L^r(\R)$ for some $r\in\N\setminus\{0\}$ then $\hat{f}(\cdot,t)\in L^r(\R)$ for $t>0$.
\end{proposition}
\begin{proof}
To fix the ideas, and without loss of generality, throughout the proof we shall assume that $p$ is non-zero.

Noticing that $r\hat{f}^{r-1}\partial_t\hat{f}=\partial_t\hat{f}^r$, we multiply~\eqref{eq:FourierBoltz} by $r\hat{f}^{r-1}$ to find
$$ \partial_t\hat{f}^r+r\hat{f}^r=r\hat{f}^{r-1}\ave{\hat{f}(p\xi,t)\hat{f}(q\xi,t)}. $$
Next, we multiply further both sides by $e^{rt}$ to obtain
$$ \partial_t\left(e^{rt}\hat{f}^r\right)=re^{rt}\hat{f}^{r-1}\ave{\hat{f}(p\xi,t)\hat{f}(q\xi,t)}. $$
Since for any complex-valued function $g$ it holds that $\partial_t\abs{g}\leq\abs{\partial_tg}$, from here we deduce
$$ \partial_t\left(e^{rt}\abs{\hat{f}}^r\right)\leq re^{rt}\abs{\hat{f}}^{r-1}\ave{\abs{\hat{f}(p\xi,t)}\cdot\abs{\hat{f}(q\xi,t)}} $$
and moreover, integrating with respect to $\xi\in\R$,
$$ \partial_t\left(e^{rt}\norm{\hat{f}(t)}{L^r}^r\right)\leq re^{rt}\ave*{\int_\R\abs{\hat{f}(\xi,t)}^{r-1}\abs{\hat{f}(p\xi,t)}\,d\xi}, $$
where we have used the fact that $\abs{\hat{f}(q\xi,t)}\leq 1$ for all $\xi\in\R$ from Lemma~\ref{lemma:Fourier_basic}.

Invoking now H\"{o}lder's inequality with exponents $\mathfrak{p}=\frac{r}{r-1}$, $\mathfrak{q}=r$, which are such that $\mathfrak{p},\,\mathfrak{q}\geq 1$ with $\frac{1}{\mathfrak{p}}+\frac{1}{\mathfrak{q}}=1$ (formally also when $r=1$), we discover
\begin{align*}
	\ave*{\int_\R\abs{\hat{f}(\xi,t)}^{r-1}\abs{\hat{f}(p\xi,t)}\,d\xi} &\leq \left(\int_\R\abs{\hat{f}(\xi,t)}^r\,d\xi\right)^\frac{r-1}{r}
		\ave*{\int_\R\abs{\hat{f}(p\xi,t)}^r\,d\xi}^\frac{1}{r} & (\text{let } \eta:=p\xi) \\
        &= \left(\int_\R\abs{\hat{f}(\xi,t)}^r\,d\xi\right)^\frac{r-1}{r}
		\ave*{\frac{1}{p^{1/r}}}\left(\int_\R\abs{\hat{f}(\eta,t)}^r\,d\eta\right)^\frac{1}{r} \\
	&= \ave*{\frac{1}{p^{1/r}}}\norm{\hat{f}(t)}{L^r}^r,
\end{align*}
whence
$$ \partial_t\left(e^{rt}\norm{\hat{f}(t)}{L^r}^r\right)\leq r\ave*{\frac{1}{p^{1/r}}}e^{rt}\norm{\hat{f}(t)}{L^r}^r. $$
Gr\"onwall's inequality applied to the function $e^{rt}\norm{\hat{f}(t)}{L^r}^r$ implies then
$$ \norm{\hat{f}(t)}{L^r}^r\leq\norm{\hat{f}_0}{L^r}^re^{r\left(\ave*{\frac{1}{p^{1/r}}}-1\right)t} $$
and finally
$$ \norm{\hat{f}(t)}{L^r}\leq\norm{\hat{f}_0}{L^r}e^{\left(\ave*{\frac{1}{p^{1/r}}}-1\right)t}. $$
This shows that if $\norm{\hat{f}_0}{L^r}<+\infty$ then $\norm{\hat{f}(t)}{L^r}<+\infty$ also for $t>0$, as desired.
\end{proof}

\begin{remark}
As a by-product of Proposition~\ref{prop:Lr_regularity}, we obtain an explicit estimate of the growth of the $L^r$-norm of $\hat{f}$ in time. In particular, by inspecting the proof carefully we see that the same argument can be repeated with $q$ in place of $p$, provided $q$ is non-zero, whereby we conclude
\begin{equation}
	\norm{\hat{f}(t)}{L^r}\leq\norm{\hat{f}_0}{L^r}e^{\left(\min\left\{\ave*{\frac{1}{p^{1/r}}},\,\ave*{\frac{1}{q^{1/r}}}\right\}-1\right)t},
		\qquad t>0.
	\label{eq:Fourier_Lm}
\end{equation}
\end{remark}

\subsubsection{\textit{A priori} regularity of the solution}
\label{sect:a_priori_reg}
For $r=2$,~\eqref{eq:Fourier_Lm} becomes straightforwardly an \textit{a priori} estimate on the $L^2$-norm of the distribution function $f$ itself, thanks to Parseval's identity. Specifically:
\begin{proposition} \label{prop:L2_regularity}
If either $p$ or $q$ in~\eqref{eq:v'} is non-zero and $f_0\in L^2(\R)$ then $f(\cdot,t)\in L^2(\R)$ for $t>0$ with
$$ \norm{f(t)}{L^2}\leq\norm{f_0}{L^2}e^{\left(\min\left\{\ave*{\frac{1}{\sqrt{p}}},\,\ave*{\frac{1}{\sqrt{q}}}\right\}-1\right)t}. $$
\end{proposition}
Therefore, $f(\cdot,t)$ is more regular than simply a probability measure in $\cP(\R)$.

Sticking to this case, we can show that if for some $m\in\N\setminus\{0\}$ it results $f_0\in H^m(\R)$, the Sobolev space of $L^2$ functions with up to their $m$-th derivative in $L^2$, then $f(\cdot,t)\in H^m(\R)$ for $t>0$ as well. In other words, the Boltzmann-type equation~\eqref{eq:Boltztype.strong} propagates in time the Sobolev regularity of the initial datum. Also in this case we rely on the Fourier-transformed Boltzmann-type equation~\eqref{eq:FourierBoltz}, thanks to the fact that the $H^m$-norm can be given the following representation in terms of the Fourier transform:
$$ \norm{f(t)}{H^m}:={\left(\int_\R{\left(1+\abs{\xi}^2\right)}^m\abs{\hat{f}(\xi,t)}^2\,d\xi\right)}^{1/2}. $$

The precise result goes as follows:
\begin{proposition} \label{prop:Hm_regularity}
Assume that either coefficient $p$, $q$ in~\eqref{eq:v'} is uniformly bounded away from zero, i.e.
$$ p\geq\delta \quad \text{or} \quad q\geq\delta $$
for some $0<\delta\leq 1$. If $f_0\in H^m(\R)$ for some $m\in\N\setminus\{0\}$ then $f(\cdot,t)\in H^m(\R)$ for $t>0$.
\end{proposition}
\begin{remark}
The requirement $\delta\leq 1$ is not restrictive but is convenient to obtain easier estimates in the proof.
\end{remark}
\begin{proof}
Without loss of generality, throughout the proof we shall assume $p\geq\delta$.

The same argument used in the proof of Proposition~\ref{prop:Lr_regularity}, applied with $r=2$, leads to
$$ \partial_t\left(e^{2t}\abs{\hat{f}}^2\right)\leq 2e^{2t}\abs{\hat{f}}\ave{\abs{\hat{f}(p\xi,t)}\cdot\abs{\hat{f}(q\xi,t)}}. $$
Multiplying both sides by ${(1+\abs{\xi}^2)}^m$ and integrating with respect to $\xi\in\R$ produces
\begin{align*}
	\partial_t\left(e^{2t}\norm{f(t)}{H^m}^2\right) &\leq 2e^{2t}\ave*{\int_\R\left(1+\abs{\xi}^2\right)^m\abs{\hat{f}(\xi,t)}\cdot\abs{\hat{f}(p\xi,t)}\,d\xi},
	\intertext{where we have used $\abs{\hat{f}(q\xi,t)}\leq 1$ from Lemma~\ref{lemma:Fourier_basic}. Cauchy-Schwarz inequality implies further}
	&\leq 2e^{2t}\ave*{\left(\int_\R\left(1+\abs{\xi}^2\right)^m\abs{\hat{f}(\xi,t)}^2\,d\xi\right)^{1/2}
		\left(\int_\R\left(1+\abs{\xi}^2\right)^m\abs{\hat{f}(p\xi,t)}^2\,d\xi\right)^{1/2}} \\
	&= 2e^{2t}\norm{f(t)}{H^m}\ave*{\left(\int_\R\left(1+\abs{\xi}^2\right)^m\abs{\hat{f}(p\xi,t)}^2\,d\xi\right)^{1/2}}.
\end{align*}
We develop now the remaining integral as
\begin{align*}
	\int_\R\left(1+\abs{\xi}^2\right)^m\abs{\hat{f}(p\xi,t)}^2\,d\xi &= \int_\R\sum_{k=0}^{m}\binom{m}{k}\abs{\xi}^{2k}\abs{\hat{f}(p\xi,t)}^2\,d\xi \\
	&= \int_\R\sum_{k=0}^{m}\binom{m}{k}\frac{\abs{\eta}^{2k}}{p^{2k+1}}\abs{\hat{f}(\eta,t)}^2\,d\eta & \text{(change of variable $\eta:=p\xi$)} \\
	&\leq \delta^{-(2m+1)}\norm{f(t)}{H^m}^2
\end{align*}
and we continue the previous estimate as
$$ \partial_t\left(e^{2t}\norm{f(t)}{H^m}^2\right)\leq 2\delta^{-(m+1/2)}e^{2t}\norm{f(t)}{H^m}^2. $$
Gr\"onwall's inequality applied to $e^{2t}\norm{f(t)}{H^m}^2$ yields
$$ e^{2t}\norm{f(t)}{H^m}^2\leq\norm{f_0}{H^m}^2e^{2\delta^{-(m+1/2)}t}, $$
i.e.
$$ \norm{f(t)}{H^m}\leq\norm{f_0}{H^m}e^{\left(\delta^{-(m+1/2)}-1\right)t}, $$
which confirms that if $\norm{f_0}{H^m}<+\infty$ then $\norm{f(t)}{H^m}<+\infty$ as well for $t>0$.
\end{proof}

For an alternative proof of Proposition~\ref{prop:Hm_regularity}, applied to a kinetic equation involving a more elaborated collisional operator, see~\cite{bisi2009CMS}.

\begin{remark}
The assumption that either $p$ or $q$ is non-zero is essential in both Propositions~\ref{prop:L2_regularity},~\ref{prop:Hm_regularity}. If $p$, $q$ vanish simultaneously, from~\eqref{eq:FourierBoltz} we get the equation $\partial_t\hat{f}=1-\hat{f}$, which is solved by $\hat{f}(\xi,t)=e^{-t}\hat{f}_0(\xi)+1-e^{-t}$. By inverse Fourier transform we find then
$$ f(v,t)=e^{-t}f_0(v)+(1-e^{-t})\delta_0(v), $$
which is a probability measure in $\cP(\R)$ but not a function in either $H^m(\R)$ or $L^2(\R)$ because of the Dirac delta.
\end{remark}

\subsection{Fourier metrics}
A crucial aspect for the development of a qualitative theory of the Boltzmann-type equation~\eqref{eq:Boltztype.strong} is the possibility to measure distances among probability measures. This is essential in order e.g., to prove existence and uniqueness of the solutions as well as to assess their stability with respect to perturbations of the initial datum or their convergence in certain asymptotic regimes.

The theory of the (optimal) transport of measures provides a prominent example of a family of metrics in the space of probability measures, the so-called \textit{Wasserstein distances}, which is thoroughly presented in~\cite{ambrosio2008BOOK,villani2009BOOK} and has also been used to address some qualitative properties of Boltzmann-type equations, see e.g.,~\cite{bisi2024PHYSD,freguglia2017CMS}. Nevertheless, owing to the Fourier representation~\eqref{eq:FourierBoltz} of~\eqref{eq:Boltztype.strong}, it turns out that another metric, based on the Fourier transform, is even more suited to treat collisional kinetic equations.

To introduce it, we define first the following family of spaces of probability measures:
$$ \cP_s(\R):=\left\{\mu\in\cP(\R)\,:\,\int_\R\abs{v}^s\,d\mu(v)<+\infty\right\}, \quad s>0. $$
Notice that $\cP_s(\R)\subseteq\cP_r(\R)$ if $0<r<s$, indeed H\"{o}lder's inequality with exponents $\mathfrak{p}=\frac{s}{r}$ and $\mathfrak{q}=\frac{s}{s-r}$, which are such that $\mathfrak{p},\,\mathfrak{q}\geq 1$ with $\frac{1}{\mathfrak{p}}+\frac{1}{\mathfrak{q}}=1$, implies
$$ \int_\R\abs{v}^r\,d\mu(v)\leq\left(\int_\R\abs{v}^s\,d\mu(v)\right)^{r/s}. $$

Next, given any two probability measures $\mu,\,\nu\in\cP(\R)$ we define their $s$-\textit{Fourier distance} as
\begin{equation}
	d_s(\mu,\nu):=\sup_{\xi\in\R\setminus\{0\}}\frac{\abs{\hat{\nu}(\xi)-\hat{\mu}(\xi)}}{{\abs{\xi}}^s}.
	\label{eq:ds}
\end{equation}

For an exhaustive review of the mathematical properties of $d_s$ and of its relationship with the Wasserstein distance we refer the interested reader to~\cite{carrillo2007RMUP}, see also~\cite{auricchio2020RLMA,duering2009RMUP,gabetta1995JSP,spiga2004AML}. Here, we collect two of its main features, which will be mostly useful in the sequel.
\begin{proposition} \label{prop:ds}
\begin{enumerate}[label=(\roman*)]
\item \label{prop:ds.finiteness} Given $\mu,\,\nu\in\cP_s(\R)$, $s>0$, let
$$ M^\mu_n:=\int_\R v^n\,d\mu(v), \qquad M^\nu_n:=\int_\R v^n\,d\nu(v) $$
be their respective $n$-th order moments ($n\in\N$, $n\leq s$). Moreover, let us denote by $[s]$ the integer part of $s$. If $M^\mu_n=M^\nu_n$ for every $n\leq [s]$ if $s\not\in\N$, or for every $n\leq s-1$ if $s\in\N$, then $d_s(\mu,\nu)<+\infty$.
\item \label{prop:cP.completeness} For $\alpha>0$, let $\cP_{s,\alpha,\cM_{s+\alpha}}(\R)$ denote the subset of $\cP_{s+\alpha}(\R)$ made of probability measures $\mu$ with prescribed moments $M^\mu_n$, $n\leq [s]$, and such that $\int_\R\abs{v}^{s+\alpha}\,d\mu(v)$ is uniformly bounded with respect to $\mu$ by a constant $\cM_{s+\alpha}>0$. Then $\cP_{s,\alpha,\cM_{s+\alpha}}(\R)$ endowed with the distance $d_s$ is a complete metric space.
\end{enumerate}
\end{proposition}
\begin{proof}
For a thorough proof of these results we refer to~\cite[Propositions~2.6,~2.7]{carrillo2007RMUP}. Here, we only show~\ref{prop:ds.finiteness} for $s\in\N$, which is the case we shall mainly deal with in the sequel.

Let then $\mu,\,\nu\in\cP_s(\R)$ with $s\in\N$. By Taylor-expanding the mapping $\xi\mapsto e^{-i\xi v}$ up to the order $s-1$ with centre in $\xi=0$ and Lagrange remainder we find
$$ e^{-i\xi v}=\sum_{n=0}^{s-1}\frac{(-iv)^n}{n!}\xi^n+\frac{(-iv)^se^{-i\bar{\xi}v}}{s!}\xi^s, $$
where $\bar{\xi}=\theta\xi$ for a certain $\theta\in [0,\,1]$. Then:
\begin{align*}
    \hat{\mu}(\xi) &= \sum_{n=0}^{s-1}\frac{(-i)^n\xi^n}{n!}\int_\R v^n\,d\mu(v)+\frac{(-i)^s\xi^s}{s!}\int_\R v^se^{-i\bar{\xi}v}\,d\mu(v) \\
    &= \sum_{n=0}^{s-1}\frac{(-i)^n\xi^n}{n!}M^\mu_n+\frac{(-i)^s\xi^s}{s!}\int_\R v^se^{-i\bar{\xi}v}\,d\mu(v)
\end{align*}
and similarly for $\hat{\nu}(\xi)$, whence, using that $M^\mu_n=M^\nu_n$ for every $n\leq s-1$, we discover
$$ \abs{\hat{\mu}(\xi)-\hat{\nu}(\xi)}=\frac{\abs{\xi}^s}{s!}\abs*{\int_\R v^se^{-i\bar{\xi}v}\,d(\mu-\nu)(v)}\leq
    \frac{\abs{\xi}^s}{s!}\int_\R\abs{v}^s\,d(\mu+\nu)(v). $$
Consequently,
$$ \frac{\abs{\hat{\mu}(\xi)-\hat{\nu}(\xi)}}{\abs{\xi}^s}\leq\frac{1}{s!}\int_\R\abs{v}^s\,d(\mu+\nu)(v) $$
and the finiteness of $\int_\R\abs{v}^s\,d\mu(v),\,\int_\R\abs{v}^s\,d\nu(v)$ yields the thesis.
\end{proof}

We stress that:
\begin{enumerate}[label=(\roman*)]
\item The finiteness of $d_s$ is not guaranteed, in general, for any two measures in $\cP_s(\R)$. The two measures need to have equal moments up to the order $s-1$ or $[s]$, depending on whether $s$ is or is not integer. Nevertheless, $d_s$ with $0<s\leq 1$ is always finite for every $\mu,\,\nu\in\cP(\R)$, because clearly $M^\mu_0=M^\nu_0=1$.
\item The second point of Proposition~\ref{prop:ds} does \textit{not} assert that $\cP_{s+\alpha}(\R)$ is complete with the metric $d_s$. The reason is twofold: on one hand, we need that the probability measures have all equal moments up to the order $[s]$; on the other hand, we need that the quantity $\int_\R\abs{v}^{s+\alpha}\,d\mu(s)$ be bounded from above by a constant $\cM_{s+\alpha}$ \textit{independent} of $\mu$.
\end{enumerate}

\begin{remark} \label{rem:P_s_s-1}
The completeness of the metric space $(\cP_{s,\alpha,\cM_{s+\alpha}}(\R),\,d_s)$ is a useful starting point to prove existence of solutions to the Boltzmann-type equation~\eqref{eq:Boltztype.strong} via fixed point arguments. 

Concerning this, we notice that by fixing $s\in (0,\,1)$ with $\alpha=1-s>0$ we obtain that all the elements of $\cP_{s+\alpha}(\R)=\cP_1(\R)$ have trivially equal moments up to the order $n=[s]=0$. To ascertain if we can find $\cM_1>0$ such that it makes sense to look for the solution to~\eqref{eq:Boltztype.strong} in $\cP_{s,1-s,\cM_1}(\R)$ endowed with the metric $d_s$ we need to check whether $\int_\R\abs{v}^{s+\alpha}f(v,t)\,dv=\int_\R\abs{v}f(v,t)\,dv$ can be bounded independently of $f$. This is clearly not true for a generic $f(\cdot,t)\in\cP_1(\R)$, for which we only know that $\int_\R\abs{v}f(v,t)\,dv$ is finite but an upper bound depends, in general, on $f$. Nevertheless, with $\varphi(v)=\abs{v}$ in~\eqref{eq:Boltztype.weak} we discover that any prospective solution $f$ to~\eqref{eq:Boltztype.strong} satisfies the \textit{a priori} estimate
\begin{align*}
	\frac{d}{dt}\int_\R\abs{v}f(v,t)\,dv &= \int_\R\int_\R\ave{\abs{pv+qv_\ast}-\abs{v}}f(v,t)f(v_\ast,t)\,dv\,dv_\ast \\
	&\leq \int_\R\int_\R\ave{(p-1)\abs{v}+q\abs{v_\ast}}f(v,t)f(v_\ast,t)\,dv\,dv_\ast \\
	&= \left(\ave{p+q}-1\right)\int_\R\abs{v}f(v,t)\,dv,
\end{align*}
whence
$$ \int_\R\abs{v}f(v,t)\,dv\leq e^{\left(\ave{p+q}-1\right)t}\int_\R\abs{v}f_0(v)\,dv. $$
Consequently, if $\ave{p+q}\leq 1$ it makes sense to look for solutions to~\eqref{eq:Boltztype.strong} in $\cP_{s,1-s,\cM_1}(\R)$ with $\cM_1:=\int_\R\abs{v}f_0(v)\,dv$.
\end{remark}

\begin{remark} \label{rem:P_1_1}
Assume that $p$, $q$ are such that $\ave{p+q}=1$. According to Section~\ref{sect:moments_evol}, the Boltzmann-type equation~\eqref{eq:Boltztype.strong} conserves in time the mean state of the system, therefore any prospective solution $f$ does not only satisfy $M_0=1$ but also $M_1=M_{1,0}$, being $M_{1,0}\in\R$ a prescribed constant (the mean value of the initial condition $f_0$). Consequently, the search for solutions to~\eqref{eq:Boltztype.strong} can be set in some $\cP_{1,\alpha,\cM_{1+\alpha}}(\R)$, which is complete with the metric $d_1$, provided proper $\alpha,\,\cM_{1+\alpha}>0$ can be found in such a way that $\int_\R\abs{v}^{1+\alpha}f(v,t)\,dv\leq\cM_{1+\alpha}$.

In particular, for $\alpha=1$ this amounts to controlling the second moment $M_2$ of $f$. From Section~\ref{sect:moments_evol} we know that
$$ \frac{dM_2}{dt}=\left(\ave{p^2+q^2}-1\right)M_2+2\ave{pq}M_{1,0}^2, $$
whence
\begin{equation}
	M_2(t)=e^{\left(\ave{p^2+q^2}-1\right)t}M_{2,0}+2\left(1-e^{\left(\ave{p^2+q^2}-1\right)t}\right)\frac{\ave{pq}}{1-\ave{p^2+q^2}}M_{1,0}^2,
	\label{eq:M2.constant_M1}
\end{equation}
which, if $\ave{p^2+q^2}<1$, is bounded for all $t>0$ by
\begin{equation}
	\cM_2:=M_{2,0}+\dfrac{2\ave{pq}}{1-\ave{p^2+q^2}}M_{1,0}^2.
	\label{eq:cM2}
\end{equation}
Therefore, in such a case we can look for solutions to~\eqref{eq:Boltztype.strong} in $(\cP_{1,1,\cM_2}(\R),\,d_1)$.

Condition $\ave{p^2+q^2}<1$ is related to the \textit{dissipation} of the energy (i.e. the second moment of $f$). If $\ave{(p+q)^2}=1$ holds simultaneously, which implies $1-\ave{p^2+q^2}=2\ave{pq}$, then from~\eqref{eq:M2.constant_M1} we see that $M_2(t)\to M_{1,0}^2$ as $t\to +\infty$, hence the internal energy (i.e. the variance of $f$) decays to zero exponentially fast in time. Still from~\eqref{eq:M2.constant_M1} we also see that, in this case, we can take $\cM_2:=M_{2,0}$, i.e. the energy is bounded for all $t>0$ by the initial energy. In the jargon of classical kinetic theory this is referred to as \textit{cooling}, because in gas dynamics the internal energy of a gas is proportional to its temperature.
\end{remark}

\begin{example}
Let us consider the prototypical case of interaction rule
$$ v'=v+\gamma(v_\ast-v)+v\eta, $$
which characterises many celebrated kinetic models of socio-economic systems, see e.g.,~\cite{cordier2005JSP}. Here, $\gamma\in (0,\,1)$ is a prescribed parameter whereas $\eta\in\R$ is a centred random variable, i.e. one with $\ave{\eta}=0$, which models a stochastic fluctuation. This rule is of the form~\eqref{eq:v'} with
$$ p=1-\gamma+\eta, \qquad q=\gamma. $$
In order for $p\geq 0$ we need that $\eta$ be supported in the interval $[\gamma-1,\,+\infty)$. Notice that such an $\eta$ is allowed to take also negative values, because $\gamma-1<0$, therefore this requirement does not conflict \textit{a priori} with the assumption $\ave{\eta}=0$.

Condition of Remark~\ref{rem:P_s_s-1} is satisfied as an equality: $\ave{p+q}=1$, whereas $\ave{p^2+q^2}=1-2\gamma(1-\gamma)+\sigma^2$, being $\sigma^2:=\ave{\eta^2}$ the variance of $\eta$. Therefore, the dissipative condition of Remark~\ref{rem:P_1_1} is satisfied provided $\gamma(1-\gamma)\geq\frac{\sigma^2}{2}$. In particular, if $\sigma^2\leq\frac{1}{2}$ there are values of $\gamma$ for which the condition holds, hence for which solutions to~\eqref{eq:Boltztype.strong} can be sought in the complete metric space $(\cP_{1,1,\cM_2}(\R),\,d_1)$.

We observe furthermore that condition $\ave{(p+q)^2}=1$ yields $\sigma^2=0$. Therefore, cooling occurs only if $\eta$ is the null random variable. Instead, if $\sigma^2>0$ the system can exhibit non-trivial emerging trends because $\eta$ acts as an external source of energy.
\end{example}

We conclude this section by recording a further simple property of the $s$-Fourier metric, which we shall use frequently in the sequel. Given $\mu,\,\nu\in\cP(\R)$ and $a\in\R\setminus\{0\}$, it results
\begin{align*}
	\sup_{\xi\in\R\setminus\{0\}}\frac{\abs{\hat{\nu}(a\xi)-\hat{\mu}(a\xi)}}{\abs{\xi}^s}
		&= \abs{a}^s\sup_{\xi\in\R\setminus\{0\}}\frac{\abs{\hat{\nu}(a\xi)-\hat{\mu}(a\xi)}}{\abs{a\xi}^s} & \text{(let $\eta:=a\xi$)} \\
	&= \abs{a}^s\sup_{\eta\in\R\setminus\{0\}}\frac{\abs{\hat{\nu}(\eta)-\hat{\mu}(\eta)}}{\abs{\eta}^s}=\abs{a}^sd_s(\mu,\nu). 
\end{align*}

\section{Basic well-posedness theory}
\label{sect:basic_theory}
Remarks~\ref{rem:P_s_s-1},~\ref{rem:P_1_1} have shown that the choice of a convenient functional space where to look for solutions to the Boltzmann-type equation~\eqref{eq:Boltztype.strong} is linked, to some extent, to certain physical properties of the interaction rule~\eqref{eq:v'}. Because of their relevance for applications and their potential to generate physically relevant trends, here we focus specifically on the properties discussed in Remark~\ref{rem:P_1_1}. Therefore, we shall assume henceforth that the coefficients $p,\,q\in\R_+$ satisfy
\begin{equation}
	\ave{p+q}=1, \qquad \ave{p^2+q^2}<1.
	\label{eq:assumptions_p.q}
\end{equation}
Nevertheless, techniques analogous to those that we shall present in the sequel can be used to repeat the theory, with the due modifications, also in the case $\ave{p+q}<1$ discussed in Remark~\ref{rem:P_s_s-1}.
\begin{remark} \label{rem:pq.average}
Assumptions~\eqref{eq:assumptions_p.q} entail precise constraints on the admissible range of the mean values of $p$, $q$. Indeed, the second assumption in~\eqref{eq:assumptions_p.q} implies $\ave{p^2},\,\ave{q^2}<1$, whence by Jensen's inequality
$$ \ave{p}^2\leq\ave{p^2}<1, \qquad \ave{q}^2\leq\ave{q^2}<1. $$
It follows $-1<\ave{p},\,\ave{q}<1$, but invoking the first assumption in~\eqref{eq:assumptions_p.q} we discover more precisely
$$ \ave{p}=1-\ave{q}>0, \qquad \ave{q}=1-\ave{p}>0, $$
thus actually $0<\ave{p},\,\ave{q}<1$. Notice that this holds independently of the sign of $p$, $q$.
\end{remark}

For $T>0$ we consider the space
$$ X:=C^0([0,\,T];\cP_{1,1,\cM_2}(\R)) $$
of continuous mappings from $[0,\,T]\subset\R$ to $\cP_{1,1,\cM_2}(\R)$, where $\cM_2>0$ is chosen as in Remark~\ref{rem:P_1_1}. If $f\in X$, we write $f(t)$ to mean the probability measure $f(\cdot,t)\in\cP_{1,1,\cM_2}(\R)$ for fixed $t\in [0,\,T]$.

Given $f,\,g\in X$, we define the mapping $\varrho:X\times X\to\R_+$
$$ \varrho(f,g):=\sup_{t\in [0,\,T]}d_1(f(t),g(t)), $$
which is a metric in $X$. Furthermore, the metric space $(X,\,\varrho)$ is complete.

\subsection{Existence and uniqueness of the solution}
\label{sect:exist.uniq}
We begin the qualitative theory of the Boltzmann-type equation~\eqref{eq:Boltztype.strong} by addressing existence and uniqueness of the solution to the initial-value problem
\begin{equation}
	\begin{cases}
		\partial_tf=Q(f,f), & t>0 \\
		f(v,0)=f_0(v),
	\end{cases}
	\label{eq:IVP}
\end{equation}
where the collisional operator $Q$ is defined in~\eqref{eq:Q.strong} and $f_0$ is a prescribed initial datum.

\begin{theorem} \label{theo:exists_unique}
Under~\eqref{eq:assumptions_p.q}, for $f_0\in\cP_2(\R)$ the initial-value problem~\eqref{eq:IVP} admits a unique solution $f\in X$.
\end{theorem}
\begin{proof}
Throughout the proof we shall use the assumption $\ave{p+q}=1$ without further notice.

We obtain the thesis via Banach's fixed-point theorem. To apply it, we observe preliminarily that the weak form of $\partial_tf=Q(f,f)$, i.e.~\eqref{eq:Boltztype.weak}, can be rewritten as
\begin{align}
	\begin{aligned}[b]
		\int_\R\varphi(v)f(v,t)\,dv &= e^{-t}\int_\R\varphi(v)f_0(v)\,dv \\
		&\phantom{=} +\int_0^te^{-(t-\tau)}\int_\R\int_\R\ave{\varphi(v')}f(v,\tau)f(v_\ast,\tau)\,dv\,dv_\ast\,d\tau
	\end{aligned}
	\label{eq:Boltztype.mild}
\end{align}
for all observable quantities $\varphi$ upon multiplying~\eqref{eq:Boltztype.weak} by $e^t$ and integrating in time over $[0,\,t]$, $0<t\leq T$. We understand~\eqref{eq:Boltztype.mild} as the weak form of the fixed-point equation
$$ f=\cQ(f), $$
where $\cQ$ is the operator defined on $X$ such that
\begin{align}
	\begin{aligned}[b]
		\int_\R\varphi(v)\cQ(f)(v,t)\,dv &= e^{-t}\int_\R\varphi(v)f_0(v)\,dv \\
		&\phantom{=} +\int_0^te^{-(t-\tau)}\int_\R\int_\R\ave{\varphi(v')}f(v,\tau)f(v_\ast,\tau)\,dv\,dv_\ast\,d\tau
	\end{aligned}
	\label{eq:cQ.weak}
\end{align}
for all $\varphi$. To show that $\cQ$ admits a unique fixed point in $X$ we check the assumptions of Banach's fixed-point theorem.
\begin{enumerate}
\item First, we show that $\cQ$ maps $X$ into itself. For this, let $f\in X$, then:
\begin{enumerate}[label=(\roman*)]
\item The non-negativity of  $\cQ(f)(t)$ is clear from~\eqref{eq:cQ.weak}, as the non negativity of $f(t)$ implies
$$ \int_\R\varphi(v)\cQ(f)(v,t)\,dv\geq 0, \qquad \forall\,\varphi\geq 0. $$
Together with the fact that with $\varphi\equiv 1$ it results
$$ \int_\R\cQ(f)(v,t)\,dv=e^{-t}+\int_0^te^{-(t-\tau)}\,d\tau=1, $$
this says that $\cQ(f)(t)$ is a probability measure for $t>0$.

\item Moreover, $\varphi(v)=v$ yields
$$ \int_\R v\cQ(f)(v,t)\,dv=\left(e^{-t}+\int_0^te^{-(t-\tau)}\,d\tau\right)M_{1,0}=M_{1,0}. $$
On the other hand, $\varphi(v)=v^2$ produces
\begin{align*}
	\int_\R v^2\cQ(f)(v,t)\,dv &= e^{-t}M_{2,0}+\ave{p^2+q^2}\int_0^te^{-(t-\tau)}M_2(\tau)\,d\tau \\
	&\phantom{=} +2\ave{pq}M_{1,0}^2\int_0^te^{-(t-\tau)}\,d\tau \\
	&\leq e^{-t}M_{2,0}+(1-e^{-t})\left(\ave{p^2+q^2}\cM_2+2\ave{pq}M_{1,0}^2\right).
\intertext{Using the expression~\eqref{eq:cM2} of $\cM_2$ we continue the estimate as}
	&= e^{-t}M_{2,0}+(1-e^{-t})\biggl[\ave{p^2+q^2}\left(M_{2,0}+\frac{2\ave{pq}}{1-\ave{p^2+q^2}}M_{1,0}^2\right) \\
	&\phantom{\leq} +2\ave{pq}M_{1,0}^2\biggr] \\
	&\leq e^{-t}M_{2,0}+(1-e^{-t})\left[M_{2,0}+2\ave{pq}\left(\frac{\ave{p^2+q^2}}{1-\ave{p^2+q^2}}+1\right)M_{1,0}^2\right] \\
	&= M_{2,0}+(1-e^{-t})\frac{2\ave{pq}}{1-\ave{p^2+q^2}}M_{1,0}^2 \\
	&\leq\cM_2.
\end{align*}
Therefore, we conclude that $\cQ(f)(t)\in\cP_{1,1,\cM_2}(\R)$ whenever $f(t)\in\cP_{1,1,\cM_2}(\R)$.

\item As for the continuity of the mapping $t\mapsto\cQ(f)(t)$, we fix $t_0\in (0,\,+\infty)$ and check that $d_1(\cQ(f)(t),\cQ(f)(t_0))\to 0$ when $t\to t_0$. For the sake of simplicity, and without loss of generality, we assume $t>t_0$. We observe that
$$ d_1(\cQ(f)(t),\cQ(f)(t_0))=\sup_{\xi\in\R\setminus\{0\}}\frac{\abs{\widehat{\cQ(f)}(t)-\widehat{\cQ(f)}(t_0)}}{\abs{\xi}} $$
and that we can obtain $\widehat{\cQ(f)}(t)$ from~\eqref{eq:cQ.weak} with $\varphi(v)=e^{-i\xi v}$:
\begin{equation}
	\widehat{\cQ(f)}(\xi,t)=e^{-t}\hat{f}_0(\xi)+\int_0^te^{-(t-\tau)}\ave{\hat{f}(p\xi,\tau)\hat{f}(q\xi,\tau)}\,d\tau.
	\label{eq:cQhat}
\end{equation}
In particular,
\begin{align*}
	\widehat{\cQ(f)}(\xi,t)-\widehat{\cQ(f)}(\xi,t_0) &=e^{-t}\left(\hat{f}_0(\xi)+\int_0^te^\tau\ave{\hat{f}(p\xi,\tau)\hat{f}(q\xi,\tau)}\,d\tau\right) \\
	&\phantom{=} -e^{-t_0}\left(\hat{f}_0(\xi)+\int_0^{t_0}e^\tau\ave{\hat{f}(p\xi,\tau)\hat{f}(q\xi,\tau)}\,d\tau\right)
\intertext{whence, adding and subtracting $\ave{\hat{f}(p\xi,t_0)\hat{f}(q\xi,t_0)}$ in the integrals,}
		&= e^{-t}\biggl(\hat{f}_0(\xi)+\int_0^te^\tau\ave{\hat{f}(p\xi,\tau)\hat{f}(q\xi,\tau)-\hat{f}(p\xi,t_0)\hat{f}(q\xi,t_0)}\,d\tau \\
	&\phantom{=} +(e^t-1)\ave{\hat{f}(p\xi,t_0)\hat{f}(q\xi,t_0)}\biggr) \\
	&\phantom{=} \resizebox{.62\textwidth}{!}{$
		\displaystyle -e^{-t_0}\biggl(\hat{f}_0(\xi)+\int_0^{t_0}e^\tau\ave{\hat{f}(p\xi,\tau)\hat{f}(q\xi,\tau)-\hat{f}(p\xi,t_0)\hat{f}(q\xi,t_0)}\,d\tau
	$} \\
	&\phantom{=} +(e^{t_0}-1)\ave{\hat{f}(p\xi,t_0)\hat{f}(q\xi,t_0)}\biggr) \\
	&= (e^{-t}-e^{-t_0})\biggl(\ave{\hat{f}_0(\xi)-\hat{f}(p\xi,t_0)\hat{f}(q\xi,t_0)} \\
	&\phantom{=} +\int_0^{t_0}e^\tau\ave{\hat{f}(p\xi,\tau)\hat{f}(q\xi,\tau)-\hat{f}(p\xi,t_0)\hat{f}(q\xi,t_0)}\,d\tau\biggr) \\
	&\phantom{=} +e^{-t}\int_{t_0}^te^\tau\ave{\hat{f}(p\xi,\tau)\hat{f}(q\xi,\tau)-\hat{f}(p\xi,t_0)\hat{f}(q\xi,t_0)}\,d\tau.
\end{align*}
Consequently,
\begin{align*}
	\frac{\abs{\widehat{\cQ(f)}(\xi,t)-\widehat{\cQ(f)}(\xi,t_0)}}{\abs{\xi}} &\leq \abs{e^{-t}-e^{-t_0}}\left(\ave*{\frac{\abs{\hat{f}_0(\xi)-\hat{f}(p\xi,t_0)\hat{f}(q\xi,t_0)}}{\abs{\xi}}}\right. \\
	&\phantom{=} \left.+\int_0^{t_0}e^{\tau}\ave*{\frac{\abs{\hat{f}(p\xi,\tau)\hat{f}(q\xi,\tau)-\hat{f}(p\xi,t_0)\hat{f}(q\xi,t_0)}}{\abs{\xi}}}\,d\tau\right) \\
	&\phantom{=} +\int_{t_0}^te^\tau\ave*{\frac{\abs{\hat{f}(p\xi,\tau)\hat{f}(q\xi,\tau)-\hat{f}(p\xi,t_0)\hat{f}(q\xi,t_0)}}{\abs{\xi}}}\,d\tau.
\end{align*}
Since
\begin{align*}
	\abs{\hat{f}(p\xi,\tau)\hat{f}(q\xi,\tau)-\hat{f}(p\xi,t_0)\hat{f}(q\xi,t_0)} &\leq
		\abs{\hat{f}(q\xi,t_0)}\cdot\abs{\hat{f}(p\xi,\tau)-\hat{f}(p\xi,t_0)} \\
	&\phantom{\leq} +\abs{\hat{f}(p\xi,\tau)}\cdot\abs{\hat{f}(q\xi,\tau)-\hat{f}(q\xi,t_0)} \\
	&\leq \abs{\hat{f}(p\xi,\tau)-\hat{f}(p\xi,t_0)}+\abs{\hat{f}(q\xi,\tau)-\hat{f}(q\xi,t_0)},
\end{align*}
we bound
\begin{align*}
	&\ave*{\frac{\abs{\hat{f}(p\xi,\tau)\hat{f}(q\xi,\tau)-\hat{f}(p\xi,t_0)\hat{f}(q\xi,t_0)}}{\abs{\xi}}} \\
	&\qquad \leq\ave*{\frac{\abs{\hat{f}(p\xi,\tau)-\hat{f}(p\xi,t_0)}}{\abs{\xi}}}
		+\ave*{\frac{\abs{\hat{f}(q\xi,\tau)-\hat{f}(q\xi,t_0)}}{\abs{\xi}}} \\
	&\qquad =\ave*{p\frac{\abs{\hat{f}(p\xi,\tau)-\hat{f}(p\xi,t_0)}}{\abs{p\xi}}}
		+\ave*{q\frac{\abs{\hat{f}(q\xi,\tau)-\hat{f}(q\xi,t_0)}}{\abs{q\xi}}} \\
	&\qquad \leq d_1(f(\tau),f(t_0)).
\end{align*}
On the other hand,
\begin{align*}
	\abs{\hat{f}_0(\xi)-\hat{f}(p\xi,t_0)\hat{f}(q\xi,t_0)} &\leq \abs{\hat{f}_0(\xi)}\cdot\abs{1-\hat{f}(p\xi,t_0)}
		+\abs{\hat{f}(p\xi,t_0)}\cdot\abs{\hat{f}_0(\xi)-\hat{f}(q\xi,t_0)} \\
	&\leq \abs{1-\hat{f}(p\xi,t_0)}+\abs{\hat{f}_0(\xi)-\hat{f}_0(q\xi)}+\abs{\hat{f}_0(q\xi)-\hat{f}(q\xi,t_0)}
\end{align*}
with\footnote{Here we use $\abs{e^{iy}-e^{ix}}\leq\abs{y-x}$ for all $x,\,y\in\R$.}
$$ \abs{1-\hat{f}(p\xi,t_0)}\leq\int_\R\abs*{e^{-i0\cdot v}-e^{-ip\xi v}}f(v,t_0)\,dv\leq\abs{p\xi}\int_\R\abs{v}f(v,t_0)\,dv $$
and likewise
$$ \abs{\hat{f}_0(\xi)-\hat{f}_0(q\xi)}\leq\int_\R\abs{e^{-i\xi v}-e^{-iq\xi v}}f_0(v)\,dv\leq\abs{(1-q)\xi}\int_\R\abs{v}f_0(v)\,dv, $$
where we notice that $\int_\R\abs{v}f_0(v)\,dv,\,\int_\R\abs{v}f(v,t_0)\,dv<+\infty$ because $f_0,\,f(t_0)\in\cP_2(\R)$.

Therefore,
\begin{align*}
	\ave*{\frac{\abs{\hat{f}_0(\xi)-\hat{f}(p\xi,t_0)\hat{f}(q\xi,t_0)}}{\abs{\xi}}} &\leq
		\ave{p}\int_\R\abs{v}f_0(v)\,dv+\ave{\abs{1-q}}\int_\R\abs{v}f(v,t_0)\,dv \\
	&\phantom{\leq} +\ave*{\frac{\abs{\hat{f}_0(q\xi)-\hat{f}(q\xi,t_0)}}{\abs{\xi}}} \\
	&\leq \ave{p}\int_\R\abs{v}f_0(v)\,dv+\ave{\abs{1-q}}\int_\R\abs{v}f(v,t_0)\,dv \\
	&\phantom{\leq} +\ave{q}d_1(f_0,f(t_0)).
\end{align*}
Collecting all these estimates we discover:
\begin{align*}
	d_1(\cQ(f)(t),\cQ(f)(t_0)) &\leq (t-t_0)\biggl(\ave{p}\int_\R\abs{v}f_0(v)\,dv+\ave{\abs{1-q}}\int_\R\abs{v}f(v,t_0)\,dv \\
	&\phantom{\leq} +\ave{q}d_1(f_0,f(t_0))+\int_0^{t_0}e^\tau d_1(f(\tau),f(t_0))\,d\tau \\
	&\phantom{\leq} +e^{\bar{t}}d_1(f(\bar{t}),f(t_0))\biggr),
\end{align*}
where we have used the Lipschitz continuity of the mapping $t\mapsto e^{-t}$ for $t\geq 0$ and the mean value theorem for integrals (in particular, $\bar{t}$ is a point in $[t_0,\,t]$). Then clearly $d_1(\cQ(f)(t),\cQ(f)(t_0))\to 0$ when $t\to t_0$ and the continuity of $t\mapsto \cQ(f)(t)$ follows from the arbitrariness of $t_0$.
\end{enumerate}
\item Second, we show that $\cQ$ is a contraction on $X$. For this, let $f,\,g\in X$ and let us compute:
\begin{align*}
	\frac{\abs{\widehat{\cQ(g)}(\xi,t)-\widehat{\cQ(f)}(\xi,t)}}{\abs{\xi}} &\leq
		\int_0^te^{-(t-\tau)}\ave*{\frac{\abs{\hat{g}(p\xi,\tau)\hat{g}(q\xi,\tau)-\hat{f}(p\xi,\tau)\hat{f}(q\xi\,\tau)}}{\abs{\xi}}}\,d\tau \\
	&\leq\int_0^te^{-(t-\tau)}\left\langle\frac{\abs{\hat{g}(q\xi,\tau)}\cdot\abs{\hat{g}(p\xi,\tau)-\hat{f}(p\xi,\tau)}}{\abs{\xi}}\right. \\
	&\phantom{\leq} \left.+\frac{\abs{\hat{f}(p\xi,\tau)}\cdot\abs{\hat{g}(q\xi,\tau)-\hat{f}(q\xi,\tau)}}{\abs{\xi}}\right\rangle\,d\tau \\
	&\leq \int_0^te^{-(t-\tau)}\ave*{\frac{\abs{\hat{g}(p\xi,\tau)-\hat{f}(p\xi,\tau)}}{\abs{\xi}}+
		\frac{\abs{\hat{g}(q\xi,\tau)-\hat{f}(q\xi,\tau)}}{\abs{\xi}}}\,d\tau \\
	&\leq\int_0^te^{-(t-\tau)}d_1(f(\tau),g(\tau))\,d\tau.
\end{align*}
Consequently,
\begin{align*}
	\varrho(\cQ(f),\cQ(g)) &= \sup_{t\in [0,\,T]}d_1(\cQ(f)(t),\cQ(g)(t)) \\
	&\leq\varrho(f,g)\sup_{t\in [0,\,T]}\int_0^te^{-(t-\tau)}\,d\tau \\
	&= (1-e^{-T})\varrho(f,g),
\end{align*}
whence $\cQ$ is a contraction on $X$ for an arbitrary $T>0$. \qedhere
\end{enumerate}
\end{proof}

\subsection{Continuous dependence on the initial datum}
We complete the basic well-posedness theory of the Boltzmann-type equation~\eqref{eq:Boltztype.strong} by showing that its solution depends continuously on the initial datum.

\begin{theorem} \label{theo:cont_dep}
Under the assumptions of Theorem~\ref{theo:exists_unique}, let $f,\,g\in C^0([0,\,T];\,\cP_2(\R))$ be the solutions to~\eqref{eq:IVP} issuing from two given initial data $f_0,\,g_0\in\cP_2(\R)$, respectively. Then
$$ \varrho(f,g)\leq d_1(f_0,g_0). $$
\end{theorem}

\begin{remark}
Theorem~\ref{theo:cont_dep} does not require that the initial data $f_0$, $g_0$ have the same mean value. Therefore, the corresponding solutions to~\eqref{eq:Boltztype.strong} do not have, in general, the same mean value for $t>0$. This is the reason why, in the statement of this theorem, we regard $f(t)$, $g(t)$ as belonging generically to $\cP_2(\R)$ instead of specifically to $\cP_{1,1,\cM_2}(\R)\subset\cP_2(\R)$, indeed each of them belongs actually to a different $\cP_{1,1,\cM_2}(\R)$-like space characterised by a different value of the mean value $M_{1,0}$ (and possibly also of $M_{2,0}$, hence of $\cM_2$, cf.~\eqref{eq:cM2}). By the way, settling the theory in a $\cP_{s,\alpha,\cM_{s+\alpha}}(\R)$-like space is useful only when completeness comes explicitly into play.
\end{remark}

\begin{proof}[Proof of Theorem~\ref{theo:cont_dep}]
We point out preliminarily that the definition of the operator $\cQ$ introduced in the proof of Theorem~\ref{theo:exists_unique} depends on the initial datum of problem~\eqref{eq:IVP}, cf.~\eqref{eq:cQ.weak}. Since here it is important to distinguish different initial data, we denote by $\cQ_{f_0},\,\cQ_{g_0}$ the operators admitting $f,\,g$ as fixed points, respectively.

Since $d_1(f(t),g(t))=d_1(\cQ_{f_0}(f)(t),\cQ_{g_0}(g)(t))$, recalling~\eqref{eq:cQhat} we estimate:
\begin{align*}
	\frac{\abs{\widehat{\cQ_{g_0}(g)}(\xi,t)-\widehat{\cQ_{f_0}(f)}(\xi,t)}}{\abs{\xi}} &\leq
		e^{-t}\frac{\abs{\hat{g}_0(\xi)-\hat{f}_0(\xi)}}{\abs{\xi}} \\
	&\phantom{\leq} +\int_0^te^{-(t-\tau)}\ave*{\frac{\abs{\hat{g}(p\xi,\tau)\hat{g}(q\xi,\tau)-\hat{f}(p\xi,\tau)\hat{f}(q\xi,\tau)}}{\abs{\xi}}}\,d\tau \\
	&\leq e^{-t}d_1(f_0,g_0)+\int_0^te^{-(t-\tau)}d_1(f(\tau),g(\tau))\,d\tau,
\end{align*}
which, taking the supremum over $\xi\neq 0$ at the left-hand side and multiplying both sides by $e^t$, yields
$$ e^td_1(f(t),g(t))\leq d_1(f_0,g_0)+\int_0^te^\tau d_1(f(\tau),g(\tau))\,d\tau. $$
Gr\"{o}nwall's inequality applied to the function $e^td_1(f(t),g(t))$ implies then
$$ e^td_1(f(t),g(t))\leq d_1(f_0,g_0)e^t, $$
i.e.
$$ d_1(f(t),g(t))\leq d_1(f_0,g_0), $$
and the thesis follows taking the supremum over $t\in (0,\,T]$ of both sides.
\end{proof}

\begin{remark}
The computations of the proof of Theorem~\ref{theo:cont_dep} can be performed also by relying on the Fourier-transformed version~\eqref{eq:FourierBoltz} of the Boltzmann-type equation~\eqref{eq:Boltztype.strong}. Indeed, by time integration~\eqref{eq:FourierBoltz} can be rewritten as
\begin{equation}
	\hat{f}(\xi,t)=e^{-t}\hat{f}_0(\xi)+\int_0^te^{-(t-\tau)}\ave{\hat{f}(p\xi,\tau)\hat{f}(q\xi,\tau)}\,d\tau,
	\label{eq:FourierBoltz.mild}
\end{equation}
which is the Fourier-transformed version of $f(t)=\cQ_{f_0}(f(t))$.
\end{remark}

\subsection{Preservation of the support}
The non-negativity of the parameters $p$, $q$ in the interaction rule~\eqref{eq:v'} entails a physically interesting property of the solution $f(t)$ to~\eqref{eq:Boltztype.strong}, which we can paraphrase as follows: if initially the states of the agents are confined in $\R_+$ (respectively, $\R_-$) they remain confined there at all successive times. Heuristically, this is quite evident from the particle model~\eqref{eq:Vt+Dt}--\eqref{eq:linsymint}. The next result proves it rigorously:
\begin{theorem} \label{theo:supp}
If $\supp{f_0}\subseteq\R_+$ then $\supp{f(t)}\subseteq\R_+$ for all $t>0$.
\end{theorem}
\begin{proof}
Since $f_0$ is a probability measure, the condition $\supp{f_0}\subseteq\R_+$ can be restated as
$$ \int_{\R_+}f_0(v)\,dv=1, $$
i.e. the whole probability mass carried by $f_0$ is distributed within $\R_+$. Similarly for $f(t)$, $t>0$.

Based on this observation, we take $\varphi(v)=\chi(v\geq 0)$ in~\eqref{eq:Boltztype.weak}, i.e. the characteristic function of the set $\{v\geq 0\}=\R_+$, to discover:
\begin{align*}
	\frac{d}{dt}\int_{\R_+}f(v,t)\,dv &= \int_\R\int_\R\ave{\chi(v'\geq 0)}f(v,t)f(v_\ast,t)\,dv\,dv_\ast-\int_{\R_+}f(v,t)\,dv \\
	&\geq \int_{\R_+}\int_{\R_+}\ave{\chi(v'\geq 0)}f(v,t)f(v_\ast,t)\,dv\,dv_\ast-\int_{\R_+}f(v,t)\,dv \\
	&= \left(\int_{\R_+}f(v,t)\,dv\right)^2-\int_{\R_+}f(v,t)\,dv,
\end{align*}
where we have used that $\chi(v'\geq 0)\geq 0$ for all $v,\,v_\ast\in\R$ and $\chi(v'\geq 0)=1$ for $v,\,v_\ast\in\R_+$ owing to the non-negativity of $p$, $q$ in~\eqref{eq:v'}. Letting
$$ F(t):=\int_{\R_+}f(v,t)\,dv, $$
we are led therefore to the differential inequality
$$ \dot{F}(t)\geq F^2(t)-F(t), $$
which, multiplying both sides by $-\frac{e^{-t}}{F^2(t)}<0$, can be rewritten as
$$ \frac{d}{dt}\left(\frac{e^{-t}}{F(t)}\right)\leq -e^{-t}. $$
An integration in time over the interval $[0,\,t]$, $t>0$, considering that $F(0)=1$ by assumption, reveals
$$ \frac{1}{F(t)}\leq 1 \quad \forall\,t>0, $$
thus $F(t)\geq 1$ for all $t>0$. On the other hand, clearly $F(t)\leq\int_\R f(v,t)\,dv=1$ for all $t>0$. Therefore, we conclude $F(t)=1$ for all $t\geq 0$ and we are done.
\end{proof}

An impressive consequence of Theorem~\ref{theo:supp} is that the theory of the Boltzmann-type equation~\eqref{eq:Boltztype.strong} with linear symmetric interaction rule~\eqref{eq:v'} that we present in this paper, although referred to the case $v\in\R$, holds straightforwardly also when the physical nature of the microscopic state of the agents requires the limitation $v\geq 0$, such as e.g., in economical~\cite{cordier2005JSP,duering2009RMUP} or epidemiological~\cite{loy2021KRM} applications. In these cases, the Boltzmann-type equation is usually written by integrating on $\R_+$ but, owing to Theorem~\ref{theo:supp}, it can be simply understood as~\eqref{eq:Boltztype.strong} supplemented by an initial datum $f_0$ supported in $\R_+$.

\section{Trend towards the equilibrium}
\label{sect:trend_equil}
The existence of global-in-time solutions to~\eqref{eq:Boltztype.strong} makes it meaningful to study the convergence of $f(t)$ to stationary distributions, the so-called \textit{Maxwellians} in the jargon of classical kinetic theory. From the point of view of applications, Maxwellians depict the aggregate behaviour emerging spontaneously from agents' interactions when the latter reach a \textit{statistical equilibrium}. By statistical equilibrium we mean a condition in which the microscopic states of the agents can still change repeatedly in time but in such a way that their statistical distribution does not, so that a stationary aggregate picture of the system is observed.

\subsection{Convergence towards stationary distributions}
To discuss the trend towards stationary distributions we observe preliminarily that if $f,\,g$ are any two solutions to~\eqref{eq:Boltztype.strong} issuing from initial conditions with the same mean value then both Fourier distances $d_1(f(t),g(t))$, $d_2(f(t),g(t))$ are well-defined for every $t\geq 0$. Indeed, $f(t),\,g(t)$ have equal zeroth and first moments, cf.~\eqref{eq:assumptions_p.q} and Proposition~\ref{prop:ds}\ref{prop:ds.finiteness}. Moreover, the following relationship holds true:
\begin{equation}
	d_1(f(t),g(t))\leq 2\sqrt{2}\left[d_2(f(t),g(t))\right]^{1/2}, \qquad \forall\,t\geq 0
	\label{eq:interp.d1_d2}
\end{equation}
as particular case of a more general metric interpolation property proved in~\cite[Proposition~2.9]{carrillo2007RMUP}.

Therefore, we can use the $2$-Fourier metric to establish the following result:
\begin{proposition} \label{prop:trend.d2}
Under~\eqref{eq:assumptions_p.q}, let $f,\,g\in C^0([0,\,+\infty);\,\cP_2(\R))$ be the solutions to~\eqref{eq:Boltztype.strong} issuing from two initial data $f_0,\,g_0\in\cP_2(\R)$, respectively, with equal mean value. Then
$$ d_2(f(t),g(t))\leq d_2(f_0,g_0)e^{\left(\ave{p^2+q^2}-1\right)t}. $$
In particular,
$$ \lim_{t\to +\infty}d_2(f(t),g(t))=0. $$
\end{proposition}

\begin{remark}
Proposition~\ref{prop:trend.d2} differs from Theorem~\ref{theo:cont_dep} in that it requires explicitly that the initial data $f_0$, $g_0$ have the same mean value. This, together with the first condition in~\eqref{eq:assumptions_p.q}, guarantees that $f(t)$, $g(t)$ have the same mean value for every $t>0$, which is essential in order for the $2$-Fourier distance between them to be well-defined for every $t>0$, cf. Proposition~\ref{prop:ds}\ref{prop:ds.finiteness}.

The continuous dependence estimate brought by Theorem~\ref{theo:cont_dep} holds instead for any pair of initial data $f_0,\,g_0\in\cP_2(\R)$, possibly with different mean value. On the other hand, in general the $1$-Fourier distance between the respective solutions does not vanish for $t\to +\infty$ unless $f_0$, $g_0$ have the same mean value, for then it is possible to apply Proposition~\ref{prop:trend.d2} and~\eqref{eq:interp.d1_d2}.
\end{remark}

\begin{proof}[Proof of Proposition~\ref{prop:trend.d2}]
We estimate the $2$-Fourier distance between $f(t)$ and $g(t)$ from~\eqref{eq:FourierBoltz.mild}. We have:
\begin{align*}
	\frac{\abs{\hat{g}(\xi,t)-\hat{f}(\xi,t)}}{\abs{\xi}^2} &\leq e^{-t}\frac{\abs{\hat{g}_0(\xi)-\hat{f}_0(\xi)}}{\abs{\xi}^2} \\
	&\phantom{\leq} +\int_0^te^{-(t-\tau)}\ave*{\frac{\abs{\hat{g}(p\xi,\tau)\hat{g}(q\xi,\tau)-\hat{f}(p\xi,\tau)\hat{f}(q\xi,\tau)}}{\abs{\xi}^2}}\,d\tau \\
	&\leq e^{-t}d_2(f_0,g_0)+\int_0^te^{-(t-\tau)}\left\langle\frac{\abs{\hat{g}(q\xi,\tau)}\cdot\abs{\hat{g}(p\xi,\tau)-\hat{f}(p\xi,\tau)}}{\abs{\xi}^2}\right. \\
	&\phantom{\leq d_2(f_0,g_0)+\int_0^te^{-(t-\tau)}\left\langle\right.} \left.+\frac{\abs{\hat{f}(p\xi,\tau)}\cdot\abs{\hat{g}(q\xi,\tau)
		-\hat{f}(q\xi,\tau)}}{\abs{\xi}^2}\right\rangle\,d\tau \\
	&\leq e^{-t}d_2(f_0,g_0) \\
	&\phantom{\leq} +\int_0^te^{-(t-\tau)}
		\ave*{\frac{\abs{\hat{g}(p\xi,\tau)-\hat{f}(p\xi,\tau)}}{\abs{\xi}^2}+\frac{\abs{\hat{g}(q\xi,\tau)-\hat{f}(q\xi,\tau)}}{\abs{\xi}^2}}\,d\tau \\
	&\leq e^{-t}d_2(f_0,g_0)+\ave{p^2+q^2}\int_0^te^{-(t-\tau)}d_2(f(\tau),g(\tau))\,d\tau.
\end{align*}
Taking the supremum over $\xi\in\R\setminus\{0\}$ and multiplying both sides by $e^t$ we obtain
$$ e^td_2(f(t),g(t))\leq d_2(f_0,g_0)+\ave{p^2+q^2}\int_0^te^\tau d_2(f(\tau),g(\tau))\,d\tau. $$
At this point, Gr\"{o}nwall's inequality applied to the function $e^td_2(f(t),g(t))$ yields
$$ e^td_2(f(t),g(t))\leq d_2(f_0,g_0)e^{\ave{p^2+q^2}t}, $$
i.e.
$$ d_2(f(t),g(t))\leq d_2(f_0,g_0)e^{\left(\ave{p^2+q^2}-1\right)t}, $$
whence the thesis follows recalling also the second assumption in~\eqref{eq:assumptions_p.q}.
\end{proof}

Proposition~\ref{prop:trend.d2} is at the basis of the characterisation of the large time trend of the solutions to~\eqref{eq:Boltztype.strong}. Assume indeed that~\eqref{eq:Boltztype.strong} admits a \textit{constant-in-time} solution, viz. an \textit{equilibrium solution}, say $f^\infty=f^\infty(v)$ such that $\partial_tf^\infty\equiv 0$, with a certain mean value $M_1^\infty\in\R$. Take then any initial condition $f_0$ having mean value $M_{1,0}=M_1^\infty$. Proposition~\ref{prop:trend.d2} implies that the solution $f(t)$ issuing from $f_0$ converges in time to $f^\infty$, because
$$ d_2(f(t),f^\infty)\leq d_2(f_0,f^\infty)e^{\left(\ave{p^2+q^2}-1\right)t}\xrightarrow{t\to +\infty}0, $$
where we have used the fact that the solution to~\eqref{eq:Boltztype.strong} issuing from the initial datum $f^\infty$ is $f^\infty$ itself by definition of constant-in-time solution. Therefore, constant-in-time solutions are the aforesaid Maxwellians depicting the emerging aggregate behaviour of the system of agents. We denote them with the superscript ``$\infty$'' to refer precisely to the fact that they are the distributions that the system converges to when $t\to +\infty$.

The necessity then arises to study constant-in-time solutions to~\eqref{eq:Boltztype.strong}. We do it by means of the following result, which asserts that under~\eqref{eq:assumptions_p.q}, plus a further technical but essentially nonintrusive assumption,~\eqref{eq:Boltztype.strong} admits always a unique constant-in-time solution of prescribed mean value.

\begin{theorem} \label{theo:const_in_time_sols}
Assume~\eqref{eq:assumptions_p.q} with furthermore $\ave{p^3+q^3}<1$ and fix $M_1^\infty\in\R$. There exists $\cM_3>0$ such that~\eqref{eq:Boltztype.strong} admits a unique constant-in-time solution $f^\infty\in\cP_{2,1,\cM_3}(\R)$ with mean value $M_1^\infty$.
\end{theorem}
\begin{remark}
Actually, it would be enough to investigate the existence of constant-in-time solutions to~\eqref{eq:Boltztype.strong} with prescribed mean value, for Proposition~\ref{prop:trend.d2} implies automatically their uniqueness. Nevertheless, Theorem~\ref{theo:const_in_time_sols} has the merit of providing a self-contained result, which asserts the existence and uniqueness of solutions to the non-evolutionary problem $Q(f,f)=0$ independently of any prior knowledge on the corresponding evolutionary problem $\partial_tf=Q(f,f)$.
\end{remark}
\begin{proof}[Proof of Theorem~\ref{theo:const_in_time_sols}]
Since constant-in-time solutions to~\eqref{eq:Boltztype.strong} are such that $\partial_tf^\infty\equiv 0$, they solve $Q(f^\infty,f^\infty)=0$ where $Q$ is the collisional operator~\eqref{eq:Q.strong}. A useful form in which to rewrite this equation is
\begin{equation}
	f^\infty=Q^+(f^\infty,f^\infty),
	\label{eq:finfty.fixedpoint}
\end{equation}
which shows that $f^\infty$ can be regarded as a fixed point of the \textit{gain operator} $Q^+$ defined as
$$ Q^+(f,f)(v,t):=\int_\R\ave*{\frac{f(\pr{v},t)f(\pr{v}_\ast,t)}{\abs{p^2-q^2}}}\,dv_\ast. $$
In weak form:
\begin{equation}
	\int_\R\varphi(v)Q^+(f^\infty,f^\infty)(v)\,dv=\int_\R\int_\R\ave{\varphi(v')}f^\infty(v)f^\infty(v_\ast)\,dv\,dv_\ast,
	\label{eq:Q+.weak}
\end{equation}
so that~\eqref{eq:finfty.fixedpoint} becomes
$$ \int_\R\varphi(v)f^\infty(v)\,dv=\int_\R\int_\R\ave{\varphi(v')}f^\infty(v)f^\infty(v_\ast)\,dv\,dv_\ast $$
for every observable quantity $\varphi$.

We notice preliminarily that~\eqref{eq:finfty.fixedpoint} does not force a particular mean value for $f^\infty$, indeed with $\varphi(v)=v$ we discover
$$ \int_\R vf^\infty(v)\,dv=\ave{p+q}\int_\R vf^\infty(v)\,dv, $$
which is satisfied by every value of $\int_\R vf^\infty(v)\,dv$ because of~\eqref{eq:assumptions_p.q}. We can therefore fix such a mean value to some $M_1^\infty\in\R$ common to all prospective solutions to~\eqref{eq:finfty.fixedpoint}. Consequently, prospective solutions turn out to have a common energy proportional to $(M_1^\infty)^2$, indeed with $\varphi(v)=v^2$ we discover
$$ \int_\R v^2f^\infty(v)\,dv=\ave{p^2+q^2}\int_\R v^2f^\infty(v)\,dv+2\ave{pq}(M_1^\infty)^2, $$
whence
$$ M_2^\infty:=\int_\R v^2f^\infty(v)\,dv=\frac{2\ave{pq}}{1-\ave{p^2+q^2}}(M_1^\infty)^2. $$
Finally, we notice that $\varphi(v)=\abs{v}^3$ yields
\begin{align*}
	\int_\R\abs{v}^3f^\infty(v)\,dv &= \int_\R\int_\R\ave{\abs{pv+qv_\ast}^3}f^\infty(v)f^\infty(v_\ast)\,dv\,dv_\ast \\
	&\leq \ave{p^3+q^3}\int_\R\abs{v}^3f^\infty(v)\,dv+3\ave{pq(p+q)}M_2^\infty\int_\R\abs{v}f^\infty(v)\,dv;
\intertext{since $\int_\R\abs{v}f^\infty(v)\,dv\leq (M_2^\infty)^{1/2}$ because of Cauchy-Schwarz inquality, we obtain further}
	&\leq \ave{p^3+q^3}\int_\R\abs{v}^3f^\infty(v)\,dv+3\ave{pq(p+q)}(M_2^\infty)^{3/2}
\end{align*}
whence, in view of the assumption $\ave{p^3+q^3}<1$,
$$ \int_\R\abs{v}^3f^\infty(v)\,dv\leq\frac{3\ave{pq(p+q)}}{1-\ave{p^3+q^3}}(M_2^\infty)^{3/2}. $$
Therefore, fixing
\begin{equation}
	\cM_3:=\frac{3\ave{pq(p+q)}}{1-\ave{p^3+q^3}}(M_2^\infty)^{3/2}
		=\frac{3\ave{pq(p+q)}}{1-\ave{p^3+q^3}}\left(\frac{2\ave{pq}}{1-\ave{p^2+q^2}}\right)^{3/2}\abs{M_1^\infty}^3
	\label{eq:cM3}
\end{equation}
we can look for solutions to~\eqref{eq:finfty.fixedpoint} in $\cP_{2,1,\cM_3}(\R)$, which equipped with the Fourier distance $d_2$ is a complete metric space, cf. Proposition~\ref{prop:ds}\ref{prop:cP.completeness}.

To show existence and uniqueness of a fixed point of the gain operator $Q^+$ in $\cP_{2,1,\cM_3}(\R)$ we rely again on Banach's fixed-point theorem.
\begin{enumerate}
\item First, we show that $Q^+$ maps $\cP_{2,1,\cM_3}(\R)$ into itself. For $f^\infty\in\cP_{2,1,\cM_3}(\R)$, from~\eqref{eq:Q+.weak} with $\varphi(v)=1$ we have clearly
$$ \int_\R Q^+(f^\infty,f^\infty)(v)\,dv=1, $$
whereas for every non-negative observable $\varphi$
$$ \int_\R\varphi(v)Q^+(f^\infty,f^\infty)(v)\,dv\geq 0. $$
Therefore, $Q^+(f^\infty,f^\infty)$ is a probability measure.

Next, we take $\varphi(v)=v$ and discover
$$ \int_\R vQ^+(f^\infty,f^\infty)(v)\,dv=\int_\R vf^\infty(v)\,dv=M_1^\infty; $$
we take $\varphi(v)=v^2$ and obtain
\begin{align*}
	\int_\R v^2Q^+(f^\infty,f^\infty)(v)\,dv &= \ave{p^2+q^2}M_2^\infty+2\ave{pq}(M_1^\infty)^2 \\
	&= 2\ave{pq}\left(\frac{\ave{p^2+q^2}}{1-\ave{p^2+q^2}}+1\right)(M_1^\infty)^2 \\
	&= \frac{2\ave{pq}}{1-\ave{p^2+q^2}}(M_1^\infty)^2=M_2^\infty;
\end{align*}
finally, we take $\varphi(v)=\abs{v}^3$ and find
\begin{align*}
	\int_\R\abs{v}^3Q^+(f^\infty,f^\infty)(v)\,dv &\leq \ave{p^3+q^3}\int_\R\abs{v}^3f^\infty(v)\,dv+3\ave{pq(p+q)}(M_2^\infty)^{3/2} \\
	&\leq \ave{p^3+q^3}\cM_3+3\ave{pq(p+q)}(M_2^\infty)^{3/2}, \\
\intertext{whence, using the expression~\eqref{eq:cM3} of $\cM_3$,}
	&\leq 3\ave{pq(p+q)}\left(\frac{\ave{p^3+q^3}}{1-\ave{p^3+q^3}}+1\right)(M_2^\infty)^{3/2} \\
	&= \frac{3\ave{pq(p+q)}}{1-\ave{p^3+q^3}}(M_2^\infty)^{3/2}=\cM_3.
\end{align*}
We conclude that $Q^+(f^\infty,f^\infty)\in\cP_{2,1,\cM_3}(\R)$.

\item Second, we show that $Q^+$ is a contraction on $\cP_{2,1,\cM_3}(\R)$. For this, we observe that from~\eqref{eq:Q+.weak} with $\varphi(v)=e^{-i\xi v}$ we obtain
$$ \widehat{Q^+}(f^\infty,f^\infty)(\xi)=\ave{\hat{f}^\infty(p\xi)\hat{f}^\infty(q\xi)}, $$
therefore, given $f^\infty,\,g^\infty\in\cP_{2,1,\cM_3}(\R)$, we have
\begin{align*}
	\frac{\abs{\widehat{Q^+}(g^\infty,g^\infty)(\xi)-\widehat{Q^+}(f^\infty,f^\infty)(\xi)}}{\abs{\xi}^2} &=
		\frac{\abs{\ave{\hat{g}^\infty(p\xi)\hat{g}^\infty(q\xi)-\hat{f}^\infty(p\xi)\hat{f}^\infty(q\xi)}}}{\abs{\xi}^2} \\
	&\leq \ave*{\frac{\abs{\hat{g}^\infty(p\xi)-\hat{f}^\infty(p\xi)}}{\abs{\xi}^2}}
		+\ave*{\frac{\abs{\hat{g}^\infty(q\xi)-\hat{f}^\infty(q\xi)}}{\abs{\xi}^2}} \\
	&= \ave{p^2+q^2}d_2(f^\infty,g^\infty).
\end{align*}
Taking the supremum over $\xi\in\R\setminus\{0\}$ yields
$$ d_2(Q^+(f^\infty,f^\infty),Q^+(g^\infty,g^\infty))\leq\ave{p^2+q^2}d_2(f^\infty,g^\infty) $$
and, owing to~\eqref{eq:assumptions_p.q}, we are done. \qedhere
\end{enumerate}
\end{proof}

\begin{remark}
The further hypothesis $\ave{p^3+q^3}<1$ of Theorem~\ref{theo:const_in_time_sols} is often less restrictive than it might seem at first glance. Observe for instance that if $p,\,q\in [0,\,1]$ then $p^3\leq p^2$ and likewise $q^3\leq q^2$, therefore $\ave{p^3+q^3}\leq\ave{p^2+q^2}<1$ by using only~\eqref{eq:assumptions_p.q}. We stress, however, that this requires the random coefficients $p,\,q$ to belong \textit{pointwise} to $[0,\,1]$ and not only \textit{on average} as discussed in Remark~\ref{rem:pq.average}. Therefore, condition $\ave{p^3+q^3}<1$ is in general not implied by~\eqref{eq:assumptions_p.q} alone.
\end{remark}

Collecting the results of Proposition~\ref{prop:trend.d2} and Theorem~\ref{theo:const_in_time_sols}, together with the observations in between, we can state:
\begin{theorem}
Assume~\eqref{eq:assumptions_p.q} with also $\ave{p^3+q^3}<1$. Every solution $f\in C^0([0,\,+\infty);\,\cP_2(\R))$ to~\eqref{eq:Boltztype.strong} issuing from an initial datum $f_0\in\cP_2(\R)$ with mean value $M_{1,0}\in\R$ is such that $f(t)$ converges in time towards a unique Maxwellian $f^\infty\in\cP_3(\R)$ with mean value $M_{1,0}$. In particular:
$$ d_2(f(t),f^\infty)\leq d_2(f_0,f^\infty)e^{\left(\ave{p^2+q^2}-1\right)t}, $$
hence the convergence is exponentially fast in the $2$-Fourier metric.
\end{theorem}

\subsection{Asymptotic trend of the moments and tails}
If $f^\infty$ has mean value $M_{1,0}$ then from the proof of Theorem~\ref{theo:const_in_time_sols} we infer that its energy is
$$ M_2^\infty=\frac{2\ave{pq}}{1-\ave{p^2+q^2}}M_{1,0}^2. $$
On the other hand, we know that the energy $M_2(t)$ of $f(t)$ is~\eqref{eq:M2.constant_M1}, hence we see that $M_2(t)\to M_2^\infty$ as $t\to +\infty$. More in general, one may wonder whether the convergence of $f(t)$ to the Maxwellian $f^\infty$ implies any properties of the time trend of the moments of $f(t)$ itself.

To investigate this issue, it turns out that a fundamental quantity is the function $S:\R_+\to\R$ defined as
\begin{equation}
	S(s):=\ave{p^s+q^s}-1,
	\label{eq:S}
\end{equation}
which is such that $S(0)=1$ and moreover, owing to~\eqref{eq:assumptions_p.q}, $S(1)=0$ and $S(2)<0$. In addition:
\begin{lemma}
The function $S$ defined in~\eqref{eq:S} is convex.
\end{lemma}
\begin{proof}
Let $s_1,\,s_2\in\R$ and $\lambda\in [0,\,1]$. Noticing for instance that
$$ p^{\lambda s_1+(1-\lambda)s_2}=e^{(\lambda s_1+(1-\lambda)s_2)\log{p}}\leq \lambda e^{s_1\log{p}}+(1-\lambda)e^{s_2\log{p}}
	=\lambda p^{s_1}+(1-\lambda)p^{s_2}, $$
where we have used the convexity of the exponential function, we deduce
\begin{align*}
	S(\lambda s_1+(1-\lambda)s_2) &\leq\lambda\ave{p^{s_1}+q^{s_1}}+(1-\lambda)\ave{p^{s_2}+q^{s_2}}-[\lambda+(1-\lambda)] \\
	&= \lambda\left(\ave{p^{s_1}+q^{s_1}}-1\right)+(1-\lambda)\left(\ave{p^{s_2}+q^{s_2}}-1\right) \\
	&= \lambda S(s_1)+(1-\lambda)S(s_2),
\end{align*}
whence the thesis follows.
\end{proof}

Because of the convexity and of the obvious continuity of $S$, together with the further properties recalled above, either of the following options is possible:
\begin{itemize}
\item $S(s)<0$ for all $s>1$;
\item there exists $\bar{s}>1$ such that $S(\bar{s})=0$, with $S(s)<0$ for $1<s<\bar{s}$ and $S(s)>0$ for $s>\bar{s}$.
\end{itemize}
In the first case, we can provide a quite general and precise characterisation of the trend of the moments of $f(t)$ for large times, as expressed by the following two results.

\begin{theorem} \label{theo:bound.moments}
Assume~\eqref{eq:assumptions_p.q} along with $S(s)<0$ for all $s>2$. Moreover, let $f_0\in\cP_2(\R)$ have all moments finite. Then the statistical moments of any order of the solution $f\in C^0([0,\,+\infty);\,\cP_2(\R))$ to~\eqref{eq:Boltztype.strong} issuing from $f_0$ are uniformly bounded in time.
\end{theorem}
\begin{proof}
We observe that~\eqref{eq:assumptions_p.q}, together with $S(s)<0$ for all $s>2$, implies actually that $S(s)<0$ for all $s>1$. Therefore, throughout the proof we shall write $S(n)=-\abs{S(n)}$, $n>1$, for the sake of clarity.

The zeroth and first moment of $f(t)$ are constant in time. Moreover, from Section~\ref{sect:moments_evol} we know that the subsequent moments ($n\geq 2$) satisfy
$$ \frac{dM_n}{dt}=-\abs{S(n)}M_n+\sum_{k=1}^{n-1}\binom{n}{k}\ave{p^kq^{n-k}}M_kM_{n-k}, $$
whence, integrating in time in the interval $[0,\,t]$, $t>0$,
$$ M_n(t)=M_{n,0}e^{-\abs{S(n)}t}+\sum_{k=1}^{n-1}\binom{n}{k}\ave{p^kq^{n-k}}\int_0^t e^{-\abs{S(n)}(t-\tau)}M_k(\tau)M_{n-k}(\tau)\,d\tau. $$

Assume that moments up to the ($n-1$)-th one, $n\geq 2$, are uniformly bounded in time, i.e. that there exist constants $\cpM_k>0$, $k=0,\,\dots,\,n-1$, such that $\abs{M_k(t)}\leq\cpM_k$ for all $t\geq 0$ (notice, in particular, that $\cpM_0=1$ and $\cpM_1=M_{1,0}$). It follows:
\begin{align*}
	\abs{M_n(t)} &\leq \abs{M_{n,0}}e^{-\abs{S(n)}t}+\frac{1-e^{-\abs{S(n)}t}}{\abs{S(n)}}\sum_{k=1}^{n-1}\binom{n}{k}\ave{p^kq^{n-k}}\cpM_k\cpM_{n-k} \\
	&\leq \abs{M_{n,0}}+\frac{1}{\abs{S(n)}}\sum_{k=1}^{n-1}\binom{n}{k}\ave{p^kq^{n-k}}\cpM_k\cpM_{n-k}=:\cpM_n.
\end{align*}
Therefore, also the $n$-th moment is uniformly bounded in time and, by induction, we obtain the thesis.
\end{proof}

The uniform boundedness of all moments of $f(t)$ asserted by Theorem~\ref{theo:bound.moments} implies that
$$ \abs*{\int_\R v^nf(v,t)\,dv}<+\infty $$
for all $n\in\N$ and all $t\geq 0$. Although this does \textit{not} mean that $f(t)\in\cP_n(\R)$ for all $n\in\N$, it nonetheless indicates that $f(t)$ has a high degree of integrability for $v\to\pm\infty$ or, as it is customary to say, that it has \textit{slim tails}.

Instead, if there exists $\bar{s}>1$ such that $S(\bar{s})=0$ we cannot characterise as much precisely, in general, the time trend of the moments of $f(t)$, essentially because we ignore \textit{a priori} the sign of the moments themselves. To be more specific, assuming $\bar{s}\not\in\N$, let us fix $\bar{n}:=[\bar{s}]+1$, so that $S(n)<0$ for all $n\leq\bar{n}-1$ while $S(\bar{n})>0$, and let us investigate the evolution of $M_{\bar{n}}$. Owing to Theorem~\ref{theo:bound.moments}, the moments $M_0$, $M_1$, $M_2(t)$, \dots, $M_{\bar{n}-1}(t)$ are uniformly bounded. If, by chance, they are all non-negative for $t\geq 0$ then $M_{\bar{n}}(t)\geq M_{\bar{n},0}e^{S(\bar{n})t}$ and if $M_{\bar{n},0}>0$ as well then $M_{\bar{n}}(t)\to +\infty$ for $t\to +\infty$. Conversely, if we ignore the sign of the moments then, proceeding like in the proof of Theorem~\ref{theo:bound.moments}, we estimate
\begin{align*}
	\abs{M_{\bar{n}}(t)} &\geq \abs{M_{\bar{n},0}}e^{S(\bar{n})t}-\frac{e^{S(\bar{n})t}-1}{S(\bar{n})}\sum_{k=1}^{\bar{n}-1}\binom{\bar{n}}{k}\ave{p^kq^{\bar{n}-k}}\cpM_k\cpM_{\bar{n}-k} \\
	&= \left(\abs{M_{\bar{n},0}}-\frac{1}{S(\bar{n})}\sum_{k=1}^{\bar{n}-1}\binom{\bar{n}}{k}\ave{p^kq^{\bar{n}-k}}\cpM_k\cpM_{\bar{n}-k}\right)e^{S(\bar{n})t} \\
	&\phantom{\leq} +\frac{1}{S(\bar{n})}\sum_{k=1}^{\bar{n}-1}\binom{\bar{n}}{k}\ave{p^kq^{\bar{n}-k}}\cpM_k\cpM_{\bar{n}-k},
\end{align*}
whence we deduce that, in principle, $\abs{M_{\bar{n}}(t)}$ can blow up for $t\to +\infty$ if $\abs{M_{\bar{n},0}}$ is large enough. When $M_{\bar{n}}(t)$ blows, if $f(t)$ is sufficiently smooth, cf. Section~\ref{sect:a_priori_reg}, we can infer that there exists $\gamma\in (\bar{n}-1,\,\bar{n}]$ such that
$$ f(v,t)\sim \frac{1}{v^{1+\gamma}} $$
for either $v\ll -1$ (viz. $v$ negatively large) or $v\gg 1$ (viz. $v$ positively large) when $t\to +\infty$. Then we say that $f$ develops \textit{fat tails}. The value $\gamma$ is called the \textit{Pareto exponent} (or \textit{index}) of $f$, from the name of the Italian economist Vilfredo Pareto, who, at the beginning of the 20th century, observed empirically a polynomial decay of the tail of wealth distribution curves in western societies. We refer the interested reader to~\cite{duering2009RMUP,matthes2008JSP} for a detailed study of fat tail formation in Boltzmann-type kinetic models of income distribution. Notice that, in the mentioned cases, a thorough analysis is possible thanks to the fact that the statistical distribution of wealth is supported in $\R_+$ at all times, thus all of its moments are \textit{a priori} non-negative.

Back to the convergence of $f(t)$ to $f^\infty$, we observe that $M_2(t)$ tends exponentially quickly to $M_2^\infty$. Indeed, from~\eqref{eq:M2.constant_M1} we have
$$ \abs{M_2(t)-M_2^\infty}=\abs{M_{2,0}-M_2^\infty}e^{S(2)t} $$
with $S(2)<0$. Now we prove that if $S(s)<0$ for all $s>1$ then all moments behave qualitatively in this way.

\begin{theorem} \label{theo:conv.moments}
Assume~\eqref{eq:assumptions_p.q} along with $S(s)<0$ for all $s>2$. Moreover, let $f_0\in\cP_2(\R)$ have all moments finite. Then the Maxwellian $f^\infty$ has finite moments of any order and the moments of the solution $f\in C^0([0,\,+\infty);\,\cP_2(\R))$ to~\eqref{eq:Boltztype.strong} issuing from $f_0$ converge exponentially quickly to the corresponding moments of $f^\infty$ when $t\to +\infty$.
\end{theorem}
\begin{proof}
Again, throughout the proof we shall write $S(n)=-\abs{S(n)}$ for clarity.
\begin{enumerate}
\item First, we show that $f^\infty$ has finite moments of any order. Notice that $f^\infty$ exists and is unique because, in the current setting, the assumptions of Theorem~\ref{theo:const_in_time_sols} are fulfilled.

Because of~\eqref{eq:finfty.fixedpoint}, the $n$-th order moment, $n\geq 2$, of $f^\infty$ satisfies
\begin{align*}
	M_n^\infty &= \int_\R\int_\R\ave{(pv+qv_\ast)^n}f^\infty(v)f^\infty(v_\ast)\,dv\,dv_\ast \\
	&= \ave{p^n+q^n}M_n^\infty+\sum_{k=1}^{n-1}\binom{n}{k}\ave{p^kq^{n-k}}M_k^\infty M_{n-k}^\infty,
\end{align*}
hence $M_n^\infty$ can be expressed in terms of the lower order moments as
$$ M_n^\infty=\frac{1}{\abs{S(n)}}\sum_{k=1}^{n-1}\binom{n}{k}\ave{p^kq^{n-k}}M_k^\infty M_{n-k}^\infty. $$
Since $S(n)\neq 0$ for all $n\geq 2$ and $M_0^\infty=1$, $M_1^\infty=M_{1,0}$ are finite, this relationship shows inductively that $M_n^\infty$ is finite for every $n\in\N$.

\item Second, we show that $M_n(t)\to M_n^\infty$ exponentially fast when $t\to +\infty$. Clearly, it suffices to consider the case $n\geq 2$. Since
$$ -\abs{S(n)}M_n^\infty+\sum_{k=1}^{n-1}\binom{n}{k}\ave{p^kq^{n-k}}M_k^\infty M_{n-k}^\infty=0 $$
and moreover $\dot{M}_n^\infty=0$, we manipulate the equation of $M_n$ to get
\begin{align*}
	\frac{d}{dt}\left(M_n-M_n^\infty\right) &= -\abs{S(n)}\left(M_n-M_n^\infty\right)+\sum_{k=1}^{n-1}\binom{n}{k}\ave{p^kq^{n-k}}\left(M_kM_{n-k}-M_k^\infty M_{n-k}^\infty\right) \\
	&= -\abs{S(n)}\left(M_n-M_n^\infty\right) \\
	&\phantom{=} +\sum_{k=1}^{n-1}\binom{n}{k}\ave{p^kq^{n-k}}\left[M_k\left(M_{n-k}-M_{n-k}^\infty\right)+M_{n-k}^\infty\left(M_k-M_k^\infty\right)\right].
\end{align*}
Integrating in time on the interval $[0,\,t]$, $t>0$, and taking the absolute value yields
\begin{align*}
	\abs{M_n-M_n^\infty} &\leq \abs{M_{n,0}-M_n^\infty}e^{-\abs{S(n)}t} \\
	&\phantom{\leq} +\sum_{k=1}^{n-1}\binom{n}{k}\ave{p^kq^{n-k}}\int_0^te^{-\abs{S(n)}(t-\tau)}\left(\abs{M_k(\tau)}\cdot\abs{M_{n-k}(\tau)-M_{n-k}^\infty}\right. \\
	&\phantom{=+\sum_{k=1}^{n-1}\binom{n}{k}\ave{p^kq^{n-k}}\int_0^te^{S(n)(t-\tau)}\left(\right.} \left.+\abs{M_{n-k}^\infty}\cdot\abs{M_k(\tau)-M_k^\infty}\right)\,d\tau \\
	&= \abs{M_{n,0}-M_n^\infty}e^{-\abs{S(n)}t} \\
	&\phantom{=} +\sum_{k=2}^{n-1}\binom{n}{k}\int_0^te^{-\abs{S(n)}(t-\tau)}\left(\ave{p^{n-k}q^k}\abs{M_{n-k}(\tau)}+\ave{p^kq^{n-k}}\abs{M_{n-k}^\infty}\right) \\
	&\phantom{= +\sum_{k=2}^{n-1}\int_0^te^{S(n)(t-\tau)}\left(\ave{p^{n-k}q^k}\abs{M_{n-k}(\tau)}\right.} \times\abs{M_k(\tau)-M_k^\infty}\,d\tau \\
	&\leq \abs{M_{n,0}-M_n^\infty}e^{-\abs{S(n)}t} \\
	&\phantom{\leq} +\sum_{k=2}^{n-1}\binom{n}{k}\left(\ave{p^{n-k}q^k}\cpM_{n-k}+\ave{p^kq^{n-k}}\abs{M_{n-k}^\infty}\right) \\
	&\phantom{\leq +\sum_{k=2}^{n-1}\binom{n}{k}\left(\ave{p^{n-k}q^k}\cpM_{n-k}\right.} \times\int_0^te^{-\abs{S(n)}(t-\tau)}\abs{M_k(\tau)-M_k^\infty}\,d\tau,
\end{align*}
where we have used the fact that $M_1(\tau)=M_1^\infty=M_{1,0}$ for all $\tau\geq 0$ (thus sums start from $k=2$ from the second passage onwards) and that, owing to Theorem~\ref{theo:bound.moments}, moments of any order of $f(t)$ are uniformly bounded in time.

Assume now, by induction, that moments of $f(t)$ up to the order $n-1$ converge exponentially fast to the corresponding moments of $f^\infty$. Therefore, there exist constants $\alpha_k,\,C_k>0$ such that
$$ \abs{M_k(t)-M_k^\infty}\leq C_ke^{-\alpha_kt}, \qquad k=2,\,\dots,\,n-1 $$
(notice that the computation performed before the statement of the theorem shows that this is indeed true for $k=2$ with $C_2=\abs{M_{2,0}-M_2^\infty}$ and $\alpha_2=\abs{S(2)}$). Then:
\begin{align*}
	\abs{M_n-M_n^\infty} &\leq \abs{M_{n,0}-M_n^\infty}e^{-\abs{S(n)}t} \\
	&\phantom{\leq} +\sum_{k=2}^{n-1}\binom{n}{k}\left(\ave{p^{n-k}q^k}\cpM_{n-k}+\ave{p^kq^{n-k}}\abs{M_{n-k}^\infty}\right) \\
	&\phantom{\leq +\sum_{k=2}^{n-1}\binom{n}{k}\left(\ave{p^{n-k}q^k}\cpM_{n-k}\right.} \times C_ke^{-\abs{S(n)}t}\int_0^te^{(\abs{S(n)}-\alpha_k)\tau}\,d\tau.
\end{align*}
In particular, it results\footnote{For $\alpha_k=\abs{S(n)}$ we use the general fact that $t<\frac{1}{a}e^{at}$ for every $a>0$ and take specifically $a=\frac{\abs{S(n)}}{2}$.}
\begin{align*}
	e^{-\abs{S(n)}t}\int_0^te^{(\abs{S(n)}-\alpha_k)\tau}\,d\tau &=
		\begin{cases}
			\dfrac{e^{-\abs{S(n)}t}-e^{-\alpha_kt}}{\alpha_k-\abs{S(n)}} & \text{if } \alpha_k\neq\abs{S(n)} \\[5mm]
			te^{-\abs{S(n)}t} & \text{if } \alpha_k=\abs{S(n)}
		\end{cases} \\[3mm]
	&\leq
		\begin{cases}
			\dfrac{e^{-\min\{\abs{S(n)},\,\alpha_k\}t}}{\abs{\alpha_k-\abs{S(n)}}} & \text{if } \alpha_k\neq\abs{S(n)} \\[5mm]
			\dfrac{2}{\abs{S(n)}}e^{-\frac{\abs{S(n)}}{2}t} & \text{if } \alpha_k=\abs{S(n)},
		\end{cases}
\end{align*}
therefore we conclude that there exist constants $\tilde{\alpha}_{k,n},\,\tilde{C}_{k,n}>0$, precisely
$$ \tilde{\alpha}_{k,n}:=
	\begin{cases}
		\min\{\abs{S(n)},\,\alpha_k\} & \text{if } \alpha_k\neq\abs{S(n)} \\[3mm]
		\dfrac{\abs{S(n)}}{2} & \text{if } \alpha_k=\abs{S(n)},
	\end{cases}
	\qquad
	\tilde{C}_{k,n}:=
		\begin{cases}
			\dfrac{1}{\abs{\alpha_k-\abs{S(n)}}} & \text{if } \alpha_k\neq\abs{S(n)} \\[5mm]
			\dfrac{2}{\abs{S(n)}} & \text{if } \alpha_k=\abs{S(n)},
		\end{cases} $$
such that
$$ e^{-\abs{S(n)}t}\int_0^te^{(\abs{S(n)}-\alpha_k)\tau}\,d\tau\leq\tilde{C}_{k,n}e^{-\tilde{\alpha}_{k,n}t}, \qquad \forall\,t\geq 0. $$
Consequently,
\begin{align*}
	\abs{M_n-M_n^\infty} &\leq \abs{M_{n,0}-M_n^\infty}e^{-\abs{S(n)}t} \\
	&\phantom{\leq} +\left(\sum_{k=2}^{n-1}\binom{n}{k}\left(\ave{p^{n-k}q^k}\cpM_{n-k}+\ave{p^kq^{n-k}}\abs{M_{n-k}^\infty}\right)C_k\tilde{C}_{k,n}\right)e^{-\tilde{\alpha}_nt},
\intertext{where $\tilde{\alpha}_n:=\min_{k=2,\,\dots,\,n-1}{\tilde{\alpha}_{k,n}}$, and further}
	&\leq \left(\abs{M_{n,0}-M_n^\infty}+\sum_{k=2}^{n-1}\binom{n}{k}\left(\ave{p^{n-k}q^k}\cpM_{n-k}+\ave{p^kq^{n-k}}\abs{M_{n-k}^\infty}\right)C_k\tilde{C}_{k,n}\right) \\
	&\phantom{\leq} \times e^{-\min\{\abs{S(n)},\,\tilde{\alpha}_n\}t}.
\end{align*}
Letting
\begin{align*}
	\alpha_n &:= \min\{\abs{S(n)},\,\tilde{\alpha}_n\}, \\
	C_n &:= \abs{M_{n,0}-M_n^\infty}+\sum_{k=2}^{n-1}\binom{n}{k}\left(\ave{p^{n-k}q^k}\cpM_{n-k}+\ave{p^kq^{n-k}}\abs{M_{n-k}^\infty}\right)C_k\tilde{C}_{k,n},
\end{align*}
we obtain
$$ \abs{M_n(t)-M_n^\infty}\leq C_ne^{-\alpha_nt}, $$
whence the thesis follows inductively. \qedhere
\end{enumerate}
\end{proof}

\section{Quasi-invariant regime and Fokker--Planck equations}
\label{sect:Fokker--Planck}
Section~\ref{sect:trend_equil} has shown that under certain assumptions on the coefficients of the interaction rule~\eqref{eq:v'} there exists a unique stationary distribution, viz. Maxwellian, $f^\infty\in\cP_3(\R)$ towards which every solution to~\eqref{eq:Boltztype.strong} with prescribed mean value converges in time. Nevertheless, apart from a characterisation in terms of boundedness and convergence of moments, we could not provide hints on how to estimate $f^\infty$ in detail. The reason is that, in general, it is difficult to solve the integral equation $Q(f^\infty,f^\infty)=0$, cf.~\eqref{eq:Q.strong}, explicitly.

This consideration is at the basis of the idea to look for simpler kinetic equations, which can \textit{approximate}~\eqref{eq:Boltztype.strong} at least in \textit{certain regimes} of the coefficients of~\eqref{eq:v'}, thereby providing models which, on one hand, have a reduced scope but, on the other hand, are more amenable to explicit analyses. The aforementioned regimes of the parameters are typically \textit{asymptotic regimes}, i.e. they are built by scaling conveniently $p,\,q$ in~\eqref{eq:v'} by means of a scale parameter, of which one considers subsequently appropriate limits.

The earliest example of a similar procedure is the so-called \textit{grazing collision regime}, introduced in the classical kinetic theory by Villani~\cite{villani1998PhD,villani1998ARMA} to study the particular case in which collisions among gas molecules produce a small exchange of momentum between the colliding particles, so that the post-interaction velocities differ slightly from the pre-interaction ones. This happens when molecules hit against one another mostly tangentially, whence the name of \textit{grazing} collisions. Subsequently, Toscani and his coworkers generalised this concept to arbitrary interactions, speaking of \textit{quasi-invariant regime} to refer to the case in which the interactions produce a small variation of the states of the interacting agents, cf. e.g.,~\cite{cordier2005JSP,toscani2006CMS}. In such a regime, it turns out that the integral operator $Q$ featured by the Boltzmann-type equation~\eqref{eq:Boltztype.strong} can be approximated by a differential operator, whose main properties depend on the adopted scaling of $p,\,q$ in~\eqref{eq:v'}.

\subsection{Formal quasi-invariant limit in the advection-diffusion regime}
\label{sect:formal_q-i_limit}
Let us introduce a small parameter $\epsilon>0$ and let us assume that the interaction rule~\eqref{eq:v'} is scaled by means of $\epsilon$ as
\begin{equation}
	v_\epsilon'=p_\epsilon v+q_\epsilon v_\ast,
	\label{eq:v'.scaled}
\end{equation}
where $p_\epsilon,\,q_\epsilon$ are scaled versions of $p,\,q$ such that $p_\epsilon\to 1$ and $q_\epsilon\to 0$ when $\epsilon\to 0^+$. This way, interactions are quasi-invariant in the limit $\epsilon\to 0^+$, because $v_\epsilon'\to v$. Let us consider, in particular, the representative case in which $p_\epsilon$ is a random variable with
$$ \ave{p_\epsilon}=1-\epsilon\lambda, \qquad \Var{p_\epsilon}=\epsilon\sigma^2, $$
where $\lambda,\,\sigma>0$ are proportionality parameters, while $q_\epsilon$ is a deterministic coefficient
\begin{equation}
	q_\epsilon=\epsilon\lambda.
	\label{eq:qeps}
\end{equation}
Notice, in particular, that in the quasi-invariant limit $\epsilon\to 0^+$ the law of $p_\epsilon$ converges to $\delta_1$ and $p_\epsilon$ itself converges to $1$ both in mean and in quadratic mean, indeed:
$$ \ave{\abs{p_\epsilon-1}}\leq{\ave{(p_\epsilon-1)^2}}^{1/2}=\left(\Var{p_\epsilon-1}+\ave{p_\epsilon-1}^2\right)^{1/2}
	=\sqrt{\epsilon\sigma^2+\epsilon^2\lambda^2}\xrightarrow{\epsilon\to 0^+}0. $$
Moreover, we remark that
\begin{equation}
	\ave{p_\epsilon+q_\epsilon}=1, \qquad \forall\,\epsilon>0
	\label{eq:peps+qeps=1}
\end{equation}
and that
$$ \ave{p_\epsilon^2+q_\epsilon^2}=\Var{p_\epsilon}+{\ave{p_\epsilon}}^2+\epsilon^2\lambda^2=1-\epsilon(2\lambda-\sigma^2)+2\epsilon^2\lambda^2, $$
therefore, under the assumption
\begin{equation}
	\sigma^2<2\lambda,
	\label{eq:sigma.lambda}
\end{equation}
it results $\ave{p_\epsilon^2+q_\epsilon^2}<1$ if $\epsilon$ is small enough, precisely
\begin{equation}
	\epsilon<\frac{1}{\lambda}\left(1-\frac{\sigma^2}{2\lambda}\right).
	\label{eq:eps.quasi-inv}
\end{equation}
Consequently, if we work with the scaled interaction rule~\eqref{eq:v'.scaled} in the regime~\eqref{eq:sigma.lambda}-\eqref{eq:eps.quasi-inv} the initial value problem~\eqref{eq:IVP} is well-posed as stated by Theorems~\ref{theo:exists_unique},~\ref{theo:cont_dep}.

It is customary to understand the random variable $p_\epsilon$ in the form
\begin{equation}
	p_\epsilon=1-\epsilon\lambda+\sqrt{\epsilon}\sigma\eta,
	\label{eq:peps.eta}
\end{equation}
where $\eta$ is an $\epsilon$-independent real-valued random variable such that
\begin{equation}
	\ave{\eta}=0, \qquad \ave{\eta^2}=1, \qquad \ave{\abs{\eta}^3}<+\infty.
	\label{eq:eta}
\end{equation}
Notice that the second property implies $\ave{\abs{\eta}}\leq\ave{\eta^2}^{1/2}=1$ owing to Jensen's inequality. In order to guarantee $p_\epsilon\geq 0$, the random variable $\eta$ has to satisfy the further requirement $\eta\geq\frac{\epsilon\lambda-1}{\sigma\sqrt{\epsilon}}$. Since the right-hand side is an increasing function of $\epsilon$, such a requirement is fulfilled for every $\epsilon$ complying with~\eqref{eq:eps.quasi-inv} if
\begin{equation}
	\eta\geq\lim_{\epsilon\to\frac{1}{\lambda}\left(1-\frac{\sigma^2}{2\lambda}\right)}\frac{\epsilon\lambda-1}{\sigma\sqrt{\epsilon}}=
		-\sqrt{\frac{\frac{\sigma^2}{2\lambda}}{2(1-\frac{\sigma^2}{2\lambda})}},
	\label{eq:eta.lowerbound}
\end{equation}
the right-hand side being an $\epsilon$-free negative value.

Since for $\epsilon$ small we have $v_\epsilon'\approx v$, writing $v_\epsilon'=v+(p_\epsilon-1)v+q_\epsilon v_\ast$ and considering a sufficiently smooth observable $\varphi$, say $\varphi\in C^3(\R)$, we can Taylor-expand the difference $\ave{\varphi(v_\epsilon')-\varphi(v)}$ at the right-hand side of~\eqref{eq:Boltztype.weak} around $v$ with Lagrange remainder as
\begin{align*}
	\ave{\varphi(v_\epsilon')-\varphi(v)} &= \left\langle\varphi'(v)\bigl((p_\epsilon-1)v+q_\epsilon v_\ast\bigr)+\frac{1}{2}\varphi''(v)\bigl((p_\epsilon-1)v
		+q_\epsilon v_\ast\bigr)^2\right. \\
	&\phantom{=} \left.+\frac{1}{6}\varphi'''(\bar{v}_\epsilon)\bigl((p_\epsilon-1)v+q_\epsilon v_\ast\bigr)^3\right\rangle \\
	&= \varphi'(v)(\ave{p_\epsilon-1}v+q_\epsilon v_\ast)+\frac{1}{2}\varphi''(v)\left(\ave{(p_\epsilon-1)^2}v^2
		+2\ave{p_\epsilon-1}q_\epsilon vv_\ast+q_\epsilon^2v_\ast^2\right) \\
	&\phantom{=} +\frac{1}{6}\ave*{\varphi'''(\bar{v}_\epsilon)\bigl((p_\epsilon-1)v+q_\epsilon v_\ast\bigr)^3} \\
	&= \epsilon\left(\lambda\varphi'(v)(v_\ast-v)+\frac{\sigma^2}{2}\varphi''(v)v^2\right)+\frac{\epsilon^2\lambda^2}{2}\varphi''(v)(v_\ast-v)^2 \\
	&\phantom{=} +\frac{1}{6}\ave*{\varphi'''(\bar{v}_\epsilon)\bigl((p_\epsilon-1)v+q_\epsilon v_\ast\bigr)^3},
\end{align*}
where $\bar{v}_\epsilon$ is a point between $\min\{v,\,v_\epsilon'\}$ and $\max\{v,\,v_\epsilon'\}$. From this computation we see that, when $\epsilon\to 0^+$, the difference $\ave{\varphi(v_\epsilon')-\varphi(v)}$ is infinitesimal of order $1$ with respect to $\epsilon$, with principal part given by
$$ \lambda\varphi'(v)(v_\ast-v)+\frac{\sigma^2}{2}\varphi''(v)v^2. $$ Therefore, we expect this expression to lead the trend of~\eqref{eq:Boltztype.weak} in the quasi-invariant regime.

This is true provided also time is properly scaled, as the expression above drives actually the \textit{large time trend} of~\eqref{eq:Boltztype.weak}. To see this, let us introduce the scaled kinetic distribution function
$$ f_\epsilon(v,t):=f(v,t/\epsilon), $$
which is such that $\partial_tf_\epsilon=\frac{1}{\epsilon}\partial_tf$.
\begin{remark}
To better understand the relationship between the ``old'' time scale $t$ and the ``new'' time scale $t/\epsilon$ it is useful to introduce a new time variable $\tau\geq 0$ such that $t=\tau/\epsilon$. Clearly, when $t=O(1)$ it results $\tau=O(\epsilon)$ whereas when $t=O(1/\epsilon)$ it results $\tau=O(1)$. Hence the new time scale $\tau$ is shrunken compared to $t$, in such a way that on $\tau$ one observes quickly the trends emerging for large $t$. We can regard $\tau$ as a time scale larger and less detailed than $t$, at which one does not catch single interactions but directly their aggregate outcomes. It is the time scale needed to compensate for the little effect of each interaction and still perceive collective changes.
\end{remark}

Evaluating~\eqref{eq:Boltztype.weak} at time $t/\epsilon$ we find that $f_\epsilon$ satisfies
\begin{align}
	\begin{aligned}[b]
		\frac{d}{dt}\int_\R\varphi(v)f_\epsilon(v,t)\,dv &= \frac{1}{\epsilon}\int_\R\int_\R\ave{\varphi(v_\epsilon')-\varphi(v)}
			f_\epsilon(v,t)f_\epsilon(v_\ast,t)\,dv\,dv_\ast \\
		&= \lambda\int_\R\varphi'(v)(M_{1,0}-v)f_\epsilon(v,t)\,dv+\frac{\sigma^2}{2}\int_\R\varphi''(v)v^2f_\epsilon(v,t)\,dv \\
		&\phantom{=} +\frac{\epsilon\lambda^2}{2}\int_\R\int_\R\varphi''(v)(v_\ast-v)^2f_\epsilon(v,t)f_\epsilon(v_\ast,t)\,dv\,dv_\ast \\
		&\phantom{=} +\frac{1}{6\epsilon}\int_\R\int_\R\ave*{\varphi'''(\bar{v}_\epsilon)
			\bigl((p_\epsilon-1)v+q_\epsilon v_\ast\bigr)^3}f_\epsilon(v,t)f_\epsilon(v_\ast,t)\,dv\,dv_\ast,
	\end{aligned}
	\label{eq:Boltztype.scaled}
\end{align}
where
$$ M_{1,0}=\int_\R vf_\epsilon(v,t)\,dv=\int_\R vf(v,t/\epsilon)\,dv $$
is, owing to~\eqref{eq:peps+qeps=1}, the constant-in-time mean value of $f$ and, consequently, also of $f_\epsilon$ for every $\epsilon>0$. Now,~\eqref{eq:Boltztype.scaled} suggests that, as $\epsilon\to 0^+$, the solution $f_\epsilon$ approaches the solution $g$ of
\begin{equation}
	\frac{d}{dt}\int_\R\varphi(v)g(v,t)\,dv=\lambda\int_\R\varphi'(v)(M_{1,0}-v)g(v,t)\,dv+\frac{\sigma^2}{2}\int_\R\varphi''(v)v^2g(v,t)\,dv,
	\label{eq:FP.weak}
\end{equation}
$\varphi\in C^3(\R)$ being arbitrary. Indeed, looking at the remainder $R_\epsilon(t)$ defined as
\begin{align*}
	R_\epsilon(t) &:= \frac{\epsilon\lambda^2}{2}\int_\R\int_\R\varphi''(v)(v_\ast-v)^2f_\epsilon(v,t)f_\epsilon(v_\ast,t)\,dv\,dv_\ast \\
	&\phantom{:=} +\frac{1}{6\epsilon}\int_\R\int_\R\ave*{\varphi'''(\bar{v}_\epsilon)\bigl((p_\epsilon-1)v+q_\epsilon v_\ast\bigr)^3}
		f_\epsilon(v,t)f_\epsilon(v_\ast,t)\,dv\,dv_\ast,
\end{align*}
we discover
\begin{align*}
	\abs{R_\epsilon(t)} &\leq \frac{\epsilon\lambda^2}{2}\norminf{\varphi''}\int_\R\int_\R(v_\ast-v)^2f_\epsilon(v,t)	f_\epsilon(v_\ast,t)\,dv\,dv_\ast \\
	&\phantom{=} +\frac{\epsilon^{1/2}}{6}\norminf{\varphi'''}\left[\sigma^3\ave{\abs{\eta}^3}\int_\R\abs{v}^3f_\epsilon(v,t)\,dv\right. \\
	&\phantom{=+\frac{\epsilon^{1/2}}{6}\norminf{\varphi'''}\left[\right.}
		+3\epsilon^{1/2}\lambda\sigma^2\int_\R\int_\R v^2\abs{v_\ast-v}f_\epsilon(v,t)f_\epsilon(v_\ast,t)\,dv\,dv_\ast \\
	&\phantom{=+\frac{\epsilon^{1/2}}{6}\norminf{\varphi'''}\left[\right.}
		+3\epsilon\lambda^2\sigma\int_\R\int_\R\abs{v}(v_\ast-v)^2f_\epsilon(v,t)f_\epsilon(v_\ast,t)\,dv\,dv_\ast \\
	&\phantom{=+\frac{\epsilon^{1/2}}{6}\norminf{\varphi'''}\left[\right.}
		\left.+\epsilon^{3/2}\lambda^3\int_\R\int_\R\abs{v_\ast-v}^3f_\epsilon(v,t)f_\epsilon(v_\ast,t)\,dv\,dv_\ast\right]
\end{align*}
for all smooth observables with bounded derivatives, such as e.g., compactly supported $C^3$-observables. Notice that the integrals on the right-hand side are bounded by either $\int_\R v^2f_\epsilon(v,t)\,dv$ or $\int_\R\abs{v}^3f_\epsilon(v,t)\,dv$. In particular, we remark that
$$ \int_\R\int_\R vv_\ast f_\epsilon(v,t)f_\epsilon(v_\ast,t)\,dv\,dv_\ast=M_{1,0}^2\leq\int_\R v^2f_\epsilon(v,t)\,dv $$
owing to $\Var{f_\epsilon}\geq 0$, while
\begin{align*}
	\int_\R\int_\R v^2\abs{v_\ast}f_\epsilon(v,t)f_\epsilon(v_\ast,t)\,dv\,dv_\ast &= \int_\R v^2f_\epsilon(v,t)\,dv\cdot\int_\R\abs{v_\ast}f_\epsilon(v_\ast,t)\,dv_\ast \\
	&\leq \left(\int_\R\abs{v}^3f_\epsilon(v,t)\,dv\right)^{2/3}\left(\int_\R\abs{v_\ast}^3f_\epsilon(v_\ast,t)\,dv_\ast\right)^{1/3} \\
	&= \int_\R\abs{v}^3f_\epsilon(v,t)\,dv
\end{align*}
owing to H\"{o}lder's inequality with exponents $\mathfrak{p}=\frac{3}{2}$, $\mathfrak{q}=3$ applied separately to each factor. Consequently, if we prove that, for fixed $t>0$, the terms $\int_\R v^2f_\epsilon(v,t)\,dv$, $\int_\R\abs{v}^3f_\epsilon(v,t)\,dv$ remain bounded when $\epsilon\to 0^+$ we can conclude that the remainder $R_\epsilon(t)$ of~\eqref{eq:Boltztype.scaled} vanishes in the quasi-invariant limit and ultimately that, for fixed $t>0$, the equation solved by $f_\epsilon$ gets closer and closer to~\eqref{eq:FP.weak} for smaller and smaller $\epsilon$.

To this purpose, invoking~\eqref{eq:Boltztype.scaled} with $\varphi(v)=v^2$ we observe that
\begin{align*}
	\frac{d}{dt}\int_\R v^2f_\epsilon(v,t)\,dv &= \frac{1}{\epsilon}\int_\R\int_\R\ave{(v_\epsilon')^2-v^2}f_\epsilon(v,t)f_\epsilon(v_\ast,t)\,dv\,dv_\ast \\
	&= \sigma^2\int_\R v^2f_\epsilon(v,t)-2\lambda(1-\epsilon\lambda)\left(\int_\R v^2f_\epsilon(v,t)\,dv-M_{1,0}^2\right);
\intertext{since in the regime~\eqref{eq:eps.quasi-inv} it results $2\lambda(1-\epsilon\lambda)>0$, we get}
	&\leq \sigma^2\int_\R v^2f_\epsilon(v,t),
\end{align*}
which says that the second moment of $f_\epsilon$ is bounded for every $t>0$ when $\epsilon\to 0^+$ provided it is so at the initial time. Similarly, invoking~\eqref{eq:Boltztype.scaled} with $\varphi(v)=\abs{v}^3$ we find
\begin{align*}
	\frac{d}{dt}\int_\R\abs{v}^3f_\epsilon(v,t)\,dv &= \frac{1}{\epsilon}\int_\R\int_\R\ave{\abs{v_\epsilon'}^3-\abs{v}^3}
		f_\epsilon(v,t)f_\epsilon(v_\ast,t)\,dv\,dv_\ast \\
	&\leq \frac{1}{\epsilon}\int_\R\int_\R\ave{(p_\epsilon\abs{v}+q_\epsilon\abs{v_\ast})^3-\abs{v}^3}f_\epsilon(v,t)f_\epsilon(v_\ast,t)\,dv\,dv_\ast,
\intertext{where we have used that $p_\epsilon,\,q_\epsilon\geq 0$ in view of~\eqref{eq:eta.lowerbound}. Developing the cube, invoking the properties~\eqref{eq:eta} of $\eta$ and recalling~\eqref{eq:eps.quasi-inv}, we find that there exists an $\epsilon$-independent constant $C_{\lambda,\sigma}>0$ such that}
	&\leq C_{\lambda,\sigma}\int_\R\abs{v}^3f_\epsilon(v,t)\,dv,
\end{align*}
whence we obtain that also $\int_\R\abs{v}^3f_\epsilon(v,t)\,dv$ is bounded for every $t>0$ when $\epsilon\to 0^+$ if it is bounded at the initial time.

In conclusion, if $f_0\in\cP_3(\R)$ then, for fixed $t>0$, the scaled Boltzmann-type equation~\eqref{eq:Boltztype.scaled} solved by $f_\epsilon$ approaches, in the limit $\epsilon\to 0^+$, equation~\eqref{eq:FP.weak} solved by $g$ for e.g., all observables $\varphi\in C^3_c(\R)$, the subscript `$c$' standing for ``compactly supported''. 

Formally, this suggests that, for fixed $t>0$, $f_\epsilon(t)$ should approach $g(t)$ as $\epsilon\to 0^+$ or, in other words, that $g(t)$ should provide a good approximation of $f_\epsilon(t)$ for $\epsilon$ small enough. A rigorous proof of the convergence of $f_\epsilon(t)$ to $g(t)$ in the Fourier metric $d_2$ is sketched in~\cite{torregrossa2018KRM} but requires non-trivial developments of the theory elaborated in~\cite{gabetta1995JSP,toscani1999JSP}. Here, we confine ourselves to the formal argument just exposed, deferring some rigorous convergence results to alternative quasi-invariant scalings of the Boltzmann-type equation, which we shall tackle with the technical tools introduced in the previous sections (cf. Section~\ref{sect:other_FP}).

\subsection{Fokker--Planck equation and steady distribution}
\label{sect:FP_steady}
With $\varphi\in C^3_c(\R)$ it is easy to recast~\eqref{eq:FP.weak} in strong form by integration-by-parts. Notice that no boundary terms appear for $v\to\pm\infty$, thanks to the compactness of the support of $\varphi$. Explicitly, we have
\begin{equation}
	\frac{\partial g}{\partial t}+\lambda\frac{\partial}{\partial v}\bigl((M_{1,0}-v)g\bigr)=\frac{\sigma^2}{2}\frac{\partial^2}{\partial v^2}(v^2g),
	\label{eq:FP.strong}
\end{equation}
which is a linear \textit{Fokker--Planck equation} with non-constant coefficients. One of the most interesting properties of this equation, linked to the modelling of multi-agent systems, is that it allows for the explicit determination of its steady distribution. Owing to the reasoning that led to establish~\eqref{eq:FP.strong} out of~\eqref{eq:Boltztype.scaled}, such a steady distribution constitutes an approximation of the Maxwellian of~\eqref{eq:Boltztype.scaled} for $\epsilon$ small enough.

Since~\eqref{eq:FP.strong} can be put in divergence form:
$$ \frac{\partial g}{\partial t}+\frac{\partial}{\partial v}\left(\lambda(M_{1,0}-v)g-\frac{\sigma^2}{2}\frac{\partial}{\partial v}(v^2g)\right)=0, $$
we look for steady distributions $g^\infty=g^\infty(v)$ by imposing that the flux vanishes. Thus $g^\infty$ solves the ordinary differential equation
$$ (v^2g^\infty)'=\frac{2\lambda}{\sigma^2}(M_{1,0}-v)g^\infty $$
for $v\in\R$. Letting $h(v):=v^2g^\infty(v)$, we rewrite this equation as
$$ h'=\frac{2\lambda}{\sigma^2}\cdot\frac{M_{1,0}-v}{v^2}h, $$
which, due to the division by $v^2$, now holds separately for $v<0$ and $v>0$. Solving by separation of variables and going back to $g^\infty$ yields
$$ g^\infty(v)=\frac{C_-\chi(v<0)+C_+\chi(v>0)}{\abs{v}^{2\left(1+\frac{\lambda}{\sigma^2}\right)}}e^{-\frac{2\lambda}{\sigma^2}\cdot\frac{M_{1,0}}{v}}, $$
where $C_-,\,C_+>0$ are integration constants for the left and right branches of the solution, respectively. We can fix $C_-,\,C_+$ by imposing the normalisation condition 
$$ \int_\R g^\infty(v)\,dv=1. $$
We notice, however, that both branches of the singularity of $g^\infty$ at the origin are not integrable simultaneously. Specifically, if $M_{1,0}>0$ then
$$ \lim_{v\to 0^+}e^{-\frac{2\lambda}{\sigma^2}\cdot\frac{M_{1,0}}{v}}=0, \qquad \lim_{v\to 0^-}e^{-\frac{2\lambda}{\sigma^2}\cdot\frac{M_{1,0}}{v}}=+\infty, $$
therefore the singularity in $v=0$ is rightwards integrable but not leftwards integrable. In this case, to guarantee $g^\infty\in L^1(\R)$ we need $C_-=0$. Conversely, if $M_{1,0}<0$ the situation is opposite and we need $C_+=0$.

To fix the ideas, let us assume $M_{1,0}>0$. Then
$$ g^\infty(v)=\frac{C_+}{v^{2\left(1+\frac{\lambda}{\sigma^2}\right)}}e^{-\frac{2\lambda}{\sigma^2}\cdot\frac{M_{1,0}}{v}}\chi(v>0) $$
and the integration constant $C_+$ can be easily determined by observing that $g^\infty$ coincides, for $v>0$, with an \textit{inverse gamma distribution} with shape parameter $1+\frac{2\lambda}{\sigma^2}$ and scale parameter $\frac{2\lambda}{\sigma^2}M_{1,0}$. Thus $C_+=\left(\frac{2\lambda}{\sigma^2}M_{1,0}\right)^{1+\frac{2\lambda}{\sigma^2}}/\Gamma(1+\frac{2\lambda}{\sigma^2})$, where $\Gamma$ denotes the gamma function, and finally
\begin{equation}
	g^\infty(v)=\frac{\left(\frac{2\lambda}{\sigma^2}M_{1,0}\right)^{1+\frac{2\lambda}{\sigma^2}}}{\Gamma(1+\frac{2\lambda}{\sigma^2})}\cdot
		\frac{e^{-\frac{2\lambda}{\sigma^2}\cdot\frac{M_{1,0}}{v}}}{v^{2\left(1+\frac{\lambda}{\sigma^2}\right)}}\chi(v>0),
	\label{eq:inv_gamma}
\end{equation}
which, as expected, has mean value $M_{1,0}$ in $\R$.
	
Some remarks are in order:
\begin{enumerate}[label=\roman*)]
\item By explicit computation, the Fokker--Planck equation~\eqref{eq:FP.strong} admits solely~\eqref{eq:inv_gamma} as steady \textit{probability} distribution, independently of the initial condition $f_0$. Using appropriate \textit{entropy functionals} as Lyapunov functionals, one can prove that~\eqref{eq:inv_gamma} is indeed an asymptotically stable equilibrium of~\eqref{eq:FP.strong}, hence in particular it is attractive. Here we do not develop this aspect but we refer the interested reader to~\cite{furioli2017M3AS} for details.
\item Under the assumption $M_{1,0}>0$, the steady distribution $g^\infty$~\eqref{eq:inv_gamma} turns out to be supported in $\R_+$ independently of the support of the initial datum $f_0$. Likewise, if $M_{1,0}<0$ then $g^\infty$ is supported in $\R_-$. Hence, in the quasi-invariant regime the sign of the conserved mean value of $f_0$ determines the half-line where the whole probability mass distributes in the long run.
\item The solution $g$ to the Fokker--Planck equation~\eqref{eq:FP.strong}, and consequently also the steady distribution $g^\infty$~\eqref{eq:inv_gamma}, does not depend on the precise law of the stochastic fluctuation $\eta$ featured by the interaction coefficient $p_\epsilon$, cf.~\eqref{eq:peps.eta}. Only properties~\eqref{eq:eta} matter. In other words, only the low order moments of the stochastic coefficient $p_\epsilon$ impact effectively on the solution to the Boltzmann-type equation in the quasi-invariant regime.
\item The steady distribution $g^\infty$~\eqref{eq:inv_gamma} exhibits a fat tail at $+\infty$ with Pareto exponent $\gamma=1+\frac{2\lambda}{\sigma^2}$, indeed
$$ g^\infty(v)\sim\frac{C_+}{v^{2\left(1+\frac{\lambda}{\sigma^2}\right)}} \quad \text{for } v\to+\infty. $$
In particular, moments $M^\infty_n$ with
$$ n\geq\bar{n}:=\left\lfloor{1+\frac{2\lambda}{\sigma^2}}\right\rfloor $$
are infinite. Owing to~\eqref{eq:sigma.lambda} we observe that $\bar{n}\geq 2$, thus the first two moments of $g^\infty$ are finite for all admissible values of $\lambda$, $\sigma$, consistently with the assumptions $\ave{p_\epsilon+q_\epsilon}=1$ and $\ave{p_\epsilon^2+q_\epsilon^2}<1$ for $\epsilon$ small enough.

We can investigate further the formation of such a fat tail by profiting from the previous remark, which allows us to fix an arbitrary law of $\eta$ fulfilling~\eqref{eq:eta}. We consider, in particular, a discrete $\eta\in\{-1,\,1\}$ with $\Prob{\eta=\pm 1}=\frac{1}{2}$, which is such that $\ave{\eta}=0$, $\ave{\eta^2}=1$, and $\ave{\abs{\eta}^3}=1<+\infty$ as required by~\eqref{eq:eta}. Furthermore, assuming for simplicity $\frac{\sigma^2}{2\lambda}\geq\frac{2}{3}$ we obtain that $\eta$ complies also with bound~\eqref{eq:eta.lowerbound}. In this setting it results $\bar{n}=2$, therefore the first unbounded moment of $g^\infty$ is $M^\infty_3$. Computing $S(3)$, cf.~\eqref{eq:S}, for the coefficients~\eqref{eq:qeps},~\eqref{eq:peps.eta} of the scaled interaction law~\eqref{eq:v'.scaled} we find
$$ S(3)=3(\sigma^2-\lambda)\epsilon(1-\epsilon\lambda), $$
where $\sigma^2-\lambda\geq\frac{\lambda}{3}>0$ in view of the assumption $\frac{\sigma^2}{2\lambda}\geq\frac{2}{3}$, which entails $\sigma^2\geq\frac{4}{3}\lambda$. Therefore, if $\epsilon$ is small enough, cf.~\eqref{eq:eps.quasi-inv}, we have $S(3)>0$ so that Theorems~\ref{theo:bound.moments},~\ref{theo:conv.moments} fail, which paves the way to the formation of a fat tail in the Maxwellian.
\item If $M_{1,0}=0$ then from~\eqref{eq:FP.weak} with $\varphi(v)=v^2$ it results
$$ \frac{d}{dt}\int_\R v^2g(v,t)\,dv=(\sigma^2-2\lambda)\int_\R v^2g(v,t)\,dv, $$
therefore, owing to~\eqref{eq:sigma.lambda}, $\int_\R v^2g(v,t)\,dv\to 0$ for $t\to +\infty$. Hence, in this case $g^\infty(v)=\delta_0(v)$, which can be checked to be indeed a steady distributional solution to the Fokker--Planck equation~\eqref{eq:FP.strong}. For the scaled Boltzmann-type equation~\eqref{eq:Boltztype.scaled} with $\varphi(v)=v^2$ it results
$$ \frac{d}{dt}\int_\R v^2f_\epsilon(v,t)\,dv=(\sigma^2-2\lambda+2\epsilon\lambda^2)\int_\R v^2f_\epsilon(v,t)\,dv, $$
which under~\eqref{eq:sigma.lambda},~\eqref{eq:eps.quasi-inv} yields in turn $\int_\R v^2f_\epsilon(v,t)\,dv\to 0$ when $t\to +\infty$. Therefore, for $\epsilon$ small enough the Maxwellian of the scaled Boltzmann-type equation is invariably $f_\epsilon^\infty(v)=\delta_0(v)$, which remains so in the quasi-invariant limit $\epsilon\to 0^+$.
\end{enumerate}

%\textcolor{magenta}{Vogliamo dire da qualche parte che nel riscalamento l'equazione evolutiva della media è la medesima su quella scala di tempi-non solo perché in questo caso c'è conservazione, e che, invece, l'energia ha un'evoluzione diversa?Magari mostrando l'equazione di $M_1^\epsilon, M_2^\epsilon$, per la seconda usando il remark (v). Possiamo anche far notare che però nel limite $\epsilon \to 0$ l'energia è la stessa (lo è l'equazione evolutiva).}

\subsection{Other quasi-invariant scalings}
\label{sect:other_FP}
Different scalings of the coefficients $p_\epsilon$, $q_\epsilon$ in~\eqref{eq:v'.scaled} can be envisaged, corresponding to different quasi-invariant regimes of the Boltzmann-type equation~\eqref{eq:Boltztype.strong} that can be fruitfully investigated by means of appropriate Fokker--Planck equations in the quasi-invariant limit. Here, we shall consider two alternatives to the regimes discussed in Sections~\ref{sect:formal_q-i_limit},~\ref{sect:FP_steady}, which can be frequently encountered in applications.

Before entering the details, we state a result that we shall often reference in the sequel.
\begin{lemma}[Gr\"{o}nwall's inequality revisited] \label{lemma:Gronwall}
Let $u=u(t),\,\alpha=\alpha(t),\,\beta=\beta(t)$ be real continuous functions defined in an interval $[0,\,T]\subset\R_+$ for some $T>0$. Assume that $\alpha$ is differentiable and $\beta$ is non-negative in $[0,\,T]$ and that
$$ u(t)\leq\alpha(t)+\int_0^t\beta(\tau)u(\tau)\,d\tau, \qquad \forall\,t\in [0,\,T]. $$
Then
$$ u(t)\leq\alpha(0)e^{\int_0^t\beta(\tau)\,d\tau}+\int_0^t\alpha'(\tau)e^{\int_\tau^t\beta(r)\,dr}\,d\tau, \qquad \forall\,t\in [0,\,T]. $$
\end{lemma}
\begin{proof}
Standard Gr\"{o}nwall's inequality implies (cf. e.g.,~\cite[Chapter 12]{mitrinovic1991BOOK})
\begin{align*}
    u(t)\leq\alpha(t)+\int_0^t\alpha(\tau)\beta(\tau)e^{\int_\tau^t\beta(r)\,dr}\,d\tau &= \alpha(t)-\int_0^t\alpha(\tau)\frac{d}{d\tau}e^{\int_\tau^t\beta(r)\,dr}\,d\tau,
\intertext{whence, integrating by parts,}
    &= \alpha(t)-\left(\alpha(\tau)e^{\int_\tau^t\beta(r)\,dr}\right\vert_{\tau=0}^{\tau=t}+\int_0^t\alpha'(\tau)e^{\int_\tau^t\beta(r)\,dr}\,d\tau \\
    &= \alpha(0)e^{\int_0^t\beta(r)\,dr}+\int_0^t\alpha'(\tau)e^{\int_\tau^t\beta(r)\,dr}\,d\tau,
\end{align*}
which gives the thesis.
\end{proof}

\subsubsection{Advection-dominated regime}
\label{sect:q.i._advection}
Assume $q_\epsilon$ is scaled like in~\eqref{eq:qeps} whereas $p_\epsilon$ is scaled in such a way that
\begin{equation}
	\ave{p_\epsilon}=1-\epsilon\lambda, \qquad \Var{p_\epsilon}=\epsilon^{1+\delta}\sigma^2
	\label{eq:peps-advection}
\end{equation}
with $\delta>0$. Again, the law of $p_\epsilon$ converges to $\delta_1$ and $p_\epsilon$ itself converges to $1$ both in mean and quadratic mean as $\epsilon\to 0^+$. Moreover, $\ave{p_\epsilon+q_\epsilon}=1$ for all $\epsilon>0$ while
$$ \ave{p_\epsilon^2+q_\epsilon^2}=1-2\epsilon\lambda(1-\epsilon\lambda)+\epsilon^{1+\delta}\sigma^2. $$
With a little algebra, considering that for $\epsilon<1$ it results $\epsilon\leq\epsilon^\delta$ if $\delta<1$ while $\epsilon^\delta\leq\epsilon$ if $\delta\geq 1$, we obtain that under the assumptions
\begin{equation}
	\sigma^2>2\lambda(1-\lambda), \qquad \epsilon<\left(\frac{2\lambda}{2\lambda^2+\sigma^2}\right)^{1/\min\{\delta,\,1\}}
	\label{eq:sigma.lambda.eps_quasi-inv-advection}
\end{equation}
we are in the general setting $\ave{p_\epsilon^2+q_\epsilon^2}<1$.

Writing, for a sufficiently smooth observable $\varphi$, say $\varphi\in C^2(\R)$,
\begin{align*}
	\ave{\varphi(v_\epsilon')-\varphi(v)} &= \varphi'(v)\ave{(p_\epsilon-1)v+q_\epsilon v_\ast}+\frac{1}{2}\ave{\varphi''(\bar{v}_\epsilon)((p_\epsilon-1)v+q_\epsilon v_\ast)^2} \\
	&= \epsilon\lambda\varphi'(v)(v_\ast-v)+\frac{1}{2}\ave{\varphi''(\bar{v}_\epsilon)((p_\epsilon-1)v+q_\epsilon v_\ast)^2},
\end{align*}
where $\bar{v}_\epsilon$ is a point between $\min\{v,\,v_\epsilon'\}$ and $\max\{v,\,v_\epsilon'\}$, we discover that the scaled Boltzmann-type equation takes the form
\begin{align}
	\begin{aligned}[b]
		\frac{d}{dt}\int_\R\varphi(v)f_\epsilon(v,t)\,dv &= \frac{1}{\epsilon}\int_\R\int_\R\ave{\varphi(v_\epsilon')-\varphi(v)}f_\epsilon(v,t)f_\epsilon(v_\ast,t)\,dv\,dv_\ast \\
		&=\lambda\int_\R\varphi'(v)(M_{1,0}-v)f_\epsilon(v,t)\,dv \\
		&\phantom{=} +\frac{1}{2\epsilon}\int_\R\int_\R\ave{\varphi''(\bar{v}_\epsilon)((p_\epsilon-1)v+q_\epsilon v_\ast)^2}f_\epsilon(v,t)f_\epsilon(v_\ast,t)\,dv\,dv_\ast,
	\end{aligned}
	\label{eq:Boltztype.scaled-advection}
\end{align}
where $M_{1,0}$ is the conserved mean value of $f_\epsilon(t)$ for every $\epsilon>0$. This equation suggests that, in the quasi-invariant limit $\epsilon\to 0^+$, the solution $f_\epsilon$ somehow approaches the solution $g$ of
\begin{equation}
	\frac{d}{dt}\int_\R\varphi(v)g(v,t)\,dv=\lambda\int_\R\varphi'(v)(M_{1,0}-v)g(v,t)\,dv.
	\label{eq:FP.weak-advection}
\end{equation}
for arbitrary observables $\varphi\in C^2(\R)$. This is indeed true in the sense specified by Theorem~\ref{theo:q.i.-advection} below, for which we need preliminarily the following
\begin{lemma} \label{lemma:g.P2-advection}
Any solution $g$ to~\eqref{eq:FP.weak-advection} issuing from an initial condition $f_0\in\cP_2(\R)$ is such that $g(t)\in\cP_2(\R)$ for all $t>0$.
\end{lemma}
\begin{proof}
Letting $\varphi(v)=v^2$ in~\eqref{eq:FP.weak-advection} and recalling that $M_{1,0}=\int_\R vf_0(v)\,dv=\int_\R vg(v,t)\,dv$ for all $t>0$ yields
$$ \frac{d}{dt}\int_\R v^2g(v,t)\,dv=2\lambda\left(M_{1,0}^2-\int_\R v^2g(v,t)\,dv\right), $$
whence
$$ \int_\R v^2g(v,t)\,dv=M_{2,0}e^{-2\lambda t}+M_{1,0}^2\left(1-e^{-2\lambda t}\right)\leq M_{2,0}, $$
being $M_{2,0}:=\int_\R v^2f_0(v)\,dv\geq M_{1,0}^2$.
\end{proof}
\begin{theorem} \label{theo:q.i.-advection}
Let $f_\epsilon\in C^0([0,\,+\infty);\,\cP_2(\R))$ be the solution to~\eqref{eq:Boltztype.scaled-advection} issuing from an initial datum $f_0\in\cP_2(\R)$. Assume moreover that~\eqref{eq:FP.weak-advection} admits a solution $g\in C^0([0,\,+\infty);\,\cP_2(\R))$ issuing from $f_0$ as well. Then
$$ \lim_{\epsilon\to 0^+}\sup_{t\in [0,\,T]}d_2(f_\epsilon(t),g(t))=0 $$
for all $T>0$.
\end{theorem}
\begin{remark}
We can use the Fourier metric $d_2$ to evaluate the distance between $f_\epsilon(t)$ and $g(t)$ because, besides $\int_\R f_\epsilon(v,t)\,dv=\int_\R g(v,t)\,dv=1$, we have $\int_\R vf_\epsilon(v,t)\,dv=\int_\R vg(v,t)\,dv=M_{1,0}$ for all admissible $\epsilon>0$, cf.~\eqref{eq:sigma.lambda.eps_quasi-inv-advection}, and all $t>0$.
\end{remark}
\begin{proof}[Proof of Theorem~\ref{theo:q.i.-advection}]
Let us denote by $Q_\epsilon(f,f)$ the collisional operator featured by the $\epsilon$-scaled Boltzmann-type equation~\eqref{eq:Boltztype.scaled-advection}, which is such that
$$ \int_\R\varphi(v)Q_\epsilon(f,f)(v,t)\,dv=\frac{1}{\epsilon}\int_\R\int_\R\ave{\varphi(v_\epsilon')-\varphi(v)}f(v,t)f(v_\ast,t)\,dv\,dv_\ast $$
for every observable quantity $\varphi$. Likewise, let us denote by $J(g)$ the operator defined, in weak form, by the right-hand side of~\eqref{eq:FP.weak-advection}, i.e. such that
$$ \int_\R\varphi(v)J(g)(v,t)\,dv=\lambda\int_\R\varphi'(v)(M_{1,0}-v)g(v,t)\,dv $$
for every sufficiently smooth observable quantity $\varphi$. Thus, the equations satisfied by $f_\epsilon$, $g$ can be rewritten as
$$ \partial_tf_\epsilon=Q_\epsilon(f_\epsilon,f_\epsilon), \qquad \partial_tg=J(g), $$
respectively. With $\varphi(v)=e^{-i\xi v}\in C^\infty(\R)$ in~\eqref{eq:Boltztype.scaled-advection},~\eqref{eq:FP.weak-advection} they become
$$ \partial_t\hat{f}_\epsilon=\widehat{Q}_\epsilon(\hat{f}_\epsilon,\hat{f}_\epsilon), \qquad \partial_t\hat{g}=\widehat{J}(\hat{g}), $$
where the Fourier-transformed version of the operator $Q_\epsilon$ reads\footnote{For completeness, the Fourier-transformed version of the operator $J$ is
$$ \widehat{J}(\hat{g})(\xi,t)=-i\lambda M_{1,0}\xi\hat{g}(\xi,t)+\lambda\xi\frac{\partial\hat{g}}{\partial\xi}(\xi,t) $$
but it is not needed explicitly in the proof.}
$$ \widehat{Q}_\epsilon(\hat{f}_\epsilon,\hat{f}_\epsilon)(\xi,t)=
	\frac{1}{\epsilon}\left(\ave{\hat{f}_\epsilon(p_\epsilon\xi,t)}\hat{f}_\epsilon(q_\epsilon\xi,t)-\hat{f}_\epsilon(\xi,t)\right). $$
In particular, we have taken into account that only the coefficient $p_\epsilon$ is stochastic. Next, we observe that
\begin{align}
	\begin{aligned}[b]
		\partial_t\bigl(\hat{f}_\epsilon-\hat{g}\bigr) &= \widehat{Q}_\epsilon(\hat{f}_\epsilon,\hat{f}_\epsilon)-\widehat{J}(\hat{g}) \\
		&= \widehat{Q}_\epsilon(\hat{f}_\epsilon,\hat{f}_\epsilon)-\widehat{Q}_\epsilon(\hat{g},\hat{g})+\widehat{Q}_\epsilon(\hat{g},\hat{g})-\widehat{J}(\hat{g}),
	\end{aligned}
	\label{eq:hatfeps-hatg}
\end{align}
where $\widehat{Q}_\epsilon(\hat{g},\hat{g})-\widehat{J}(\hat{g})$ is the Fourier-transformed version of the operator $Q_\epsilon(g,g)-J(g)$ which, by comparing the right-hand sides of~\eqref{eq:Boltztype.scaled-advection} and~\eqref{eq:FP.weak-advection}, can be written in weak form as
\begin{multline}
	\int_\R\varphi(v)\bigl(Q_\epsilon(g,g)(v,t)-J(g)(v,t)\bigr)\,dv \\
	=\frac{1}{2\epsilon}\int_\R\int_\R\ave*{\varphi''(\bar{v}_\epsilon)((p_\epsilon-1)v+q_\epsilon v_\ast)^2})g(v,t)g(v_\ast,t)\,dv\,dv_\ast
	\label{eq:Qeps-J}
\end{multline}
for every observable quantity $\varphi\in C^2(\R)$.

Dividing both sides of~\eqref{eq:hatfeps-hatg} by $\abs{\xi}^2$ and rearranging the terms yields
\begin{align*}
	\partial_t\frac{\hat{f}_\epsilon(\xi,t)-\hat{g}(\xi,t)}{\abs{\xi}^2}+\frac{1}{\epsilon}\frac{\hat{f}_\epsilon(\xi,t)-\hat{g}(\xi,t)}{\abs{\xi}^2}
		&= \frac{1}{\epsilon}\left[\ave*{\frac{\bigl(\hat{f}_\epsilon(p_\epsilon\xi,t)-\hat{g}(p_\epsilon\xi,t)\bigr)\hat{f}_\epsilon(q_\epsilon\xi,t)}{\abs{\xi}^2}}\right. \\
	&\phantom{=} \left.+\ave*{\frac{\hat{g}(p_\epsilon\xi,t)\bigl(\hat{f}_\epsilon(q_\epsilon\xi,t)-\hat{g}(q_\epsilon\xi,t)\bigr)}{\abs{\xi}^2}}\right] \\
	&\phantom{=} +\frac{\widehat{Q}_\epsilon(\hat{g},\hat{g})(\xi,t)-\widehat{J}(\hat{g})(\xi,t)}{\abs{\xi}^2},
\end{align*}
whence
\begin{align}
	\begin{aligned}[b]
		\partial_t\left(e^{t/\epsilon}\frac{\abs{\hat{f}_\epsilon(\xi,t)-\hat{g}(\xi,t)}}{\abs{\xi}^2}\right)
			&\leq \frac{e^{t/\epsilon}}{\epsilon}\ave{p_\epsilon^2+q_\epsilon^2}d_2(f_\epsilon(t),g(t)) \\
		&\phantom{=} +e^{t/\epsilon}\sup_{\xi\in\R\setminus\{0\}}\frac{\abs{\widehat{Q}_\epsilon(\hat{g},\hat{g})(\xi,t)-\widehat{J}(\hat{g})(\xi,t)}}{\abs{\xi}^2}.
	\end{aligned}
	\label{eq:d2.proof_FP}
\end{align}
The last term at the right-hand side can be estimated out of~\eqref{eq:Qeps-J} with $\varphi(v)=e^{-i\xi v}$, considering that $\abs{\varphi''(v)}\leq\abs{\xi}^2$:
\begin{align*}
	\frac{\abs{\widehat{Q}_\epsilon(\hat{g},\hat{g})(\xi,t)-\widehat{J}(\hat{g})(\xi,t)}}{\abs{\xi}^2} &\leq
		\frac{1}{2\epsilon}\int_\R\int_\R\ave{((p_\epsilon-1)v+q_\epsilon v_\ast)^2}g(v,t)g(v_\ast,t)\,dv\,dv_\ast \\
	&= \frac{1}{2}\left((\epsilon^\delta\sigma^2+2\epsilon\lambda^2)\int_\R v^2g(v,t)\,dv-2\epsilon\lambda^2M_{1,0}^2\right) \\
	&\leq \frac{\epsilon^\delta\sigma^2+2\epsilon\lambda^2}{2}\int_\R v^2g(v,t)\,dv \\
	&\leq \epsilon^{\min\{\delta,\,1\}}\frac{\sigma^2+2\lambda^2}{2}\int_\R v^2g(v,t)\,dv \\
	&\leq \epsilon^{\min\{\delta,\,1\}}\frac{\sigma^2+2\lambda^2}{2}M_{2,0},
\end{align*}
where we have assumed $\epsilon<1$ and, in the last passage, we have used Lemma~\ref{lemma:g.P2-advection}. In the rest of the proof, we shall denote
$$ \bar{K}:=\frac{\sigma^2+2\lambda^2}{2}M_{2,0} $$
for brevity.

Back to~\eqref{eq:d2.proof_FP}, integrating in time on the interval $[0,\,t]$, $t>0$, we get
$$ e^{t/\epsilon}d_2(f_\epsilon(t),g(t))\leq\frac{\ave{p_\epsilon^2+q_\epsilon^2}}{\epsilon}\int_0^te^{\tau/\epsilon}d_2(f_\epsilon(\tau),g(\tau))\,d\tau+
	\bar{K}\epsilon^{1+\min\{\delta,\,1\}}(e^{t/\epsilon}-1), $$
where we have taken into account that $d_2(f_\epsilon(0),g(0))=0$ as $f_\epsilon(0)=g(0)=f_0$ by assumption.

Now, applying the revisited version of Gr\"{o}nwall's inequality established in Lemma~\ref{lemma:Gronwall} with
$$ u(t)=e^{t/\epsilon}d_2(f_\epsilon(t),g(t)), \quad \alpha(t)=\bar{K}\epsilon^{1+\min\{\delta,\,1\}}(e^{t/\epsilon}-1), \quad
	\beta=\frac{\ave{p_\epsilon^2+q_\epsilon^2}}{\epsilon} $$
yields
$$ d_2(f_\epsilon(t),g(t))\leq\bar{K}\epsilon^{\min\{\delta,\,1\}}\int_0^te^{\frac{\ave{p_\epsilon^2+q_\epsilon^2}-1}{\epsilon}(t-\tau)}\,d\tau
	\leq\bar{K}\epsilon^{\min\{\delta,\,1\}}t, $$
where in the second inequality we have used that $\ave{p_\epsilon^2+q_\epsilon^2}-1<0$, hence $e^{\frac{\ave{p_\epsilon^2+q_\epsilon^2}-1}{\epsilon}(t-\tau)}\leq 1$ for $\tau\leq t$, if $\epsilon>0$ is small enough. Then, for $T>0$ arbitrary and finite,
$$ \sup_{t\in [0,\,T]}d_2(f_\epsilon(t),g(t))\leq\bar{K}T\epsilon^{\min\{\delta,\,1\}}\xrightarrow{\epsilon\to 0^+}0 $$
and we are done.
\end{proof}

Considering smooth observables with compact support, so as to get rid of boundary terms in the integration-by-parts, from~\eqref{eq:FP.weak-advection} we obtain the following strong form:
\begin{equation}
	\frac{\partial g}{\partial t}+\lambda\frac{\partial}{\partial v}\left((M_{1,0}-v)g\right)=0,
	\label{eq:FP.strong-advection}
\end{equation}
which, compared with~\eqref{eq:FP.strong}, is a linear variable-coefficient Fokker--Planck equation \textit{without diffusion}. This equation can be solved explicitly for all $t>0$ by the method of characteristics, starting from an initial condition $g(v,0)=f_0(v)$, $v\in\R$, satisfying
$$ \int_\R f_0(v)\,dv=1, \qquad \int_\R vf_0(v)\,dv=M_{1,0}. $$
Specifically, introducing the auxiliary function $\tilde{g}(v,t):=e^{-\lambda t}g(v,t)$ reduces~\eqref{eq:FP.strong-advection} to the linear advection equation $\partial_t\tilde{g}+\lambda(M_{1,0}-v)\partial_v\tilde{g}=0$, which indicates that $\tilde{g}$ is constant along the curves $v=v(t)$ such that $\dot{v}=\lambda(M_{1,0}-v)$, i.e. $v=e^{-\lambda t}v_0+M_{1,0}(1-e^{-\lambda t})$ with $v_0\in\R$ arbitrary. Consequently, it results $\tilde{g}(v,t)=\tilde{g}(v_0,0)=g(v_0,0)=f_0(v_0)$ for all $v\in\R$, $t>0$, whence, going back to $g$,
\begin{equation}
	g(v,t)=e^{\lambda t}f_0\bigl(M_{1,0}+e^{\lambda t}(v-M_{1,0})\bigr).
	\label{eq:g.FP-advection}
\end{equation}

It is not difficult to check that such a $g$ satisfies $\int_\R g(v,t)\,dv=1$, $\int_\R vg(v,t)\,dv=M_{1,0}$ for every $t>0$. Furthermore, $g(t)$ converges in distributional sense to
\begin{equation}
	g^\infty(v)=\delta_{M_{1,0}}(v)
	\label{eq:ginf.FP-advection}
\end{equation}
when $t\to +\infty$, which is indeed the only stationary probability distribution of~\eqref{eq:FP.strong-advection}.

This can be also seen from Lemma~\ref{lemma:g.P2-advection}, which yields $\int_\R v^2g(v,t)\,dv\to M_{1,0}^2$ and consequently $\Var{g(t)}\to 0$ as $t\to +\infty$. On the other hand, the scaled Boltzmann-type equation~\eqref{eq:Boltztype.scaled-advection} with $\varphi(v)=v^2$ gives
$$ \int_\R v^2f_\epsilon(v,t)\,dv=-\left(2\lambda(1-\epsilon\lambda)-\epsilon^\delta\sigma^2\right)\int_\R v^2f_\epsilon(v,t)\,dv+2\lambda(1-\epsilon\lambda)M_{1,0}^2, $$
which, for $\epsilon,\,\lambda,\,\sigma$ complying with~\eqref{eq:sigma.lambda.eps_quasi-inv-advection}, yields
$$ \int_\R v^2f_\epsilon(v,t)\,dv\to\frac{2\lambda(1-\epsilon\lambda)}{2\lambda(1-\epsilon\lambda)-\epsilon^\delta\sigma^2}M_{1,0}^2
	\quad \text{as } t\to +\infty. $$
Notice that for $\epsilon>0$ this asymptotic value is strictly greater than $M_{1,0}^2$, meaning that when $\epsilon$ is not infinitesimal the Maxwellian of the scaled Boltzmann-type equation~\eqref{eq:Boltztype.scaled-advection} is not~\eqref{eq:ginf.FP-advection}. However, it results clearly $\frac{2\lambda(1-\epsilon\lambda)}{2\lambda(1-\epsilon\lambda)-\epsilon^\delta\sigma^2}M_{1,0}^2\to M_{1,0}^2$ as $\epsilon\to 0^+$.

\begin{remark}
The distribution~\eqref{eq:g.FP-advection} is also called a \textit{self-similar} solution to~\eqref{eq:FP.strong-advection}, because it is obtained at each time $t>0$ as a deformation of a prescribed profile (that of $f_0$).
\end{remark}

%\textcolor{magenta}{Forse spenderei qualche parola in più sulla soluzione autosimilare con questa regola per parlare dopo del limite quasi invariante nel label switching e anche per far emergere il fatto che è la non conservazione dell'energia che fa emergere la coda grassa. Questo è già stato detto in sezione 5, ma potremmo proprio fare un breve accenno al lavoro di Pareschi Toscani in cui fanno il limite quasi-invariante a partire da profili autosimilari.}

\subsubsection{Conserved energy regime}
\label{sect:q.i.-cons_en}
Now we consider a regime of the coefficients $p_\epsilon$, $q_\epsilon$, and a related quasi-invariant limit, departing from the general setting $\ave{p_\epsilon+q_\epsilon}=1$, $\ave{p_\epsilon^2+q_\epsilon^2}<1$ addressed so far and nonetheless representative of an important class of applications mimicking the classical one of gas molecules.

In more detail, the regime we refer to is the one in which both the mean value and the energy of the kinetic distribution function are conserved in time, as it happens in the case of elastic collisions among gas particles. To obtain such a regime, we observe, from the moment equations reported in Section~\ref{sect:moments_evol}, that $M_1=0$ is an equilibrium of the mean value independently of $p$, $q$. Hence, if $f_0$ has null mean value then such a null mean value is conserved at all successive times for whatever $p$, $q$. At that point, if $p$, $q$ satisfy $\ave{p^2+q^2}=1$ then also the energy of $f_0$ is conserved in time.

Therefore, in this section we fix
$$ \int_\R vf_0(v)\,dv=0 $$
and we devise appropriate scalings of $p_\epsilon$, $q_\epsilon$ matching the condition $\ave{p_\epsilon^2+q_\epsilon^2}=1$. Specifically, we set
\begin{equation}
    p_\epsilon=1-\epsilon\lambda+\sqrt{\epsilon}\sigma\eta, \qquad
	q_\epsilon=\sqrt{2\lambda\epsilon}\sqrt{1-\frac{\sigma^2}{2\lambda}-\frac{\epsilon\lambda}{2}},
    \label{eq:peps.qeps-cons_en}
\end{equation}
where $\lambda,\,\sigma>0$ fulfil~\eqref{eq:sigma.lambda} while $\eta$ is a random variable satisfying~\eqref{eq:eta}. It is not difficult to see that if $\epsilon$ is sufficiently small like in~\eqref{eq:eps.quasi-inv} and $\eta$ is bounded from below like in~\eqref{eq:eta.lowerbound} then $p_\epsilon,\,q_\epsilon\geq 0$ are well-defined and moreover $p_\epsilon\to 1$, $q_\epsilon\to 0$ as $\epsilon\to 0^+$. Additionally, since
\begin{equation}
	\ave{p_\epsilon^3+q_\epsilon^3}=
		1+\left[\sigma^3\ave{\eta^3}+\sqrt{2\lambda\left(1-\frac{\sigma^2}{2\lambda}\right)}\right]\epsilon^{3/2}+o(\epsilon^{3/2}),
	\label{eq:p3+q3.cons_en}
\end{equation}
if $\ave{\eta^3}$ is sufficiently negative then $\ave{p_\epsilon^3+q_\epsilon^3}<1$ for $\epsilon>0$ small enough.

Using the same ideas as in Section~\ref{sect:exist.uniq} but working now in $\cP_{2,1,\cM_3}(\R)\subset\cP_3(\R)$, for a suitable constant $\cM_3>0$ which exists in view of $\ave{p_\epsilon^3+q_\epsilon^3}<1$, with the Fourier metric $d_3$ it can be shown that for $\epsilon>0$ sufficiently small the scaled Boltzmann-type equation
\begin{equation}
	\frac{d}{dt}\int_\R\varphi(v)f_\epsilon(v,t)\,dv=\frac{1}{\epsilon}\int_\R\int_\R\ave{\varphi(v_\epsilon')-\varphi(v)}
		f_\epsilon(v,t)f_\epsilon(v_\ast,t)\,dv\,dv_\ast
	\label{eq:Boltztype.scaled-cons_en}
\end{equation}
admits a unique solution $f_\epsilon\in C^0([0,\,T];\,\cP_{2,1,\cM_3}(\R))$ issuing from an initial condition $f_0\in\cP_3(\R)$ with $M_{1,0}=0$.

Taylor-expanding the difference $\ave{\varphi(v_\epsilon')-\varphi(v)}$ around $v$ for a smooth observable, say $\varphi\in C^3(\R)$, we obtain
\begin{align*}
	\ave{\varphi(v_\epsilon')-\varphi(v)} &= \varphi'(v)\left(-\epsilon\lambda v+\sqrt{\epsilon}\sqrt{2\lambda-\sigma^2-\epsilon\lambda^2}v_\ast\right) \\
	&\phantom{=} +\frac{1}{2}\varphi''(v)\left(\epsilon(\sigma^2+\epsilon\lambda^2)(v^2-v_\ast^2)+2\epsilon\lambda v_\ast^2
		-2\epsilon\lambda\sqrt{\epsilon}\sqrt{2\lambda-\sigma^2-\epsilon\lambda^2}vv_\ast\right) \\
	&\phantom{=} +\frac{1}{6}\ave{\varphi'''(\bar{v}_\epsilon)((p_\epsilon-1)v+q_\epsilon v_\ast)^3},
\end{align*}
where, as usual, $\bar{v}_\epsilon$ is a point comprised between $\min\{v,\,v_\epsilon'\}$ and $\max\{v,\,v_\epsilon'\}$. Next, plugging this expansion into~\eqref{eq:Boltztype.scaled-cons_en} and recalling that $\int_\R vf_\epsilon(v,t)\,dv=0$, $\int_\R v^2f_\epsilon(v,t)\,dv=M_{2,0}$ for every $t>0$, $M_{2,0}>0$ being the energy of the initial condition $f_0$, we get
\begin{align}
	\begin{aligned}[b]
		\frac{d}{dt}\int_\R\varphi(v)f_\epsilon(v,t)\,dv &= -\lambda\int_\R\varphi'(v)vf_\epsilon(v,t)\,dv \\
		&\phantom{\leq} +\frac{\sigma^2}{2}\int_\R\varphi''(v)(v^2-M_{2,0})f_\epsilon(v,t)\,dv+\lambda M_{2,0}\int_\R\varphi''(v)f_\epsilon(v,t)\,dv \\
		&\phantom{\leq} +\frac{\epsilon\lambda^2}{2}\int_\R\varphi''(v)(v^2-M_{2,0})f_\epsilon(v,t)\,dv \\
		&\phantom{\leq} +\frac{1}{6\epsilon}\int_\R\int_\R\ave{\varphi'''(\bar{v}_\epsilon)((p_\epsilon-1)v+q_\epsilon v_\ast)^3}
			f_\epsilon(v,t)f_\epsilon(v_\ast,t)\,dv\,dv_\ast.
	\end{aligned}
	\label{eq:Boltztype.scaled-cons_en.2}
\end{align}
By inspecting this equation we see that it is reasonable to expect that, in the quasi-invariant limit $\epsilon\to 0^+$, the solution $f_\epsilon$ somehow approaches the solution $g$ of
\begin{align}
	\begin{aligned}[b]
		\frac{d}{dt}\int_\R\varphi(v)g(v,t)\,dv &= -\lambda\int_\R\varphi'(v)vg(v,t)\,dv \\
		&\phantom{\leq} +\frac{\sigma^2}{2}\int_\R\varphi''(v)(v^2-M_{2,0})g(v,t)\,dv+\lambda M_{2,0}\int_\R\varphi''(v)g(v,t)\,dv
	\end{aligned}
	\label{eq:FP.weak-cons_en}
\end{align}
for an arbitrary observable $\varphi\in C^3(\R)$. To prove that this is indeed true we need preliminarily the following
\begin{lemma} \label{lemma:g.P3-cons_en}
Any solution $g$ to~\eqref{eq:FP.weak-cons_en} issuing from an initial condition $f_0\in \cP_3(\R)$ is such that $g(t)\in\cP_3(\R)$ for all $t>0$.
\end{lemma}
\begin{proof}
Letting $\varphi(v)=\abs{v}^3$ in~\eqref{eq:FP.weak-cons_en}, along with $\varphi'(v)=3v\abs{v}$ and $\varphi''(v)=6\abs{v}$, produces
\begin{align*}
	\frac{d}{dt}\int_\R\abs{v}^3g(v,t)\,dv &= 3(\sigma^2-\lambda)\int_\R\abs{v}^3g(v,t)\,dv+6M_{2,0}(2\lambda-\sigma^2)\int_\R\abs{v}g(v,t)\,dv; \\
\intertext{since $2\lambda-\sigma^2>0$ in view of~\eqref{eq:sigma.lambda}, we can apply H\"{o}lder's inequality with exponents $\mathfrak{p}=3$, $\mathfrak{q}=\frac{3}{2}$ to the second term at the right-hand side to get}
	&\leq 3(\sigma^2-\lambda)\int_\R\abs{v}^3g(v,t)\,dv+6M_{2,0}(2\lambda-\sigma^2)\left(\int_\R\abs{v}^3g(v,t)\,dv\right)^{1/3}.
\end{align*}
This is a Bernoulli-like differential inequality, which can be solved by standard methods to find
$$ \int_\R\abs{v}^3g(v,t)\,dv\leq\left[e^{-2(\lambda-\sigma^2)t}M_{3,0}^{2/3}+\frac{2(2\lambda-\sigma^2)}{\lambda-\sigma^2}M_{2,0}
	\left(1-e^{-2(\lambda-\sigma^2)t}\right)\right]^{3/2} $$
with $M_{3,0}:=\int_\R\abs{v}^3f_0(v)\,dv$, whence the thesis follows.
\end{proof}

We are now in a position to show that
\begin{theorem} \label{theo:q.i.-cons_en}
Let $f_\epsilon\in C^0([0,\,+\infty);\,\cP_3(\R))$ be the solution to~\eqref{eq:Boltztype.scaled-cons_en} issuing from an initial datum $f_0\in\cP_3(\R)$. Assume moreover that~\eqref{eq:FP.weak-cons_en} admits a solution $g\in C^0([0,\,+\infty);\,\cP_3(\R))$ issuing from $f_0$ as well. Then
$$ \lim_{\epsilon\to 0^+}\sup_{t\in [0,\,T]}d_3(f_\epsilon(t),g(t))=0 $$
for all $T>0$.
\end{theorem}
\begin{remark}
We can use the Fourier metric $d_3$ to estimate the distance between $f_\epsilon(t)$ and $g(t)$ because, besides $\int_\R f_\epsilon(v,t)\,dv=\int_\R g(v,t)\,dv=1$, we also have $\int_\R vf_\epsilon(v,t)\,dv=\int_\R vg(v,t)\,dv=0$ and $\int_\R v^2f_\epsilon(v,t)\,dv=\int_\R v^2g(v,t)\,dv=M_{2,0}$ for all admissible $\epsilon>0$, cf.~\eqref{eq:eps.quasi-inv}, and all $t>0$.
\end{remark}

\begin{proof}[Proof of Theorem~\ref{theo:q.i.-cons_en}]
Proceeding like in the proof of Theorem~\ref{theo:q.i.-advection} but with the operator $J(g)$ defined as
\begin{align*}
	\int_\R\varphi(v)J(g)(v,t)\,dv &= -\lambda\int_\R\varphi'(v)vg(v,t)\,dv \\
	&\phantom{=} +\frac{\sigma^2}{2}\int_\R\varphi''(v)(v^2-M_{2,0})g(v,t)\,dv+\lambda M_{2,0}\int_\R\varphi''(v)g(v,t)\,dv
\end{align*}
for every sufficiently smooth observable $\varphi$, we arrive again at~\eqref{eq:hatfeps-hatg} with $\widehat{Q}_\epsilon(\hat{g},\hat{g})-\widehat{J}(\hat{g})$ the Fourier-transformed version of the operator $Q_\epsilon(g,g)-J(g)$, which now, by inspecting the right-hand sides of~\eqref{eq:Boltztype.scaled-cons_en.2},~\eqref{eq:FP.weak-cons_en}, can be written in weak form as
\begin{align}
	\begin{aligned}[b]
		\int_\R\varphi(v)(Q_\epsilon(g,g)(v,t)-J(g)(v,t))\,dv &= \frac{\epsilon\lambda^2}{2}\int_\R\varphi''(v)(v^2-M_{2,0})g(v,t)\,dv \\
		&\phantom{=} +\frac{1}{6\epsilon}\int_\R\int_\R\ave{\varphi'''(\bar{v}_\epsilon)((p_\epsilon-1)v+q_\epsilon v_\ast)^3}
			g(v,t)g(v_\ast,t)\,dv\,dv_\ast
	\end{aligned}
	\label{eq:Qeps-J.2}
\end{align}
for every observable quantity $\varphi\in C^3(\R)$. Next, dividing both sides of~\eqref{eq:hatfeps-hatg} by $\abs{\xi}^3$ and rearranging the terms we obtain the equivalent of~\eqref{eq:d2.proof_FP}, which in this case reads
\begin{align}
	\begin{aligned}[b]
		\partial_t\left(e^{t/\epsilon}\frac{\abs{\hat{f}_\epsilon(\xi,t)-\hat{g}(\xi,t)}}{\abs{\xi}^3}\right)
			&\leq \frac{e^{t/\epsilon}}{\epsilon}\ave{p_\epsilon^3+q_\epsilon^3}d_3(f_\epsilon(t),g(t)) \\
		&\phantom{=} +e^{t/\epsilon}\sup_{\xi\in\R\setminus\{0\}}\frac{\abs{\widehat{Q}_\epsilon(\hat{g},\hat{g})(\xi,t)-\widehat{J}(\hat{g})(\xi,t)}}{\abs{\xi}^3}.
	\end{aligned}
	\label{eq:d3.proof_FP}
\end{align}
To estimate the second term at the right-hand side we choose $\varphi(v)=e^{-i\xi v}$ in~\eqref{eq:Qeps-J.2}, noticing that $\abs{\varphi''(v)}\leq\abs{\xi}^2$ and $\abs{\varphi''(v)}\leq\abs{\xi}^3$:
\begin{align*}
	\frac{\abs{\widehat{Q}_\epsilon(\hat{g},\hat{g})(\xi,t)-\widehat{J}(\hat{g})(\xi,t)}}{\abs{\xi}^3} &\leq
		\frac{\epsilon\lambda^2}{2\abs{\xi}}\abs*{\int_\R e^{-i\xi v}(v^2-M_{2,0})g(v,t)\,dv} \\
	&\phantom{\leq} +\frac{1}{6\epsilon}\int_\R\int_\R\ave{\abs{(p_\epsilon-1)v+q_\epsilon v_\ast}^3}g(v,t)g(v_\ast,t)\,dv\,dv_\ast.
\end{align*}
To handle the first term at the right-hand side we observe that $\int_\R(v^2-M_{2,0})g(v,t)\,dv=0$ because $g$ has conserved energy $M_{2,0}$, hence
\begin{align*}
	\frac{\epsilon\lambda^2}{2\abs{\xi}}\abs*{\int_\R e^{-i\xi v}(v^2-M_{2,0})g(v,t)\,dv} &=
		\frac{\epsilon\lambda^2}{2\abs{\xi}}\abs*{\int_\R(e^{-i\xi v}-1)(v^2-M_{2,0})g(v,t)\,dv} \\
	&\leq \frac{\epsilon\lambda^2}{2\abs{\xi}}\int_\R\abs{e^{-i\xi v}-1}\cdot\abs{v^2-M_{2,0}}g(v,t)\,dv
\intertext{and further, since $\abs{e^{-i\xi v}-1}\leq\abs{\xi v}$,}
	&\leq \frac{\epsilon\lambda^2}{2}\int_\R\abs{v}\cdot\abs{v^2-M_{2,0}}g(v,t)\,dv \\
	&\leq \frac{\epsilon\lambda^2}{2}\left(\int_\R\abs{v}^3g(v,t)\,dv+M_{2,0}\int_\R\abs{v}g(v,t)\,dv\right).
\end{align*}
Owing to Lemma~\ref{lemma:g.P3-cons_en}, the quantity
$$ K_1(t):=\frac{\lambda^2}{2}\left(\int_\R\abs{v}^3g(v,t)\,dv+M_{2,0}\int_\R\abs{v}g(v,t)\,dv\right) $$
is finite for every $t>0$, therefore we conclude
$$ \frac{\epsilon\lambda^2}{2\abs{\xi}}\abs*{\int_\R e^{-i\xi v}(v^2-M_{2,0})g(v,t)\,dv}\leq\epsilon K_1(t)<+\infty, \qquad \forall\,t>0. $$

On the other hand, an explicit computation shows that for $\epsilon$ small, cf.~\eqref{eq:eps.quasi-inv}, there exists a constant $C_{\lambda,\sigma,\eta}>0$ depending only on $\lambda,\,\sigma$ and on the random variable $\eta$ through the quantities $\ave{\abs{\eta}}\leq\ave{\eta^2}^{1/2}=1$, $\ave{\abs{\eta}^3}<+\infty$, cf.~\eqref{eq:eta}, such that
$$ \frac{1}{6\epsilon}\int_\R\int_\R\ave{\abs{(p_\epsilon-1)v+q_\epsilon v_\ast}^3}g(v,t)g(v_\ast,t)\,dv\,dv_\ast
	\leq \sqrt{\epsilon}C_{\lambda,\sigma,\eta}\int_\R\abs{v}^3g(v,t)\,dv. $$
Also in this case, owing to Lemma~\ref{lemma:g.P3-cons_en} the quantity
$$ K_2(t):=C_{\lambda,\sigma,\eta}\int_\R\abs{v}^3g(v,t)\,dv $$
is finite for every $t>0$, therefore
$$ \frac{1}{6\epsilon}\int_\R\int_\R\ave{\abs{(p_\epsilon-1)v+q_\epsilon v_\ast}^3}g(v,t)g(v_\ast,t)\,dv\,dv_\ast
	\leq\sqrt{\epsilon}K_2(t)<+\infty, \qquad \forall\,t>0. $$

Back to~\eqref{eq:d3.proof_FP}, upon integrating on $[0,\,t]$, $t>0$, and taking into account that $f_\epsilon(0)=g(0)=f_0$ we obtain
\begin{align*}
	e^{t/\epsilon}d_3(f_\epsilon(t),g(t)) &\leq \frac{\ave{p_\epsilon^3+q_\epsilon^3}}{\epsilon}\int_0^te^{\tau/\epsilon}d_3(f_\epsilon(\tau),g(\tau))\,d\tau \\
	&\phantom{\leq} +\int_0^te^{\tau/\epsilon}\left(\epsilon K_1(\tau)+\sqrt{\epsilon}K_2(\tau)\right)\,d\tau
\intertext{and, as soon as $\epsilon<1$,}
	&\leq \frac{\ave{p_\epsilon^3+q_\epsilon^3}}{\epsilon}\int_0^te^{\tau/\epsilon}d_3(f_\epsilon(\tau),g(\tau))\,d\tau
		+\sqrt{\epsilon}\int_0^te^{\tau/\epsilon}K(\tau)\,d\tau
\end{align*}
with $K(t):=K_1(t)+K_2(t)$. Invoking now Lemma~\ref{lemma:Gronwall} with
$$ u(t)=e^{t/\epsilon}d_3(f_\epsilon(t),g(t)), \quad \alpha(t)=\sqrt{\epsilon}\int_0^te^{\tau/\epsilon}K(\tau)\,d\tau,
	\quad \beta=\frac{\ave{p_\epsilon^3+q_\epsilon^3}}{\epsilon} $$
we discover
$$ d_3(f_\epsilon(t),g(t))\leq\sqrt{\epsilon}\int_0^tK(\tau)e^{\frac{\ave{p_\epsilon^3+q_\epsilon^3}-1}{\epsilon}(t-\tau)}\,d\tau
	\leq\sqrt{\epsilon}\int_0^tK(\tau)\,d\tau, $$
where we have used that $\ave{p_\epsilon^3+q_\epsilon^3}<1$, thus $e^{\frac{\ave{p_\epsilon^3+q_\epsilon^3}-1}{\epsilon}(t-\tau)}\leq 1$ for $\tau\leq t$, when $\epsilon>0$ is sufficiently small.

Lemma~\ref{lemma:g.P3-cons_en} implies that $\int_\R\abs{v}^3g(v,t)\,dv$, thus also $\int_\R\abs{v}g(v,t)\,dv\leq\left(\int_\R\abs{v}^3g(v,t)\,dv\right)^{1/3}$, is uniformly bounded on every interval $[0,\,T]$, $T>0$. Therefore, so is the function $K$: there exists a constant $\bar{K}_T>0$ such that $K(t)\leq\bar{K}_T$ for all $t\in [0,\,T]$. Consequently,
$$ \sup_{t\in [0,\,T]}d_3(f_\epsilon(t),g(t))\leq\bar{K}_TT\sqrt{\epsilon}\xrightarrow{\epsilon\to 0^+}0 $$
for all $T>0$, whence the thesis follows.
\end{proof}

If in~\eqref{eq:FP.weak-cons_en} we restrict the observable $\varphi$ to compactly supported $C^3$-functions we obtain the following strong form of the equation:
\begin{equation}
	\frac{\partial{g}}{\partial{t}}-\lambda\frac{\partial}{\partial{v}}(vg)=
		\frac{\partial^2}{\partial{v^2}}\left[\left(\frac{\sigma^2}{2}v^2+\left(\lambda-\frac{\sigma^2}{2}\right)M_{2,0}\right)g\right],
	\label{eq:FP.strong-cons_en}
\end{equation}
i.e. a linear Fokker--Planck equation with non-constant coefficients. The steady distribution $g^\infty=g^\infty(v)$ can be sought proceeding similarly to Section~\ref{sect:FP_steady}, which produces the following results:
\begin{enumerate}[label=\roman*)]
\item If $\sigma>0$ we obtain
$$ g^\infty(v)=\frac{C_{\lambda,\sigma,M_{2,0}}}{\left[\frac{\sigma^2}{2}v^2+\left(\lambda-\frac{\sigma^2}{2}\right)M_{2,0}\right]^{1+\frac{\lambda}{\sigma^2}}}, $$
the constant $C_{\lambda,\sigma,M_{2,0}}>0$ being such that the normalisation condition $\int_\R g^\infty(v)\,dv=1$ holds true. For the records, the explicit expression of $C_{\lambda,\sigma,M_{2,0}}$ is
$$ C_{\lambda,\sigma,M_{2,0}}=\frac{2^{1+\frac{\lambda}{\sigma^2}}\sqrt{\pi}\Gamma\left(\frac{1}{2}+\frac{\lambda}{\sigma^2}\right)M_{2,0}}
	{(2\lambda-\sigma^2)^{\frac{1}{2}+\frac{\lambda}{\sigma^2}}\sigma^2\Gamma\left(1+\frac{\lambda}{\sigma^2}\right)}, $$
$\Gamma$ denoting, as usual, the gamma function. We observe that
$$ g^\infty(v)\sim\frac{C_{\lambda,\sigma,M_{2,0}}}{\left(\frac{\sigma^2}{2}\right)^{1+\frac{\lambda}{\sigma^2}}v^{2\left(1+\frac{\lambda}{\sigma^2}\right)}}, $$
therefore $g^\infty$ develops fat tails for $\abs{v}\to\infty$ with Pareto exponent $\gamma=1+\frac{2\lambda}{\sigma^2}>2$, cf.~\eqref{eq:sigma.lambda}. This is the same Pareto exponent as that of the case discussed in Section~\ref{sect:FP_steady}. Furthermore, we notice that if the tighter condition $\sigma^2<\lambda$ holds then $\gamma>3$, which ensures the integrability also of the third moment of $g^\infty$. By inspecting the proof of Lemma~\ref{lemma:g.P3-cons_en}, we see that this is consistent with the fact that for $\sigma^2<\lambda$ one gets the uniform boundedness in time of $\int_\R\abs{v}^3g(v,t)\,dv$.
\item If $\sigma=0$, i.e. if the coefficient $p_\epsilon$ of the interaction rule is fully deterministic, the Fokker--Planck equation~\eqref{eq:FP.strong-cons_en} reduces to
$$ \frac{\partial{g}}{\partial{t}}-\lambda\frac{\partial}{\partial{v}}(vg)=\lambda M_{2,0}\frac{\partial^2g}{\partial v^2}, $$
which admits as steady distribution the Gaussian density (Maxwell-type distribution)
\begin{equation}
    g^\infty(v)=\frac{1}{\sqrt{2\pi M_{2,0}}}e^{-\frac{v^2}{2M_{2,0}}}
    \label{eq:Gaussian}
\end{equation}  
with zero mean, (internal) energy (i.e. variance) $M_{2,0}$, and slim tails.

We can check that, in this case, this is exactly a Maxwellian of the scaled Boltzmann-type equation~\eqref{eq:Boltztype.scaled-cons_en} for \textit{every} admissible $\epsilon>0$, cf.~\eqref{eq:eps.quasi-inv}, in particular every $\epsilon<\frac{1}{\lambda}$. Indeed, passing to the Fourier-transformed version of~\eqref{eq:Boltztype.scaled-cons_en} we obtain that the Maxwellian $f_\epsilon^\infty$ satisfies
\begin{equation}
	\hat{f}_\epsilon^\infty(\xi)=\hat{f}_\epsilon^\infty(p_\epsilon\xi)\hat{f}_\epsilon^\infty(q_\epsilon\xi),
	\label{eq:hatf_eps.Gaussian}
\end{equation}
where we have taken into account that under the assumption $\sigma=0$ neither $p_\epsilon$ nor $q_\epsilon$ is stochastic, thus the mean $\ave{\cdot}$ on the right-hand side disappears. Since
$$ \hat{g}^\infty(\xi)=e^{-\frac{M_{2,0}}{2}\xi^2}, $$
we have $\hat{g}^\infty(p_\epsilon\xi)\hat{g}^\infty(q_\epsilon\xi)=e^{-\frac{M_{2,0}}{2}(p_\epsilon^2+q_\epsilon^2)\xi^2}$, which coincides with $\hat{g}^\infty(\xi)$ due to $p_\epsilon^2+q_\epsilon^2=1$. Consequently, $\hat{g}^\infty$ solves~\eqref{eq:hatf_eps.Gaussian}, i.e. $\widehat{Q}_\epsilon(\hat{g}^\infty,\hat{g}^\infty)=0$, thus $g^\infty$ solves $Q_\epsilon(g^\infty,g^\infty)=0$.

Notice that if $\sigma=0$ then from~\eqref{eq:p3+q3.cons_en} it results $p_\epsilon^3+q_\epsilon^3=1+\sqrt{2\lambda}\epsilon^{3/2}+o(\epsilon^{3/2})$, whence  we deduce that $p_\epsilon^3+q_\epsilon^3>1$ when $\epsilon>0$ is small. Consequently, in this case the proof of Theorem~\ref{theo:q.i.-cons_en} fails in the last passages. Nevertheless, it still holds that
$$ d_3(f_\epsilon(t),g(t))\leq\sqrt{\epsilon}\int_0^tK(\tau)e^{\frac{p_\epsilon^3+q_\epsilon^3-1}{\epsilon}(t-\tau)}\,d\tau
	\leq\bar{K}_T\sqrt{\epsilon}te^{\frac{p_\epsilon^3+q_\epsilon^3-1}{\epsilon}t}, $$
whence
$$ \sup_{t\in [0,\,T]}d_3(f_\epsilon(t),g(t))\leq\bar{K}_TTe^{\frac{p_\epsilon^3+q_\epsilon^3-1}{\epsilon}T}\sqrt{\epsilon}, $$
so that, since $p_\epsilon^3+q_\epsilon^3-1\sim\sqrt{2\lambda}\epsilon^{3/2}$ when $\epsilon\to 0^+$, we recover again
$$ \lim_{\epsilon\to 0^+}\sup_{t\in [0,\,T]}d_3(f_\epsilon(t),g(t))\leq
	\bar{K}_TT\lim_{\epsilon\to 0^+}e^{\sqrt{2\lambda}T\sqrt{\epsilon}}\sqrt{\epsilon}=0 $$
for every $T>0$.
\end{enumerate}

\begin{remark}
The case $\sigma=0$ above is one in which the Maxwellian can be computed exactly from the Boltzmann-type equation in every regime of the parameters, hence with no need to resort to the Fokker--Planck asymptotics. It is worth mentioning that an analogously precise characterisation of the asymptotic trend of the Boltzmann-type equation is possible also in other relatively simple cases, in which the parameters $p$, $q$ are deterministic, the mean value of the kinetic distribution function is constantly zero from the initial time onwards, but the energy is \textit{not} conserved. In particular, when $p^2+q^2\neq 1$ one can introduce the following scaled kinetic distribution function:
$$ \tilde{f}(v,t):=\sqrt{M_2(t)}f(\sqrt{M_2(t)}v,t), $$
$M_2$ being the second moment of $f$, which is such that $\int_\R\tilde{f}(v,t)\,dv=1$, $\int_\R v\tilde{f}(v,t)\,dv=0$, and
$$ \int_\R v^2\tilde{f}(v,t)\,dv=\frac{1}{M_2(t)}\int_\R w^2f(w,t)\,dw=1, $$
as it can be checked via the change of variable $w:=\sqrt{M_2(t)}v$. Thus, $\tilde{f}$ restores a constant-in-time energy. Using~\eqref{eq:Boltztype.weak}, one discovers that $\tilde{f}$ satisfies a Boltzmann-type equation similar to that satisfied by $f$ but with an additional drift term proportional to $p^2+q^2-1$, which, in practice, represents the contribution needed to conserve in time the energy of $\tilde{f}$. Such an equation can be solved explicitly via the Fourier transform, whereby one sees that $\tilde{f}$ approaches, asymptotically in time, a universal steady profile. The technical details can be found in~\cite{pareschi2006JSP,pareschi2013BOOK}. Interestingly, this approach leads to justify the formation of fat tails in the Maxwellian as a consequence of the non-conservation of the energy, as opposed to the slim tail of the Gaussian distribution~\eqref{eq:Gaussian} found instead in the regime of conserved energy.
\end{remark}

\subsection{Uniqueness and continuous dependence of the solution in the quasi-invariant limit}
The technique employed in the proofs of Theorems~\ref{theo:q.i.-advection},~\ref{theo:q.i.-cons_en}, which is based on approximating the Fokker--Planck operator $J$ with the $\epsilon$-scaled collisional operator $Q_\epsilon$, can be profitably used also to show that the solutions to the Fokker--Planck equations~\eqref{eq:FP.strong-advection},~\eqref{eq:FP.strong-cons_en} obtained in the quasi-invariant limit $\epsilon\to 0^+$ are unique and depend continuously on their respective initial data. This idea is due originally to Torregrossa and Toscani, cf.~\cite{torregrossa2018KRM}.

\begin{theorem}
Each equation~\eqref{eq:FP.strong-advection},~\eqref{eq:FP.strong-cons_en} admits at most one solution $g\in C^0([0,\,T];\,\cP_s(\R))$, where $T>0$ is arbitrary and $s=2$ in the case of~\eqref{eq:FP.strong-advection} and $s=3$ in the case of ~\eqref{eq:FP.strong-cons_en}, which depends continuously on the initial datum. That is:
\begin{enumerate}[label=(\roman*)]
\item for~\eqref{eq:FP.strong-advection}, if $g_{1,0},\,g_{2,0}\in\cP_2(\R)$ are initial data with
$$ \int_\R vg_{1,0}(v)\,dv=\int_\R vg_{2,0}(v)\,dv=M_{1,0} $$
and $g_1,\,g_2\in C^0([0,\,T];\,\cP_2(\R))$ are respective solutions to~\eqref{eq:FP.strong-advection} then
$$ \sup_{t\in [0,\,T]}d_2(g_1(t),g_2(t))\leq d_2(g_{1,0},g_{2,0}); $$
\item for~\eqref{eq:FP.strong-cons_en}, if $g_{1,0},\,g_{2,0}\in\cP_3(\R)$ are initial data with
$$ \int_\R vg_{1,0}(v)\,dv=\int_\R vg_{2,0}(v)\,dv=0, \qquad
    \int_\R v^2g_{1,0}(v)\,dv=\int_\R v^2g_{2,0}(v)\,dv=M_{2,0} $$
and $g_1,\,g_2\in C^0([0,\,T];\,\cP_3(\R))$ are respective solutions to~\eqref{eq:FP.strong-cons_en} then
$$ \sup_{t\in [0,\,T]}d_3(g_1(t),g_2(t))\leq d_3(g_{1,0},g_{2,0}). $$
\end{enumerate}
\end{theorem}
\begin{proof}
We treat both cases simultaneously.

Each $g_k\in C^0([0,\,T];\,\cP_s(\R))$, where $k=1,\,2$ and $s=2,\,3$ as needed, solves a Fokker--Planck equation of the form
$\partial_tg_k=J(g_k)$, the operator $J$ being defined by the right-hand side of either~\eqref{eq:FP.strong-advection} or~\eqref{eq:FP.strong-cons_en}. Hence:
\begin{align*}
    \partial_t\bigl(\hat{g}_2-\hat{g}_1\bigr) &= \widehat{J}(\hat{g}_2)-\widehat{J}(\hat{g}_1) \\
    &= \widehat{J}(\hat{g}_2)-\widehat{Q}_\epsilon(\hat{g}_2,\hat{g}_2) \\
    &\phantom{=} +\widehat{Q}_\epsilon(\hat{g}_2,\hat{g}_2)-\widehat{Q}_\epsilon(\hat{g}_1,\hat{g}_1) \\
    &\phantom{=} +\widehat{Q}_\epsilon(\hat{g}_1,\hat{g}_1)-\widehat{J}(\hat{g}_1),
\end{align*}
where $\widehat{Q}_\epsilon(\hat{g}_k,\hat{g}_k)(\xi,t)=\frac{1}{\epsilon}\left(\ave{\hat{g}_k(p_\epsilon\xi,t)\hat{g}_k(q_\epsilon\xi,t)}-\hat{g}_k(\xi,t)\right)$. Therefore, dividing both sides by $\abs{\xi}^s$ and rearranging the terms we find
\begin{align*}
    \partial_t\frac{\hat{g}_2(\xi,t)-\hat{g}_1(\xi,t)}{\abs{\xi}^s}+\frac{1}{\epsilon}\cdot\frac{\hat{g}_2(\xi,t)-\hat{g}_1(\xi,t)}{\abs{\xi}^s}
        &= \frac{1}{\epsilon}\ave*{\frac{\hat{g}_2(p_\epsilon\xi,t)\hat{g}_2(q_\epsilon\xi,t)-\hat{g}_1(p_\epsilon\xi,t)\hat{g}_1(q_\epsilon\xi,t)}{\abs{\xi}^s}} \\
        &\phantom{=} +\frac{\widehat{J}(\hat{g}_2)(\xi,t)-\widehat{Q}_\epsilon(\hat{g}_2,\hat{g}_2)(\xi,t)}{\abs{\xi}^s} \\
        &\phantom{=} +\frac{\widehat{Q}_\epsilon(\hat{g}_1,\hat{g}_1)(\xi,t)-\widehat{J}(\hat{g}_1)(\xi,t)}{\abs{\xi}^s},
\end{align*}
i.e.
\begin{align*}
    \partial_t\left(e^{t/\epsilon}\frac{\abs{\hat{g}_2(\xi,t)-\hat{g}_1(\xi,t)}}{\abs{\xi}^s}\right) &\leq
        \frac{e^{t/\epsilon}}{\epsilon}\ave{p_\epsilon^s+q_\epsilon^s}d_s(g_1(t),g_2(t)) \\
    &\phantom{\leq} +e^{t/\epsilon}\sup_{\xi\in\R\setminus\{0\}}
        \frac{\abs{\widehat{Q}_\epsilon(\hat{g}_1,\hat{g}_1)(\xi,t)-\widehat{J}(\hat{g}_1)(\xi,t)}}{\abs{\xi}^s} \\
    &\phantom{\leq} +e^{t/\epsilon}\sup_{\xi\in\R\setminus\{0\}}
        \frac{\abs{\widehat{Q}_\epsilon(\hat{g}_2,\hat{g}_2)(\xi,t)-\widehat{J}(\hat{g}_2)(\xi,t)}}{\abs{\xi}^s}.
\end{align*}

By inspecting the proofs of Theorems~\ref{theo:q.i.-advection},~\ref{theo:q.i.-cons_en} we see that, as soon as $\epsilon>0$ is small enough, in both cases there exist:
\begin{enumerate*}[label=(\roman*)]
\item a non-negative function $\psi=\psi(\epsilon)$ with $\psi(\epsilon)\to 0$ for $\epsilon\to 0^+$;
\item a non-negative function $K=K(t)$ bounded above by a constant $\bar{K}_T>0$ for $t\in [0,\,T]$,
\end{enumerate*}
such that 
$$ \sup_{\xi\in\R\setminus\{0\}}\frac{\abs{\widehat{Q}_\epsilon(\hat{g},\hat{g})(\xi,t)-\widehat{J}(\hat{g})(\xi,t)}}{\abs{\xi}^s}
    \leq \psi(\epsilon)K(t), $$
where $g$ is any solution to either Fokker--Planck equation~\eqref{eq:FP.strong-advection},~\eqref{eq:FP.strong-cons_en}. Consequently, we continue the previous estimate as
$$ \partial_t\left(e^{t/\epsilon}\frac{\abs{\hat{g}_2(\xi,t)-\hat{g}_1(\xi,t)}}{\abs{\xi}^s}\right)\leq
        \frac{e^{t/\epsilon}}{\epsilon}\ave{p_\epsilon^s+q_\epsilon^s}d_s(g_1(t),g_2(t))+2e^{t/\epsilon}\psi(\epsilon)K(t), $$
and further, integrating both sides in the interval $[0,\,t]$, $0<t\leq T$,
\begin{align*}
    e^{t/\epsilon}d_s(g_1(t),g_2(t)) &\leq d_s(g_{1,0},g_{2,0})
        +\frac{\ave{p_\epsilon^s+q_\epsilon^s}}{\epsilon}\int_0^te^{\tau/\epsilon}d_s(g_1(\tau),g_2(\tau))\,d\tau \\
    &\phantom{\leq} +2\psi(\epsilon)\int_0^te^{\tau/\epsilon}K(\tau)\,d\tau. 
\end{align*}
We notice that, unlike the proofs of Theorems~\ref{theo:q.i.-advection},~\ref{theo:q.i.-cons_en}, here we need to take into account explicitly also the distance between the initial data, which in general is non-zero.

At this point, Lemma~\ref{lemma:Gronwall} produces
\begin{align*}
    d_s(g_1(t),g_2(t)) &\leq e^{\frac{\ave{p_\epsilon^s+q_\epsilon^s}-1}{\epsilon}t}d_s(g_{1,0},g_{2,0})
        +2\psi(\epsilon)\int_0^te^{\frac{\ave{p_\epsilon^s+q_\epsilon^s}-1}{\epsilon}(t-\tau)}K(\tau)\,d\tau,
\intertext{namely, considering that $\ave{p_\epsilon^s+q_\epsilon^s}<1$ when $\epsilon$ is sufficiently small,}
    &\leq d_s(g_{1,0},g_{2,0})+2\bar{K}_TT\psi(\epsilon).
\end{align*}
Since this holds for every small enough $\epsilon>0$, letting $\epsilon\to 0^+$ and taking the supremum of both sides for $0\leq t\leq T$ we get the thesis.
\end{proof}

\section{Basics of the Monte Carlo numerical approach}
\label{sect:MonteCarlo}
The Boltzmann-type equation~\eqref{eq:Boltztype.strong} can be solved numerically by means of stochastic particle-based algorithms, which closely follow the agent-based model introduced in Section~\ref{sect:agent-based_model} and its statistical properties. This way, modelling, analysis, and numerics of the Boltzmann-type equation turn out to be intimately correlated, a fact which is not that common for other equations and theories of Mathematical Physics.

\subsection{Nanbu--Babovsky-type algorithm}

The stochastic algorithms typically employed for the numerical solution of~\eqref{eq:Boltztype.strong.gen} belong to the broad class of \textit{Monte Carlo methods}, i.e. probabilistic computational methods based on random samplings. The basic idea underlying them is to discretise only the time variable, letting then evolve in discrete time a finite set of particles that are identified to their microscopic state $v$. The latter changes in consequence of random interactions among the particles driven by some prescribed microscopic rules. One of the most popular of such Monte Carlo methods is implemented by the \textit{Nanbu--Babovsky algorithm}~\cite{babovsky1986M2AS,nanbu1980JPSJ}, which we describe in the following.

Given a discretisation of the time interval $[0,\,T]$, where $T>0$ is the final time, in $\lfloor T/\Delta{t}\rfloor$ intervals of size $\Delta{t}>0$, the approximation of the kinetic distribution function $f(v,n\Delta{t})$, with $n=0,\,1,\,\dots,\,\lfloor T/\Delta{t}\rfloor$, is denoted by $f^n(v)$. The latter is built numerically as a histogram based on a finite set of, say, $N_p\in\N$ particles. To introduce the appropriate discretisation, it is customary to rewrite~\eqref{eq:Boltztype.strong.gen} as
\begin{equation}
    \partial_t f(v,t)=\mu Q^+(f,f)(v,t)- \mu f(v,t)\label{eq:Boltztype.strongQ+}
\end{equation}
where $Q^+$ is the so-called \textit{gain term} of the collision operator $Q$:
$$ Q^+(f,f):=\int_\R\ave*{\frac{1}{\abs{J}}f(\pr{v},t)f(\pr{v}_\ast,t)}\,dv_\ast, $$
that is, in particular,
$$ Q^+(f,f):=\int_\R\ave*{\frac{1}{\abs{p^2-q^2}}f(\pr{v},t)f(\pr{v}_\ast,t)}\,dv_\ast $$
for the binary interaction rules~\eqref{eq:v'}. Since $\int_\R Q(f,f)(v,t)\,dv=0$ and $\int_\R f(v,t)\,dv=1$ for all $t>0$, it results
$$ \int_\R Q^+(f,f)(v,t)\,dv=1 \qquad \forall\,t>0, $$
therefore $Q^+(f,f)(\cdot,t)$ is a probability distribution for all $t>0$.

The forward Euler scheme applied to~\eqref{eq:Boltztype.strongQ+} gives
$$ f^{n+1}(v)=f^{n}(v)+\mu\Delta{t}\left[Q^+(f^n,f^n)(v)-f^n(v)\right], $$
which, rearranging the terms, becomes
\begin{equation}
    f^{n+1}(v)=(1-\mu\Delta{t})f^{n}(v)+\mu\Delta{t}Q^+(f^n,f^n)(v).
    \label{eq:discr}
\end{equation}
The probabilistic interpretation of~\eqref{eq:discr} is the following: in the time step $n\to n+1$ a particle with state $v$ either does not interact with probability $1-\mu\Delta{t}$, thereby maintaining its state $v$ also at time $n+1$, or interacts with probability $\mu\Delta{t}$, thereby changing its state at time $n+1$ according to the microscopic rule embodied in $Q^+$. Consequently, at each iteration of the algorithm the expected number of interacting particles is $N_p \mu \Delta t$. To implement such stochastic particle dynamics, in the Nanbu--Babovski algorithm $\lfloor N_p\mu\Delta{t}/2\rfloor$ independent pairs of particles are randomly sampled at each time iteration. Those particles update then their states at time $n+1$ based on the prescribed interaction rule. It is worth stressing that the pairs of interacting particles need to be sampled independently of one another and of the respective microscopic states. This is essential in order to guarantee that, at the time of the interaction, the pre-interaction states $v$, $v_\ast$ are indeed statistically independent as assumed by the factorisation $f(v,t)f(v_\ast,t)$, cf.~\eqref{eq:Q.strong},~\eqref{eq:Q.weak}. In the case of the original Boltzmann equation~\eqref{eq:Boltz}-\eqref{eq:Boltz.coll_rules},\commentout{for Maxwell molecules} this particular aspect guarantees the energy conservation in each collision, which was not granted in previous versions of the algorithm~\cite{pareschi1999AN}.

As a matter of fact, the mentioned probabilistic interpretation of~\eqref{eq:discr}, that allows one to define the microscopic dynamics underlying the Nanbu--Babovski particle algorithm, can be formulated without the need for the time discretisation introduced above. Indeed, such an interpretation corresponds to the intrinsically random microscopic dynamics underlying~\eqref{eq:Boltztype.strongQ+}, that is~\eqref{eq:Vt+Dt}. Consequently, we present a form of the Nanbu--Babovski algorithm based on integrating in time directly the stochastic process~\eqref{eq:Vt+Dt}. Notice that the time step $\Delta{t}$ defines the parameter of the random variable $\Theta$ in~\eqref{eq:Theta}, which discriminates between an interaction taking place ($\Theta=1$) or not ($\Theta=0$). At each time iteration, random pairs of agents are formed, which update their microscopic states according to~\eqref{eq:Vt+Dt}. For each pair of agents, a corresponding value of the random variable $\Theta$ is sampled to decide whether the agents in that pair actually interact or not. In sampling the pairs of agents and the random variable $\Theta$ it is fundamental to preserve the statistical properties on which the derivation of the Boltzmann-type equation~\eqref{eq:Boltztype.weak} relies. Otherwise, the simulated particle system would not be consistent with the statistical description brought by~\eqref{eq:Boltztype.weak} and, as such, it would return results which do not reproduce the solutions of that equation. In more detail, the pairs of (potentially) interacting particles need to be selected independently (like in the classical case and for the same reason): for each of the randomly chosen $N_p/2$ pairs $(V_t=v,\,V^\ast_t=v_\ast)$ of microscopic states of the finite set of $N_p$ particles, $\Theta$ has to be sampled independently of $V_t$, $V^\ast_t$. This way, the structure~\eqref{eq:expectation_w.r.t._Theta}, which is at the basis of~\eqref{eq:Boltztype.weak} in the continuous-time limit, is statistically reproduced.

\begin{algorithm}[!t]
\caption{Nanbu--Babovsky-type algorithm}
\label{alg:nanbu}
\KwData{$N_p\in\mathbb{N}$ (number of agents); $\Delta{t}\leq 1$; $T>0$ (final time); $f=f_0(v)$}
Sample a set of $N_p$ states $\cV^0:=\{v_1^0,\,v_2^0,\,\dots,\,v_{N_p}^0\}\subseteq\R$ from $f_0$\; \label{alg:sample_f0}
\For{$n=0,\,1,\,2,\,\dots,\,\lfloor\frac{T}{\Delta{t}}\rfloor$}{
	\Repeat{no pairs of states are left in $\cV^n$}{
		Pick randomly and independently $v_i^n,\,v_j^n\in\cV^n$ with $i\neq j$\;
		Sample $\Theta\sim\operatorname{Bernoulli}{(\Delta{t})}$\; \label{alg:sample_Theta}
		\eIf{$\Theta=1$}{
			Sample $p$, $q$ from their respective distributions\; \label{alg:sample_pq}
			Compute $v_i^{n+1}=pv_i^n+qv_j^n$\; \label{alg:interaction.1}
			Compute $v_j^{n+1}=pv_j^n+qv_i^n$\; \label{alg:interaction.2}
		}{
			Set $v_i^{n+1}=v_i^n$\;
			Set $v_j^{n+1}=v_j^n$\;
		}
		Discard $v_i^n,\,v_j^n$ from $\cV^n$\; \label{alg:discard}
	}
	Form the new set of states $\cV^{n+1}:=\{v_1^{n+1},\,v_2^{n+1},\,\dots,\,v_{N_p}^{n+1}\}$\;
	Build a histogram out of the dataset $\cV^{n+1}$ to approximate $f(v,(n+1)\Delta{t})$\;
}
\end{algorithm}

Algorithm~\ref{alg:nanbu} reports the conceptual implementation of a Nanbu--Babovsky-type scheme to approximate numerically the solutions of~\eqref{eq:Boltztype.strong}, namely~\eqref{eq:Boltztype.strong.gen} with interaction rules~\eqref{eq:v'} and $\mu=1$. The sampling of random variables from prescribed distributions invoked in lines~\ref{alg:sample_f0},~\ref{alg:sample_Theta},~\ref{alg:sample_pq} is a major topic in the realm of Monte Carlo methods, which requires in general \textit{ad-hoc} numerical procedures~\cite{pareschi2001ESAIM}. However, most programming languages possess built-in routines that perform samplings from some popular probability distributions, such as e.g., the binomial, normal, beta, uniform, and gamma distributions. Samplings from further distributions can be realised by exploiting simple relationships among random variables. For example:
\begin{enumerate}[label=\roman*)]
\item since a binomial random variable with parameters $n\in\N$, $r\in [0,\,1]$ is the sum of $n$ independent and identically distributed Bernoulli random variables with parameter $r$, one obtains a Bernoulli sample by invoking the routine for the binomial sampling with the desired parameter $r\in [0,\,1]$ and $n=1$;
\item since, by definition, a random variable $X$ has lognormal distribution if $Y=\log{X}$ has normal distribution, one obtains a lognormal sample by first constructing a sample $y_1,\,y_2,\,y_3,\,\dots$ of $Y$ via the routine for the normal sampling and then computing the corresponding values $e^{y_1},\,e^{y_2},\,e^{y_3},\,\dots$ of $X$;
\item if a random variable $X$ has gamma distribution then the random variable $Y=1/X$ has inverse gamma distribution with the same parameters of the distribution of $X$. Therefore, it is possible to obtain a sample distributed according to an inverse gamma law by first sampling some gamma-distributed values $x_1,\,x_2,\,x_3,\,\dots$ of $X$ via the routine for the gamma sampling and then computing the corresponding values $\frac{1}{x_1},\,\frac{1}{x_2},\,\frac{1}{x_3},\,\dots$ of $Y$.
\end{enumerate}

At each time iteration, the formation of random and independent pairs of (potentially) interacting agents can be realised, in practice, with the following simple method: first, one performs a random permutation of the elements of the set $\cV^n=\{v_1^n,\,v_2^n,\,\dots,\,v_{N_p}^n\}$, so that the $v_i^n$'s are mixed randomly and independently. Next, assuming that the permuted elements are relabelled ordinately from $1$ to $N_p$ and taking advantage of the evenness of $N_p$, cf. the section ``Data'' of Algorithm~\ref{alg:nanbu}, one forms the pairs $(v_1^n,\,v_{N_p/2+1}^n)$, $(v_2^n,\,v_{N_p/2+2}^n)$, \dots, $(v_{N_p/2}^n,\,v_{N_p}^n)$.

In order to use Algorithm~\ref{alg:nanbu} for the approximate solution of the $\epsilon$-scaled Boltzmann-type equation in the quasi-invariant regime it is necessary to replace $p,\,q$ with $p_\epsilon,\,q_\epsilon$ in lines~\ref{alg:interaction.1},~\ref{alg:interaction.2}. Moreover, the parameter of the Bernoulli random variable $\Theta$ in line~\ref{alg:sample_Theta} has to be changed in $\Delta{t}/\epsilon$ to take into account the time scaling $t\rightsquigarrow t/\epsilon$ typical of the quasi-invariant regime. Consequently, the constraint on $\Delta{t}$ (cf. the section ``Data'' in Algorithm~\ref{alg:nanbu}) becomes $\Delta{t}\leq\epsilon$. For completeness, we report in Algorithm~\ref{alg:nanbu.eps} the detail also of this case.

\begin{algorithm}[!t]
\caption{Nanbu--Babovsky-type algorithm in the quasi-invariant regime}
\label{alg:nanbu.eps}
\KwData{$N_p\in\mathbb{N}$ (number of agents); $\epsilon>0$; $\Delta{t}\leq\epsilon$; $T>0$ (final time); $f=f_0(v)$}
Sample a set of $N_p$ states $\cV^0:=\{v_1^0,\,v_2^0,\,\dots,\,v_{N_p}^0\}\subseteq\R$ from $f_0$\;
\For{$n=0,\,1,\,2,\,\dots,\,\lfloor\frac{T}{\Delta{t}}\rfloor$}{
	\Repeat{no pairs of states are left in $\cV^n$}{
		Pick randomly and independently $v_i^n,\,v_j^n\in\cV^n$ with $i\neq j$\;
		Sample $\Theta\sim\operatorname{Bernoulli}{(\Delta{t}/\epsilon)}$\;
		\eIf{$\Theta=1$}{
			Sample $p_\epsilon$ from its distribution\;
			Compute $v_i^{n+1}=p_\epsilon v_i^n+q_\epsilon v_j^n$\;
			Compute $v_j^{n+1}=p_\epsilon v_j^n+q_\epsilon v_i^n$\;
		}{
			Set $v_i^{n+1}=v_i^n$\;
			Set $v_j^{n+1}=v_j^n$\;
		}
		Discard $v_i^n,\,v_j^n$ from $\cV^n$\;
	}
	Form the new set of states $\cV^{n+1}:=\{v_1^{n+1},\,v_2^{n+1},\,\dots,\,v_{N_p}^{n+1}\}$\;
	Build a histogram out of the dataset $\cV^{n+1}$ to approximate $f_\epsilon(v,(n+1)\Delta{t})$\; \label{alg:histogram}
}
\end{algorithm}

As usual in time-discrete numerical schemes, the trade-off between speed and accuracy in time of Algorithm~\ref{alg:nanbu} is dictated by the time step $\Delta{t}$: the larger $\Delta{t}$ the quicker but less accurate the algorithm. In particular, we notice that by fixing $\Delta{t}$ to its maximum possible value one obtains $\Theta\sim\operatorname{Bernoulli}{(1)}$, i.e. $\Theta=1$ deterministically, meaning that all the $N_p/2$ pairs of agents interact. Heuristically, such a number of effective interactions entails a large variation of the statistical distribution of the microscopic states in a single iteration of the algorithm, which is ultimately the source of the possibly poor numerical accuracy. In general, instead, only $N_p\Delta{t}/2$ (and $N_p\Delta{t}/(2\epsilon)$ in the quasi-invariant regime) interactions take place, on average, in a single iteration of the algorithm. Therefore, the lower $\Delta{t}$ the lower the number of effective interactions, hence the smaller the variation of the statistical distribution of the microscopic states in one iteration, which allows for a better numerical accuracy.

Another parameter which impacts considerably on the accuracy of Algorithm~\ref{alg:nanbu} is the number $N_p$ of sampled microscopic states. On one hand, this number has to be large enough in order for the statistics reconstructed from $\cV^n=\{v_1^n,\,v_2^n,\,\dots,\,v_{N_p}^n\}$ at each iteration, including the approximation of the kinetic distribution function $f(v,n\Delta{t})$, to be meaningful. On the other hand, the larger $N_p$ the slower the algorithm, because at each iteration $N_p/2$ interactions need to be tested (some of which are probably even ``useless'' if $\Delta{t}<1$ or $\Delta{t}<\epsilon$ in the quasi-invariant regime). Typically, in Monte Carlo schemes one expects the error of the method to scale with the number of samples as $1/\sqrt{N_p}$. In the case of Algorithm~\ref{alg:nanbu}, considering that each sample $v_i^n$ is one-dimensional, a good trade-off between speed and accuracy is usually obtained with $N_p=O(10^6)$. For multi-dimensional samples, like in the case of the original Boltzmann equation, one might be forced to reduce $N_p$ to reach an acceptable speed of the algorithm.

\begin{remark}
Besides the Nanbu--Babovsky scheme, another popular particle-based algorithm to approximate numerically the solutions to~\eqref{eq:Boltztype.strong} is \textit{Bird's scheme}~\cite{bird1970PF}, which differs from Nanbu--Babovsky essentially in that it does not discard already used agents from $\cV^n$, cf. line~\ref{alg:discard} of Algorithm~\ref{alg:nanbu}. Consequently, in Bird's scheme an agent can interact multiple times in a single iteration, whereas in the Nanbu--Babovsky scheme every agent interacts at most once in each iteration. We refer to~\cite{pareschi2001ESAIM} and references therein for more formal details about Bird's scheme and for a closer comparison with the Nanbu--Babovsky scheme.
\end{remark}

\subsection{Numerical tests}
\subsubsection{Advection-diffusion quasi-invariant regime}
\label{sect:num-advdiff}
\begin{figure}[!t]
\centering
\begin{tabular}{c|c|c}
$\boldsymbol{\epsilon=4\cdot 10^{-2}}$ & $\boldsymbol{\epsilon=10^{-2}}$ & $\boldsymbol{\epsilon=10^{-3}}$ \\
\hline
\phantom{I} & \phantom{I} \\
\includegraphics[scale=.32]{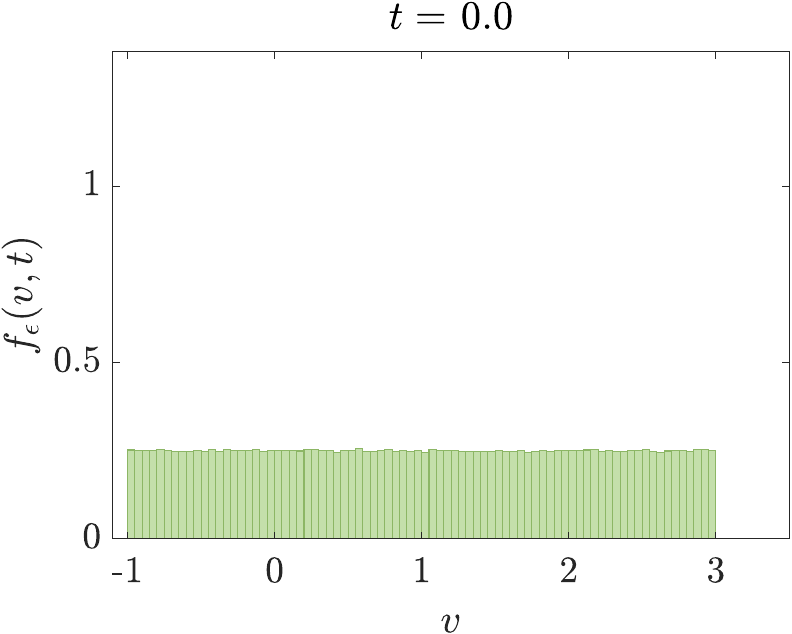} & \includegraphics[scale=.32]{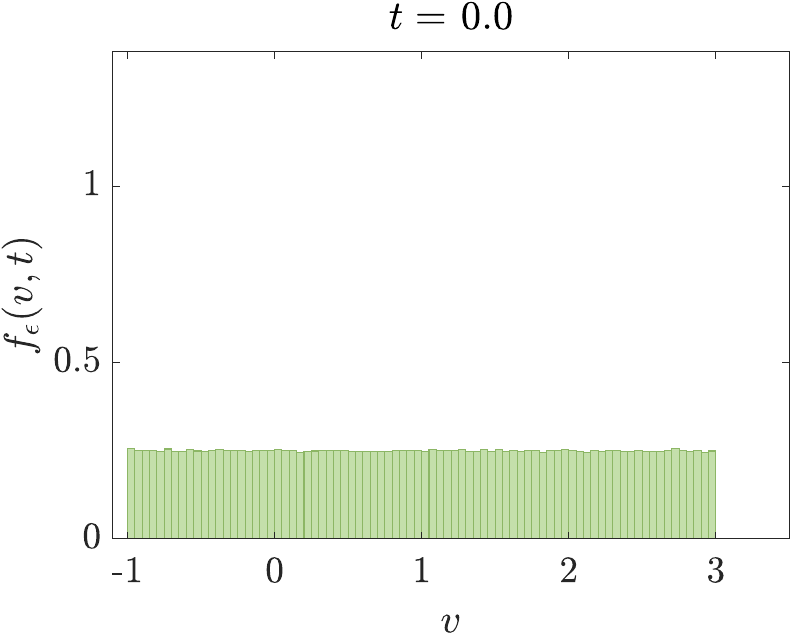} & \includegraphics[scale=.32]{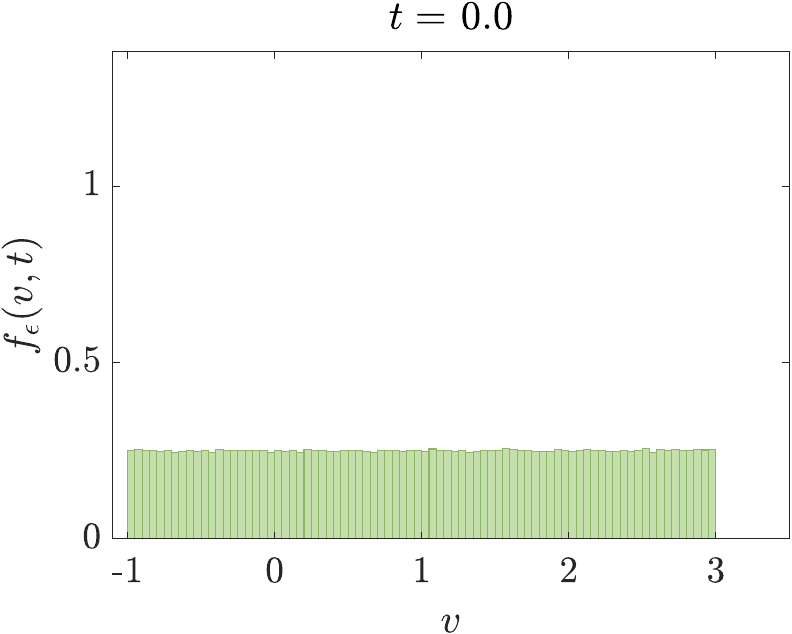} \\
\includegraphics[scale=.32]{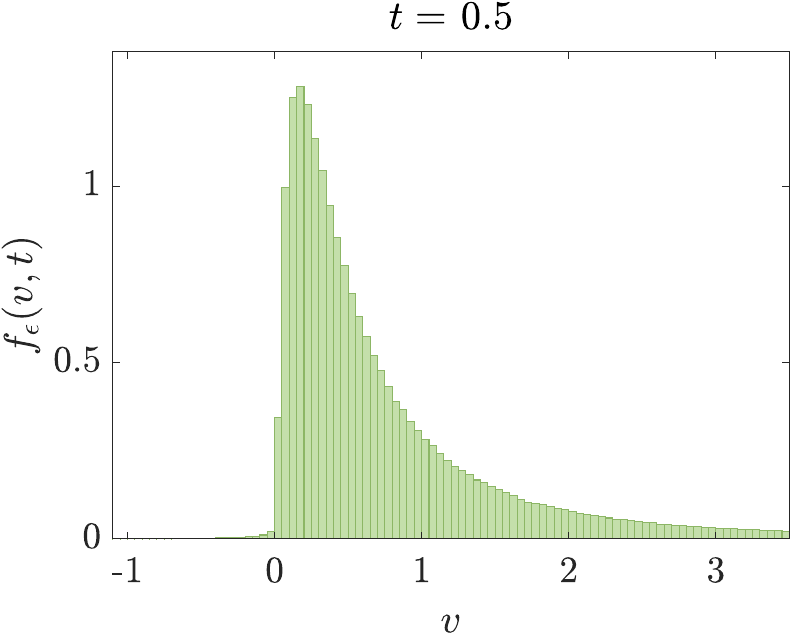} & \includegraphics[scale=.32]{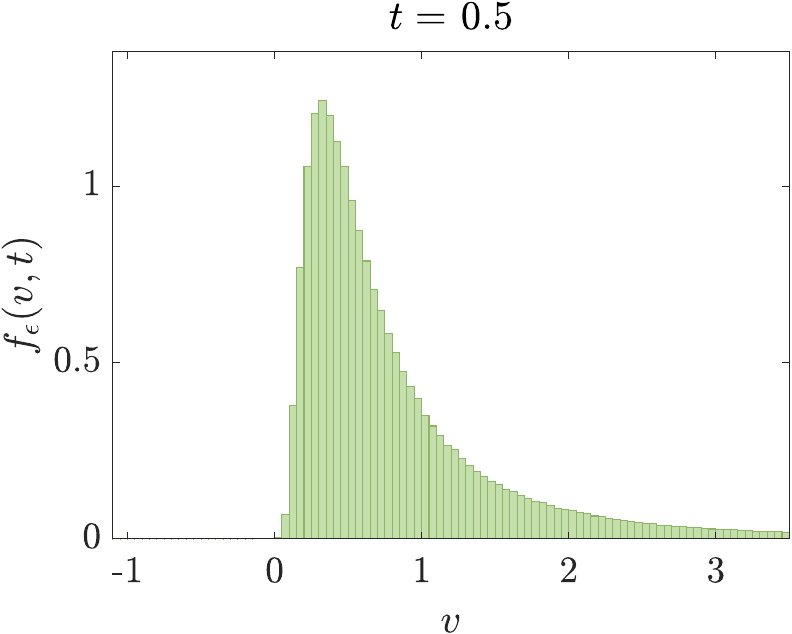} & \includegraphics[scale=.32]{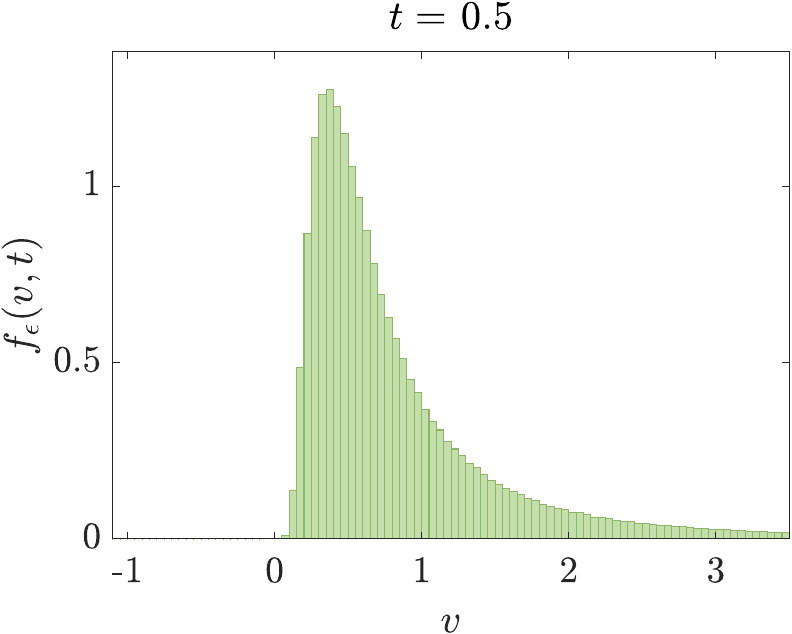} \\
\includegraphics[scale=.32]{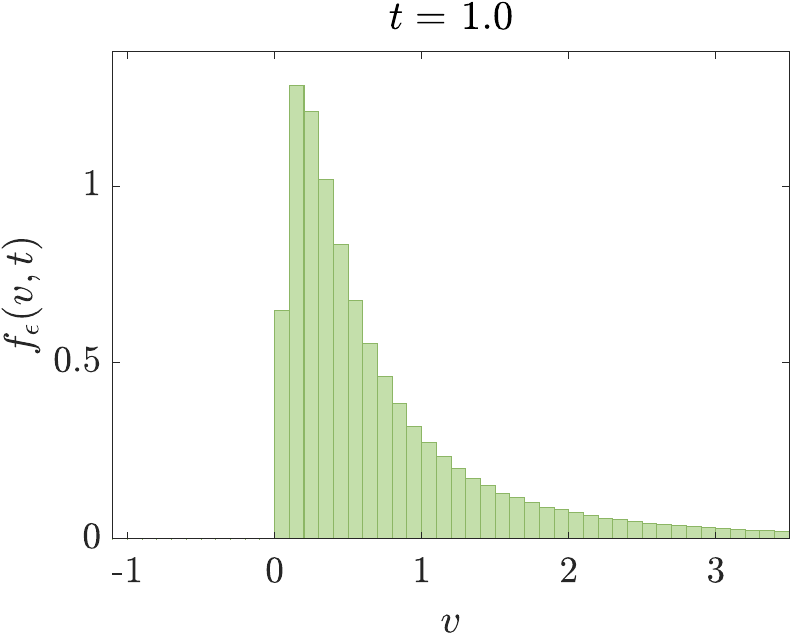} & \includegraphics[scale=.32]{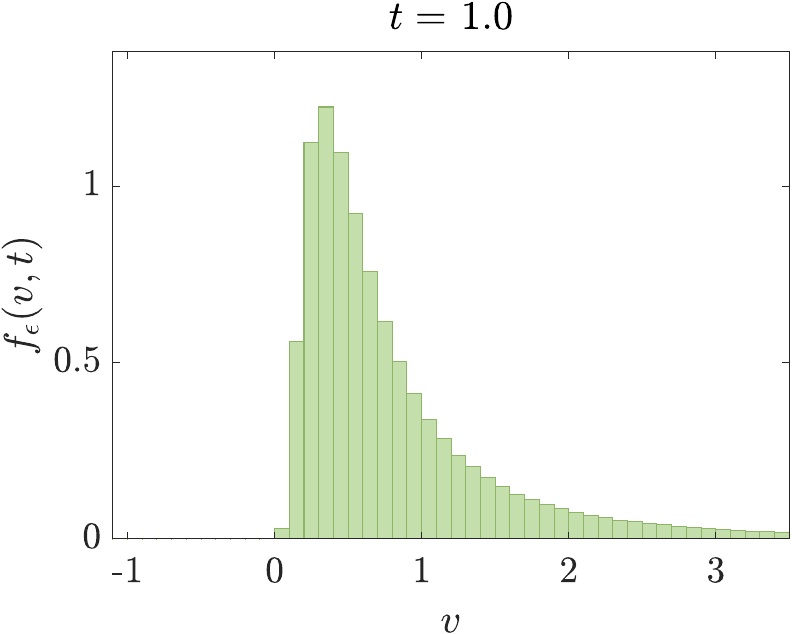} & \includegraphics[scale=.32]{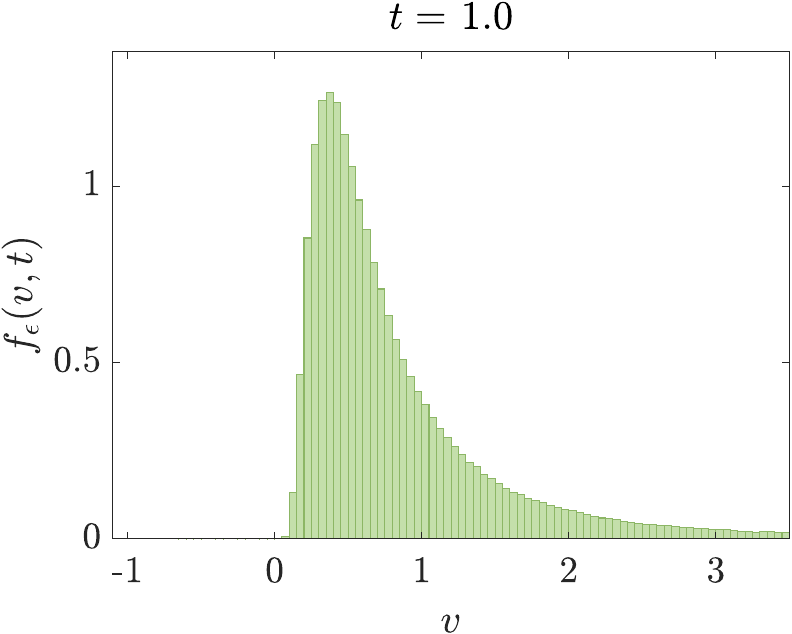} \\
\includegraphics[scale=.32]{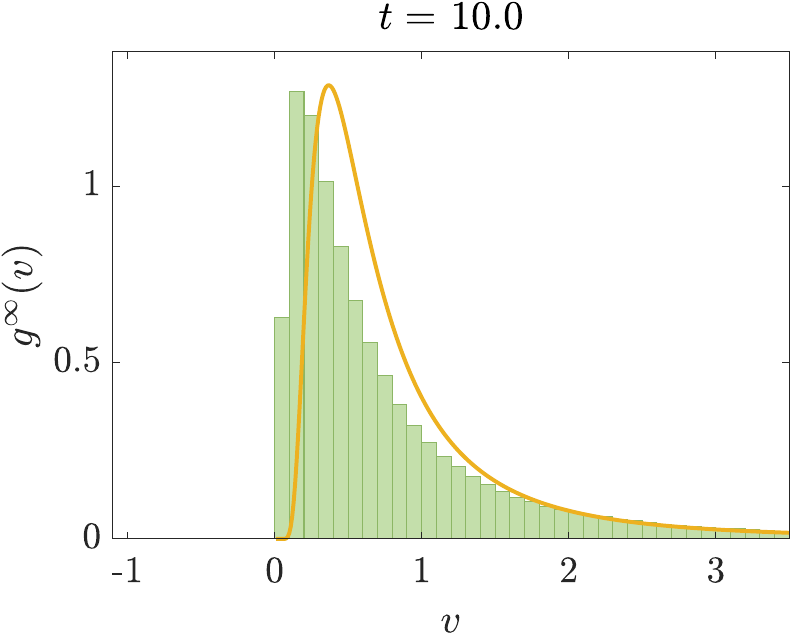} & \includegraphics[scale=.32]{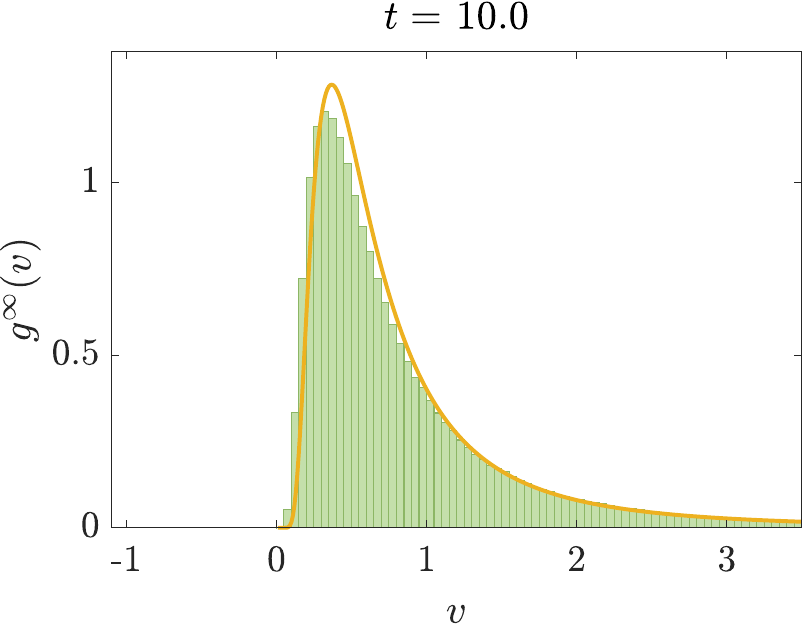} & \includegraphics[scale=.32]{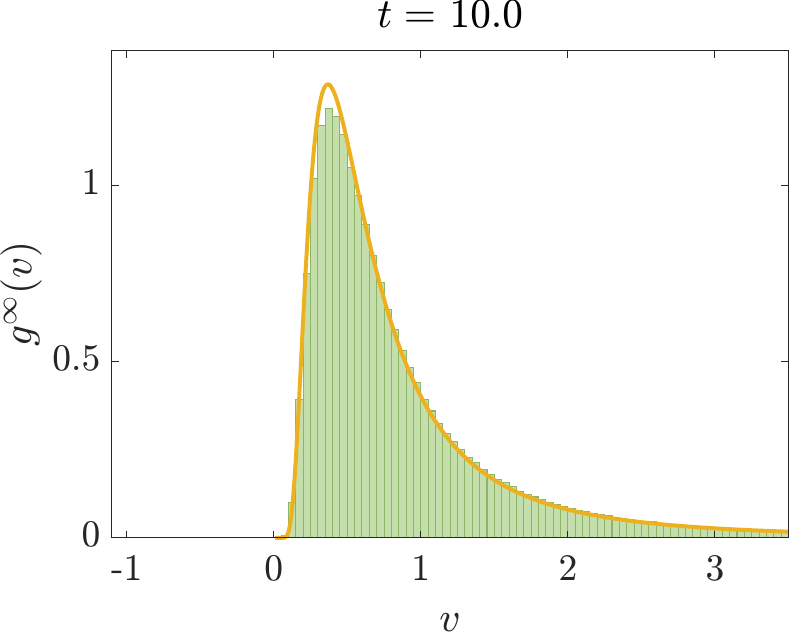} \\
\end{tabular}
\caption{Numerical solution at successive computational times of the Boltzmann-type equation~\eqref{eq:Boltztype.scaled} in the quasi-invariant regime with interaction coefficients~\eqref{eq:qeps},~\eqref{eq:peps.eta} and: left column $\epsilon=4\cdot 10^{-2}$, center column $\epsilon=10^{-2}$, right column $\epsilon=10^{-3}$.}
\label{fig:invgamma}
\end{figure}

In Figure~\ref{fig:invgamma} we show the numerical solution, obtained by means of Algorithm~\ref{alg:nanbu.eps}, of the $\epsilon$-scaled Boltzmann-type equation~\eqref{eq:Boltztype.scaled} with interaction coefficients $p_\epsilon$, $q_\epsilon$ of the form~\eqref{eq:qeps},~\eqref{eq:peps.eta}. At each computational time, the kinetic distribution function is built numerically as a normalised histogram of the computed microscopic states of the particles, cf. line~\ref{alg:histogram} of Algorithm~\ref{alg:nanbu.eps}.

In particular, we choose the parameters $\lambda=\frac{7}{2}$ and $\sigma=\sqrt{6}$, which comply with~\eqref{eq:sigma.lambda}. Moreover, for the random variable $\eta$ featured by $p_\epsilon$ we fix a uniform distribution in the interval $[-\sqrt{3},\,\sqrt{3}]$, so that conditions~\eqref{eq:eta},~\eqref{eq:eta.lowerbound} are satisfied.

Concerning the scaling parameter $\epsilon$, we consider the three cases $\epsilon=4\cdot 10^{-2},\,10^{-2},\,10^{-3}$, all complying with~\eqref{eq:eps.quasi-inv}. The goal is to check that for decreasing $\epsilon$ the equilibrium distribution reached for large times is more and more well approximated by the stationary solution~\eqref{eq:inv_gamma} to the Fokker--Planck equation~\eqref{eq:FP.strong} obtained in the quasi-invariant limit $\epsilon\to 0^+$.

As initial condition we take the uniform probability density $f_0(v)=\frac{1}{4}\chi(-1\leq v\leq 3)$ on the interval $[-1,\,3]$ with $M_{1,0}=1$. Therefore, $f_0$ has positive mean value but is not fully supported in $\R_+$.

We run Algorithm~\ref{alg:nanbu.eps} with $N_p=10^6$ agents up to the final computational time $T=10$. In the cases $\epsilon=10^{-2},\,10^{-3}$ we set $\Delta{t}=\epsilon$, as such a time step is reasonably small for one to expect qualitatively accurate numerical solutions despite the fact that, in each iteration, all the $N_p/2$ pairs of agents interacts. Conversely, in the case $\epsilon=4\cdot 10^{-2}$ we set $\Delta{t}=\frac{\epsilon}{2}$, meaning that in each iteration only half of the pairs of agents interact on average.

From Figure~\ref{fig:invgamma} it is apparent that the support of the distribution $f(t)$ tends to move in time towards $\R_+$, driven by $M_{1,0}>0$. In addition to this, the last row of Figure~\ref{fig:invgamma} shows clearly that the smaller $\epsilon$ the better the qualitative matching between the (numerical approximation of the) large time solution to the $\epsilon$-scaled Boltzmann-type equation~\eqref{eq:Boltztype.scaled} and the inverse gamma distribution~\eqref{eq:inv_gamma} (solid line in the figure).

\subsubsection{Advection-dominated quasi-invariant regime}
\begin{figure}[!t]
\centering
\includegraphics[scale=.32]{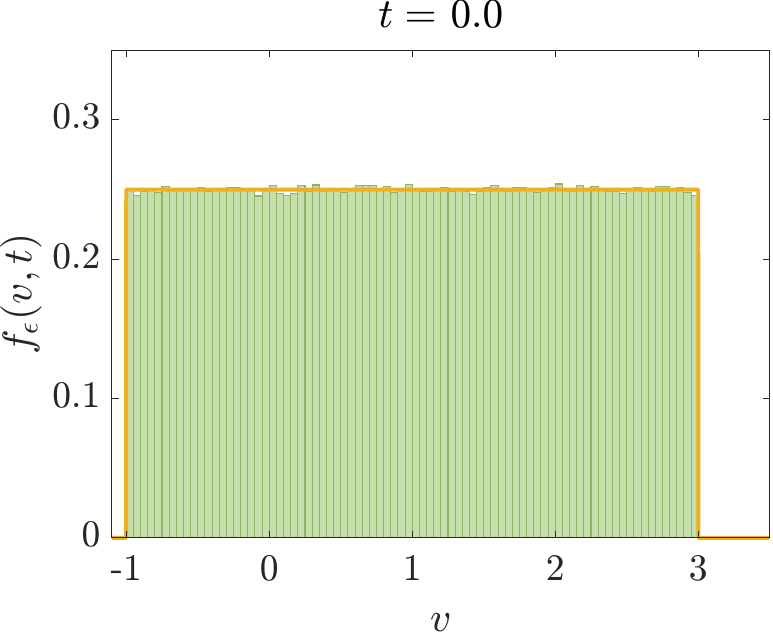} \quad
\includegraphics[scale=.32]{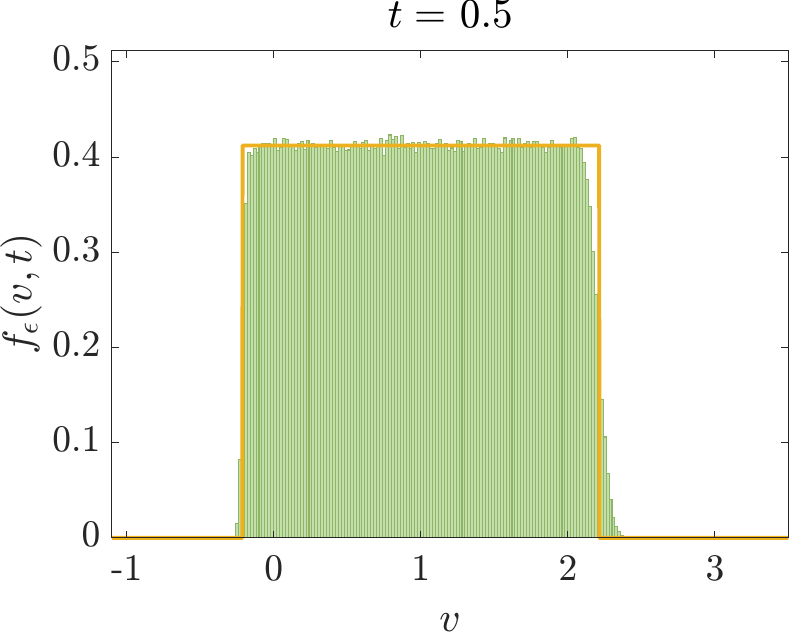} \quad
\includegraphics[scale=.32]{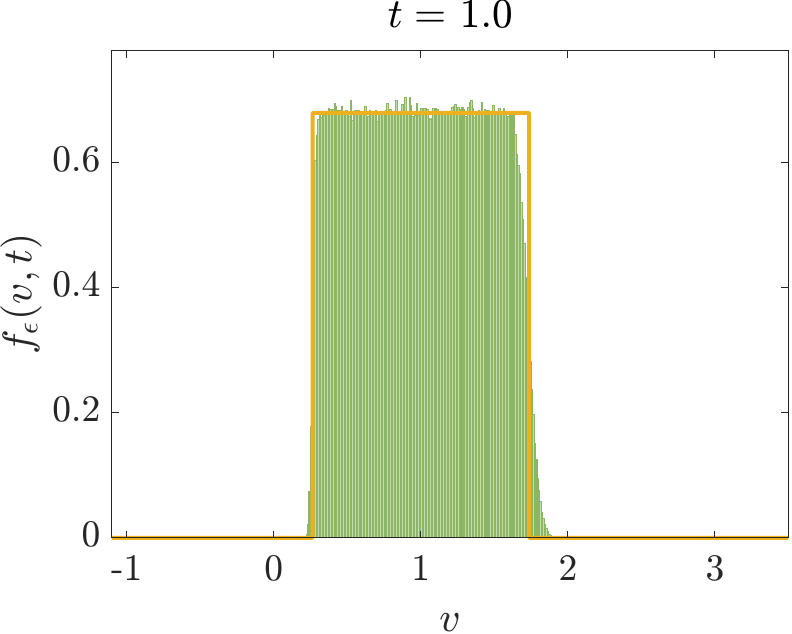} \\
\includegraphics[scale=.32]{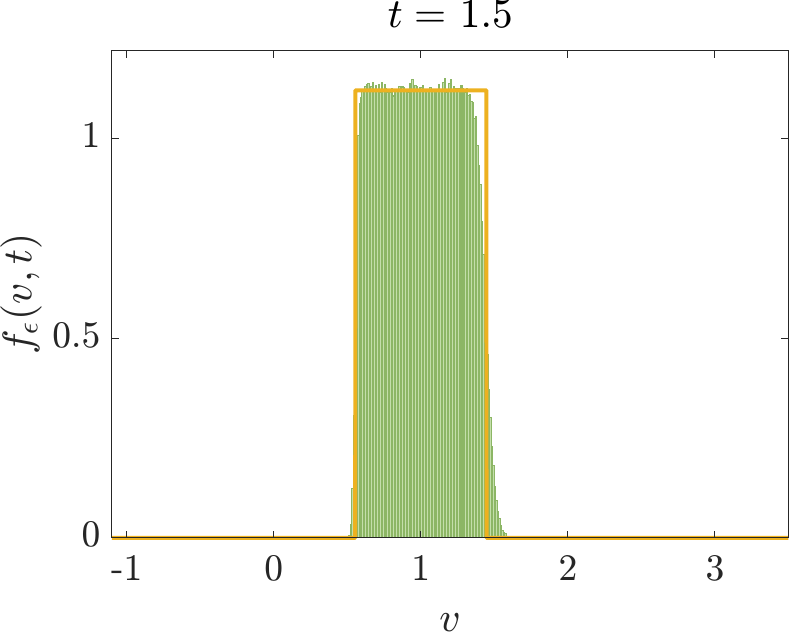} \quad
\includegraphics[scale=.32]{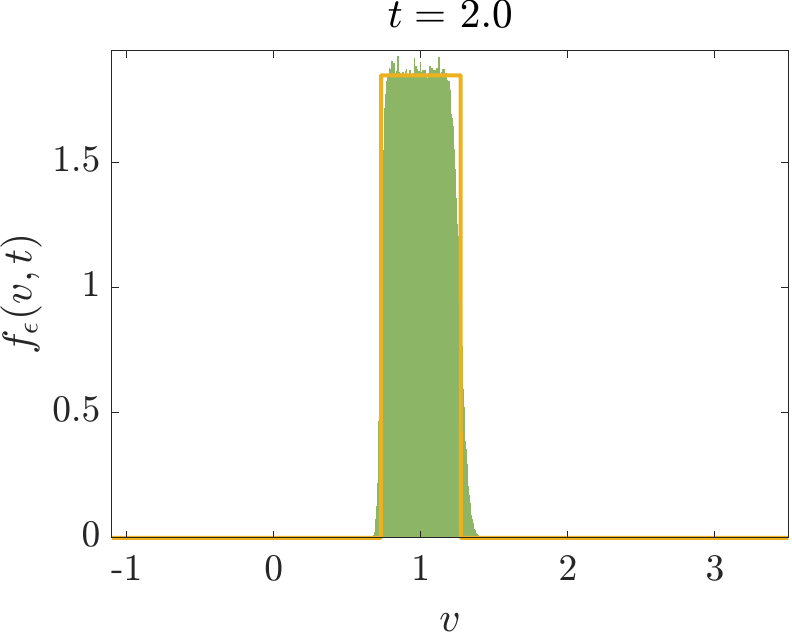} \quad
\includegraphics[scale=.32]{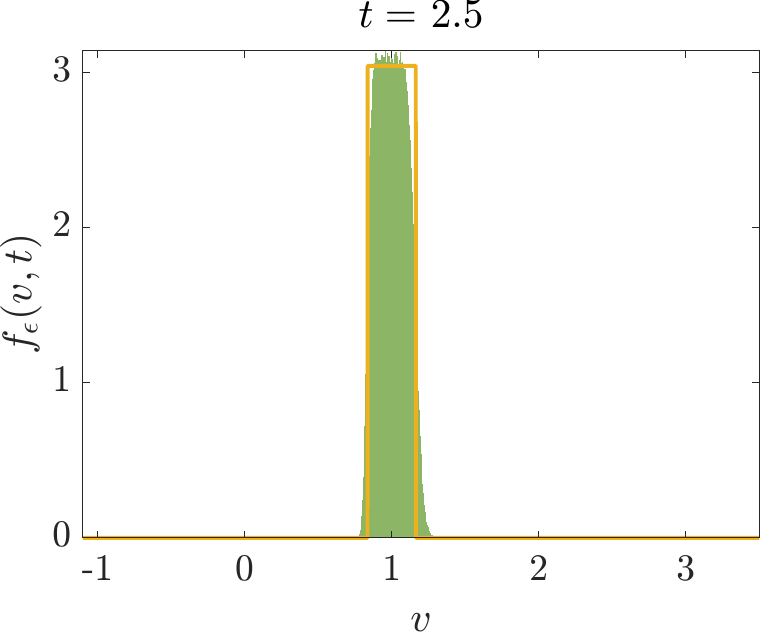} \\
\includegraphics[scale=.32]{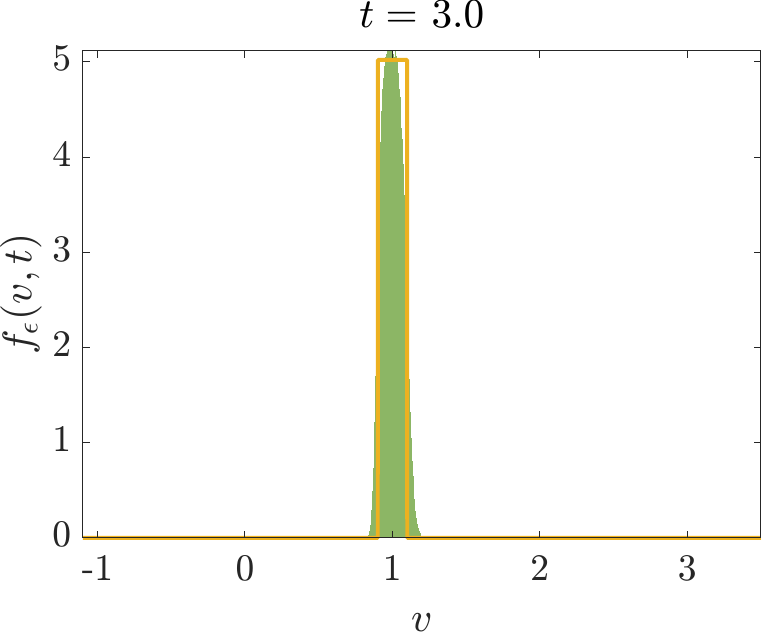} \quad
\includegraphics[scale=.32]{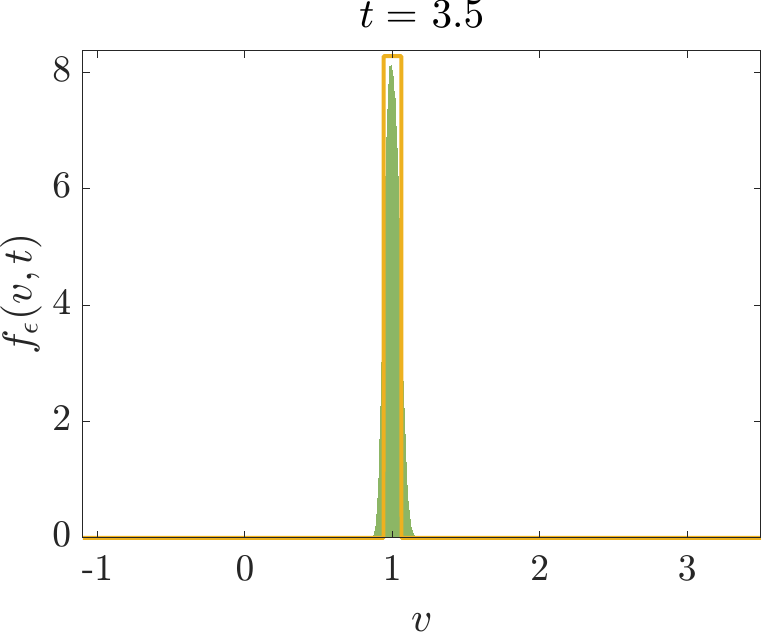} \quad
\includegraphics[scale=.32]{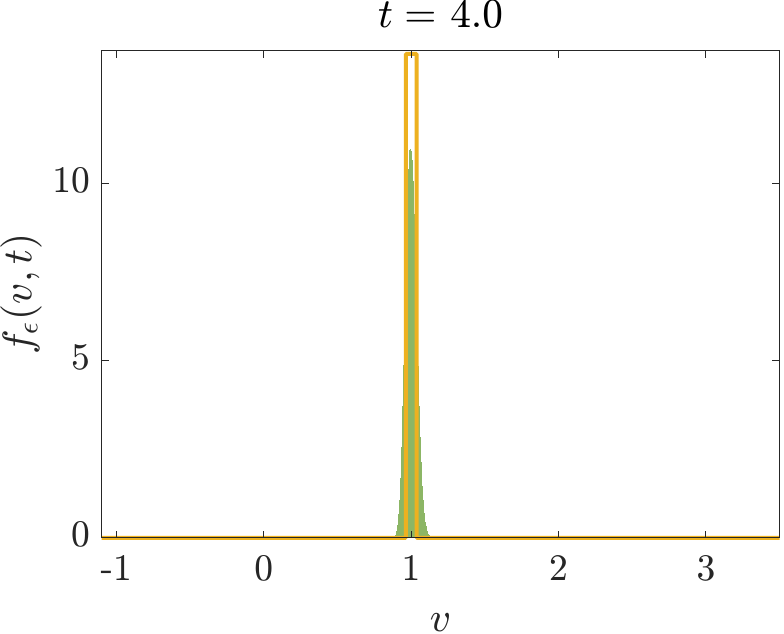}
\caption{Numerical solution at successive computational times of the Boltzmann-type equation~\eqref{eq:Boltztype.scaled-advection} in the advection-dominated quasi-invariant regime, cf. Section~\ref{sect:q.i._advection}, and comparison with the exact solution~\eqref{eq:g.FP-advection} of the limit Fokker--Planck equation~\eqref{eq:FP.strong-advection} (solid line).}
\label{fig:delta}
\end{figure}

In Figure~\ref{fig:delta} we show instead the numerical solution, still obtained by means of Algorithm~\ref{alg:nanbu.eps}, of the $\epsilon$-scaled Boltzmann-type equation~\eqref{eq:Boltztype.scaled-advection} with interaction coefficients satisfying~\eqref{eq:qeps},~\eqref{eq:peps-advection}. In particular, we fix the parameters $\lambda=1$, $\sigma=1.3$, which comply with~\eqref{eq:sigma.lambda.eps_quasi-inv-advection}, and $\delta=1$. Moreover, we take $\epsilon=10^{-3}$, which complies with~\eqref{eq:sigma.lambda.eps_quasi-inv-advection} as well. These values of the parameters ensure also that, setting $p_\epsilon=1-\epsilon\lambda+\epsilon\sigma\eta$ with $\eta\sim\pazocal{U}(-\sqrt{3},\,\sqrt{3})$, it results $p_\epsilon\geq 0$ as needed.

The initial condition $f_0$ and the numerical parameters $N_p$, $\Delta{t}$ of the algorithm are the same as in Section~\ref{sect:num-advdiff}. The final computational time we consider here is instead $T=4$.

Figure~\ref{fig:delta} displays clearly the consistency of the Monte Carlo solution with the exact solution~\eqref{eq:g.FP-advection} in the quasi-invariant limit $\epsilon\to 0^+$, as expected from Theorem~\ref{theo:q.i.-advection}. Notice, however, that the correspondence between the Monte Carlo reconstruction of $f_\epsilon(t)$ and the distribution function $g(t)$ becomes looser and looser as time increases. The ``physical'' reason for this is that whenever $\epsilon>0$ is non-infinitesimal the solution of the $\epsilon$-scaled Boltzmann-type equation is affected by diffusive, viz. \textit{anti-dissipative}, effects due to $\Var{p_\epsilon}>0$, cf.~\eqref{eq:peps-advection}, which, as observed at the end of Section~\ref{sect:q.i._advection}, cause ultimately the Dirac delta~\eqref{eq:ginf.FP-advection} \textit{not} to be a Maxwellian of~\eqref{eq:Boltztype.scaled-advection}. It is only in the limit $\epsilon\to 0^+$ that diffusive effects disappear, owing to $\Var{p_\epsilon}\propto\epsilon^{1+\delta}=o(\epsilon)$. The presence of diffusive effects in the numerical reconstruction of $f_\epsilon(t)$ is apparent in Figure~\ref{fig:delta} by observing that the support of the histogram tends to be invariably larger and less sharp than that of $g(t)$.

\subsubsection{Quasi-invariant regime with conserved energy}
The scaling discussed in Section~\ref{sect:q.i.-cons_en} is hard to explore numerically, because when $p_\epsilon,\,q_\epsilon$ are given by~\eqref{eq:peps.qeps-cons_en} and $\epsilon$ is small the point $0$ is an \textit{unstable} equilibrium of the equation for the mean value $M_{1,\epsilon}$ of $f_\epsilon$. Indeed, we recall that (cf. Section~\ref{sect:moments_evol})
$$ \frac{dM_{1,\epsilon}}{dt}=\left(\ave{p_\epsilon+q_\epsilon}-1\right)M_{1,\epsilon} $$
and, in view of~\eqref{eq:peps.qeps-cons_en}, we notice that
$$ \ave{p_\epsilon+q_\epsilon}=1+\sqrt{(2\lambda-\sigma^2)\epsilon}+o(\sqrt{\epsilon}), $$
i.e. $\ave{p_\epsilon+q_\epsilon}-1>0$ for $\epsilon$ small enough, which causes $\abs{M_{1,\epsilon}(t)}\to +\infty$ for $t\to +\infty$ if $M_{1,\epsilon}(0)\neq 0$. This implies in turn that not even the energy $M_{2,\epsilon}$ remains constant in time despite condition $\ave{p_\epsilon^2+q_\epsilon^2}=1$, indeed from Section~\ref{sect:moments_evol} we deduce
$$ \frac{dM_{2,\epsilon}}{dt}=2\ave{p_\epsilon q_\epsilon}M_{1,\epsilon}^2\neq 0. $$
Another way to see this issue is considering that the constraint $M_{2,\epsilon}\geq M_{1,\epsilon}^2$ forces $M_{2,\epsilon}(t)\to +\infty$ when $t\to +\infty$.

Consequently, any numerical error in the computation of $M_{1,\epsilon}$, such as e.g., the possible error in the initial value $M_{1,\epsilon}(0)$ produced by the sampling of $f_0$, propagates dramatically through the iterations of the Monte Carlo algorithm, driving the numerical solution of Algorithm~\ref{alg:nanbu.eps} far from the regime of constantly null mean value and constant energy at the basis of the results of Section~\ref{sect:q.i.-cons_en}.

For these reasons, we refrain from producing numerical simulations of the conserved energy regime of Section~\ref{sect:q.i.-cons_en}, which would require advanced techniques to stabilise the numerical computation of the mean value.

\section{Boltzmann-type kinetic equations on graphs}
\label{sect:graph}
In this section, we show that the tools and methods discussed so far can be fruitfully employed to address Boltzmann-type equations \textit{on graphs}, a recent extension of the standard kinetic approach conceived to model \textit{networked} multi-agent systems~\cite{loy2021KRM}.

The prototypical model inspiring this theoretical development consists in a finite collection of locations, represented by the vertices of the graph, populated by agents that, besides interacting within a given location, migrate from one location to another based on the available connections among the locations. The connections are described by the edges of the graph. Typically, the vertices represent instead spatial locations, such as e.g., cities, regions, countries depending on the spatial scale of interest. See~\cite{loy2021MBE} for a specific application to the transmission of an infectious disease. Ideally, in each vertex a Boltzmann-type kinetic equation describes the evolution of the statistical distribution of the state of the agents due to binary interactions taking place in that vertex. Nevertheless, since in each vertex agents also come and go following the migration process, the Boltzmann-type equation of a given vertex is coupled to those of the (topologically) adjacent vertices, which are either the origins or the destinations of the migrating agents. Hence, on the whole, a \textit{system} of Boltzmann-type equations is produced incorporating terms which account for \textit{mass transfers} among the equations.

\subsection{Derivation of the equations}
We consider a large system of interacting agents located in the vertices of a \textit{finite} graph, through which they can migrate. The agents are identified by the usual microscopic variable $v$, which changes in consequence of binary interactions, but also by the vertex of the graph where they are located. In particular, the agents are indistinguishable within every given vertex of the graph.

We represent the graph as a triplet $(\cI,\,\cE,\,\bA)$, where $\cI=\{1,\,\dots,\,N\}\subset\mathbb{N}$ is the set of vertices, which are $N\in\N$ in total, $\cE\subset\cI\times\cI$ is the set of edges, and $\bA=(A_{ij})_{i,\,j\in\cI}$ is the $N\times N$ real-valued matrix of non-negative weights assigned to each edge of the graph. By means of $\bA$ we define the further $N\times N$ matrix $\bP$, whose entries are
$$  P_{ij}:=\frac{A_{ij}}{\displaystyle\sum_{i=1}^NA_{ij}}\in [0,\,1], \qquad i,\,j\in\cI. $$
Clearly,
\begin{gather}
    0\leq P_{ij}\leq 1, \qquad \forall\,i,\,j \in \cI, \label{eq:graph.Pij_01} \\
    \sum_{i=1}^{N}P_{ij}=1, \qquad \forall\,j\in \cI, \label{eq:graph.Pij_sum}
\end{gather}
so that $\bP$ is a left stochastic matrix. It is called the \textit{transition matrix}, because it can be seen as a transition probability distribution defined on the vertices of the graph, which describes the probability to migrate from one vertex to another:
$$ P_{ij}=\Prob{j\to i}, $$
i.e. the entry $(i,\,j)$ of $\bP$ is the probability for an agent to migrate from vertex $j\in\cI$ to vertex $i\in\cI$. Notice that, with this description, the migration of the agents on the graph turns out to be a Markov-type jump process.

We shall consider \textit{strongly connected graphs}, meaning that there exists at least one directed path connecting any two vertices. Equivalently, we say that the matrix $\bP$ is \emph{irreducible}.

Agent dynamics on the graph can be described as follows. Each agent is identified by its microscopic state $V_t\in\R_+$ at time $t$ and also by the vertex $X_t\in\cI$ where it is located at the same time $t$. Similarly to Section~\ref{sect:stat_descr}, we describe the evolution of the state $(X_t,\,V_t)$ by means of discrete-in-time random processes with the following update rules:
\begin{subequations}
    \begin{align}
        V_{t+\Delta{t}} &= \left(1-\Theta_{X_t}\delta_{X_t,X^\ast_t}\right)V_t+\Theta_{X_t}\delta_{X_t,X^\ast_t}V'_t, \label{eq:graph.Vt} \\
        X_{t+\Delta{t}} &= (1-\Xi)X_t+\Xi J_t. \label{eq:graph.Xt}
    \end{align}
    \label{eq:graph.particle.gen}
\end{subequations}

Equation~\eqref{eq:graph.Vt} describes the update of the pre-interaction state $V_t$ produced by a symmetric binary interaction with another agent with pre-interaction state $(X_t^\ast,\,V_t^\ast)$. Similarly to~\eqref{eq:Vt+Dt}, $\Theta_{X_t}\in\{0,\,1\}$ is the Bernoulli random variable taking into account whether a binary interaction between two agents in $X_t$ occurs or not. We consider the case in which its law includes explicitly an interaction rate $\mu_{X_t}>0$, cf.~\eqref{eq:Bern.mu}, which here we assume to depend on the vertex $X_t$:
$$ \Theta_{X_t}\sim\operatorname{Bernoulli}(\mu_{X_t}\Delta{t}). $$
The Kronecker delta $\delta_{X_t,X^\ast_t}$, defined as
$$	\delta_{X_t,X^\ast_t}:=
		\begin{cases}
			1 & \text{if } X_t=X^\ast_t \\
			0 & \text{otherwise},
		\end{cases} $$
expresses the fact that only individuals within the same vertex can interact. The random variable $V'_t\in\R$ is the microscopic state after a binary interaction. In the same spirit as~\eqref{eq:linsymint}, it is defined as
$$ V'_t=p_{X_t}V_t+q_{X_t}V^\ast_t, $$
where also the mixing parameters $p_{X_t}$, $q_{X_t}$ can depend on the vertex $X_t$. 

In~\eqref{eq:graph.Xt}, $\Xi\in\{0,\,1\}$ is a second Bernoulli random variable, independent of $\Theta_{X_t}$, discriminating whether a vertex jump takes place ($\Xi=1$) or not ($\Xi=0$) in the time step $\Delta{t}$. Specifically, we let
$$ \Xi\sim\operatorname{Bernoulli}(\chi\Delta{t}), $$
$\chi>0$ being the rate of migration through the vertices. Clearly, we assume $\Delta{t}\leq\min\{\frac{1}{\chi},\,\frac{1}{\mu}\}$ for consistency. Furthermore, $J_t\in\cI$ is a random variable indicating the new vertex after a jump, with
$$ \Prob{J_t=i\vert X_t=j}=\Prob{j\to i}=P_{ij}, \qquad i,\,j\in\cI. $$ 

In order to derive the kinetic description of these microscopic dynamics, we introduce the kinetic distribution function on the graph, say $f=f(x,v,t)\geq 0$, with $x\in\cI$, $v\in\R_+$, $t\geq 0$, such that $f(x,v,t)\,dv$ gives the proportion of agents that at time $t$ are in the vertex $x$ with a microscopic state beloging to $[v,\,v+dv]$. Since the total number of agents in the graph is constant in time, $f$ can be normalised to a probability density:
\begin{equation}
    \sum_{x\in\cI}\int_\R f(x,v,t)\,dv=1, \qquad \forall\,t\geq 0.
    \label{eq:graph.f_int}
\end{equation}
Moreover, since the variable $x$ indicating the vertex of the graph is discrete, we can give $f$ the following form:
\begin{equation}
    f(x,v,t)=\sum_{i=1}^{N}f_i(v,t)\delta_{x,i},
    \label{def:graph.f}
\end{equation}
where $\delta_{x,i}$ is again a Kronecker delta while $f_i=f_i(v,t)\geq 0$ is the kinetic distribution function of the microscopic state $v$ in vertex $i$ at time $t$. Owing to~\eqref{eq:graph.f_int}, we deduce
\begin{equation}
    \sum_{i=1}^{N}\int_{\R}f_i(v,t)\,dv=1, \qquad \forall\,t\geq 0,
    \label{eq:graph.fi_sum_int}
\end{equation}
which highlights that the $f_i$'s are in general \textit{not} probability densities in $v$, as their integrals are in principle not unitary and possibly variable in time. We denote by
\begin{equation}
    \rho_i(t):=\int_{\R}f_i(v,t)\,dv
    \label{eq:graph.rhoi}
\end{equation}
the \textit{density} of the agents in vertex $i$ at time $t$. We shall sometimes call $\rho_i$ the mass carried by the distribution function $f_i$.

To obtain a Boltzmann-type equation for $f$, we apply a procedure analogous to that of Section~\ref{sect:stat_descr}. Specifically, we pick an observable quantity (test function) $\phi=\phi(x,v):\cI\times\R_+\to\R$ and carry out the same computations, considering the microscopic dynamics~\eqref{eq:graph.particle.gen} along with the definition of $f$ given in~\eqref{def:graph.f}. We obtain then:
\begin{align*}
	\frac{d}{dt}\sum_{i=1}^{N}\int_{\R}\phi(i,v)f_i(v,t)\,dv &= \chi\sum_{i=1}^{N}\sum_{j=1}^{N}\int_{\R}\phi(i,v)P_{ij}f_j(v,t)\,dv \\
	&\phantom{=} +\sum_{i=1}^{N}\mu_i\int_{\R}\int_{\R}\ave{\phi(i,v_i')}f_i(v,t)f_i(v_\ast,t)\,dv\,dv_\ast \\
	&\phantom{=} -\sum_{i=1}^{N}(\chi+\mu_i)\int_{\R}\phi(i,v)f_i(v,t)\,dv,
\end{align*}
which is required to hold for every $\phi$. Choosing $\phi(x,v)=\psi(x)\varphi(v)$ with $\psi$ such that $\psi(i)=1$ for a certain $i\in\cI$ and $\psi(x)=0$ for all $x\in\cI\setminus\{i\}$, 
we end up with the following system of Boltzmann-type equations for the $f_i$'s in weak form:
\begin{align}
    \begin{aligned}[b]
        \frac{d}{dt}\int_{\R}\varphi(v)f_i(v,t)\,dv &= \chi\int_{\R}\varphi(v)\left(\sum_{j=1}^{N}P_{ij}f_j(v,t)-f_i(v,t)\right)dv \\
        &\phantom{=} +\mu_i\int_{\R}\varphi(v)Q_i(f_i,f_i)(v,t)\,dv, \qquad i=1,\,\dots,\,N,
    \end{aligned}
    \label{eq:graph.Boltztype.weak}
\end{align}
where $Q_i$ denotes the collisional operator in the $i$th vertex of the graph, while the first term on the right-hand side is a \textit{jump operator} characterised by the Markovian transition matrix $\bP$. The analytical expression of $Q_i$ is analogous to that given in~\eqref{eq:Q.strong},~\eqref{eq:Q.weak} but for the fact that now the coefficients of the interaction rule implemented in $Q_i$ are in general vertex-dependent:
\begin{equation}
    v_i'=p_iv+q_iv_\ast,
    \label{eq:graph.intrules}
\end{equation}
with $p_i,\,q_i\geq 0$ random parameters. Due to this, the expression of the collisional operator changes formally from vertex to vertex. Notice that also the collision rates $\mu_i>0$ are assumed to be, in general, vertex-dependent.

\begin{remark}[Notation]
In the sequel, we shall write
$$ \bbf(v,t):=(f_1(v,t),\,f_2(v,t),\,\dots,\,f_N(v,t)) $$
for the vector-valued distribution function of the whole graph $\bbf:\R\times [0,\,+\infty)\to\R_+^N$. We shall say that $\bbf(\cdot,t)\in (L^r(\R))^N$ for some $r\in\N$ and $t>0$ if $f_i(\cdot,t)\in L^r(\R)$ for all $i=1,\,\dots,\,N$. In such a case, we shall denote by
$$ \norm{\bbf(t)}{(L^r)^N}:=\left(\sum_{i=1}^{N}\norm{f_i(t)}{L^r}^r\right)^{1/r} $$
the norm of $\bbf(\cdot,t)$ in $(L^r(\R))^N$. Moreover, we shall occasionally use the following quantities:
$$ \ubar{p}:=\min_{i=1,\,\dots,\,N}p_i, \qquad \ubar{q}:=\min_{i=1,\,\dots,\,N}q_i, \qquad
    \bar{\mu}:=\max_{i=1,\,\dots,\,N}\mu_i. $$
\end{remark}

\subsection{Evolution of mass and momentum in the vertices}
Letting $\varphi\equiv 1$ in~\eqref{eq:graph.Boltztype.weak}, we obtain a system of equations for the evolution in time of the densities $\{\rho_i\}_{i=1}^{N}$ of the agents in the vertices of the graph. Since intra-vertex interactions conserve the number of agents, i.e.
$$ \int_{\R}Q_i(f_i,f_i)(v,t)\,dv=0, \qquad \forall\,i=1,\,\dots,\,N, $$
we get
\begin{equation}
    \frac{d\rho_i}{dt}=\chi\left(\sum_{j=1}^{N}P_{ij}\rho_j-\rho_i\right).
    \label{eq:graph.system_rhoi}
\end{equation}
This is a linear system of ordinary differential equations that can be put in the vector-matrix form
$$ \frac{d\brho}{dt}=\chi(\bP-\bI)\brho, $$
where $\brho=\brho(t):=(\rho_1(t),\,\dots,\,\rho_N(t))$ and  $\bI$ denotes the $N\times N$ identity matrix. From the basic theory of linear ODE systems it is well known that there exists a unique solution issuing from a prescribed initial condition $\brho_0\in\R^N_+$, which can be written via the matrix exponential as $\brho(t)=e^{\chi(\bP-\bI)t}\brho_0$, $t>0$.

Owing to~\eqref{eq:graph.rhoi} and to the non-negativity of $f_i$, in principle it results $\rho_i(t)=\norm{f_i(t)}{L^1}$, therefore it is important to ensure that~\eqref{eq:graph.system_rhoi} admits only non-negative solutions at least when the initial condition $\brho_0$ is non-negative. This consistency result is provided by the following:
\begin{proposition} \label{prop:graph.rhoi_nonneg}
Let $\brho=\brho(t)$ be the unique solution to~\eqref{eq:graph.system_rhoi} issuing from an initial condition $\brho_0\in\R^N_+$. Then $\rho_i(t)\geq 0$ for all $t>0$ and all $i=1,\,\dots,\,N$. Moreover, if $\rho_{i,0}>0$ for some $i$ then $\rho_i(t)>0$ for all $t>0$.
\end{proposition}
\begin{proof}
Let $\rho_i^+:=\max\{\rho_i,\,0\}\geq 0$ and $\rho_i^-:=\max\{0,\,-\rho_i\}\geq 0$ be the positive and negative parts, respectively, of $\rho_i$. Writing $\rho_i=\rho_i^+-\rho_i^-$ and multiplying~\eqref{eq:graph.system_rhoi} by $\rho_i^-$ gives, after rearranging the terms,
$$ \frac{d}{dt}(\rho_i^-)^2+2\chi(\rho_i^-)^2=-2\chi\sum_{j=1}^{N}P_{ij}\rho_j\rho_i^-, $$
where we have used that $\rho_i^+\rho_i^-=0$ and $\frac{d\rho_i^+}{dt}\rho_i^-=0$. Writing also $\rho_j=\rho_j^+-\rho_j^-$ we find further
$$ \frac{d}{dt}(\rho_i^-)^2+2\chi(\rho_i^-)^2=-2\chi\sum_{j=1}^{N}P_{ij}\rho_j^+\rho_i^-+2\chi\sum_{j=1}^{N}P_{ij}\rho_j^-\rho_i^-
    \leq 2\chi\sum_{j=1}^{N}\rho_j^-\rho_i^-, $$
whence, summing both sides over $i=1,\,\dots,\,N$ and setting $u:=\sum_{i=1}^{N}(\rho_i^-)^2$ for brevity,
$$ \frac{du}{dt}+2\chi u\leq 2\chi\left(\sum_{i=1}^{N}\rho_i^-\right)^2\leq 2\chi Nu, $$
the last inequality following from the Cauchy-Schwartz inequality. Thus $u(t)\leq u(0)e^{2\chi(N-1)t}$, which in terms of the $\rho_i^-$'s entails
$$ \sum_{i=1}^{N}(\rho_i^-)^2(t)\leq e^{2\chi(N-1)t}\sum_{i=1}^{N}(\rho_{i,0}^-)^2=0 $$
because by assumption $\rho_{i,0}\geq 0$, hence $\rho_{i,0}^-=0$, for all $i=1,\,\dots,\,N$. Consequently, $\rho_i^-(t)=0$, i.e. $\rho_i(t)\geq 0$, for all $i=1,\,\dots,\,N$ and all $t>0$ and the first assertion is proved.

Using this, from~\eqref{eq:graph.system_rhoi} we deduce
$$ \frac{d\rho_i}{dt}\geq -\chi\rho_i, $$
whence $\rho_i(t)\geq\rho_{i,0}e^{-\chi t}>0$ for all $t>0$ whenever $\rho_{i,0}>0$, which proves the second assertion.
\end{proof}

A more detailed study of~\eqref{eq:graph.system_rhoi} reveals also two significant \textit{a priori} estimates:
\begin{theorem}[$L^1$ regularity] \label{theo:graph.L1_regularity}
Let~\eqref{eq:graph.Boltztype.weak} be complemented with an initial condition $\bbf_0\in (L^1(\R))^N$. Then $\bbf(\cdot,t)\in (L^1(\R))^N$ for all $t>0$ and
$$ \norm{\bbf(t)}{(L^1)^N}=\norm{\bbf_0}{(L^1)^N}, \qquad \forall\,t>0. $$

In particular, $\rho_i(t)\leq\norm{\bbf_0}{(L^1)^N}$ for all $t>0$.
\end{theorem}
\begin{proof}
Since $\norm{f_i(t)}{L^1}=\rho_i(t)$, summing both sides of~\eqref{eq:graph.system_rhoi} over $i=1,\,\dots,\,N$ while recalling~\eqref{eq:graph.Pij_sum} yields
$$ \frac{d}{dt}\sum_{i=1}^{N}\norm{f_i(t)}{L^1}=0, $$
whence the thesis.
\end{proof}

From Theorem~\ref{theo:graph.L1_regularity} we see that the modelling requirement~\eqref{eq:graph.fi_sum_int} is straightforwardly satisfied if the initial datum carries a unitary mass on the graph. Therefore, from now on we shall invariably consider only initial conditions such that
\begin{equation}
    \norm{\bbf_0}{(L^1)^N}=1.
    \label{eq:graph.f0_norm}
\end{equation}

\begin{proposition} \label{prop:graph.continuous_dep_L1}
Let $\bbf(\cdot,t),\,\bg(\cdot,t)\in (L^1(\R))^N$ be two solutions of~\eqref{eq:graph.Boltztype.weak} issuing from the respective initial conditions $\bbf_0,\,\bg_0\in (L^1(\R))^N$. Denote
$$ \rho_i^f(t)=\int_{\R}f_i(v,t)\,dv, \qquad \rho_i^g(t)=\int_{\R}g_i(v,t)\,dv. $$
Then
$$ \sum_{i=1}^{N}\abs{\rho_i^g(t)-\rho_i^f(t)}\leq\norm{\bg_0-\bbf_0}{(L^1)^N}, \qquad \forall\,t>0. $$
\end{proposition}
\begin{proof}
Let $u_i(t):=\rho_i^g(t)-\rho_i^f(t)$. Since $\rho_i^f$, $\rho_i^g$ satisfy~\eqref{eq:graph.system_rhoi}, by linearity $u_i$ satisfies in turn~\eqref{eq:graph.system_rhoi}. Thus, multiplying both sides of the equation by $e^{\chi t}$, integrating on $[0,\,t]$, $t>0$, and taking the absolute value we obtain
$$ e^{\chi t}\abs{u_i(t)}\leq\abs{u_{i,0}}+\chi\sum_{j=1}^{N}P_{ij}\int_0^t e^{\chi\tau}\abs{u_j(\tau)}\,d\tau. $$
Next, summing over $i=1,\,\dots,\,N$ and using~\eqref{eq:graph.Pij_sum} we deduce
$$ e^{\chi t}\sum_{i=1}^{N}\abs{u_i(t)}\leq\sum_{i=1}^{N}\abs{u_{i,0}}+\chi\int_0^t e^{\chi\tau}\sum_{j=1}^{N}\abs{u_j(\tau)}\,d\tau $$
so that Gr\"{o}nwall's inequality applied to the function $e^{\chi t}\sum_{i=1}^{N}\abs{u_i(t)}$ yields
$$ \sum_{i=1}^{N}\abs{u_i(t)}\leq\sum_{i=1}^{N}\abs{u_{i,0}}. $$
But
$$ \sum_{i=1}^{N}\abs{u_{i,0}}=\sum_{i=1}^{N}\abs*{\int_{\R}(g_{i,0}(v)-f_{i,0}(v))\,dv}
    \leq\sum_{i=1}^{N}\norm{g_{i,0}-f_{i,0}}{L^1}=\norm{\bg_0-\bbf_0}{(L^1)^N} $$
and we are done.
\end{proof}

The big picture about the evolution of the $\rho_i$'s is completed by their asymptotic trend for large times, which is provided by the following result from~\cite{loy2021MBE}:
\begin{theorem} \label{theo:graph.rhoi_equil}
There exists a unique constant-in-time density distribution $\brho^\infty=(\rho^\infty_1,\,\dots,\,\rho^\infty_N)$, with $\rho^\infty_i>0$ for all $i=1,\,\dots,\,N$ and $\sum_{i=1}^{N}\rho^\infty_i=1$, which is a stable and globally attractive equilibrium of~\eqref{eq:graph.system_rhoi}.
\end{theorem}
\begin{proof}
Constant-in-time solutions to~\eqref{eq:graph.system_rhoi} satisfy $(\bP-\bI)\brho^\infty=0$. Therefore, $\brho^\infty$ should be an eigenvector of $\bP$ corresponding to the eigenvalue $1$.
\begin{enumerate}
\item First, we show that $\brho^\infty$ exists and is unique. For this, we observe that since $\bP$ is a stochastic matrix, cf.~\eqref{eq:graph.Pij_sum}, it admits the eigenvalue $1$. This gives the existence of $\brho^\infty$. Moreover, since $\bP$ is irreducible because the graph is strongly connected, Perron-Frobenius theory, cf. e.g.,~\cite{minc1988BOOK}, implies that such an eigenvalue is simple and that there exists a corresponding eigenvector with strictly positive components. Therefore, we can assume $\rho^\infty_i>0$ for all $i=1,\,\dots,\,N$ and moreover we deduce that all the eigenvectors associated with the eigenvalue $1$ are parellel. The uniqueness of $\brho^\infty$ follows then from the constraint $\sum_{i=1}^{N}\rho^\infty_i=1$.
\item Second, we show that $\brho^\infty$ is a stable equilibrium of~\eqref{eq:graph.system_rhoi}. Still from Perron-Frobenius theory we have that $1$ is the maximum real part of the eigenvalues of $\bP$. Therefore, all the eigenvalues of $\bP-\bI$ have non-positive real part; moreover, owing to what we have established at the previous point, the eigenvalue with null real part is simple. This says that $\brho^\infty$ is stable.
\item Third, we show that $\brho^\infty$ is globally attractive. For this, we notice that the eigenvalue of $\bP-\bI$ with null real part is associated with the conservation of $\sum_{i=1}^{N}\rho_i(t)$ in time, cf. Theorem~\ref{theo:graph.L1_regularity}, which causes the trajectories of~\eqref{eq:graph.system_rhoi} to lie in an $N$-dimensional simplex. In such a simplex, $\brho^\infty$ is the unique equilibrium point and the system dynamics are ruled by the other eigenvalues of $\bP-\bI$ with negative real parts, which gives the global attractiveness. \qedhere
\end{enumerate}
\end{proof}

Similarly, defining the first statistical moment, viz. the mean, of the microscopic state $v$ in the vertex $i$ at time $t$ as
\begin{equation}
    M_{1,i}(t):=\frac{1}{\rho_i(t)}\int_{\R}vf_i(v,t)\,dv,
    \label{eq:graph.M1i}
\end{equation}
we obtain a system of evolution equations for the $M_{1,i}$'s by letting $\varphi(v)=v$ in~\eqref{eq:graph.Boltztype.weak}. Rearranging the terms conveniently, those equations can be given the form
\begin{equation}
    \frac{d}{dt}(\rho_iM_{1,i})=[\chi(P_{ii}-1)+\mu_i(\ave{p_i+q_i}-1)\rho_i]\rho_iM_{1,i}
        +\chi\sum_{j\neq i}P_{ij}\rho_jM_{1,j}.
    \label{eq:graph.momentum}
\end{equation}
Notice that each of them is more properly an equation for the quantity $\rho_iM_{1,i}$, which is called the \textit{momentum} in the vertex $i$. This terminology is borrowed from classical mechanics, when $v$ is the (signed) speed of the agents. As a matter of fact,~\eqref{eq:graph.momentum} can be easily converted into an equation for $M_{1,i}$ by developing the derivative on the left-hand side and invoking~\eqref{eq:graph.system_rhoi} to express $\frac{d\rho_i}{dt}$. We get:
\begin{equation}
    \frac{dM_{1,i}}{dt}=\mu_i\rho_i(\ave{p_i+q_i}-1)M_{1,i}+\chi\sum_{\substack{j=1,\,\dots,\,N \\ j\neq i}}P_{ij}\frac{\rho_j}{\rho_i}(M_{1,j}-M_{1,i}),
    \label{eq:graph.system_M1i}
\end{equation}
which is equivalent to~\eqref{eq:graph.momentum} for all $t>0$ such that $\rho_i(t)>0$. If $\rho_{i,0}>0$ then, owing to Proposition~\ref{prop:graph.rhoi_nonneg},~\eqref{eq:graph.system_M1i} is equivalent to~\eqref{eq:graph.momentum} for all $t>0$. \begin{remark}\label{rk:graph.Mcons}
From~\eqref{eq:graph.system_M1i} we see that even if $\ave{p_i+q_i}=1$ for all $i=1,\,\dots,\,N$ the first moment $M_{1,i}$ is not conserved in general in the vertices because of the migration mechanism (cf. the second term on the right-hand side of~\eqref{eq:graph.system_M1i}), that causes a variation of the $M_{1,i}$'s regardless of the interaction dynamics. However, if $\ave{p_i+q_i}=1$ for all $i=1,\,\dots,\,N$ then constant-in-time solutions are possible if e.g., $M_{1,i}$ is also constant with respect to $i$, i.e., if the first moments are the same in all vertices. In turn, this can be achieved by prescribing initial conditions with the same mean, say $\cpM_1$, in all vertices, for then, under the said assumptions, $M_{1,i}(t)=\cpM_1$ for all $t>0$ and all $i=1,\,\dots,\,N$ is a solution to the system above.
\end{remark}

\subsection{Continuous dependence and uniqueness estimates}
In this section, we discuss \textit{a priori} estimates developed in~\cite{bisoglio2024THESIS}, which are at the basis of the theory of well-posedness and trend to equilibrium of~\eqref{eq:graph.Boltztype.weak}. In so doing, we shall have an example of how the analytical tools and techniques introduced in the previous sections can be profitably revisited and adapted to tackle the qualitative analysis of more complicated kinetic equations.

A first difficulty intrinsic to~\eqref{eq:graph.Boltztype.weak} is that the $f_i$'s are in general not probability measures, for they have variable-in-time masses $\rho_i$. This prevents from a direct application to~\eqref{eq:graph.Boltztype.weak} of the Fourier metrics, which requires instead to work with probability measures. For this reason, part of the forthcoming estimates will rely on Lebesgue metrics, particularly the $L^2$ metric which interfaces well with the use of the Fourier transform to deal with~\eqref{eq:graph.Boltztype.weak}.

We begin by a result reminiscent of Proposition~\ref{prop:L2_regularity}, which ensures that working with the $L^2$ norm makes indeed sense also in the case of~\eqref{eq:graph.Boltztype.weak}:
\begin{proposition}[$L^2$ regularity] \label{prop:graph.L2_regularity}
Let~\eqref{eq:graph.Boltztype.weak} be complemented with an initial datum $\bbf_0\in (L^1(\R))^N\cap (L^2(\R))^N$ satisfying~\eqref{eq:graph.f0_norm}. Assume moreover that either $\ubar{p}$ or $\ubar{q}$ is non-zero. Then $\bbf(\cdot,t)\in (L^2(\R))^N$ for all $t>0$ and
\begin{equation}
    \norm{\bbf(t)}{(L^2)^N}\leq\norm{\bbf_0}{(L^2)^N}e^{\left(\chi(N-1)+\bar{\mu}\min\left\{\ave*{\frac{1}{\sqrt{\ubar{p}}}},\,
        \ave*{\frac{1}{\sqrt{\ubar{q}}}}\right\}\right)t},
        \qquad t>0.
    \label{eq:graph.L2_estimate}
\end{equation}
\end{proposition}
\begin{remark}
The ideas behind the proof of this Proposition are largely similar to those of Proposition~\ref{prop:Lr_regularity}. Here, however, we have to take extra care in handling the non-constant mass of the $f_i$'s and the inter-vertex variability of the rates $\mu_i$ and of the interaction coefficients $p_i$, $q_i$.
\end{remark}
\begin{proof}[Proof of Proposition~\ref{prop:graph.L2_regularity}]
Without loss of generality, throughout the proof we assume $\ubar{p}>0$.

Letting\footnote{A superposition of notation is produced between the index $i=1,\,\dots,\,N$ of the generic vertex of the graph and the imaginary unit $i\in\mathbb{C}$. Here and henceforth, we allow ourselves this little abuse as there is no possibility of confusion.} $\varphi(v)=e^{-i\xi v}$ in~\eqref{eq:graph.Boltztype.weak} and using the weak form~\eqref{eq:Q.weak} of the collisional operator we obtain the following Fourier-transformed system of kinetic equations:
\begin{equation}
    \partial_t\hat{f}_i+(\chi+\mu_i\rho_i)\hat{f}_i=\chi\sum_{j=1}^{N}P_{ij}\hat{f}_j+\mu_i\ave{\hat{f}_i(p_i\xi,t)\hat{f}_i(q_i\xi,t)},
    \label{eq:graph:Boltz.Fourier-transformed}
\end{equation}
where we stress in particular the appearance of the non-constant coefficient $\rho_i$ in the loss term of the collisional operator $Q_i$.

Let
$$ \pP_i(t):=\int_0^t(\mu_i\rho_i)(\tau)\,d\tau $$
be the primitive of $\mu_i\rho_i$ vanishing at $t=0$. Since both $\mu_i$ and $\rho_i$ are bounded and non-negative, $\pP_i$ is continuous and non-decreasing, therefore $0\leq\pP_i(t_1)\leq\pP_i(t_2)$ for every $0\leq t_1\leq t_2$.

Multiplying both sides of~\eqref{eq:graph:Boltz.Fourier-transformed} by $2e^{2(\chi t+\pP_i(t))}\hat{f}_i$ yields
\begin{equation}
    \partial_t\left(e^{2(\chi t+\pP_i(t))}\hat{f}_i^2\right)=2e^{2(\chi t+\pP_i(t))}\left(\chi\sum_{j=1}^{N}P_{ij}\hat{f}_j\hat{f}_i+
        \mu_i\ave{\hat{f}_i(p_i\xi,t)\hat{f}_i(q_i\xi,t)}\hat{f}_i\right)
    \label{eq:graph.theo.L2_regularity.proof}
\end{equation}
whence, using $\partial_t\abs{\cdot}\leq\abs{\partial_t(\cdot)}$, property~\eqref{eq:graph.Pij_01}, and moreover the fact that $\abs{\hat{f}_i(q_i\xi,t)}\leq\rho_i(t)\leq 1$ (cf. Theorem~\ref{theo:graph.L1_regularity}) and $\mu_i\leq\bar{\mu}$ for all $i=1,\,\dots,\,N$ and all $t>0$, we obtain
$$ \partial_t\left(e^{2(\chi t+\pP_i(t))}\abs{\hat{f}_i}^2\right)\leq 2e^{2(\chi t+\pP_i(t))}
    \left(\chi\sum_{j=1}^{N}\abs{\hat{f}_j}\cdot\abs{\hat{f}_i}+\bar{\mu}\ave{\abs{\hat{f}_i(p_i\xi,t)}}\abs{\hat{f}_i}\right). $$

Integrating in time on $[0,\,t]$, $t>0$, and multiplying then both sides by $e^{-2(\chi t+\pP_i(t))}$ produces
\begin{align*}
    \abs{\hat{f}_i}^2 &\leq \abs{\hat{f}_{i,0}}^2e^{-{2(\chi t+\pP_i(t))}} \\
    &\phantom{\leq} +2\int_0^t e^{2(\chi(\tau-t)+\pP_i(\tau)-\pP_i(t))}
            \left(\chi\sum_{j=1}^{N}\abs{\hat{f}_j}\cdot\abs{\hat{f}_i}+\bar{\mu}\ave{\abs{\hat{f}_i(p_i\xi,\tau)}}\abs{\hat{f}_i}\right)d\tau \\
    &\leq \abs{\hat{f}_{i,0}}^2e^{-{2\chi t}}
        +2\int_0^t e^{2\chi(\tau-t)}
            \left(\chi\sum_{j=1}^{N}\abs{\hat{f}_j}\cdot\abs{\hat{f}_i}+\bar{\mu}\ave{\abs{\hat{f}_i(p_i\xi,\tau)}}\abs{\hat{f}_i}\right)d\tau,
\end{align*}
because $-2\pP_i(t)\leq 0$ and likewise $\pP_i(\tau)-\pP_i(t)\leq 0$ for all $i=1,\,\dots,\,N$, all $t>0$, and all $\tau\leq t$. Next, summing over $i=1,\,\dots,\,N$ both sides and noticing that
$$ \sum_{i=1}^{N}\left(\sum_{j=1}^{N}\abs{\hat{f}_j}\cdot\abs{\hat{f}_i}\right)=\left(\sum_{i=1}^{N}\abs{\hat{f}_i}\right)^2
    \leq N\sum_{i=1}^{N}\abs{\hat{f}_i}^2 $$
owing to Cauchy-Schwartz inequality, we obtain
$$ \sum_{i=1}^{N}\abs{\hat{f}_i}^2\leq\sum_{i=1}^{N}\abs{\hat{f}_{i,0}}^2e^{-{2\chi t}}
    +2\int_0^t e^{2\chi(\tau-t)}
        \left(\chi N\sum_{j=1}^{N}\abs{\hat{f}_i}^2+\bar{\mu}\sum_{i=1}^{N}\ave{\abs{\hat{f}_i(p_i\xi,\tau)}}\abs{\hat{f}_i}\right)d\tau $$
whence, integrating with respect to $\xi\in\R$,
\begin{multline*}
    \norm{\bbfhat(t)}{(L^2)^N}^2\leq\norm{\bbfhat_0}{(L^2)^N}^2e^{-{2\chi t}} \\
        +2\int_0^t e^{2\chi(\tau-t)}\left(\chi N\norm{\bbfhat(\tau)}{(L^2)^N}^2
            +\bar{\mu}\sum_{i=1}^{N}\int_{\R}\ave{\abs{\hat{f}_i(p_i\xi,\tau)}}\abs{\hat{f}_i}\,d\xi\right)d\tau.
\end{multline*}
The $\xi$-integral on the right-hand side can be treated similarly to the analogous term in the proof of Proposition~\ref{prop:Lr_regularity} (with $r=2$), which produces:
$$ \sum_{i=1}^{N}\int_{\R}\ave{\abs{\hat{f}_i(p_i\xi,\tau)}}\abs{\hat{f}_i(\xi,\tau)}\,d\xi\leq
    \sum_{i=1}^{N}\ave*{\frac{1}{\sqrt{p_i}}}\norm{\hat{f}_i(\tau)}{L^2}^2\leq
        \ave*{\frac{1}{\sqrt{\ubar{p}}}}\norm{\bbfhat(\tau)}{(L^2)^N}^2 $$
and finally, pre-multiplying by $e^{2\chi t}$,
$$ e^{2\chi t}\norm{\bbfhat(t)}{(L^2)^N}^2\leq\norm{\bbfhat_0}{(L^2)^N}^2 \\
        +2\left(\chi N+\bar{\mu}\ave*{\frac{1}{\sqrt{\ubar{p}}}}\right)\int_0^t e^{2\chi\tau}\norm{\bbfhat(\tau)}{(L^2)^N}^2\,d\tau. $$
Gr\"{o}nwall's inequality applied to the function $e^{2\chi t}\norm{\bbfhat(t)}{(L^2)^N}^2$ yields then
$$ \norm{\bbfhat(t)}{(L^2)^N}^2\leq\norm{\bbfhat_0}{(L^2)^N}^2
    e^{2\left(\chi(N-1)+\bar{\mu}\ave*{\frac{1}{\sqrt{\ubar{p}}}}\right)t}, $$
whence taking the square root of both sides and invoking Parseval's identity gives
$$ \norm{\bbf(t)}{(L^2)^N}\leq\norm{\bbf_0}{(L^2)^N}e^{\left(\chi(N-1)+\bar{\mu}\ave*{\frac{1}{\sqrt{\ubar{p}}}}\right)t}, \qquad t>0. $$

If we assume instead $\ubar{q}>0$ then, estimating $\abs{\hat{f}_i(p_i\xi,t)}\leq\rho_i(t)\leq 1$ in~\eqref{eq:graph.theo.L2_regularity.proof} and proceeding with subsequent computations entirely analogous to those performed so far, we end up with the same upper bound on $\norm{\bbf(t)}{(L^2)^N}$ but with $\ave{\frac{1}{\sqrt{\ubar{p}}}}$ replaced by $\ave{\frac{1}{\sqrt{\ubar{q}}}}$:
$$ \norm{\bbf(t)}{(L^2)^N}\leq\norm{\bbf_0}{(L^2)^N}e^{\left(\chi(N-1)+\bar{\mu}\ave*{\frac{1}{\sqrt{\ubar{q}}}}\right)t}, \qquad t>0. $$
The thesis then follows from the simultaneous validity of these two estimates and the fact that the exponential term is monotonically increasing.
\end{proof}
The \textit{a priori} $L^2$ regularity of the solution to~\eqref{eq:graph.Boltztype.weak} asserted by Proposition~\ref{prop:graph.L2_regularity} is at the basis of a continuous dependence estimate in the same space, which implies the uniqueness of the solution as by-product:

\begin{theorem}[Continuous dependence and uniqueness] \label{theo:graph.cont_dep}
Let $\bbf(\cdot,t),\,\bg(\cdot,t)\in (L^2(\R))^N$ be two solutions of~\eqref{eq:graph.Boltztype.weak} issuing from the respective initial conditions $\bbf_0,\,\bg_0\in (L^1(\R))^N\cap (L^2(\R))^N$ satisfying~\eqref{eq:graph.f0_norm} and assume that $\min\{\ubar{p},\,\ubar{q}\}>0$. There exists a non-negative, non-decreasing function $\psi=\psi(t)$, with $\psi(0)=0$, such that
$$ \norm{\bg(t)-\bbf(t)}{(L^2)^N}^2\leq\Bigl(\norm{\bg_0-\bbf_0}{(L^2)^N}^2+\psi(t)\norm{\bg_0-\bbf_0}{(L^1)^N}\Bigr)
    e^{2\left[\chi(N-1)+\bar{\mu}\left(\ave*{\frac{1}{\sqrt{\ubar{p}}}+\frac{1}{\sqrt{\ubar{q}}}}+\frac{1}{2}\right)\right]t} $$
for all $t\geq 0$.

In particular,~\eqref{eq:graph.Boltztype.weak} admits at most one solution issuing from a given initial condition.
\end{theorem}
\begin{proof}
Throughout the proof, we shall use the symbols $\rho_i^f$, $\rho_i^g$ as introduced in the proof of Proposition~\ref{prop:graph.continuous_dep_L1}.

Let $h_i(v,t):=g_i(v,t)-f_i(v,t)$. Subtracting the corresponding terms of~\eqref{eq:graph:Boltz.Fourier-transformed} computed for $\hat{f}_i$ and $\hat{g}_i$, we get 
\begin{align*}
    \partial_t\hat{h}_i+(\chi+\mu_i\rho_i^g)\hat{h}_i &= \chi\sum_{j=1}^{N}P_{ij}\hat{h}_j
        +\mu_i\ave*{\hat{g}_i(p_i\xi,t)\hat{h}_i(q_i\xi,t)+\hat{f}_i(q_i\xi,t)\hat{h}_i(p_i\xi,t)} \\
    &\phantom{=} -\mu_i(\rho_i^g-\rho_i^f)\hat{f}_i, 
\end{align*}
whence, proceeding similarly to the proof of Proposition~\ref{prop:graph.L2_regularity},
\begin{align*}
    \abs{\hat{h}_i}^2 &\leq \abs{\hat{h}_{i,0}}^2e^{-2\chi t} \\
    &\phantom{\leq} +2\int_0^t e^{2\chi(\tau-t)}\left(\chi\sum_{j=1}^{N}\abs{\hat{h}_j}\cdot\abs{\hat{h}_i}
        +\mu_i\ave*{\abs{\hat{h}_i(p_i\xi,t)}+\abs{\hat{h}_i(q_i\xi,t)}}\abs{\hat{h}_i}\right)d\tau \\
    &\phantom{\leq} +2\mu_i\int_0^t e^{2\chi(\tau-t)}\abs{\rho_i^g-\rho_i^f}\cdot\abs{\hat{f}_i}\cdot\abs{\hat{h}_i}\,d\tau
\end{align*}
and further, summing both sides over $i$, integrating with respect to $\xi\in\R$, and pre-multiplying by $e^{2\chi t}$,
\begin{align}
    \begin{aligned}[b]
        e^{2\chi t}\norm{\bhhat(t)}{(L^2)^N}^2 &\leq \norm{\bhhat_0}{(L^2)^N}^2
            +2\left(\chi N+\bar{\mu}\ave*{\frac{1}{\sqrt{\ubar{p}}}+\frac{1}{\sqrt{\ubar{q}}}}\right)\int_0^t e^{2\chi\tau}\norm{\bhhat(\tau)}{(L^2)^N}\,d\tau \\
        &\phantom{\leq} +2\bar{\mu}\int_0^t e^{2\chi\tau}\sum_{i=1}^{N}\abs{\rho_i^g-\rho_i^f}\left(\int_{\R}\abs{\hat{f}_i}\cdot\abs{\hat{h}_i}\,d\xi\right)d\tau.
    \end{aligned}
    \label{eq:graph.partial_estimate.h}
\end{align}

The Cauchy-Schwartz inequality produces
\begin{align*}
    \int_{\R}\abs{\hat{f}_i(\xi,\tau)}\cdot\abs{\hat{h_i}(\xi,\tau)}\,d\xi &\leq \norm{\hat{f}_i(\tau)}{L^2}\norm{\hat{h}_i(\tau)}{L^2} \\
    &\leq \frac{1}{2}\left(\norm{\hat{f}_i(\tau)}{L^2}^2+\norm{\hat{h}_i(\tau)}{L^2}^2\right).
\intertext{Moreover, since $\norm{\hat{f}_i(\tau)}{L^2}^2=\norm{f_i(\tau)}{L^2}^2\leq\norm{\bbf(\tau)}{(L^2)^N}^2$, owing to~\eqref{eq:graph.L2_estimate} we deduce}
    &\leq \frac{1}{2}\left[\norm{\bbf_0}{(L^2)^N}^2e^{2\left(\chi(N-1)+\bar{\mu}\ave*{\frac{1}{\sqrt{\ubar{p}}}}\right)\tau}+\norm{\hat{h}_i(\tau)}{L^2}^2\right],
\end{align*}
thus we continue~\eqref{eq:graph.partial_estimate.h} as
\begin{align*}
    e^{2\chi t}\norm{\bhhat(t)}{(L^2)^N}^2 &\leq \norm{\bhhat_0}{(L^2)^N}^2
        +2\left(\chi N+\bar{\mu}\ave*{\frac{1}{\sqrt{\ubar{p}}}+\frac{1}{\sqrt{\ubar{q}}}}\right)\int_0^t e^{2\chi\tau}\norm{\bhhat(\tau)}{(L^2)^N}\,d\tau \\
    &\phantom{\leq} +\bar{\mu}\norm{\bbf_0}{(L^2)^N}^2\int_0^t e^{2\left(\chi N
        +\bar{\mu}\ave*{\frac{1}{\sqrt{\ubar{p}}}}\right)\tau}\sum_{i=1}^{N}\abs{\rho_i^g-\rho_i^f}\,d\tau \\
    &\phantom{\leq} +\bar{\mu}\int_0^t e^{2\chi \tau}\sum_{i=1}^{N}\abs{\rho_i^g-\rho_i^f}\cdot\norm{\hat{h}_i(\tau)}{L^2}^2\,d\tau.
\end{align*}

Let
$$ \psi(t):=\bar{\mu}\norm{\bbf_0}{(L^2)^N}^2\int_0^t e^{2\left(\chi N+\bar{\mu}\ave*{\frac{1}{\sqrt{\ubar{p}}}}\right)\tau}\,d\tau, $$
which is non-negative, non-decreasing and vanishes at $t=0$. We use it in the third term on the right-hand side of the previous inequality together with Proposition~\ref{prop:graph.continuous_dep_L1}, while in the fourth term we observe that, based on Proposition~\ref{prop:graph.rhoi_nonneg} and Theorem~\ref{theo:graph.L1_regularity}, it results $\abs{\rho_i^g-\rho_i^f}\leq 1$ for all $i=1,\,\dots,\,N$. This way, we discover
\begin{align*}
    e^{2\chi t}\norm{\bhhat(t)}{(L^2)^N}^2 &\leq \norm{\bhhat_0}{(L^2)^N}^2+\psi(t)\norm{\bh_0}{(L^1)^N} \\
    &\phantom{\leq} +2\left[\chi N+\bar{\mu}\left(\ave*{\frac{1}{\sqrt{\ubar{p}}}+\frac{1}{\sqrt{\ubar{q}}}}
        +\frac{1}{2}\right)\right]\int_0^t e^{2\chi\tau}\norm{\bhhat(\tau)}{(L^2)^N}^2\,d\tau,
\end{align*}
whence we get the claimed continuous dependence estimate by applying Gr\"{o}nwall's inequality to the function $e^{2\chi t}\norm{\bhhat(t)}{(L^2)^N}^2$ and appealing to Parseval's identity.

Finally, if $\bbf_0=\bg_0$ then $\bh_0=\mathbf{0}$. The continuous dependence estimate just proved yields then $\bh\equiv\mathbf{0}$, which implies the uniqueness of the solution.
\end{proof}

\begin{remark} \label{rem:graph.time_dependent_mui}
By carefully inspecting the proofs of Proposition~\ref{prop:graph.L2_regularity} and Theorem~\ref{theo:graph.cont_dep} we see that those results hold also in the more general case of \textit{time-dependent} $\mu_i$'s, provided the non-negativity and the boundedness of the latter are guaranteed.
\end{remark}

\subsection{Trend to equilibrium}
\label{sect:graph.trend_equil}
The continuous dependence estimate provided by Theorem~\ref{theo:graph.cont_dep} does not disclose any particular trend of the solutions to~\eqref{eq:graph.Boltztype.weak} for large times, because the exponential term contained in the estimate diverges invariably for $t\to +\infty$. Conversely, an estimate like that of Proposition~\ref{prop:trend.d2}, by indicating that any two solutions approach each other in the long run, would imply that there exists at most one equilibrium distribution, towards which every solution would converge asymptotically in time.

Proving a similar result on a graph requires first of all a suitable adaptation of the analytical tools at the basis of Proposition~\ref{prop:trend.d2}. In particular, it is necessary to handle the fact that, as already mentioned, Fourier metrics are conceived for probability measures, which the $f_i$'s in~\eqref{eq:graph.Boltztype.weak} are not due to their time-varying masses. To overcome this difficulty, we introduce the \textit{normalised kinetic distribution functions} $F_i=F_i(v,t)$ such that
$$ f_i(v,t)=\rho_i(t)F_i(v,t), $$
which are clearly non-negative and moreover satisfy
$$ \int_{\R}F_i(v,t)\,dv=1, \qquad \forall\,t\geq 0,\ \forall\,i=1,\,\dots,\,N $$
by construction. Briefly,
\begin{quote}
the $F_i$'s are probability measures expressing the statistical distribution of the microscopic state $v$ within the vertex $i$ of the graph regardless of the mass of agents populating that vertex.
\end{quote}
Notice that $F_i(\cdot,t)$ is technically indefinite if $\rho_i(t)=0$, for then also $f_i(\cdot,t)\equiv 0$.

We consider now two initial distributions $\bbf_0$, $\bg_0$ featuring the \textit{same mass distribution} on the graph, i.e. such that $\rho_{i,0}^f=\rho_{i,0}^g$ for all $i=1,\,\dots,\,N$. Owing to the uniqueness of the solution to~\eqref{eq:graph.system_rhoi}, we deduce that $\rho_i^f(t)=\rho_i^g(t)$ for all $t>0$ and all $i=1,\,\dots,\,N$, hence that the solutions $\bbf$, $\bg$ to~\eqref{eq:graph.Boltztype.weak} issuing from $\bbf_0$, $\bg_0$, respectively, have the same mass distribution on the graph at all times. Therefore, we can write
$$ f_i(v,t)=\rho_i(t)F_i(v,t), \qquad g_i(v,t)=\rho_i(t)G_i(v,t), $$
the coefficient $\rho_i(t)$ being the same in both expressions. On the other hand, the kinetic distribution functions $f_i$, $g_i$ differ in general from each other on the whole, because the statistical distribution of the microscopic state $v$ in the vertex $i$ at time $t$ can change depending on whether the initial condition is $\bbf_0$ or $\bg_0$. Within this perspective, it is clear that quantifying the distance between $\bbf$ and $\bg$ amounts to evaluating the distance between $\bF:=(F_1,\,\dots,\,F_N)$ and $\bG:=(G_1,\,\dots,\,G_N)$, whose components are probability distributions in $v$ for all $t\geq 0$.

For $s>0$, let us define
\begin{equation}
    D_s(\bbf(t),\bg(t)):=\sum_{i=1}^{N}\rho_i(t)d_s(F_i(t),G_i(t)),
    \label{eq:graph.Ds}
\end{equation}
where $d_s$ denotes the $s$-Fourier metric~\eqref{eq:ds}. Notice that $d_s(F_i(t),G_i(t))$ is well defined at least for $s\leq 1$, because the zeroth-order $v$-moments of $F_i$ and $G_i$ equal both $1$ by construction (cf. Proposition~\ref{prop:ds}\ref{prop:ds.finiteness}). Specific properties of the coefficients $p_i$, $q_i$ of the intra-vertex interaction rules~\eqref{eq:graph.intrules} can possibly allow for larger values of $s$. For the moment, we defer their detailed discussion and keep $s$ generic.

Invoking the properties of $d_s$ as a metric, it is not difficult to see that also the quantity $D_s$ in~\eqref{eq:graph.Ds} is a metric. Furthermore, using the $F_i$'s the Fourier-transformed version of~\eqref{eq:graph.Boltztype.weak} reads
\begin{equation}
    \partial_t(\rho_i\hat{F}_i)=\chi\left(\sum_{j=1}^{N}P_{ij}\rho_j\hat{F}_j-\rho_i\hat{F}_i\right)
        +\mu_i\rho_i^2\left(\ave{\hat{F}_i(p_i\xi,t)\hat{F}_i(q_i\xi,t)}-\hat{F}_i\right).
    \label{eq:graph.Fi_Fourier}
\end{equation}
Considering that $\abs{\hat{F}_i}\leq 1$ for all $i=1,\,\dots,\,N$ by construction, this equation puts in evidence that the effective rate of the inter-vertex jumps scales linearly with the density of the agents in the vertices while that of the intra-vertex interactions scales quadratically. This suggests that in densely populated vertices, where the density is close to $1$, the two effective rates are of the same order of magnitude, whereas in poorly populated vertices, where the density is close to $0$, intra-vertex interactions are much less frequent than inter-vertex jumps.

On the whole, the fact that the agents interact possibly less than they move might hinder the emergence of \textit{universal} statistical distributions of their microscopic state in the vertices, i.e. of equilibrium distributions independent of the initial conditions. For this reason, we consider henceforth rates $\mu_i$ of the form
\begin{equation}
    \mu_i=\frac{\mu}{\rho_i}, \qquad i=1,\,\dots,\,N,
    \label{eq:graph.mui}
\end{equation}
where $\mu>0$ is a constant, which make the effective intra-vertex interaction rates comparable to the effective inter-vertex jump rates in every density regime.

Assumption~\eqref{eq:graph.mui} implies that the $\mu_i$'s are time-dependent but certainly non-negative. Furthermore, if the initial density distribution is such that $\rho_{i,0}>0$ for all $i=1,\,\dots,\,N$ then, owing to Proposition~\ref{prop:graph.rhoi_nonneg} and Theorem~\ref{theo:graph.rhoi_equil}, the $\mu_i$'s in~\eqref{eq:graph.mui} are also bounded, because the $\rho_i$'s are uniformly bounded away from zero for all $t>0$ (the constant bounding the $\mu_i$'s from above depending in general on the initial condition $\brho_0$, as the latter affects the minimum values taken by the $\rho_i$'s over time). In conclusion, under reasonable assumptions the $\mu_i$'s in~\eqref{eq:graph.mui} are well defined. Furthermore, as asserted by Remark~\ref{rem:graph.time_dependent_mui}, they still guarantee the validity of the theory developed so far, in particular of Proposition~\ref{prop:graph.L2_regularity} and Theorem~\ref{theo:graph.cont_dep}.

Finally, concerning the choice of the index $s$ in~\eqref{eq:graph.Ds}, we can mimic the regime considered in Sections~\ref{sect:basic_theory}--\ref{sect:Fokker--Planck} by requiring the conservation of the mean state and the dissipation of the energy in each vertex. To this purpose, we observe that the $M_{1,i}$'s introduced in~\eqref{eq:graph.M1i} are precisely the first moments of the $F_i$'s. Under~\eqref{eq:graph.mui}, they satisfy the system of equations
$$ \frac{dM_{1,i}}{dt}
    =\mu(\ave{p_i+q_i}-1)M_{1,i}+\chi\sum_{\substack{j=1,\,\dots,\,N \\ j\neq i}}P_{ij}\frac{\rho_j}{\rho_i}(M_{1,j}-M_{1,i}). $$
We observe that Remark~\ref{rk:graph.Mcons} holds in this case as well, i.e., constant-in-time solutions exist if $\ave{p_i+q_i}=1$ for all $i=1,\,\dots,\,N$ and if the initial conditions have all the same mean, say $\cpM_1$.
Proceeding similarly, we discover from~\eqref{eq:graph.Boltztype.weak} with $\varphi(v)=v^2$, and taking~\eqref{eq:graph.mui} into account, that the energies
$$ M_{2,i}(t)=\int_{\R}v^2F_i(v,t)\,dv $$
in the vertices of the graph satisfy the system of equations
$$ \frac{dM_{2,i}}{dt}=\mu(\ave{p_i^2+q_i^2}-1)M_{2,i}+\chi\sum_{\substack{j=1,\,\dots,\,N \\ j\neq i}}P_{ij}\frac{\rho_j}{\rho_i}(M_{2,j}-M_{2,i})
    +2\mu\ave{p_iq_i}M_{1,i}^2, $$
which show that $\ave{p_i^2+q_i^2}<1$ for all $i=1,\,\dots,\,N$ is a condition to obtain that the intra-vertex interactions dissipate the energy.

On the whole, under the conditions discussed so far, we can let $s=2$ in~\eqref{eq:graph.Ds}. Therefore, after fixing
\begin{equation}
    \ave{p_i+q_i}=1, \qquad \ave{p_i^2+q_i^2}<1, \qquad \forall\,i=1,\,\dots,\,N
    \label{eq:graph.assumptions_pi.qi}
\end{equation}
we are in a position to prove:
\begin{theorem} \label{theo:graph.trend.D2}
Under~\eqref{eq:graph.assumptions_pi.qi}, let $\bbf(t),\,\bg(t)\in(\cP(\R))^N$ be two solutions to~\eqref{eq:graph.Boltztype.weak} issuing from two initial conditions $\bbf_0,\,\bg_0\in(\cP(\R))^N$ such that $\rho_{i,0}^f=\rho_{i,0}^g>0$ for every $i=1,\,\dots,\,N$. Assume moreover that $\bF_0$, $\bG_0$ are such that $M_{1,i}^F(0)=M_{1,i}^G(0)=\cpM_1$ for every $i=1,\,\dots,\,N$, being $\cpM_1\in\R$ a prescribed constant independent of $i$. Finally, let the intra-vertex interaction rates $\mu_i$ be given by~\eqref{eq:graph.mui}. Then
$$ D_2(\bbf(t),\bg(t))\leq D_2(\bbf_0,\bg_0)e^{\mu\left(\max\limits_{i=1,\,\dots,\,N}\ave{p_i^2+q_i^2}-1\right)t},
    \qquad \forall\,t>0. $$

In particular,
$$ \lim_{t\to +\infty}D_2(\bbf(t),\bg(t))=0. $$
\end{theorem}
\begin{proof}
The assumptions of the theorem ensure that $\bbf(t)$, $\bg(t)$ have the same mass for all $t>0$ and that the distances $d_2(F_i(t),G_i(t))$, $i=1,\,\dots,\,N$, and $D_2(\bbf(t),\bg(t))$ are well-defined.

Let $H_i(\xi,t):=\frac{\hat{G}_i(\xi,t)-\hat{F}_i(\xi,t)}{\abs{\xi}^2}$. From~\eqref{eq:graph.Fi_Fourier}, letting $\rho_i(t):=\rho_i^f(t)=\rho_i^g(t)$, it is not difficult to see that $H_i$ satisfies
$$ \partial_t(\rho_iH_i)+(\chi+\mu)\rho_iH_i=\chi\sum_{j=1}^{N}P_{ij}\rho_jH_j
    +\mu\ave{p_i^2\rho_iH_i(p_i\xi,t)\hat{G}_i(q_i\xi,t)+q_i^2\rho_iH_i(q_i\xi,t)\hat{F}_i(p_i\xi,t)}. $$
Multiplying both sides by $e^{(\chi+\mu)t}$ and taking the absolute value yields
$$ \partial_t\left(e^{(\chi+\mu)t}\rho_i\abs{H_i}\right)\leq e^{(\chi+\mu)t}\left(\chi\sum_{j=1}^{N}P_{ij}\rho_j\abs{H_j}
    +\mu\ave{p_i^2\rho_i\abs{H_i(p_i\xi,t)}+q_i^2\rho_i\abs{H_i(q_i\xi,t)}}\right), $$
where we have used that $\abs{\hat{F}_i(p_i\xi,t)},\,\abs{\hat{G}_i(q_i\xi,t)}\leq 1$. Since $\sup_{\xi\in\R\setminus\{0\}}\abs{H_i(\xi,t)}=d_2(F_i(t),G_i(t))$, on the right-hand side we can estimate $\abs{H_i}\leq d_2(F_i,G_i)$, which implies
\begin{align*}
    \partial_t\left(e^{(\chi+\mu)t}\rho_i\abs{H_i}\right) &\leq e^{(\chi+\mu)t}\left(\chi\sum_{j=1}^{N}P_{ij}\rho_jd_2(F_j,G_j)
        +\mu\ave{p_i^2+q_i^2}\rho_id_2(F_i,G_i)\right) \\
    &\leq e^{(\chi+\mu)t}\left[\chi\sum_{j=1}^{N}P_{ij}\rho_jd_2(F_j,G_j)
        +\mu\left(\max_{i=1,\,\dots,\,N}\ave{p_i^2+q_i^2}\right)\rho_id_2(F_i,G_i)\right].
\end{align*}
Summing now over $i=1,\,\dots,\,N$ while recalling~\eqref{eq:graph.Pij_sum} produces
$$ \partial_t\left(e^{(\chi+\mu)t}\sum_{i=1}^{N}\rho_i\abs{H_i}\right)\leq
    \left(\chi+\mu\max_{i=1,\,\dots,\,N}\ave{p_i^2+q_i^2}\right)e^{(\chi+\mu)t}D_2(\bbf,\bg); $$
integrating then both sides in time on $[0,\,t]$, $t>0$, gives
\begin{align*}
    e^{(\chi+\mu)t}D_2(\bbf(t),\bg(t)) &\leq D_2(\bbf_0,\bg_0) \\
    &\phantom{\leq} +\left(\chi+\mu\max_{i=1,\,\dots,\,N}\ave{p_i^2+q_i^2}\right)\int_0^t e^{(\chi+\mu)\tau}D_2(\bbf(\tau),\bg(\tau))\,d\tau,
\end{align*}
whence the thesis follows by applying Gr\"{o}nwall's inequality to $e^{(\chi+\mu)t}D_2(\bbf(t),\bg(t))$ and observing furthermore that $\max_{i=1,\,\dots,\,N}\ave{p_i^2+q_i^2}<1$ because of~\eqref{eq:graph.assumptions_pi.qi}.
\end{proof}

Theorem~\ref{theo:graph.trend.D2} implies that, given the trajectory $t\mapsto\brho(t)$ of~\eqref{eq:graph.system_rhoi} issuing from a prescribed initial condition $\brho_0$ and evolving towards the equilibrium $\brho^\infty$, there exists at most one Maxwellian $\bbf^\infty$ towards which every solution to~\eqref{eq:graph.Boltztype.weak} evolves in time.

We notice that despite the global asymptotic stability of $\brho^\infty$ asserted by Theorem~\ref{theo:graph.rhoi_equil}, which implies that every trajectory of~\eqref{eq:graph.system_rhoi} evolves towards $\brho^\infty$ regardless of the initial condition $\brho_0$, the result of Theorem~\ref{theo:graph.trend.D2} is linked to a specific mass trajectory $t\mapsto\brho(t)$. From the technical point of view, this is due to the definition itself of the metric $D_2$, cf.~\eqref{eq:graph.Ds}. The interpretation from the modelling point of view is that one should expect the emergence of a \textit{unique} equilibrium statistical distribution of the microscopic state $v$ in the vertices of the graph only for given dynamics of mass transfer across the vertices over time. In other words, the equilibrium density $\brho^\infty$ is not sufficient by itself to determine a unique $\bF^\infty$, because in principle the way in which the system gets to $\brho^\infty$ from a certain $\brho_0$ through successive exchanges of agents from vertex to vertex might also matter.

\medskip

The theoretical framework of Theorem~\ref{theo:graph.trend.D2}, particularly the use of the metric $D_2$, turns out to be essential to bring to light the asymptotic trends discussed so far. Alternatively, we could rely also in this case on the $L^2$ functional framework. In a setting analogous to that of Theorem~\ref{theo:graph.trend.D2} but without the requirement of identical first moments of $\bF_0$, $\bG_0$ and moreover with assumption~\eqref{eq:graph.assumptions_pi.qi} replaced by $\min\{\ubar{p},\,\ubar{q}\}>0$, we can improve the continuous dependence estimate provided by Theorem~\ref{theo:graph.cont_dep} taking advantage of $\rho_i^f(t)=\rho_i^g(t)$ for all $t\geq 0$ and of the normalisation~\eqref{eq:graph.mui}. The improved estimate reads:
$$ \norm{\bg(t)-\bbf(t)}{(L^2)^N}\leq\norm{\bg_0-\bbf_0}{(L^2)^N}e^{\left[\chi(N-1)+\mu\left(\ave*{\frac{1}{\sqrt{\ubar{p}}}+\frac{1}{\sqrt{\ubar{q}}}}-1\right)\right]t},
    \qquad \forall\,t>0, $$
whence we conclude on the existence of at most one attractive equilibrium distribution if the right-hand side decreases to zero in time, which happens if and only if
\begin{equation}
    \ave*{\frac{1}{\sqrt{\ubar{p}}}+\frac{1}{\sqrt{\ubar{q}}}}<1-\frac{\chi}{\mu}(N-1).
    \label{eq:graph.condition_decrease_trinorm_L2}
\end{equation}
Notice that, owing to the positiveness of $\ubar{p}$, $\ubar{q}$, this is possible if $\frac{\chi}{\mu}<\frac{1}{N-1}$, i.e. if the rate of intra-vertex interactions is sufficiently larger than that of the inter-vertex jumps. Conversely, the more general estimate of Theorem~\ref{theo:graph.cont_dep} cannot lead to an analogous conclusion for any choice of the parameters.

Similarly, we can improve the estimate of time propagation of the $L^2$-norm provided by Proposition~\ref{prop:graph.L2_regularity} as
$$ \norm{\bbf(t)}{(L^2)^N}\leq\norm{\bbf_0}{(L^2)^N}
    e^{\left[\chi(N-1)+\mu\left(\min\left\{\ave*{\frac{1}{\sqrt{\ubar{p}}}},\,\ave*{\frac{1}{\sqrt{\ubar{q}}}}\right\}-1\right)\right]t},
        \qquad \forall\,t>0, $$
which, since
$$ \min\left\{\ave*{\frac{1}{\sqrt{\ubar{p}}}},\,\ave*{\frac{1}{\sqrt{\ubar{q}}}}\right\}
    \leq\ave*{\frac{1}{\sqrt{\ubar{p}}}+\frac{1}{\sqrt{\ubar{q}}}}, $$
under condition~\eqref{eq:graph.condition_decrease_trinorm_L2} implies $\norm{\bbf(t)}{(L^2)^N}\to 0$ as $t\to +\infty$. 

In conclusion, the $L^2$-metric allows us to find mainly conditions for certain trivial universal trends to emerge, such as the decay to zero of the statistical distribution of $v$ in every vertex of the graph. This occurs when the coefficients of~\eqref{eq:graph.intrules} are so large that the agents disperse greatly their microscopic states in every intra-vertex interaction, as confirmed by condition~\eqref{eq:graph.condition_decrease_trinorm_L2}. Conditions~\eqref{eq:graph.assumptions_pi.qi} allow instead for physically more significant trends, as the analysis of the first two statistical moments of $\bF$ suggests. In view of Theorem~\ref{theo:graph.trend.D2}, the metric $D_2$ can reveal asymptotically such trends.

\subsection{Quasi-invariant limit and explicit equilibria on the graph}
In Section~\ref{sect:graph.trend_equil} we have investigated the trend to equilibrium of~\eqref{eq:graph.Boltztype.weak}, however without providing any closed form for equilibrium distributions on the graph. The reason is clearly that the system of Boltzmann-type equations~\eqref{eq:graph.Boltztype.weak} is not straightforwardly amenable to explicit analytical computations. To cope with this difficulty of the theory, here we adapt to~\eqref{eq:graph.Boltztype.weak} the quasi-invariant limit technique discussed in Section~\ref{sect:Fokker--Planck}. 

We refer again to the quasi-invariant regime of the interaction parameters defined by \eqref{eq:qeps}, \eqref{eq:peps.eta}, assuming that the coefficients $\lambda$, $\sigma$ are possibly vertex-dependent: $\lambda_i,\,\sigma_i>0$, $i\in\cI$, consistently with the inter-vertex variability of the interaction coefficients $p_i$, $q_i$. For the sake of clarity, we report here the scaled parameters:
\begin{equation}
    p_i^\epsilon=1-\epsilon\lambda_i+\sqrt{\epsilon}\sigma_i\eta_i,
        \qquad q_i^\epsilon=\epsilon\lambda_i,
    \label{eq:graph.q_i}
\end{equation}
where the $\eta_i$'s are independent random variables satisfying~\eqref{eq:eta} for all $i\in\cI$ and $\epsilon>0$ is the usual small scaling parameter. We analyse such a regime on the slow time scale $\epsilon t$. Then, in each vertex of the graph we define the time-scaled kinetic distribution function $f_i^\epsilon(v,t):=f_i(v,t/\epsilon)$, which carries the mass
$$ \rho_i^\epsilon(t):=\int_\R f_i^\epsilon(v,t)\,dv $$
and has mean
$$ M_{1,i}^\epsilon(t):=\frac{1}{\rho_i^\epsilon(t)}\int_\R vf_i^\epsilon(v,t)\,dv. $$
Parallelly, also the migration dynamics across the vertices have to be reformulated consistently with the quasi-invariant regime. To this aim, we scale the entries of the transition matrix $\bP$ as
\begin{equation}
    P_{ij}^\epsilon:=\epsilon P_{ij} \quad \text{if } i\neq j, \qquad
        P_{jj}^\epsilon:=1-\sum_{\substack{i=1,\,\dots,\,N \\ i\neq j}}P_{ij}^\epsilon,
    \label{eq:Peps}
\end{equation}
so that migrating to a different vertex has a small probability of order $\epsilon$ whereas staying in a given vertex has a probability close to $1$. Recalling~\eqref{eq:graph.Pij_sum}, we rewrite~\eqref{eq:graph.Boltztype.weak} as
\begin{align*}
    \frac{d}{dt}\int_\R\varphi(v)f_i(v,t)\,dv &= 
        \chi\sum_{\substack{j=1,\,\dots,\,N \\ j\neq i}}\int_\R\varphi(v)(P_{ij}f_j(v,t)-P_{ji}f_i(v,t))\,dv \\
    &\phantom{=} +\mu_i\int_\R\varphi(v)Q_i(f_i,f_i)(v,t)\,dv
\end{align*}
for $i\in\cI$, which in the quasi-invariant regime~\eqref{eq:graph.q_i},~\eqref{eq:Peps} on the time scale $\epsilon t$ becomes
\begin{align}
    \begin{aligned}[b]
        \frac{d}{dt}\int_\R\varphi(v)f_i^\epsilon(v,t)\,dv &= \frac{\chi}{\epsilon}\sum_{\substack{j=1,\,\dots,\,N \\ j\neq i}}
            \int_\R\varphi(v)\left(P^\epsilon_{ij}f^\epsilon_j(v,t)-P^\epsilon_{ji}f^\epsilon_i(v,t)\right)dv \\
        &\phantom{=} +\frac{\mu_i}{\epsilon}\int_\R\int_\R\ave{\varphi({v_i^\epsilon}')-\varphi(v)}
            f_i^\epsilon(v,t)f_i^\epsilon(v_\ast,t)\,dv\,dv_\ast \\
        &= \chi\sum_{\substack{j=1,\,\dots,\,N \\ j\neq i}}\int_\R\varphi(v)\left(P_{ij}f^\epsilon_j(v,t)-P_{ji}f^\epsilon_i(v,t)\right)dv \\
        &\phantom{=} +\frac{\mu_i}{\epsilon}\int_\R\int_\R\ave{\varphi({v_i^\epsilon}')-\varphi(v)}
            f_i^\epsilon(v,t)f_i^\epsilon(v_\ast,t)\,dv\,dv_\ast \\
        &= \chi\int_\R\varphi(v)\left(\sum_{j=1}^{N}P_{ij}f^\epsilon_j(v,t)-f^\epsilon_i(v,t)\right)dv \\
        &\phantom{=} +\frac{\mu_i}{\epsilon}\int_\R\int_\R\ave{\varphi({v_i^\epsilon}')-\varphi(v)}
            f_i^\epsilon(v,t)f_i^\epsilon(v_\ast,t)\,dv\,dv_\ast,
    \end{aligned}
    \label{eq:graph.Boltztype.scaled}
\end{align}
where ${v_i^\epsilon}'=p_i^\epsilon v+q_i^\epsilon v_\ast$. Letting $\varphi(v)=1,\,v$ in~\eqref{eq:graph.Boltztype.scaled}, we see that the time evolutions of $\rho_i^\epsilon$, $M_{1,i}^\epsilon$ are still ruled by~\eqref{eq:graph.system_rhoi},~\eqref{eq:graph.system_M1i}, thus they are the same as those on the original time scale $t$ for the unscaled migration and interaction dynamics. It follows that $\rho_i^\epsilon=\rho_i$ and $M_{1,i}^\epsilon=M_{1,i}$ for every $\epsilon>0$, being $\rho_i$, $M_{1,i}$ the density and mean of the unscaled distribution function $f_i$. In particular, from Remark~\ref{rk:graph.Mcons} we know that $M_{1,i}$ is not conserved in general.

Performing in~\eqref{eq:graph.Boltztype.scaled} computations analogous to those in~\eqref{eq:Boltztype.scaled} we obtain:
\begin{align}
    \begin{aligned}[b]
        \frac{d}{dt}\int_\R\varphi(v)f_i^\epsilon(v,t)\,dv &= \chi\int_\R\varphi(v)\left(\sum_{j=1}^{N}P_{ij}f^\epsilon_j(v,t)-
            f^\epsilon_i(v,t)\right)dv \\
	&\phantom{=} +\lambda_i\mu_i\rho_i(t)\int_\R\varphi'(v)(M_{1,i}(t)-v)f_i^\epsilon(v,t)\,dv \\
        &\phantom{=} +\frac{\mu_i\sigma_i^2}{2}\rho_i(t)\int_\R\varphi''(v)v^2f_i^\epsilon(v,t)\,dv \\
	&\phantom{=} +\frac{\epsilon\lambda_i^2\mu_i}{2}\int_\R\int_\R\varphi''(v)
            (v_\ast-v)^2f_i^\epsilon(v,t)f_i^\epsilon(v_\ast,t)\,dv\,dv_\ast \\
	&\phantom{=} +\frac{\mu_i}{6\epsilon}\int_\R\int_\R\ave*{\varphi'''(\bar{v}_\epsilon)
		\bigl((p_i^\epsilon-1)v+q_i^\epsilon v_\ast\bigr)^3}f_i^\epsilon(v,t)f_i^\epsilon(v_\ast,t)\,dv\,dv_\ast,
    \end{aligned}
    \label{eq:graph.Boltztype.scaled.bis}
\end{align}
where now, unlike~\eqref{eq:Boltztype.scaled}, the non-constant terms $\rho_i(t)$, $M_{1,i}(t)$ appear owing to the fact that the $f_i^\epsilon$'s carry the masses $\rho_i$ rather than constant unitary masses and that, as discussed before, their first statistical moments are not conserved in general. By estimating the last two terms on the right-hand side of~\eqref{eq:graph.Boltztype.scaled.bis} as done in Section~\ref{sect:formal_q-i_limit} with the remainder $R_\epsilon(t)$, we see that, in the limit $\epsilon\to 0^+$, $f_i^\epsilon$ approaches formally the solution $g_i$ of
\begin{align*}
    \frac{d}{dt}\int_\R\varphi(v)g_i(v,t)\,dv &= \chi\int_\R\varphi(v)\left(\sum_{j=1}^{N}P_{ij}g_j(v,t)-g_i(v,t)\right)dv \\
    &\phantom{=} +\lambda_i\mu_i\rho_i(t)\int_\R\varphi'(v)(M_{1,i}(t)-v)g_i(v,t)\,dv \\
    &\phantom{=} +\frac{\mu_i\sigma_i^2}{2}\rho_i(t)\int_\R\varphi''(v)v^2g_i(v,t)\,dv,
\end{align*}
for every $\varphi\in C^3(\R)$, where $\rho_i$, $M_{1,i}$ play also the role of mass and mean of $g_i$:
\begin{equation}
    \int_\R g_i(v,t)\,dv=\rho_i(t), \qquad \int_\R vg_i(v,t)\,dv=\rho_i(t)M_{1,i}(t).
    \label{eq:graph.gi_mass_mean}
\end{equation}
With $\varphi\in C^3_c(\R)$ we obtain, in particular, the strong form
\begin{equation}
    \frac{\partial g_i}{\partial t}+\lambda_i\mu_i\rho_i(t)\frac{\partial}{\partial v}\bigl((M_{1,i}(t)-v)g_i\bigr)=
        \frac{\mu_i\sigma_i^2}{2}\rho_i(t)\frac{\partial^2}{\partial v^2}(v^2g_i)
            +\chi\left(\sum_{j=1}^{N}P_{ij}g_j-g_i\right), \quad i\in\cI
    \label{eq:graph.FP.strong}
\end{equation}
to be coupled to~\eqref{eq:graph.system_rhoi},~\eqref{eq:graph.system_M1i} yielding the evolutions of the $\rho_i$'s and the $M_{1,i}$'s.

In practice, we have transformed the system~\eqref{eq:graph.Boltztype.weak} of $N$ Boltzmann-type equations on the graph in a system of $N$ Fokker--Planck equations on the graph coupled by the same jump operator. There are two main differences with respect to the quasi-invariant limit procedure leading to a standard Fokker--Planck equation such as~\eqref{eq:FP.strong}, both due to the presence of the migration dynamics. On one hand, the jump operator makes \eqref{eq:graph.FP.strong} a Fokker--Planck equation with \textit{reaction term}. On the other hand, the drift and diffusion terms in~\eqref{eq:graph.FP.strong} are analogous to the corresponding ones in~\eqref{eq:FP.strong} but feature \textit{non-constant} coefficients $\rho_i=\rho_i(t)$ and $M_{1,i}=M_{1,i}(t)$, i.e. the density and the mean state of the agents in the $i$-th vertex, which change in time because of the migrations across the vertices.

\paragraph{An explicitly solvable case}
Now we present a specific example, in which we take advantage of~\eqref{eq:graph.FP.strong} to determine explicitly the equilibrium distribution on the graph and the time evolution towards it in the quasi-invariant regime.

We consider a very simple graph made of $N=2$ vertices, hence $\cI=\{1,\,2\}$, assuming that only the agents in vertex $i=1$ can interact with interaction parameters~\eqref{eq:graph.assumptions_pi.qi}, whereas agents in vertex $i=2$ do not interact. In particular, we fix the interaction rates $\mu_1=1$, $\mu_2=0$. Moreover, we allow for migrations across the vertices according to the following transition matrix:
$$  \bP=
    \begin{pmatrix}
        1-\beta & 0 \\
        \beta & 1
    \end{pmatrix} $$
with $\beta\in (0,\,1)$; we also set the migration rate $\chi=1$. Notice that such a $\bP$ is \textit{not} irreducible, or equivalently the graph we are considering is \textit{not} strongly connected. Indeed, from vertex $i=2$ it is impossible to reach vertex $i=1$ because $P_{12}=0$. In other words, agents can migrate only from vertex $i=1$ to vertex $i=2$ but once they reach vertex $i=2$ they remain stuck there. Consequently, the results obtained in the previous sections under the assumption of strong connection of the graph do not hold, but the Fokker--Planck equation~\eqref{eq:graph.FP.strong} does and can be profitably used to investigate a case not covered by the general qualitative theory.

In particular, from~\eqref{eq:graph.FP.strong} we have, for $i=1$,
\begin{equation}
    \dfrac{\partial g_1}{\partial t}+\lambda_1\rho_1(t)\dfrac{\partial}{\partial v}\bigl((M_{1,1}(t)-v)g_1\bigr)
        =\dfrac{\sigma_1^2}{2}\rho_1(t)\dfrac{\partial^2}{\partial v^2}(v^2g_1)-\beta g_1
    \label{eq:graph.FP_g1}
\end{equation}
plus, from~\eqref{eq:graph.system_rhoi},~\eqref{eq:graph.system_M1i},
$$ \frac{d\rho_1}{dt}=-\beta\rho_1, \qquad \dfrac{dM_{1,1}}{dt}=0, $$
whence $\rho_1(t)=\rho_{1,0}e^{-\beta t}$ and $M_{1,1}(t)=M_{1,1}^0$ for all $t>0$, where $M_{1,1}^0:=M_{1,1}(0)$. In the sequel, we shall assume $M_{1,1}^0>0$ to fix the ideas. Since the mean value of $g_1$ is conserved in time, it is reasonable to look for a solution of the form
\begin{equation}
    g_1(v,t)=\frac{\rho_1(t)}{M_{1,1}^0}h\!\left(\frac{v}{M_{1,1}^0}\right),
    \label{eq:graph.g1_self-similar}
\end{equation}
where $h:\R\to\R_+$ satisfies the normalisation conditions
\begin{equation}
    \int_{\R}h(w)\,dw=1, \qquad \int_{\R}wh(w)\,dw=1.
    \label{eq:graph.h_normalisations}
\end{equation}
In practice, the distribution function~\eqref{eq:graph.g1_self-similar} is built from the fixed profile of $h$ modulated by the time-varying coefficient $\rho_1(t)$. For this reason,~\eqref{eq:graph.g1_self-similar} is called a \textit{self-similar solution} to~\eqref{eq:graph.FP_g1}. Conditions~\eqref{eq:graph.h_normalisations} are imposed to ensure that $g_1$ fulfils~\eqref{eq:graph.gi_mass_mean}.

Plugging~\eqref{eq:graph.g1_self-similar} into~\eqref{eq:graph.FP_g1}, we discover that $h$ satisfies the following stationary Fokker--Planck equation:
$$ \lambda_1\frac{\partial}{\partial w}((1-w)h)=\frac{\sigma_1^2}{2}\frac{\partial^2}{\partial w^2}(w^2h), $$
whose unique solution with unitary mass is (cf. Section~\ref{sect:FP_steady})
\begin{equation}
    h(w)=\frac{\left(\frac{2\lambda_1}{\sigma_1^2}\right)^{1+\frac{2\lambda_1}{\sigma_1^2}}}{\Gamma\!\left(1+\frac{2\lambda_1}{\sigma_1^2}\right)}
        \cdot\frac{e^{-\frac{2\lambda_1}{\sigma_1^2}\cdot\frac{1}{w}}}{w^{2\left(1+\frac{\lambda_1}{\sigma_1^2}\right)}}\chi(w>0).
    \label{eq:graph.h}
\end{equation}
Consequently, from~\eqref{eq:graph.g1_self-similar} we determine
$$ g_1(v,t)=\rho_{1,0}e^{-\beta t}
    \frac{\left(\frac{2\lambda_1}{\sigma_1^2}M_{1,1}^0\right)^{1+\frac{2\lambda_1}{\sigma_1^2}}}
        {\Gamma\!\left(1+\frac{2\lambda_1}{\sigma_1^2}\right)}
            \cdot\frac{e^{-\frac{2\lambda_1}{\sigma_1^2}\cdot\frac{M_{1,1}^0}{v}}}{v^{2\left(1+\frac{\lambda_1}{\sigma_1^2}\right)}}
                \chi(v>0), $$
which provides the exact evolution of the distribution function in vertex $i=1$ in the quasi-invariant regime. Notice that the ansatz~\eqref{eq:graph.g1_self-similar} forces $g_1(v,0)=\rho_{1,0}/M_{1,1}^0h(v/M_{1,1}^0)$. Therefore, the exact evolution just found corresponds to a specific choice of the family of initial distribution functions in vertex $i=1$, precisely a two-parameter ($\rho_{1,0}$, $M_{1,1}^0$) family of distributions with profile~\eqref{eq:graph.h}.

As far as the steady distribution is concerned, for $t\to +\infty$ we observe that $g_1(v,t)\to 0$ for all $v\in\R$ and also $\norm{g_1(\cdot,t)}{L^1}\to 0$, hence $g_1(\cdot, t)$ tends to $g_1^\infty\equiv 0$ both pointwise and in $L^1(\R)$. The physical meaning is clear: since only migrations from vertex $i=1$ to vertex $i=2$ are allowed with a constant-in-time probability, vertex $i=1$ shall empty in the long run.

From~\eqref{eq:graph.FP.strong} with $i=2$ get
$$ \frac{\partial g_2}{\partial t}=\beta g_1, $$
which is the Fokker--Planck equation in vertex $i=2$, where agents migrate from vertex $i=1$ and do not interact. Integrating in time we obtain
\begin{align*}
    g_2(v,t) &= g_{2,0}(v)+\frac{\rho_{1,0}(1-e^{-\beta t})}{M_{1,1}^0}h\!\left(\frac{v}{M_{1,1}^0}\right) \\
    &= g_{2,0}(v)+\rho_{1,0}(1-e^{-\beta t})\frac{\left(\frac{2\lambda_1}{\sigma_1^2}M_{1,1}^0\right)^{1+\frac{2\lambda_1}{\sigma_1^2}}}
        {\Gamma\!\left(1+\frac{2\lambda_1}{\sigma_1^2}\right)}
            \cdot\frac{e^{-\frac{2\lambda_1}{\sigma_1^2}\cdot\frac{M_{1,1}^0}{v}}}{v^{2\left(1+\frac{\lambda_1}{\sigma_1^2}\right)}}
                \chi(v>0),
\end{align*}
where $g_{2,0}=g_{2,0}(v)\geq 0$ with $\int_{\R}g_{2,0}(v)\,dv=1-\rho_{1,0}$ is the initial distribution function in vertex $i=2$. For $t\to +\infty$ we have that $g_2$ converges to $g_{2,0}(v)+\rho_{1,0}/M_{1,1}^0h(v/M_{1,1}^0)$ both pointwise and in $L^1(\R)$. Since $\rho_{1,0}/M_{1,1}^0h(v/M_{1,1}^0)$ is the initial profile of $g_1$, we see that in the long run the agent distribution of vertex $i=1$ is fully ``copied'' in vertex $i=2$. The rate of this copy is $\beta=\Prob{1\to 2}$.

\subsection{Further developments}
The vertices of the graph need not represent necessarily spatial locations but can group the agents according to certain \textit{structural features}, which partition the system in \textit{compartments}. Migrations across the vertices model then changes of compartment, which can be either spontaneous or triggered by the interactions and are typically correlated to a certain level of expression of the microscopic state of the agents. Within this interpretation, the microscopic state is frequently meant as a \textit{phenotypic trait} of the agents, which evolves according to interaction rules possibly different from vertex to vertex depending on the structural features of the agents in the various compartments. This formalism has been used to provide, for instance, consistent mathematical derivations and extensions of compartmental epidemiological models from stochastic agent-based models, stressing the role of the viral load (phenotypic trait) in the transmission of the infection and explaining the microscopic origin of aggregate parameters such as the basic reproduction number. See e.g.,~\cite{dellamarca2022NHM,dellamarca2023JMB}. A further generalisation of this approach has consisted in allowing for different microscopic states in different vertices of the graph, considering that a specific phenotypic trait can be more representative than others of the structural feature of the agents in a certain compartment. Still with reference to compartmental epidemiological models, this is the case of e.g., the viral load in the compartment of the infectious individuals and the resistance to infection in the compartment of the susceptible individuals, see~\cite{bernardi2025PREPRINT,lorenzi2024CMS}. In all these cases, the theory developed in the previous sections either applies straightforwardly or can be easily adapted.

A quite different, yet natural, way of understanding Boltzmann-type equations on graphs consists instead in using the graph as a descriptor of the links among the agents. In this case, each vertex of the graph coincides with a single agent while the edges determine which agents are directly linked and can therefore interact. One speaks then more properly of \textit{networked interactions} rather than of a networked multi-agent system. The mathematical formalisation of this idea results in a single kinetic equation incorporating the information about the links among the agents, usually by means of a proper interaction kernel. Specifically, the number of links becomes part of the microscopic state of an agent along with the variable describing the trait which changes in consequence of the interactions. This approach has been used to model social interactions~\cite{burger2021VJM}, in particular those leading to opinion formation on social networks~\cite{albi2024EJAM,burger2025SIADS,loy2022PTRSA,toscani2018PRE}, which are usually not all-to-all but adapt to the distribution of the contacts among the social network users. An interesting issue, which departs significantly from the theory developed in the previous sections, is the derivation of a statistical description of the distribution of contacts, to be embedded in the Boltzmann-type description of the interactions, out of the information encoded in the adjacency matrix of the graph in the limit of an infinite number of vertices, viz. agents. Some results have been obtained in~\cite{duering2024JNS,nurisso2024EJAM}.

\section*{Acknowledgements}
The authors are members of GNFM (Gruppo Nazionale per la Fisica Matematica) of INdAM (Istituto Nazionale di Alta Matematica), Italy.
	
\bibliographystyle{plain}
\bibliography{biblio}

\begin{thebibliography}{10}

\bibitem{albi2024EJAM}
G.~Albi, E.~Calzola, and G.~Dimarco.
\newblock A data-driven kinetic model for opinion dynamics with social network
  contacts.
\newblock {\em European J. Appl. Math.}, pages 1--27, 2024.

\bibitem{ambrosio2008BOOK}
L.~Ambrosio, N.~Gigli, and G.~Savar{\'e}.
\newblock {\em Gradient flows in metric spaces and in the space of probability
  measures}.
\newblock Lectures in Mathematics ETH Z\"urich. Birkh\"auser Verlag, Basel,
  2008.

\bibitem{auricchio2020RLMA}
G.~Auricchio, A.~Codegoni, S.~Gualandi, G.~Toscani, and M.~Veneroni.
\newblock The equivalence of {F}ourier-based and {W}asserstein metrics on
  imaging problems.
\newblock {\em Atti Accad. Naz. Lincei Rend. Lincei Mat. Appl.},
  31(3):627--649, 2020.

\bibitem{babovsky1986M2AS}
H.~Babovsky and H.~Neunzert.
\newblock On a simulation scheme for the {B}oltzmann equation.
\newblock {\em Math. Methods Appl. Sci.}, 8(1):223--233, 1986.

\bibitem{bernardi2025PREPRINT}
E.~Bernardi, T.~Lorenzi, M.~Sensi, and A.~Tosin.
\newblock Heterogeneously structured compartmental models of epidemiological
  systems: from individual-level processes to population-scale dynamics.
\newblock Preprint, 2025.

\bibitem{bhatnagar1954PR}
P.~L. Bhatnagar, E.~P. Gross, and M.~Krook.
\newblock A model for collision processes in gases. {I}. {S}mall amplitude
  processes in charged and neutral one-component systems.
\newblock {\em Phys. Rev.}, 94:511--525, 1954.

\bibitem{bird1970PF}
G.~A. Bird.
\newblock Direct simulation and the {B}oltzmann equation.
\newblock {\em Phys. Fluids}, 13(11):2676--2681, 1970.

\bibitem{bisi2024PHYSD}
M.~Bisi and N.~Loy.
\newblock Kinetic models for systems of interacting agents with multiple
  microscopic states.
\newblock {\em Phys. D}, 457:133967/1--23, 2024.

\bibitem{bisi2009CMS}
M.~Bisi, G.~Spiga, and G.~Toscani.
\newblock Kinetic models of conservative economies with wealth redistribution.
\newblock {\em Commun. Math. Sci.}, 7(4):901--916, 2009.

\bibitem{bisoglio2024THESIS}
D.~Bisoglio.
\newblock A priori estimates for {B}oltzmann-type equations on graphs.
\newblock Master's thesis, Politecnico di Torino, 2024.

\bibitem{bobylev1975DANSSSR}
A.~V. Bobylev.
\newblock Fourier transform method in the theory of the {B}oltzmann equation
  for {M}axwellian molecules.
\newblock {\em Dokl. Akad. Nauk SSSR}, 225(5):1041--1044, 1975.

\bibitem{boltzmann1970CHAPTER}
L.~Boltzmann.
\newblock Weitere {S}tudien \"{u}ber das {W}\"{a}rmegleichgewicht unter
  {G}asmolek\"{u}len.
\newblock In {\em Kinetische Theorie II. WTB Wissenschaftliche
  Taschenb\"{u}cher}. Vieweg+Teubner Verlag, Wiesbaden, 1970.

\bibitem{burger2021VJM}
M.~Burger.
\newblock Network structured kinetic models of social interactions.
\newblock {\em Vietnam J. Math.}, 49(3):937--956, 2021.

\bibitem{burger2025SIADS}
M.~Burger, N.~Loy, and A.~Rossi.
\newblock Asymptotic and stability analysis of kinetic models for opinion
  formation on networks: an {A}llen-{C}ahn approach.
\newblock {\em SIAM J. Appl. Dyn. Syst.}, 2025.
\newblock To appear.

\bibitem{canizo2011M3AS}
J.~A. Ca\~{n}izo, J.~A. Carrillo, and J.~Rosado.
\newblock A well-posedness theory in measures for some kinetic models of
  collective motion.
\newblock {\em Math. Models Methods Appl. Sci.}, 21(3):515--539, 2011.

\bibitem{carrillo2010SIMA}
J.~A. Carrillo, M.~Fornasier, J.~Rosado, and G.~Toscani.
\newblock Asymptotic flocking dynamics for the kinetic {C}ucker-{S}male model.
\newblock {\em SIAM J. Math. Anal.}, 42(1):218--236, 2010.

\bibitem{carrillo2010MSSET}
J.~A. Carrillo, M.~Fornasier, G.~Toscani, and F.~Vecil.
\newblock Particle, kinetic, and hydrodynamic models of swarming.
\newblock In G.~Naldi, L.~Pareschi, and G.~Toscani, editors, {\em Mathematical
  Modeling of Collective Behavior in Socio-Economic and Life Sciences},
  Modeling and Simulation in Science, Engineering and Technology, pages
  297--336. Birkh\"{a}user, Boston, 2010.

\bibitem{carrillo2007RMUP}
J.~A. Carrillo and G.~Toscani.
\newblock Contractive probability metrics and asymptotic behavior of
  dissipative kinetic equations.
\newblock {\em Riv. Mat. Univ. Parma}, 7(6):75--198, 2007.

\bibitem{cercignani1988BOOK}
C.~Cercignani.
\newblock {\em The {B}oltzmann equation and its applications}.
\newblock Springer, 1988.

\bibitem{cercignani1994BOOK}
C.~Cercignani, R.~Illner, and M.~Pulvirenti.
\newblock {\em The mathematical theory of dilute gases}, volume 106 of {\em
  Applied Mathematical Sciences}.
\newblock Springer, 1994.

\bibitem{cercignani1993CHAPTER}
C.~Cercignani and M.~Pulvirenti.
\newblock Nonequilibrium problems in many-particle systems. {A}n introduction.
\newblock In C.~Cercignani and M.~Pulvirenti, editors, {\em Nonequilibrium
  Problems in Many-Particle Systems}, volume 1551 of {\em Lecture Notes in
  Mathematics}, pages 71--305. Springer, Berlin, Heidelberg, 1993.

\bibitem{cordier2005JSP}
S.~Cordier, L.~Pareschi, and G.~Toscani.
\newblock On a kinetic model for a simple market economy.
\newblock {\em J. Stat. Phys.}, 120(1):253--277, 2005.

\bibitem{cucker2007TAC}
F.~Cucker and S.~Smale.
\newblock Emergent behavior in flocks.
\newblock {\em IEEE Trans. Automat. Control}, 52(5):852--862, 2007.

\bibitem{cucker2007JJM}
F.~Cucker and S.~Smale.
\newblock On the mathematics of emergence.
\newblock {\em Japan. J. Math.}, 2(1):197--227, 2007.

\bibitem{dellamarca2022NHM}
R.~Della~Marca, N.~Loy, and A.~Tosin.
\newblock An {SIR}-like kinetic model tracking individuals' viral load.
\newblock {\em Netw. Heterog. Media}, 17(3):467--494, 2022.

\bibitem{dellamarca2023JMB}
R.~Della~Marca, N.~Loy, and A.~Tosin.
\newblock An {SIR} model with viral load-dependent transmission.
\newblock {\em J. Math. Biol.}, 86(4):61/1--28, 2023.

\bibitem{pareschi1999AN}
G.~Dimarco and L.~Pareschi.
\newblock Numerical methods for kinetic equations.
\newblock {\em Acta Numer.}, 23:369--520, 2014.

\bibitem{dobrushin1979FAA}
R.~L. Dobrushin.
\newblock Vlasov equations.
\newblock {\em Funct. Anal. Appl.}, 13(2):115--123, 1979.

\bibitem{duering2024JNS}
B.~D\"{u}ring, J.~Franceschi, M.-T. Wolfram, and M.~Zanella.
\newblock Breaking consensus in kinetic opinion formation models on graphons.
\newblock {\em J. Nonlinear Sci.}, 34(79), 2024.

\bibitem{duering2009RMUP}
B.~D\"{u}ring, D.~Matthes, and G.~Toscani.
\newblock A {B}oltzmann-type approach to the formation of wealth distribution
  curves.
\newblock {\em Riv. Mat. Univ. Parma}, 8(1):199--261, 2009.

\bibitem{freguglia2017CMS}
P.~Freguglia and A.~Tosin.
\newblock Proposal of a risk model for vehicular traffic: {A} {B}oltzmann-type
  kinetic approach.
\newblock {\em Commun. Math. Sci.}, 15(1):213--236, 2017.

\bibitem{furioli2017M3AS}
G.~Furioli, A.~Pulvirenti, E.~Terraneo, and G.~Toscani.
\newblock {F}okker--{P}lanck equations in the modeling of socio-economic
  phenomena.
\newblock {\em Math. Models Methods Appl. Sci.}, 27(1):115--158, 2017.

\bibitem{gabetta1995JSP}
G.~Gabetta, G.~Toscani, and B.~Wennberg.
\newblock Metrics for probability distributions and the trend to equilibrium
  for solutions of the {B}oltzmann equation.
\newblock {\em J. Stat. Phys.}, 81(5--6):901--934, 1995.

\bibitem{gatignol1975BOOK}
R.~Gatignol.
\newblock {\em Th\'eorie cin\'etique des gaz \`a r\'epartition discr\`ete de
  vitesses}, volume~36 of {\em Lecture Notes in Physics}.
\newblock Springer-Verlag, Berlin, 1975.

\bibitem{kac1956CHAPTER}
M.~Kac.
\newblock Foundations of kinetic theory.
\newblock In J.~Neyman, editor, {\em Proceedings of the Third Berkeley
  Symposium on Mathematical Statistics and Probability}, volume III, pages
  173--200. University of California Press, 1956.

\bibitem{klar1997JSP}
A.~Klar and R.~Wegener.
\newblock Enskog-like kinetic models for vehicular traffic.
\newblock {\em J. Stat. Phys.}, 87(1--2):91--114, 1997.

\bibitem{lorenzi2024CMS}
T.~Lorenzi, E.~Paparelli, and A.~Tosin.
\newblock Modelling coevolutionary dynamics in heterogeneous {SI}
  epidemiological systems across scales.
\newblock {\em Commun. Math. Sci.}, 22(8):2131--2165, 2024.

\bibitem{loy2022PTRSA}
N.~Loy, M.~Raviola, and A.~Tosin.
\newblock Opinion polarization in social networks.
\newblock {\em Philos. Trans. Roy. Soc. A}, 380(2224):20210158/1--15, 2022.

\bibitem{loy2021KRM}
N.~Loy and A.~Tosin.
\newblock {B}oltzmann-type equations for multi-agent systems with label
  switching.
\newblock {\em Kinet. Relat. Models}, 14(5):867--894, 2021.

\bibitem{loy2021MBE}
N.~Loy and A.~Tosin.
\newblock A viral load-based model for epidemic spread on spatial networks.
\newblock {\em Math. Biosci. Eng.}, 18(5):5635--5663, 2021.

\bibitem{matthes2008JSP}
D.~Matthes and G.~Toscani.
\newblock On steady distributions of kinetic models of conservative economies.
\newblock {\em J. Stat. Phys.}, 130(6):1087--1117, 2008.

\bibitem{minc1988BOOK}
H.~Minc.
\newblock {\em Nonnegative matrices}.
\newblock Wiley-Interscience, 1988.

\bibitem{mitrinovic1991BOOK}
D.~S. Mitrinovi\'{c}, J.~E. Pe\v{c}ri\'{c}, and A.~M. Fink.
\newblock {\em {I}nequalities {I}nvolving {F}unctions and {T}heir {I}ntegrals
  and {D}erivatives}.
\newblock Springer, 1991.

\bibitem{nanbu1980JPSJ}
K.~Nanbu.
\newblock Direct simulation scheme derived from the {B}oltzmann equation. {I}.
  {M}onocomponent gases.
\newblock {\em J. Phys. Soc. Japan}, 49(5):2042--2049, 1980.

\bibitem{nurisso2024EJAM}
M.~Nurisso, M.~Raviola, and A.~Tosin.
\newblock Network-based kinetic models: Emergence of a statistical description
  of the graph topology.
\newblock {\em European J. Appl. Math.}, pages 1--22, 2024.

\bibitem{pareschi2001ESAIM}
L.~Pareschi and G.~Russo.
\newblock An introduction to {M}onte {C}arlo method for the {B}oltzmann
  equation.
\newblock {\em ESAIM: Proc.}, 10:35--75, 2001.

\bibitem{pareschi2006JSP}
L.~Pareschi and G.~Toscani.
\newblock Self-similarity and power-like tails in nonconservative kinetic
  models.
\newblock {\em J. Stat. Phys.}, 124(2--4):747--779, 2006.

\bibitem{pareschi2013BOOK}
L.~Pareschi and G.~Toscani.
\newblock {\em Interacting {M}ultiagent {S}ystems: {K}inetic equations and
  {M}onte {C}arlo methods}.
\newblock Oxford University Press, 2013.

\bibitem{paveri1975TR}
S.~L. Paveri-Fontana.
\newblock On {B}oltzmann-like treatments for traffic flow: a critical review of
  the basic model and an alternative proposal for dilute traffic analysis.
\newblock {\em Transportation Res.}, 9(4):225--235, 1975.

\bibitem{perthame2004BAMS}
B.~Perthame.
\newblock Mathematical tools for kinetic equations.
\newblock {\em Bull. Amer. Math. Soc. (N.S.)}, 41(2):205--244, 2004.

\bibitem{prigogine1960OR}
I.~Prigogine and F.~C. Andrews.
\newblock A {B}oltzmann-like approach for traffic flow.
\newblock {\em Operations Res.}, 8(6):789--797, 1960.

\bibitem{prigogine1971BOOK}
I.~Prigogine and R.~Herman.
\newblock {\em Kinetic theory of vehicular traffic}.
\newblock American Elsevier Publishing Co., New York, 1971.

\bibitem{puppo2019RMUP}
G.~Puppo.
\newblock Kinetic models of {BGK} type and their numerical integration.
\newblock {\em Riv. Mat. Univ. Parma}, 10(2):299--349, 2019.

\bibitem{saint-raymond2009BOOK}
L.~Saint-Raymond.
\newblock {\em {H}ydrodynamic limits of the {B}oltzmann equation}.
\newblock Springer, 2009.

\bibitem{spiga2004AML}
G.~Spiga and G.~Toscani.
\newblock The dissipative linear {B}oltzmann equation.
\newblock {\em Appl. Math. Lett.}, 17(3):295--301, 2004.

\bibitem{torregrossa2018KRM}
M.~Torregrossa and G.~Toscani.
\newblock On a {F}okker--{P}lanck equation for wealth distribution.
\newblock {\em Kinet. Relat. Models}, 11(2):337--355, 2018.

\bibitem{toscani1989CMP}
G.~Toscani.
\newblock On the {C}auchy problem for the discrete {B}oltzmann equation with
  initial values in {$L_{+}^{1}(\mathbb{R})$}.
\newblock {\em Commun. Math. Phys.}, 121(1):121--142, 1989.

\bibitem{toscani2006CMS}
G.~Toscani.
\newblock Kinetic models of opinion formation.
\newblock {\em Commun. Math. Sci.}, 4(3):481--496, 2006.

\bibitem{toscani2018PRE}
G.~Toscani, A.~Tosin, and M.~Zanella.
\newblock Opinion modeling on social media and marketing aspects.
\newblock {\em Phys. Rev. E}, 98(2):022315/1--15, 2018.

\bibitem{toscani1999JSP}
G.~Toscani and C.~Villani.
\newblock Probability metrics and uniqueness of solution to the {B}oltzmann
  equation for a {M}axwell gas.
\newblock {\em J. Stat. Phys.}, 94(3--4):619--637, 1999.

\bibitem{villani1998PhD}
C.~Villani.
\newblock {\em Contribution \`{a} l'\'{e}tude math\'{e}matique des
  \'{e}quations de Boltzmann et de Landau en th\'{e}orie cin\'{e}tique des gaz
  et des plasmas}.
\newblock Ph{D} thesis, Paris 9, 1998.

\bibitem{villani1998ARMA}
C.~Villani.
\newblock On a new class of weak solutions to the spatially homogeneous
  {B}oltzmann and {L}andau equations.
\newblock {\em Arch. Ration. Mech. Anal.}, 143(3):273--307, 1998.

\bibitem{villani2002CHAPTER}
C.~Villani.
\newblock A review of mathematical topics in collisional kinetic theory.
\newblock In S.~Friedlander and D.~Serre, editors, {\em Handbook of
  Mathematical Fluid Dynamics}, volume~I, chapter~2, pages 71--305. Elsevier,
  2002.

\bibitem{villani2009BOOK}
C.~Villani.
\newblock {\em Optimal transport -- {O}ld and new}.
\newblock Grundlehren der mathematischen Wissenschaften. Springer-Verlag,
  Berlin, 2009.

\bibitem{vlasov1945URMSU}
A.~A. Vlasov.
\newblock Theory of vibrational properties of electron gas and its
  applications.
\newblock {\em Uch. Rec. MSU}, 1945.

\end{thebibliography}
\end{document}